\documentclass{article}
\usepackage{my_style}

\begin{document}

\title{Classifying 2D topological phases: mapping ground states to string-nets}
\author[1]{Isaac H.~Kim}
\author[2]{Daniel Ranard}
\affil[1]{Department of Computer Science, University of California, Davis}
\affil[2]{Center for Theoretical Physics, Massachusetts Institute of Technology}

\maketitle
\begin{abstract}
            We prove the conjectured classification of topological phases in two spatial dimensions with gappable boundary, in a simplified setting.  Two gapped ground states of lattice Hamiltonians are in the same quantum phase of matter, or topological phase, if they can be connected by a constant-depth quantum circuit.  It is conjectured that the Levin-Wen string-net models exhaust all possible gapped phases with gappable boundary, and these phases are labeled by unitary modular tensor categories. We prove this under the assumption that every phase has a representative state with zero correlation length satisfying the entanglement bootstrap axioms, or a strict form of area law. Our main technical development is to transform these states into string-net states using constant-depth quantum circuits.  
\end{abstract}

\section{Introduction}
\label{sec:intro}

Ground states of quantum many-body systems exhibit diverse phenomena.  When two Hamiltonians exhibit fundamentally distinct behavior at zero temperature, we generally say these systems belong to distinct quantum phases of matter. A basic problem in condensed matter physics is to classify the quantum phases of gapped lattice Hamiltonians.

We focus on phases that are robust to perturbation by any local operator, i.e., genuine topological phases, rather than those protected by symmetry.  For the case of one spatial dimension, without imposing symmetries, there are no non-trivial topological phases~\cite{Chen2011,Schuch2011}. While conjectured classifications exist for two dimensions and higher, the justification in the literature usually relies on analogies to topological quantum field theory~\cite{Witten1998}, or assumptions about the form of fixed-point wavefunctions~\cite{levin2005string}.

More precisely, two gapped Hamiltonians are said to be in the same topological phase when one Hamiltonian can be deformed to the other without closing the gap, hence without inducing a phase transition.  A roughly equivalent characterization is that two ground states belong to the same phase when one ground state can be connected to the other by a constant-depth, quasi-local circuit~\cite{Hastings2005,Osborne2007,Bachmann2012}.  Thus gapped phases may be loosely defined as equivalence classes of ground states under constant-depth circuits \cite{chen2010local}. The equivalence class captures the long-range entanglement properties of the state, while disregarding local differences.  By reducing the classification of phases to a question of equivalence under circuits, it becomes subject to the tools of quantum information theory.

We focus on gapped 2D lattice systems of qudits with ``gappable boundary,'' or systems that can remain gapped when placed on a disk with boundary [Section \ref{sec:gapped-boundaries-conceptual}]. 
One famous example is the ``toric code,'' a 2D lattice Hamiltonian whose ground state exhibits long-range entanglement~\cite{kitaev2003fault}. The Levin-Wen Hamiltonians~\cite{levin2005string}, with their ``string-net'' ground states, furnish a zoo of further examples, with a different model associated to each unitary fusion category (UFC)~\cite{etingof2017fusion}. Do these constitute all possible 2D gapped phases with gappable boundary?  We answer yes, proving an exhaustive classification\footnote{Our goal is only to classify topological phases by proving the equivalence classes of ground states correspond to certain tensor categories, \textit{not} to then classify those tensor categories by enumerating them.  The classification of unitary fusion categories or unitary modular tensor categories themselves is believed to be an extremely difficult mathematical problem; it contains the famous classification of finite groups as a special case.} in a simplified setting.

One feature shared by all states in the same phase is the braiding and fusion properties of their anyonic excitations.  The anyons can be analyzed using the ground state alone, and they help identify the phase.  In fact, for gapped 2D qudit systems with gappable boundary, two ground states are conjectured to be in the same phase \textit{if and only if} their associated anyons behave identically.  The behavior of the anyons, or the ``anyon content,'' is described by a unitary modular tensor category (UMTC).  One direction of the above ``if and only if'' is  more conceptually straightforward: if two systems are in the same phase, then their anyons should behave isomorphically, because we can map anyon operators from one system to the other via the circuit connecting the ground states.  The other direction is less obvious: do the anyon contents provide a \textit{complete} set of invariants to determine the topological phase, or could there be some other unknown invariant?  For instance, if two systems both exhibit exactly four anyon types that braid and fuse like those of the toric code, is it obvious their ground states are connected by a constant-depth circuit, or could there be some obstruction?  

To make these questions more tractable, we work with gapped ground states satisfying the ``entanglement bootstrap axioms'' \cite{shi2020fusion,shi2021entanglement}. These are a class states satisfying certain conditions on their entanglement entropies. 
A sufficient condition is to assume  they satisfy a ``strict'' area law for the entanglement entropy $S(X)$ of region $X$, \cite{Kitaev2006,Levin2006}
\begin{align} \label{eq:strict-area}
S(X)=\alpha |\partial X| - \gamma
\end{align}
for constants $\alpha, \gamma$ independent of $X$.  Such states have zero correlation length (that is, their correlations are precisely zero beyond some radius), unlike generic gapped ground states, which are only expected to satisfy \eqref{eq:strict-area} up to some error that decreases for large regions. In that sense, the states we consider are ``fixed-point states'' that one obtains after sufficiently coarse-graining the lattice, and this simplification allows us to avoid arguments involving ``$\epsilon$'s and $\delta$'s.''  However, we emphasize that our assumptions only involve entanglement entropies, without invoking anyons, fusion categories, and so on.  We also do not require translation-invariance.

Loosely speaking, our main technical result states
\begin{mdframed}
\begin{theorem*}[Informal]\label{thm:mapping-informal}
     For any 2D state satisfying the entanglement bootstrap axioms (including the existence of a gapped boundary), there exists a constant-depth, geometrically local, unitary circuit mapping the state to a string-net, i.e., to the ground state of a Levin-Wen Hamiltonian.
\end{theorem*} 
\end{mdframed}
See Theorem \ref{thm:mapping-to-sn} for a more precise statement. The string-net we obtain lives on a hexagonal coarse-graining of the original lattice.  The original state need not have a gapped boundary present, but we assume that various disk-like subregions can be \textit{given} a gapped boundary.  That is, we assume there exists some pure state $|\psi_A\rangle$ on each disk-like region $A$ that matches the original ground state on the interior of $A$ and satisfies certain entropic axioms near the boundary $\partial A$, which then guarantees the existence of some gapped parent Hamiltonian for $|\psi_A\rangle$. 

The first step in mapping a ground state to a string-net state is to ask \textit{which} string-net to target, since there is a different string-net associated to every UFC.  It turns out an appropriate UFC is the one describing the anyons living on a gapped boundary of the state. Using techniques from \cite{shi2021entanglement}, we can define operations for creating and manipulating boundary anyons, and we extract the full UFC using ideas from \cite{kawagoe2020microscopic}. The bulk of our work then lies in constructing the circuit.

String-nets associated to distinct UFCs may actually occupy the same phase.  In particular, string-nets are in the same phase if and only if they have the same anyon content (labeled by a unitary modular fusion category, or UMTC), subject to caveats discussed in Section \ref{sec:setup-results}. We can then interpret Theorem \ref{thm:mapping-informal} as a classification of gapped phases with gappable boundary:
 \begin{mdframed}
\begin{corollary*}(Informal)
If every 2D gapped phase with gappable boundary has a representative satisfying the entanglement bootstrap axioms, then by Theorem \ref{thm:mapping-informal}, the Levin-Wen models exhaust all such topological phases, and these phases are labeled by UMTCs.
\end{corollary*}
\end{mdframed}
Of course, it may be difficult to prove every gapped phase with gappable boundary has a representative satisfying the entanglement bootstrap axioms. 
We leave this question to future work.

\section{Conceptual background}
\label{sec:background}

We review the notion of gapped systems, topological phases, and gapped boundaries.  These can be viewed from the perspective of both the Hamiltonian and the ground state.  Our technical assumptions will involve states, not Hamiltonians.  However, both perspectives are helpful to understand terminology and intuition.  This Section does not attempt to be rigorous, instead serving as background to our technical results.

\subsection{Gapped Hamiltonians}
For the purpose of this conceptual discussion, we consider translation-invariant local Hamiltonians on the lattice, so that we can consider the same Hamiltonian on systems of different size or topology. The notion of a gapped Hamiltonian formally applies to infinite systems, or families of Hamiltonians on systems of increasing size. We say
\begin{definition}[Gapped Hamiltonian] The gap of a Hamiltonian refers to the difference in energy between the ground state, or ground space, and the first excited state.  A family of Hamiltonians is called ``gapped'' when the gap is lower bounded by a positive constant, independent of system size. 
\end{definition}
Moreover, when a system is called gapped, the ground space degeneracy is usually understood to be independent of system size.

On the other hand, when the gap decreases with system size, the system is called gapless.  For instance, the toric code Hamiltonian is gapped, but under a sufficiently large perturbation, it might become gapless.  When this happens, the system may undergo a phase transition.

\subsection{Defining topological phases}
Given a gapped local Hamiltonian $H$, we can ask what properties of the system are preserved when perturbing the Hamiltonian.  Consider $H \mapsto H' = H + \epsilon \Delta H$, where $\Delta H$ is another local Hamiltonian (a sum of uniformly bounded terms on every site), and $\epsilon$ is a small but nonvanishing constant.  Then under mild conditions, $H'$ will remain gapped for all sufficiently small $\epsilon$.  

More generally, one considers a path through the space of Hamiltonians $H(s)$ for $s \in [0,1]$.  Imagine a physical system in the ground state, whose Hamiltonian is slowly (adiabatically) varied along the path $H(s)$, so that it remains in the ground state $|\psi(s)\rangle$. Perturbation theory suggests that as long as the system remains gapped, the long-range properties of the ground state are invariant, and the local properties vary smoothly. Therefore one adopts the following definition:
\begin{definition}[Phases of gapped Hamiltonians] \label{def:phase-Ham} Two gapped, local Hamiltonians $H_0$, $H_1$ are said to be in the same topological phase if and only if there exists a continuous path of gapped Hamiltonians $H(s)$ such that $H(0)=H_0$ and $H(1)=H_1$. 
\end{definition}
Thus topological phases are like connected components in the space of gapped local Hamiltonians.  The connected component including $H = \sum_i \sigma^z_i$, whose ground state is a product state, is called the ``trivial phase.''

We can ask precisely how the family of ground states $|\psi(s)\rangle$ of $H(s)$ varies with $s$.  It turns out this evolution is well-approximated by the action of some time-dependent local Hamiltonian $V(s)$, or $\partial_t |\psi(s)\rangle \approx V(s) |\psi(s)\rangle$~\cite{Osborne2007,Hastings2005,Bachmann2012}.  (Note this Hamiltonian $V(s)$ is not $H(s)$, though they are related.)  For two gapped Hamiltonians in the same phase, their ground states therefore are related by some local Hamiltonian evolution. In general, to make this statement exact, the terms of the local Hamiltonian $V(s)$ cannot be strictly local; instead, they are ``quasi-local,'' with spatially decaying tails. 

By Trotterizing the above time-evolution, we find the ground states of two Hamiltonians in the same phase are approximately related by a circuit of local unitaries.  Loosely speaking, we can use a circuit of constant depth (depth independent of system size), because the Hamiltonian evolution $V(s)$ runs for a constant time.  However, to obtain a good approximation, one requires a circuit of at least logarithmic depth, or alternatively a constant-depth but quasi-local circuit (composed of quasi-local unitary gates).

Therefore one can adopt an alternative definition of topological phase:
\begin{definition}[Phases of gapped ground states] \label{def:phase-state} Two states $|\psi_0\rangle$, $|\psi_1\rangle$ are said to be in the same topological phase if and only if there exists a constant-depth, quasi-local circuit $U$ such that $U|\psi_0\rangle = |\psi_1\rangle$.
\end{definition}
This definition is informal and does not fully specify ``quasi-local circuit.''
Note the definition does not refer directly to Hamiltonians, though we restrict our attention to the case that $|\psi_0\rangle$, $|\psi_1\rangle$ do have gapped parent Hamiltonians (i.e., they are the ground states of some gapped Hamiltonians).  Under the right notion of ``quasi-local circuit,'' Definitions \ref{def:phase-Ham} and \ref{def:phase-state} should be equivalent.  We have already discussed how a path of gapped Hamiltonian implies a circuit connecting the ground states, up to details involving quasi-locality.  To see the other direction, given a circuit $U_1$ connecting two ground states, consider a continuous path of circuits, $U(s)$, with $U(0)=I$ and $U(1)=U_1$.  (Imagine implementing each circuit gate continuously, layer by layer.)  Then if $|\psi_0\rangle$ has gapped parent Hamiltonian $H_0$, we obtain a path of gapped parent Hamiltonians $H(s) = U(s)^\dagger H_0 U(s)$. 

In general, one allows local ``unitary'' circuits that may both introduce and discard ancillas.  The ancillas are required to be in a pure state $|0\rangle$ both when introduced and discarded. Equivalently, the circuit may include local isometries, as well as projections that preserve the state. This effectively allows one to change the size of the local Hilbert space. Sometimes these operations are called generalized local unitaries~\cite{zeng2019quantum}.  When we discuss local unitary circuits, we actually refer to this more general notion. 

If one were further allowed to discard ancillas in a non-trivial state (entangled with each other), this produces a more coarse-grained equivalence; states which can be reached from the product state in this way are called invertible.  Here, we do not consider this more coarse-grained notion of topological phase.  

Apparently the precise definition of a topological phase in the sense of Definition \ref{def:phase-state} deserves further analysis. Regardless, we can already conclude:
\begin{proposition} \label{prop:same-phase}
    If two gapped ground states $|\psi_0\rangle$, $|\psi_1\rangle$ are connected by a constant-depth, strictly local circuit $U$, then they occupy the same topological phase in the sense of Definition \ref{def:phase-Ham}.
\end{proposition}
This notion will be the most relevant to our paper, because we will construct such circuits between states.

\subsection{Gapped boundaries}\label{sec:gapped-boundaries-conceptual}
We can also consider a system like the toric code on a manifold with boundary, like a finite disk.  Given the Hamiltonian $H$ on the plane, one might define the Hamiltonian $H_A$ on a disk-like sub-region $A$ by simply deleting terms outside the disk. In the case of the toric code, this $H_A$ generally has a ground state degeneracy that grows with the the size of the boundary,  $|\partial A|$.  Alternatively, we might define a new Hamiltonian 
\begin{align}
    H'_A = H_A + H_{\partial A}
\end{align}
for some choice of new or modified Hamiltonian terms $H_{\partial A}$ near the boundary $\partial A$.  This choice is called a boundary condition. For the toric code, various choices of boundary condition are well-understood, and they can lead to a Hamiltonian $H'_A$ that is gapped with non-degenerate ground state~\cite{kitaev2003fault,bravyi1998quantum}; see Figure~\ref{fig:toric_code}. 

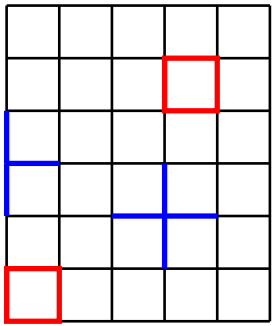
\begin{figure}[htbp]
    \centering
    \begin{tikzpicture}[line width=1pt, scale=0.7]
        \foreach \x in {0, ..., 5}
        {
        \draw[] (\x, 0) -- (\x, 6);
        }
        \foreach \y in {0, ..., 6}
        {
        \draw[] (0, \y) -- (5, \y);
        }
        \draw[red, line width=2pt] (0, 0) -- ++ (0, 1) -- ++ (1, 0) -- ++ (0, -1) -- cycle;
        \draw[blue, line width=2pt] (0,3) -- ++(1,0);
        \draw[blue, line width=2pt] (0,3) -- ++(0,1);
        \draw[blue, line width=2pt] (0,3) -- ++(0,-1);

        \draw[red, line width=2pt] (3,4) -- ++ (1,0) -- ++ (0,1) -- ++ (-1, 0) -- cycle;
        \draw[blue, line width=2pt] (3,2) -- ++ (1,0);
        \draw[blue, line width=2pt] (3,2) -- ++ (-1,0);
        \draw[blue, line width=2pt] (3,2) -- ++ (0,1);
        \draw[blue, line width=2pt] (3,2) -- ++ (0,-1);
    \end{tikzpicture}
    \caption{A toric code with a gapped boundary condition \cite{kitaev2012models}. The Hamiltonian consists of a tensor product of Pauli-$Z$s over edges surrounding a plaquette (red) and Pauli-$X$s over ``stars'' (blue), i.e., edges incident on a vertex. The boundary terms are incomplete stars.}
    \label{fig:toric_code}
\end{figure}

However, for other choices of boundary condition, $H'_A$ will be gapless.  That is, the gap may decrease with increasing $A$. In this case, often the low-energy eigenstates (those causing the system to be gapless) approximately match the reduced density matrix of the true ground state on the interior of $A$, with their excitation energy concentrated near the boundary.  These eigenstates constitute ``gapless edge modes.'' While $H'_A$ as a whole is gapless, often one emphasizes that the gaplessness is due to excitations living near the boundary. For this reason we say the system $H'_A$ has a ``gapless boundary,'' while in some sense it is still ``gapped in the bulk.'' For a more precise definition of bulk gap, see Ref.~\cite{nachtergaele2023stability}. 

Meanwhile, when $H_{\partial A}$ is chosen such that $H'_A$ has a unique, gapped ground state, then $H'_A$ is said to have a ``gapped boundary.''  Under the definitions offered above, it would be equivalent to simply say $H'_A$ is gapped.  However, when one says that a system with boundary is gapped, it may not be clear whether one means it is genuinely gapped, or simply gapped in the bulk.  Hence one often specifies ``gapped with gapped boundary.''

For systems like the toric code, where there does exist a choice of gapped boundary condition, we say the system has ``gappable boundary.''  Evidently this is a property of $H$, not of any particular boundary condition. Perhaps surprisingly, not all gapped systems have gappable boundary.  For instance, it is believed that the Kitaev honeycomb model~\cite{kitaev2006anyons} cannot be given boundary conditions such that it has a unique, gapped ground state on the disk.  Such a system has ``ungappable boundary''; see Ref.~\cite{Levin2013}  for a related discussion.

If a Hamiltonian has gappable boundary, one can argue that other Hamiltonians in the same topological phase also have gappable boundary. Thus the property of gappable or ungappable boundary is actually a property of the whole phase.  Chiral topological phases are expected to have protected gapless edge modes, hence ungappable boundary. By assuming gappable boundaries, this paper neglects phases with nonzero chiral central charge, and more generally those that have ungappable boundary~\cite{Levin2013}.  

It is conjectured that 2D topological phases are labeled by both a UMTC and a number $c_-$ called the chiral central charge.  When restricting attention to phases with gappable boundary, then $c_-=0$, and the topological phases are conjectured to be labeled only by the UMTC.

\subsection{Classifying Hamiltonians versus states}\label{sec:Hamiltonians_vs_states}
We already saw that the notion of a topological phase can be defined with an emphasis on either the space of gapped Hamiltonians (via Definition \ref{def:phase-Ham}) or gapped ground states (via Definition \ref{def:phase-state}).  

Can we discuss the space of gapped ground states, and specifically those with gappable boundary, without referring directly to parent Hamiltonians? 
The ground state of a gapped Hamiltonian has various special properties: it has exponentially decaying correlations~\cite{hastings2006spectral}, and its entanglement entropy is conjectured to satisfy an area law. Throughout the paper, we will work with states that satisfy certain strong versions of these properties: the entanglement bootstrap axioms~\cite{shi2020anyon,shi2020fusion,shi2021entanglement}. While these axioms do guarantee the state has a gapped parent Hamiltonian \cite{kim2024strict}, we will not generally make use of the Hamiltonian.

For a ground state $|\psi_A\rangle$ on a disk $A$, what does it mean for the state to have a gapped boundary?  Following Section \ref{sec:gapped-boundaries-conceptual}, roughly this should mean the state has a fully gapped parent Hamiltonian (without gapless edge modes).  Then $|\psi_A\rangle$ has exponentially decaying correlations, including for operators located near the boundary.  (In contrast, when gapless edge modes are present, one expects algebraically decaying correlations near the boundary.)  We therefore expect the area law for entanglement entropy to hold near the boundary as well.  When applying the heuristic area law [Eq.~\eqref{eq:strict-area}] for subregions $X \subset A$ overlapping the boundary of $A$, one does not count the boundary of $A$ in the calculation of the length $|\partial X|$.  That is, $|\partial X|$ is replaced with  $|\partial X \backslash \partial A|$.  This modification occurs because $|\psi_A\rangle$ is pure, so there are no local correlations between $X$ and its complement along the interval $|\partial X \cap \partial A|$.  These properties are codified in the entanglement bootstrap axioms for systems with gapped boundary~\cite{shi2021entanglement}, which we discuss in more detail in Section~\ref{sec:eb_review}. States satisfying these axioms do in fact have parent Hamiltonians with gapped boundary \cite{kim2024strict}, but again we need not make use of the Hamiltonian.

Finally, for some ground state $\rho$ on a system that may not have boundary, what does it mean to for $\rho$ to have gappable boundary?  In this paper, we mean that for every disk-like subregion $X$, there exists a state $|\psi_X\rangle$ with gapped boundary, such that $|\psi_X\rangle$ matches $\rho$ on the interior of $X$, i.e., $\rho_X := \Tr_{\bar{X}}{\rho} = \Tr_{\bar{X}} |\psi_X\rangle \langle \psi_X|.$  In this way, we can classify ground states with gappable boundary, without direct concern for the parent Hamiltonian.

\section{Prior work}
To what extent has it already been shown that 2D gapped phases with gappable boundary always have representatives given by string-net models?  The idea appears widely believed, and it was developed by Refs.~\cite{levin2005string, kitaev2012models, lin2014generalizations} among others, with an emphasis on gappable boundaries in Ref.~\cite{lin2014generalizations}.

To substantiate this idea, one approach is to simply define a topological phase as the algebraic data characterizing the anyons. Then one can concoct a string-net model with the same anyon data, concluding there is indeed a string-net in the same phase.  In contrast, here we use the term ``topological phase'' to refer to equivalence classes of ground states under constant-depth circuits, or of gapped Hamiltonians under continuous paths.  

Regardless of terminology, the point is that we want to study these equivalence classes, and their connection to anyon data is a priori unclear. First, while it is intuitive that states connected by circuits share the same anyon data, this can already be difficult to formalize: it requires extracting the anyon data from the state or Hamiltonian in a circuit-invariant way, i.e., defining rigorous invariants of the phase \cite{haah2016invariant, Kato2020entropic, ogata2021classification,ogata2023boundary}.  Meanwhile, the converse may seem less intuitive, and less progress has been made: why are states with the same anyon data connected by circuits?  For instance, how would one obtain the circuit?   The idea of ``mapping the anyons in one state to the anyons in the other state'' alone is too vague: the ground states are not excited, and they do not ``contain'' the anyons in a straightforward way, so it is unclear how to build a circuit that implements the mapping.

One variety of folk argument relies on analogies to topological quantum field theories (TQFTs).  While helpful, these analogies present a few issues.  First, when discussing abstract TQFTs, one does not generally consider microscopic Hamiltonians,  paths of gapped Hamiltonians, or circuits.   Therefore it is hard to build a concrete connection between (1) classifications of abstract TQFTs  and (2) the existence of quantum circuits, or continuous paths of gapped Hamiltonians.  Moreover, it appears TQFTs may simply fail to describe some quantum phases on the lattice, at least when the latter are understood as equivalence classes under paths of gapped Hamiltonians. For instance, in three dimensions, fracton phases appear to break this analogy, or at least question its predictive power.

Another variety of argument relies on intuitions from the renormalization group, viewed as a transformation or coarse-graining of the ground state. Levin and Wen \cite{levin2005string} introduced string-net wavefunctions partly as an ansatz for the fixed points of a renormalization group procedure.  However, the renormalization procedure itself was not concretely specified, so it was not clear how a general ground state might be mapped to a string-net.  There has been further work on renormalization group procedures involving string-nets \cite{gu2008tensor,konig2009exact,chen2010local}. However, it appears unclear whether these procedures can map a general class of ground states to string-nets.

Another obstacle to many of these renormalization or coarse-graining arguments is that they do not clearly invoke the assumption of a gappable boundary.  Yet this assumption should be crucial: there are phases with ungappable boundary, which are not captured by string-nets.  Meanwhile, the discussion in Ref.~\cite{swingle2016renormalization} does account for gapped boundaries, and their discussion of invertible states bears some relation to our ultimate strategy.

Finally, a rigorous classification of phases (in the sense of equivalence under constant-depth circuits) has been completed for a special class of 2D systems, namely stabilizer states with prime local dimension; see Refs.~\cite{haah2021classification}.  It turns out those phases are all equivalent to stacks of the toric code.

\section{Setup and results} \label{sec:setup-results}
We introduce the assumptions needed to state our results.  For more conceptual background, see Section \ref{sec:background}.  For more thorough exposition of entanglement bootstrap techniques, see Section \ref{sec:eb_review}.

We work with quantum states on a 2D spatial lattice.  We do not impose translation-invariance, and the microscopic details of the lattice will not be important; it may even be irregular.  Our results will apply to large but finite systems.\footnote{While some results could probably be formalized for infinite systems, we also view the nonnecessity the infinite-size limit as a benefit.} 

Consider a state $\sigma$ on a disk $X$ of arbitrary finite size.  Let $X^-$ denote the interior of $X$, consisting of all sites larger than some constant distance from the boundary. As a preview, after introducing some assumptions on $\sigma_X$, our main technical result states that there exists a constant-depth, geometrically local, unitary circuit $U$ such that 
\begin{align}
    \Tr_{X \backslash X^-}(U \sigma U^\dagger) = \Tr_{X \backslash X^-} (\rho_{SN})
\end{align}
where $\rho_{SN}$ is a string-net state, i.e., the ground state of a Levin-Wen Hamiltonian on $X$. The Levin-Wen models are reviewed in Section \ref{sec:sn-review}.
%we use ``Levin-Wen model'' and ``string-net'' somewhat interchangeably.  

In the assumptions below, we are not precise about the $O(1)$ constants involved when specifying the sizes of various subregions.  However, at the price of some tedium, we expect these assumptions can be made more precise, following discussion in \cite{shi2020anyon, kim2024strict}.

For input state $\sigma$ on disk $X$, first we impose conditions on all $O(1)$-size subregions that are separated from the boundary. Following~\cite{shi2020fusion}, we refer to these as the entanglement bootstrap axioms, \textbf{A0} and \textbf{A1}, defined in Figure~\ref{fig:eb_axioms_bulk}. 
\begin{definition}[Bulk entanglement bootstrap axioms] \label{def:bulk-EB}
    We say that a state $\sigma_X$ on disk $X$ satisfies the entanglement bootstrap axioms ``in the bulk'' when axioms \textbf{A0} and \textbf{A1}, defined in Figure~\ref{fig:eb_axioms_bulk}, hold for all $O(1)$-size regions that are disjoint from the boundary and topologically equivalent to those in Figure~\ref{fig:eb_axioms_bulk}.
\end{definition}
In particular, we say that $\sigma_X$ satisfies the entanglement bootstrap axioms in the bulk because we have not yet imposed assumptions on regions touching the boundary.

 Both axioms are phrased in terms of a certain linear combination of entanglement entropies.  They can also be interpreted as a mutual information (for \textbf{A0}) or conditional mutual information (for \textbf{A1}), by introducing a purifying system. For \textbf{A0}, let $A$ be the purifying system. Then the axiom is equivalent to $I(A:C)_{\sigma}=0$, where $I(A:C)_{\sigma} := \left(S(A) + S(C) - S(AC) \right)_{\sigma}$ is the mutual information. Because mutual information is a measure of correlation~\cite{Wolf2008}, this tells us that there is no correlation between $A$ and any subsystems that is non-adjacent to $A$. Similarly, for \textbf{A1}, let $A$ be the purifying system. Then the axiom reduces to $I(A:C|B)_{\sigma}=0$, where $I(A:C|B)_{\sigma} := \left(S(AB) + S(BC) - S(B) - S(ABC) \right)$ is the conditional mutual information.

These axioms will not be exactly satisfied for gapped ground states with nonzero correlation length.  However, generically, we expect them to be satisfied approximately after sufficient coarse-graining of the lattice.  (Then whenever two coarse-grained sites are non-adjacent, they are actually far apart.) Extending our techniques to the approximate case remains the subject of future work.\footnote{Actually, the axioms themselves are not circuit-invariant: axiom $\textbf{A1}$ may be violated after applying special constant-depth circuits to a state already satisfying $\textbf{A0}$ and $\textbf{A1}$.  See discussion in Section~\ref{sec:discussion}.}  

It will be important to distinguish between (1) the boundary of a subregion of the bulk, as an imagined delimiter, and (2) a genuine physical boundary, at the physical edge of a material.  We refer to the latter as the ``physical boundary'' or often just the boundary.  To codify the notion of a \textit{gapped} physical boundary, the entanglement bootstrap axioms are slightly altered for regions touching the boundary. For more conceptual background, see Section \ref{sec:gapped-boundaries-conceptual}.
We refer to the corresponding boundary axioms again as $\textbf{A0}$ and $\textbf{A1}$, and they are described in Figure~\ref{fig:eb_axioms_boundary}~\cite{shi2021entanglement}.
\begin{figure}[htbp]
    \centering
    \begin{tikzpicture}[scale=0.7]
    \filldraw[blue!30!white] (-4, -2.5) --++ (18, 0) -- ++ (0, 5) -- ++ (-18, 0) -- cycle;
        \draw[very thick, dashed] (0,0) circle (1cm);
        \draw[very thick, dashed] (0,0) circle (2cm);
        \node[] () at (0,0) {$C$};
        \node[] () at (0, 1.5) {$B$};
        \node[] () at (0, -3) {\textbf{A0}: $\left( S(BC) + S(C) - S(B)\right)_{\sigma}=0 $};

        \begin{scope}[xshift=10cm]
        \draw[very thick, dashed] (0,0) circle (1cm);
        \draw[very thick, dashed] (0,0) circle (2cm);
        \draw[very thick, dashed] (0,1) -- ++ (0,1);
        \draw[very thick, dashed] (0,-1) -- ++ (0, -1);
        \node[] () at (-1.5, 0) {$B$};
        \node[] () at (0, 0) {$C$};
        \node[] () at (1.5, 0) {$D$};
        \node[] () at (0, -3) {\textbf{A1}: $\left( S(BC) + S(CD) - S(B) - S(D) \right)_{\sigma}=0$};
        \end{scope}
    \end{tikzpicture}
    \caption{Axioms of the entanglement bootstrap in the bulk. The blue region represents the bulk. Note $\textbf{A0}$ is equivalent to $I(A:C) = 0$ for any purification of $\sigma_{BC}$ onto $ABC$ with auxiliary space $A$.  Likewise $\textbf{A1}$ is equivalent to $I(A:C|B)_\sigma = 0$ for any purification of $\sigma_{BC}$ onto $ABCD$.}
    \label{fig:eb_axioms_bulk}
\end{figure}
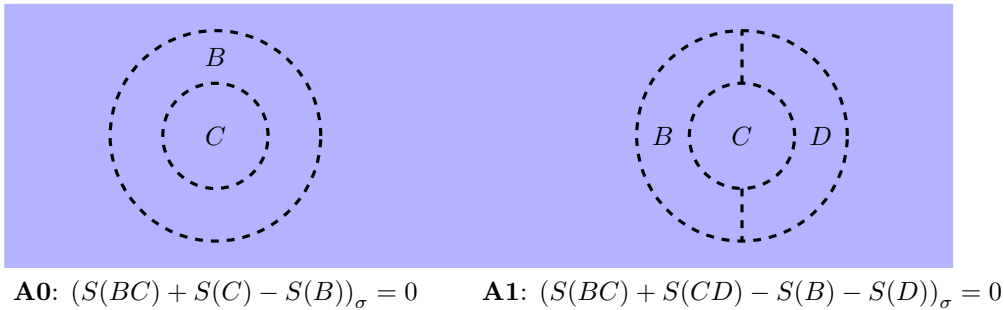
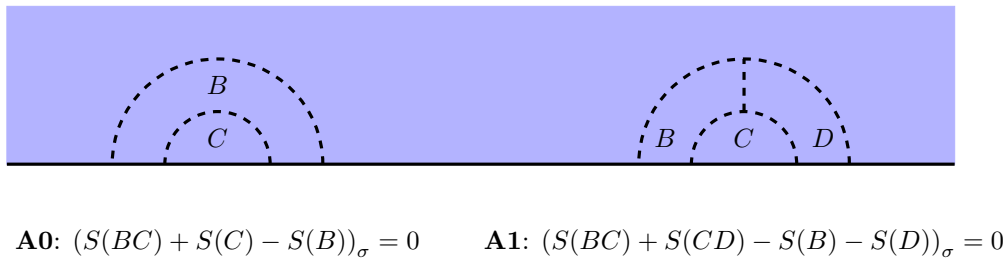
\begin{figure}[htbp]
    \centering
    \begin{tikzpicture}[scale=0.7]

    \filldraw[blue!30!white] (-4, 0) --++ (18, 0) -- ++ (0, 3) -- ++ (-18, 0) -- cycle;
    \draw[very thick] (-4, 0) -- ++ (18, 0);
    
    \begin{scope}
    \clip (-5,0) rectangle (5,5);
        \draw[very thick, dashed] (0,0) circle (1cm);
        \draw[very thick, dashed] (0,0) circle (2cm);
        \node[] () at (0,0.5) {$C$};
        \node[] () at (0, 1.5) {$B$};
    \end{scope}
        \node[] () at (0, -1.5) {\textbf{A0}: $\left( S(BC) + S(C) - S(B)\right)_{\sigma}=0$};

        \begin{scope}[xshift=10cm]
         \clip (-5,0) rectangle (5,5);
        \draw[very thick, dashed] (0,0) circle (1cm);
        \draw[very thick, dashed] (0,0) circle (2cm);
        \draw[very thick, dashed] (0,1) -- ++ (0,1);
        \draw[very thick, dashed] (0,-1) -- ++ (0, -1);
        \node[] () at (-1.5, 0.5) {$B$};
        \node[] () at (0, 0.5) {$C$};
        \node[] () at (1.5, 0.5) {$D$};
        \end{scope}
        \begin{scope}[xshift=10cm]
            \node[] () at (0, -1.5) {\textbf{A1}: $\left( S(BC) + S(CD) - S(B) - S(D) \right)_{\sigma}=0$};
        \end{scope}
    \end{tikzpicture}
    \caption{Axioms of the entanglement bootstrap in the boundary. The blue region represents the bulk and the solid line represents the physical boundary.}
    \label{fig:eb_axioms_boundary}
\end{figure}

\begin{definition}[Boundary entanglement bootstrap axioms] \label{def:boundary-EB}
    We say that a state $\sigma_Y$ on a disk $Y$ satisfies the entanglement bootstrap axioms on the boundary when 
    the axioms \textbf{A0} and \textbf{A1}, defined in Figure~\ref{fig:eb_axioms_bulk}, hold for all $O(1)$-size regions that overlap the boundary in a way topologically equivalent to those in Figure~\ref{fig:eb_axioms_boundary}.
\end{definition}

Note the bulk entanglement bootstrap axioms follow from the strict area law \eqref{eq:strict-area}. In fact, so do the boundary axioms, for the appropriate notion the area law. In particular, for subregions $A$ overlapping the physical boundary, the term $|\partial A|$ is modified to only count the length of $\partial A$ where it does not overlap the physical boundary $\partial X$, as motivated in Section \ref{sec:Hamiltonians_vs_states}.  Then the strict area law implies the boundary entanglement bootstrap axioms as well. 

So far, we have developed the notion of a state $\sigma$ on a disk $X$, satisfying the entanglement bootstraps axioms in the bulk and possibly the boundary.   Now we want to develop the notion that $\sigma$ has a \textit{gappable} boundary.  Under our terminology, in some sense this will be a stronger condition than $\sigma$ simply having a single gapped boundary at a fixed location. Actually, we will not require that $\sigma$ itself has any physical boundary.  Instead, we require that every $O(1)$-size contractible sub-region $A \subset X$ can be \textit{given} a gapped boundary.\footnote{For a slightly more parsimonious assumption, note our results will only require that \textit{certain} disks $O(1)$-size disks can be given gapped boundary. We explain this in Appendix~\ref{sec:elementary-fragments}.}

\begin{definition}[Gappable boundary]
\label{def:gappable-boundary-EB}
    We say a state $\sigma$ on a disk $X$ satisfies the entanglement bootstrap axioms with gappable boundary if (1) it satisfies the entanglement bootstrap axioms in the bulk (Definition \ref{def:bulk-EB}) and (2) for every $O(1)$-size contractible disk $A \subset X$, there exists a state $\sigma'_A$ on $A$ that satisfies the  boundary entanglement bootstrap axioms for $\partial A$ (Definition \ref{def:boundary-EB}) and also matches $\sigma$ on the interior, $\Tr_{A\backslash A^-} \sigma'_A = \Tr_{A\backslash A^-} \sigma$.  For overlapping disks $A$ and $B$, the associated states $\sigma'_A$ and $\sigma'_B$ are assumed to be consistent, in the sense of Figure~\ref{fig:fragment_assumption}. 
\end{definition}
We often call a state satisfying this assumption a ``reference state,'' almost always denoted by $\sigma$.  The states $\sigma'_A$ on disks $A$ with gappable boundary are often called ``fragments.''

Note that while string-net states do satisfy the above axioms, more generic examples will look nothing like string-nets. In fact, note we do not even require translation-invariance.  This may be surprising -- for instance, if the original state included a domain wall between two distinct phases, how could we map it to a single, homogeneous string-net with a constant-depth circuit?  It turns out that axiom $\textbf{A1}$ for the bulk prevents such domain walls and ensures homogeneity.\footnote{Because we do not assume translation-invariance, the assumption of a gappable boundary for multiple disks $A \subset X$ in Definition \ref{def:gappable-boundary-EB} is possibly non-trivial.  On the other hand, due to the homogeneity mentioned above already guaranteed by the bulk entanglement bootstrap axioms, it is possible Definition \ref{def:gappable-boundary-EB} could be weakened to only require that a gapped boundary exist for a single disk.}

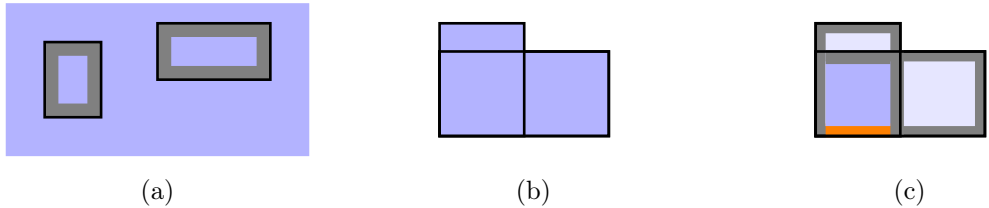
\begin{figure}
    \centering
    \begin{tikzpicture}[line width=1pt]
        \filldraw[blue!30!white] (-2cm, -1cm) -- ++ (4cm, 0) -- ++ (0, 2cm) -- ++ (-4cm, 0) -- cycle;
        \draw[fill=black!50!white] (-1.5cm, -0.5cm) -- ++ (0.75cm, 0) -- ++ (0, 1cm) -- ++(-0.75cm, 0) -- cycle;
        \filldraw[blue!30!white] (-1.3cm, -0.3cm) -- ++ (0.35cm, 0) -- ++ (0, 0.6cm) -- ++ (-0.35cm, 0) -- cycle;

        \draw[fill=black!50!white] (0, 0) -- ++ (1.5cm, 0) -- ++ (0, 0.75cm) -- ++ (-1.5cm, 0) -- cycle;
        \filldraw[blue!30!white] (0.2, 0.2) -- ++ (1.1cm, 0) -- ++ (0, 0.35cm) -- ++ (-1.1cm, 0) -- cycle;

        \node[] () at (0, -1.5cm) {(a)};

        \begin{scope}[xshift=5cm]
        \begin{scope}[xshift=1cm, scale=1.5]
        \draw[fill=blue!30!white] (-1.5cm, -0.5cm) -- ++ (1.5cm, 0) -- ++ (0, 0.75cm) -- ++ (-1.5cm, 0) -- cycle;
        \draw[fill=blue!30!white] (-1.5cm, -0.5cm) -- ++ (0.75cm, 0) -- ++ (0, 1cm) -- ++(-0.75cm, 0) -- cycle;
        \draw[] (-1.5cm, -0.5cm) -- ++ (1.5cm, 0) -- ++ (0, 0.75cm) -- ++ (-1.5cm, 0) -- cycle;
        \end{scope}
        \node[] () at (0, -1.5cm) {(b)};
        \end{scope}

        \begin{scope}[xshift=10cm]
        \begin{scope}[xshift=1cm, scale=1.5]
        \draw[fill=black!50!white] (-1.5cm, -0.5cm) -- ++ (1.5cm, 0) -- ++ (0, 0.75cm) -- ++ (-1.5cm, 0) -- cycle;
        \draw[fill=black!50!white] (-1.5cm, -0.5cm) -- ++ (0.75cm, 0) -- ++ (0, 1cm) -- ++(-0.75cm, 0) -- cycle;
        \filldraw[blue!10!white] (-1.4cm, -0.4cm) -- ++ (1.3cm, 0) -- ++ (0, 0.55cm) -- ++ (-1.3cm, 0) -- cycle;
        \filldraw[blue!10!white] (-1.4cm, -0.4cm) -- ++ (0.55cm, 0) -- ++ (0, 0.8cm) -- ++ (-0.55cm, 0) -- cycle;
         \filldraw[black!50!white] (-0.73cm, -0.5cm) -- ++ (-0.1cm, 0) -- ++ (0, 0.75cm) -- ++ (0.1cm, 0) -- cycle;

         \filldraw[orange] (-1.4cm, -0.5cm)  -- ++ (0.55cm, 0) -- ++ (0, 0.1cm) -- ++ (-0.55cm, 0) -- cycle;
         \filldraw[blue!30!white] (-1.4cm, -0.4cm) -- ++ (0.55cm, 0) -- ++ (0, 0.55cm) -- ++ (-0.55cm, 0) -- cycle;
          \filldraw[black!50!white] (-1.4cm, 0.15cm)  -- ++ (0.55cm, 0) -- ++ (0, 0.1cm) -- ++ (-0.55cm, 0) -- cycle;

        \draw[] (-1.5cm, -0.5cm) -- ++ (1.5cm, 0) -- ++ (0, 0.75cm) -- ++ (-1.5cm, 0) -- cycle;
        \draw[] (-1.5cm, -0.5cm) -- ++ (0.75cm, 0) -- ++ (0, 1cm) -- ++(-0.75cm, 0) -- cycle;
        
        \end{scope}
        \node[] () at (0, -1.5cm) {(c)};
        \end{scope}
    \end{tikzpicture}

    \caption{We assume the reference state has ``gappable boundary,'' including the existence of states on certain disk-like sub-regions (``fragments'') with gapped boundary.
    (\textbf{a}) Two examples of fragments outlined in black.  We assume there exist states on these fragments that match the reference state on their interiors (blue) and have gapped boundary (grey). \textbf{(b)} Two rectangular fragments on overlapping regions.  \textbf{(c)}  Fragments on overlapping regions are assumed to match on the interior of their overlap (blue) and also wherever they share a boundary (orange boundary). }
    \label{fig:fragment_assumption}
\end{figure}

Finally we are ready to state our main technical result.
\begin{theorem}[Mapping to string-nets]
\label{thm:mapping-to-sn}
    Assume a state $\sigma$ on a 2D disk $X$ satisfies the entanglement bootstrap axioms for a state with gappable boundary, as in Definition \ref{def:gappable-boundary-EB}. Then there exists a unitary fusion category $\mathcal{C}$ and geometrically local, unitary circuit $U$ of constant depth such that 
    \begin{align}
    \Tr_{X \backslash X^-}(U \sigma U^\dagger) = \Tr_{X \backslash X^-} (\rho_{SN})
\end{align}
where $\rho_{SN}$ is the canonical string-net state associated to category $\mathcal{C}$, i.e., the ground state of the associated Levin-Wen Hamiltonian.  Here $\rho_{SN}$ lives on a hexagonal lattice obtained from a constant-size coarse-graining of the original lattice.
\end{theorem}
The proof of Theorem \ref{thm:mapping-to-sn} appears in Section \ref{sec:proof_main}. Though we refer to the circuit as unitary, it may actually involve local isometries (to increase the local Hilbert space dimension) or local projections that preserve the state (to decrease the local Hilbert space dimension).

Given Theorem \ref{thm:mapping-to-sn} and Proposition \ref{prop:same-phase} regarding the notion of a topological phase, the following corollary is essentially a tautology.  We spell it out for emphasis.

\begin{corollary}[Informal]
    Suppose the conjecture that every 2D gapped phase with gappable boundary has a representative state satisfying the entanglement bootstrap axioms in the sense of Definition \ref{def:gappable-boundary-EB}.  Then the Levin-Wen Hamiltonians exhaust all possible 2D gapped phases with gappable boundary.
\end{corollary}
The conjecture entering the above corollary could only be made precise by specifying the exact meaning of ``gapped phase with gappable boundary,'' but the underlying ideas are addressed in Section \ref{sec:background}. The conclusion would be that when we restrict to gapped phases with gappable boundary, there are no more exotic phases to be found, beyond those given by the Levin-Wen models.

While Levin-Wen models exhaust all possible phases in our setting, note not all string-nets are actually in the same phase.  For a string-net specified by unitary fusion category $\mathcal{C}$, it is believed that the bulk anyon contents are labeled by a UMTC denoted $Z(\mathcal{C})$, called the center (or Drinfeld center) of $\mathcal{C}$ \cite{kitaev2012models}.  Two fusion categories $\mathcal{C}, \mathcal{D}$ can have isomorphic centers, $Z(\mathcal{C}) \cong Z(\mathcal{D})$, in which case they are called Morita-equivalent; we call the associated string-nets Morita-equivalent as well.  When two string-nets are Morita-equivalent, they are connected by a constant-depth circuit \cite{lootens2022mapping}, thus occupy the same phase.  Conversely, when two string-nets are are connected by a constant-depth circuit, they have the same bulk anyon content, and thus they are Morita-equivalent.  

To make this discussion rigorous, we need two missing ingredients: (1) a precise notion of the ``bulk anyon contents'' that is isomorphic for any two states connected by a constant-depth circuit, and (2) a proof that for string-nets built on fusion category $\mathcal{C}$, the bulk anyon contents are actually given by $Z(\mathcal{C})$. As it stands, we are able to provide (1) but not (2). The latter is posited in Conjecture \ref{conjecture:SN-bulk-anyons}.  Assuming this conjecture, then Theorem \ref{thm:mapping-to-sn} can ultimately be used to show:
\begin{corollary}
\label{cor:classification-main}
    Suppose Conjecture \ref{conjecture:SN-bulk-anyons} holds, regarding certain calculations using the Levin-Wen model.
    % which formalizes the claim that the bulk anyon contents of a string-net are given by the center of its fusion category.  
    Then two gapped ground states with gappable boundary satisfying the axioms of Definition \ref{def:gappable-boundary-EB} are connected by a constant-depth circuit if and only if they have the same bulk anyon contents (as defined in Theorem \ref{thm:bulk-anyons}). The corresponding topological phases are in bijection with doubled UMTCs, i.e.\ those of the form $Z(\mathcal{C})$ for UFC $\mathcal{C}$.
\end{corollary}
 The proof is given in Section \ref{sec:classification}, after developing the notion of bulk anyon contents in Theorem \ref{thm:bulk-anyons}.  If one could further show that every gapped phase with gappable boundary had a representative satisfying the entanglement bootstrap axioms, then the classification project for these phases would be complete.  

As an aside, we also consider explicitly ``doubled'' states, obtained by stacking two copies of the 2D system. The stacking is analogous to two stacked sheets of paper, but with one copy spatially reflected. Physically, one expects the doubled system has gappable boundary.  Accordingly, we demonstrate that for doubled states, the explicit assumption about gappable boundary in Theorem~\ref{thm:mapping-to-sn} may be dropped.  This makes for a particularly simple theorem statement, discussed in Section~\ref{sec:double}.

 \section{Proof summary}
 \label{sec:proof_summary}

 We offer a high-level summary of the proof technique, before developing the main technical tools.  We intend this as a loosely readable standalone summary.
 
 In Theorem \ref{thm:mapping-to-sn} we map  ground states satisfying certain assumptions to string-net states.  The mapping is a constant-depth, geometrically local unitary circuit.  This circuit uses only three layers, acting on coarse-grained regions.  (In fact every 2D local circuit can be implemented with three layers on suitably coarse-grained regions, though our construction does naturally occur in three steps.)   We call the original state $\sigma$, with
 \begin{align} \label{eq:circ-sequence}
     \sigma \underset{U_1}{\mapsto} \sigma^{(1)} \underset{U_2}{\mapsto} \sigma^{(2)} \underset{U_3}{\mapsto} \sigma^{(3)}
 \end{align}
 producing state $\sigma^{(3)}$ that can be viewed as a canonical string-net state living on a hexagonal coarse-graining of the original lattice.
 These steps are illustrated in Figure~\ref{fig:circuit-summary}.

 \begin{figure}[tbhp]
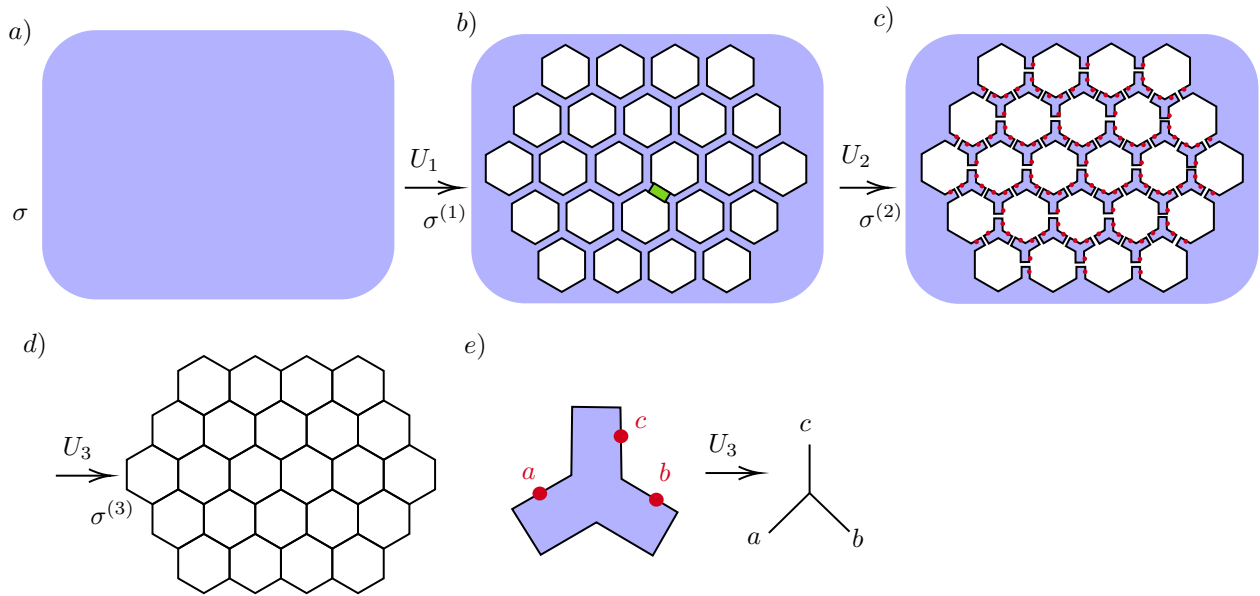

        \centering
        \include{figs/circuit-summary}
    \caption{Parts (a-d) illustrate the transformations in Eq.~\eqref{eq:circ-sequence}, using unitary circuit layers $U_1, U_2, U_3$.}
    \label{fig:circuit-summary}
\end{figure}

Each of the steps $U_1$, $U_2$, $U_3$ is a single circuit layer: a product of unitaries on disjoint regions.  These are illustrated in Figure \ref{fig:circuit-summary}.

First $U_1$ punches holes with gapped boundary. These holes are large but $O(1)$-size from the perspective of the microscopic lattice. We choose hexagonal holes, in order to ultimately to produce a string-net on a hexagonal lattice.  
These holes are \textit{not} produced by simply tracing out degrees of freedom. Each hole is formed by a local unitary operation, and the interior of the hole a pure state.  We then disregard the interiors.  (Recall that although we call the circuit unitary, we actually allow local isometries and also local projections that preserve the state, corresponding to enlarging or shrinking the local Hilbert space.)

The resulting state $\sigma^{(1)}$ lives on a ``fattened'' hexagonal lattice; the edges are thick 1D strips, which we also call ``edge regions.''  See Figure~\ref{fig:circuit-summary}(b), where such an edge region is illustrated in green.  The reduced density matrix of the edge region turns out to be a mixture of orthogonal sectors, and each sector can be identified with a type of boundary anyon (the anyon excitations can live along the gapped boundary).  This identification uses the tools in Section \ref{sec:eb_review}.  The boundary anyons are described by objects in a unitary fusion category (UFC). A UFC is \textit{also} precisely the data required to specify a string-net model, where the edges are also associated to objects in a UFC.  This analogy guides our construction.  

In the second step, $U_2$ disentangles the edge regions, leaving the vertex regions in a product state [Figure \ref{fig:circuit-summary}(c)]. More precisely, the edge regions can only be disentangled conditional on their sector, and they generally occupy a mixture of sectors.  Then the result of $U_2$ is a superposition of product states over vertex regions.  Each vertex region has three anyon excitations, labeled by red dots in the figure, one associated to each incident edge.

While string-net states are often defined on a hexagonal lattice with degrees of freedom living on both the edges and vertices, one can also define them with only vertex degrees of freedom.  The local Hilbert space on each vertex then has a basis states are labeled by a choice of three anyons, $(a,b,c)$, as well as an element of their ``fusion space.''  

The state $\sigma^{(2)}$ in Figure \ref{fig:circuit-summary}(c) therefore looks analogous to a string-net state.  In fact, the only remaining step is essentially a local change of basis, implemented by $U_3$.  
A subspace of the physical Hilbert space on the vertex region (the subspace that supports $\sigma^{(2)}$) is identified with canonical string net vertex degrees of freedom.  This identification is emphasized in Figure~\ref{fig:circuit-summary}(e).  The final state $\sigma^{(3)}$ lives in a string-net Hilbert space imposed over a hexagonal coarse-graining of the original lattice.

What remains is to show $\sigma^{(3)}$ is actually the  the ground state of the Levin-Wen Hamiltonian living on this embedded string-net Hilbert space.  To this end, we must develop a dictionary between the (1) the abstract Levin-Wen Hamiltonian, with its fusion category diagrammatics, and (2) the anyon operations on the physical Hilbert space of $\sigma$.  Then we show $\sigma^{(3)}$ is stabilized by the Levin-Wen Hamiltonian.  These steps occupy much of our technical work. 

\section{Review of entanglement bootstrap}
\label{sec:eb_review}

In this Section, we briefly review the essentials of entanglement bootstrap~\cite{shi2020fusion,shi2021entanglement} that are pertinent to this work. The entanglement bootstrap is a set of tools for analyzing states satsisfying certain axioms concerning their entanglement entropy.  These axioms have already been introduced in Section \ref{sec:setup-results}.  The axioms for bulk regions are summarized in Figure~\ref{fig:eb_axioms_bulk}, and those for regions touching the boundary are summarized in Figure~\ref{fig:eb_axioms_boundary}.

Why are these axioms useful? To answer this question, we state the following two facts. First, while the axioms are imposed only on balls of constant sizes, they \emph{imply} axioms at an arbitrarily larger scale~\cite{shi2020fusion,shi2021entanglement}. Second, for any density matrix $\rho_{ABCD}$, 
\begin{equation}
    I(A:C|B)_{\rho} \leq \left(S(BC) + S(CD) - S(B) - S(D) \right)_{\sigma}, \label{eq:cmi_upper_bound}
\end{equation}
which follows straightforwardly from SSA. The importance of Eq.~\eqref{eq:cmi_upper_bound} is that one can bound the conditional mutual information of a potentially large system ($ABC$) in terms of the linear combination of a smaller subsystem ($BCD$); note that Eq.~\eqref{eq:cmi_upper_bound}  holds for any $A$ which is a subsystem in the complement of $BCD$. It can even be the entire complement of the region $BCD$ shown in Figure~\ref{fig:eb_axioms_bulk}. Observations like this constrain the space of states locally indistinugishable from the reference state in significantly, often leading to surprisingly powerful implications. In the rest of this section, we will review such implications, focusing on the ones pertinent to this paper.

\subsection{Information convex set}
\label{subsec:information_convex_set}

A crucial concept used in the entanglement bootstrap program is the notion of \emph{information convex set}. Let $\Lambda$ be a set of sites on which the reference state is defined. Without loss of generality, let $\Omega\subset \Lambda$ be a subsystem. We can define the information convex set of $\Omega$ as follows. First, enlarge $\Omega$ to include neighboring sites. Second, identify the set of reduced density matrices on this enlarged subsystem that are each locally indistinguishable from the reference state. Third, trace out these density matrices over the neighbors used for the enlargement. The set we obtain this way is the information convex set of $\Omega$, denoted as $\Sigma(\Omega)$ [Figure~\ref{fig:info_convex_def}].

\begin{figure}[htbp]
    \centering
    \includegraphics[width=0.3\textwidth]{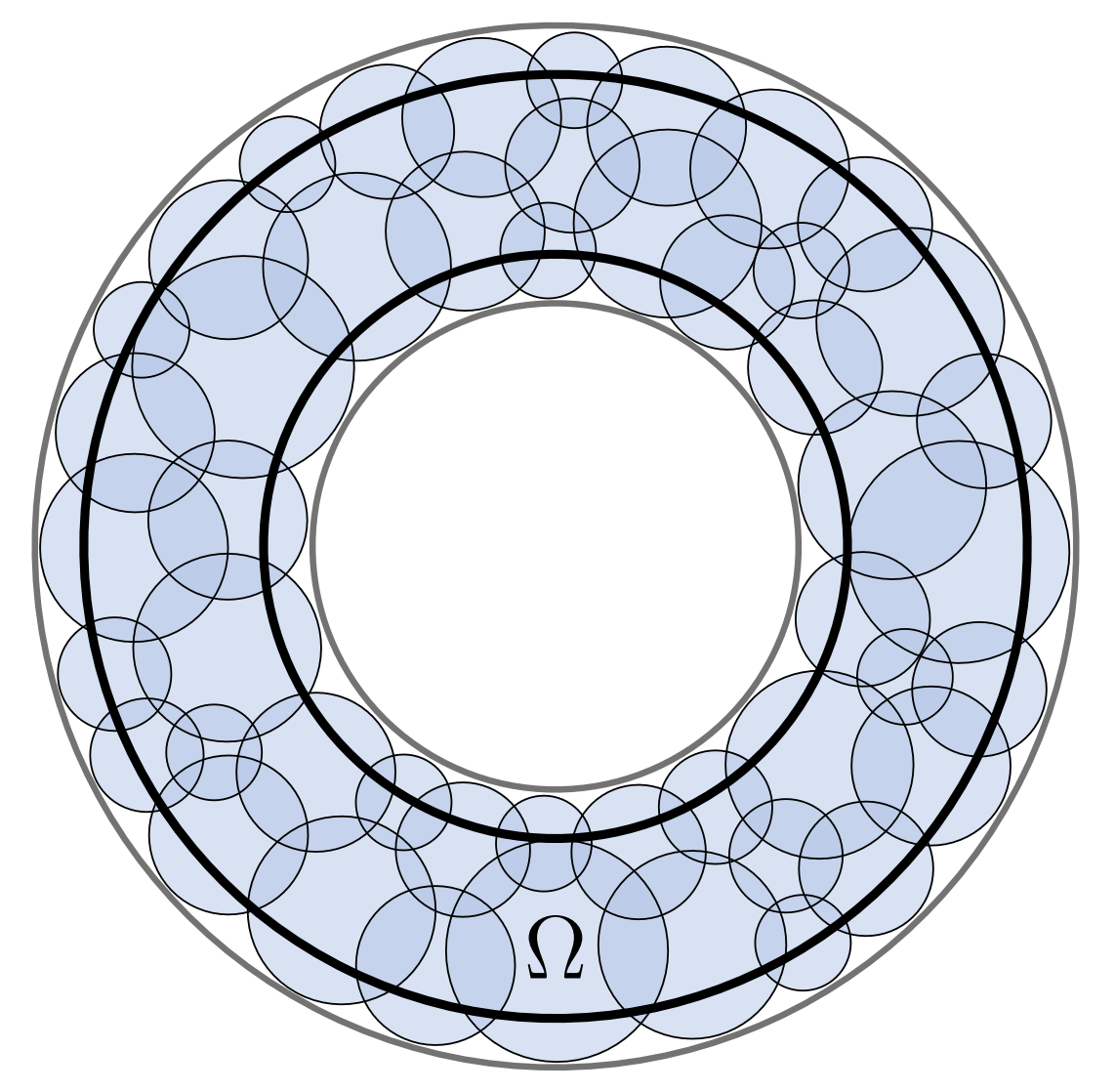}
    \caption{Subsystems used in the definition of the information convex set. (Figure obtained from Ref.~\cite{shi2020fusion} under authors' permission.)} 
    \label{fig:info_convex_def}
\end{figure}

To help intuition, we can use the fact that the entanglement bootstrap state has some commuting parent Hamiltonians  \cite{kim2024strict}.  Then one has the following equivalent formulation. For region $\Omega$, we can also define the information convex set $\Sigma(\Omega)$ as
\begin{align}
\Sigma(\Omega) = \{\Tr_{\Omega^+ \backslash \Omega} (\rho_{\Omega^+}) : \Tr(\rho_{\Omega^+} H_{\Omega^+}) = 0\}
\end{align}
where $\Omega^+$ is the enlargement of region $\Omega$, and $H_{\Omega^+}$ consists of the Hamiltonian terms on that region. In other words, $\Sigma(\Omega)$ just consists of zero-energy states $\rho_{\Omega^+}$ on $\Omega^+$ that have been reduced to $\Omega$.

An important fact is that the information convex sets, up to an isomorphism, only depend on the topology of the underlying subsystem. Let $X$ and $X'$ be two subsystems that can be smoothly deformed into each other. Then there is an isomorphism between $\Sigma(X)$ and $\Sigma(X')$ realized by a quantum channel. This is called as the \emph{isomorphism theorem}.
\begin{theorem}[Isomorphism theorem, Ref.~\cite{shi2020fusion,shi2021entanglement}]
     \label{thm:isomorphism_theorem}
     If $\Omega^0$ and $\Omega^1$ are connected by a path $\{\Omega^t\}_{t\in [0, 1]}$, there is an isomorphism $\Phi$ between $\Sigma(\Omega^0)$ and $\Sigma(\Omega^1)$ uniquely determined by the path. Moreover, the isomorphism preserves the distance and entropy difference between two elements of the information convex sets: for any $\rho, \lambda \in \Sigma(\Omega^0)$,
     \begin{equation}
         \begin{aligned}
             D(\rho, \lambda) &= D(\Phi(\rho), \Phi(\lambda)), \\
             S(\rho) - S(\lambda) &= S(\Phi(\rho)) - S(\Phi(\lambda)),
         \end{aligned}
     \end{equation}
     where $D(\cdot, \cdot)$ is any distance measure that is non-increasing under completely-positive trace-preserving map. 
\end{theorem}

Since we will be dealing with a physical system with boundaries, let us make a remark on what it means for two subsystems to be topologically equivalent. Given a subsystem $A, B\subset \Lambda$, we say $A$ and $B$ are topologically equivalent if $A$ and $B$ \emph{as well as their restrictions to the physical boundary} can be smoothly deformed into each other. For instance, a ball in the bulk is topologically inequivalent to a ball anchored on a boundary, even though they are both balls. This is because their restriction on the boundary is inequivalent; the former is an empty set whereas the latter is nonempty.

In this paper, we shall primarily deal with two types of topologies. In what follows, we specify these and review the facts that are relevant to this paper. Throughout this paper, we shall use the following diagrammatic convention [Figure~\ref{fig:figure_conventions}]. The blue color will be reserved for regions that satisfy the bulk EB axioms [Figure~\ref{fig:eb_axioms_bulk}].   Solid lines represent the physical boundary and the dashed lines represent boundaries between subsystems. We shall often specify a localized region, e.g., a region enclosed in the red rectangle in Figure~\ref{fig:figure_conventions}(a), without specifying the global system. For those localized regions, the absence of lines on certain boundaries mean that the shown region is connected to a larger system through those boundaries [Figure~\ref{fig:figure_conventions}(b)].

\begin{figure}[htbp]
    \centering
    \begin{tikzpicture}
    \draw [fill=blue!30!white] plot [smooth cycle] coordinates {(0,0) (1,1) (3,1) (4.5,0) (5, -0.5) (2,-1)};
        \begin{scope}
            \clip plot [smooth cycle] coordinates {(0,0) (1,1) (3,1) (4.5,0) (5, -0.5) (2,-1)};
            \draw[dashed, fill=green!30!white, dash pattern=on 1pt off 1.25pt] (3,-0.9) circle (0.25cm);
            \draw[dashed, fill=blue!30!white, dash pattern=on 1pt off 1.25pt] (3,-0.9) circle (0.125cm);
        \end{scope}
        \draw[fill=white, thick] (2, 0) circle (0.35cm);
        \draw[thick] plot [smooth cycle] coordinates {(0,0) (1,1) (3,1) (4.5,0) (5, -0.5) (2,-1)};
        \draw[red] (3-0.35, -0.9-0.35) -- ++ (0.7, 0) -- ++ (0, 0.7) -- ++ (-0.7, 0) -- cycle;
    \node[] () at (2.5, -1.5) {(a)};
    \node[] () at (9, -1.5) {(b)};

    \begin{scope}[xshift=0cm, yshift=2.25cm, scale=3]
    \clip (3-0.35, -0.9-0.35) -- ++ (0.7, 0) -- ++ (0, 0.7) -- ++ (-0.7, 0) -- cycle;
    \draw [fill=blue!30!white] plot [smooth cycle] coordinates {(0,0) (1,1) (3,1) (4.5,0) (5, -0.5) (2,-1)};
        \begin{scope}
            \clip plot [smooth cycle] coordinates {(0,0) (1,1) (3,1) (4.5,0) (5, -0.5) (2,-1)};
            \draw[dashed, fill=green!30!white, very thick] (3,-0.9) circle (0.25cm);
            \draw[dashed, fill=blue!30!white, very thick] (3,-0.9) circle (0.125cm);
        \end{scope}
        \draw[very thick] plot [smooth cycle] coordinates {(0,0) (1,1) (3,1) (4.5,0) (5, -0.5) (2,-1)};
    \end{scope}

    \end{tikzpicture}
    \caption{Diagrammatic convention of this paper. The solid and the dashed lines are the physical boundary and the boundary between subsystems, respectively. The blue region is topologically ordered and the green region is a subsystem being specified. The subsystem enclosed in the red rectangle in (a) is shown in (b).}
    \label{fig:figure_conventions}
\end{figure}
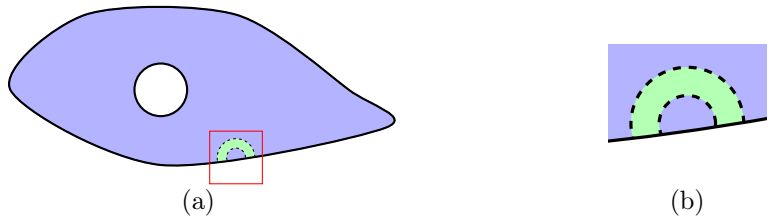

The first topology type is a half-annulus (denoted $N$, and shaped roughly like the letter $n$) [Figure~\ref{fig:IC_types_original}(a)]. Its information convex set is  a simplex~\cite{shi2021entanglement}
\begin{equation}
    \Sigma(N) = \left\{\bigoplus_{i} p_i \rho_i: p_i\geq 0, \sum_i p_i=1 \right\},
\end{equation}
where $\{\rho_i\}$ is a set of extreme points of the convex set $\Sigma(N)$ which are furthermore orthogonal to each other, e.g., $\rho_i \rho_j \propto \delta_{ij}$., and $\{ p_i\}$ is a probability distribution.
By the isomorphism theorem, there is a one-to-one map between the extreme points of \emph{any} two half-annulus. Up to this map, we can denote the extreme points in terms of a label set $\mathcal{C}=\{1, a, b, c, \ldots \}$. Here the sector $1$ corresponds to the reduced density matrix of the reference state on $N$. Physically, the label set $\mathcal{C}$ should be viewed as the label set of the boundary anyons.  (In the category-theoretic language, $\mathcal{C}$ will be the category describing the boundary anyon theory, and the labels $\{1,a,b,\ldots\}$ denote simple objects.)

The second topology we discuss is a $M$-shaped subsystem anchored at the boundary (denoted as $M$) [Figure~\ref{fig:IC_types_original}(b)]. (While the $M$-shaped regions are literally shaped like the letter $M$, the $N$-shaped regions are shaped more like the lowercase letter $n$.) This subsystem includes three half-annuli $N_1, N_2,$ and $N_3$  [Figure~\ref{fig:IC_types_original}(c)], and as such, the information convex set $\Sigma(M)$ can be further decomposed in terms of the extreme points of $\Sigma(N_1), \Sigma(N_2),$ and $\Sigma(N_3)$. Without loss of generality, let $a, b,$ and $c$ be the labels for these extreme points. Following Ref.~\cite{shi2021entanglement}, we shall denote the set of density matrices in $\Sigma(M)$ which are locally indistinguishable from the corresponding extreme points as $\Sigma_{ab}^c(M)$. It was shown in Ref.~\cite{shi2021entanglement} that $\Sigma_{ab}^c(M)$ is isomorphic to the convex hull of the state space of some finite-dimensional Hilbert space. We denote the underlying Hilbert space as $\mathbb{V}_{ab}^c$ and let $\dim(\mathbb{V}_{ab}^c) = N_{ab}^c$, known as the \emph{fusion multiplicity}. Physically, one should view $\mathbb{V}_{ab}^c$ as the fusion space of the boundary anyon $a$ and $b$ fusing into $c$ (or $\Hom{a \otimes b}{c}$ in the category-theoretic language).

\begin{figure}[htbp]
    \centering
    \begin{tikzpicture}[scale=0.25]
        \filldraw[blue!30!white] (-10, 0) -- (10, 0) -- ++ (0, 8) -- ++ (-20, 0) -- cycle;
        \filldraw[green!30!white] (-5, 0) arc (180:0:5) -- cycle;
        \filldraw[blue!30!white] (-2, 0) arc (180:0:2) -- cycle;
        \draw[very thick] (-10, 0) -- (10, 0);
        \draw[dashed, very thick] (-2, 0) arc (180:0:2);
        \draw[dashed, very thick] (-5, 0) arc (180:0:5);
        \node[] () at (0, 3.5) {$N$};
        \node[] () at (0, -1.5) {(a)};
        
        \begin{scope}[xshift=22.5cm]
        \filldraw[blue!30!white] (-10, 0) -- (10, 0) -- ++ (0, 8) -- ++ (-20, 0) -- cycle;
        \filldraw[green!30!white] (-6.5, 0) arc (180:0:6.5) -- cycle;
        \filldraw[blue!30!white] (-5, 0) arc (180:0:2) -- cycle;
        \filldraw[blue!30!white] (1, 0) arc (180:0:2) -- cycle;
        \draw[very thick] (-10, 0) -- (10, 0);
        \draw[dashed, very thick] (-5, 0) arc (180:0:2);
        \draw[dashed, very thick] (1, 0) arc (180:0:2);
        \draw[dashed, very thick] (-6.5, 0) arc (180:0:6.5);
        \node[] () at (0, 3.5) {$M$};
        \node[] () at (0, -1.5) {(b)};
        \end{scope}
        \begin{scope}[xshift=45cm]
        \filldraw[blue!30!white] (-10, 0) -- (10, 0) -- ++ (0, 8) -- ++ (-20, 0) -- cycle;
        \filldraw[orange] (-6.5, 0) arc (180:0:6.5) -- cycle;
        \filldraw[green!30!white] (-6, 0) arc (180:0:6) -- cycle;
        \filldraw[orange] (-5.5, 0) arc (180:0:2.5) -- cycle;
        \filldraw[blue!30!white] (-5, 0) arc (180:0:2) -- cycle;
        \filldraw[orange] (0.5, 0) arc (180:0:2.5) -- cycle;
        \filldraw[blue!30!white] (1, 0) arc (180:0:2) -- cycle;
        \draw[very thick] (-10, 0) -- (10, 0);
        \draw[dashed, very thick] (-5, 0) arc (180:0:2);
        \draw[dashed, very thick] (1, 0) arc (180:0:2);
        \draw[dashed, very thick] (-6.5, 0) arc (180:0:6.5);

        \node[inner sep = 0pt] (n1) at (-2, 4) {$N_1$};
        \node[inner sep = 0pt] (n2) at (2, 4) {$N_2$};
        \node[inner sep = 0pt] (n3) at (8, 4) {$N_3$};
        \draw[->] (n1) -- (-3, 2.5);
        \draw[->] (n2) -- (3, 2.5);
        \draw[->] (n3) --++ (-3, 0);
        
        \node[] () at (0, -1.5) {(c)};
        \end{scope}
    \end{tikzpicture}
    \caption{(a) Half-annulus $N$ (b) An $M$-shaped region, obtained by puncturing a boundary-anchored disk with two (smaller) boundary-anchored disks. (c) The $M$-shaped region contains three disjoint half-annuli $N_1, N_2,$ and $N_3$ (orange).}
    \label{fig:IC_types_original}
\end{figure}
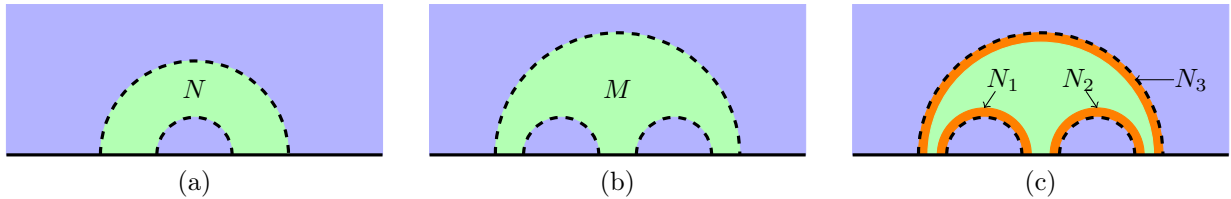

We remark that the information convex sets $\Sigma(N)$ and $\Sigma(M)$ are isomorphic to information convex sets $\Sigma(N')$ and $\Sigma(M')$ shown in Figure~\ref{fig:IC_types_this_paper}(a) and (b), respectively. For instance, the following argument establishes $\Sigma(N') \cong \Sigma(N)$. First, by the definition of the information convex set, it is insensitive to the change in the density matrix sufficiently far away from the given region. Therefore, even though Figure~\ref{fig:equivalence_ICs}(a) and (b) have different boundary conditions, the information convex set $\Sigma(N')$ in both diagrams are identical. Then we can use the isomorphism theorem [Theorem~\ref{thm:isomorphism_theorem}] to conclude that $\Sigma(N') \cong \Sigma(N)$ [Figure~\ref{fig:equivalence_ICs}(b-c)]. A similar argument can be used to prove $\Sigma(M') \cong \Sigma(M)$ [Figure~\ref{fig:equivalence_ICs}(d-f)].

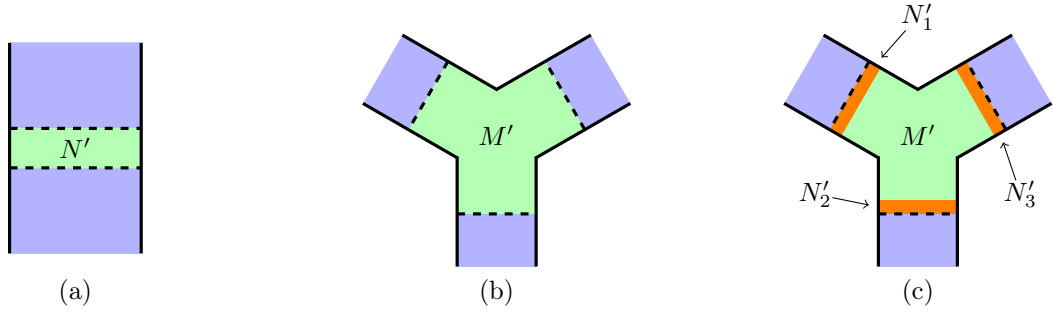
\begin{figure}[htbp]
    \centering
    \begin{tikzpicture}[scale=0.35]
        \filldraw[blue!30!white] (-2.5, 0) -- (2.5, 0) -- ++ (0, 8)  -- ++ (-5, 0)-- cycle;
        \filldraw[green!30!white] (-2.5, 3.25) -- ++ (5,0) -- ++ (0, 1.5) -- ++ (-5, 0)-- cycle;
        \draw[very thick] (-2.5, 0) -- ++ (0,8);
        \draw[very thick] (2.5, 0) -- ++ (0,8);

        \draw[dashed, very thick] (-2.5, 4-0.75) -- ++ (5,0);
        \draw[dashed, very thick] (-2.5, 5.5-0.75) -- ++ (5,0);

        \node[] () at (0, 4) {$N'$};
        
        \node[] () at (0, -1.5) {(a)};

        \begin{scope}[xshift=16cm, yshift=4.5cm]
            \filldraw[fill=blue!30!white, draw=none, rotate=60] (-1.5cm,0) --++ (3cm, 0)  -- ++ (0, 5cm) -- ++ (-3cm, 0) -- cycle;
            \filldraw[fill=blue!30!white, draw=none, rotate=-60] (-1.5cm,0) --++ (3cm, 0)  -- ++ (0, 5cm) -- ++ (-3cm, 0) -- cycle;
            \filldraw[fill=blue!30!white, draw=none, rotate=180] (-1.5cm,0) --++ (3cm, 0)  -- ++ (0, 5cm) -- ++ (-3cm, 0) -- cycle;

            \filldraw[fill=green!30!white, draw=none, rotate=60] (-1.5cm,0) --++ (3cm, 0)  -- ++ (0, 3cm) -- ++ (-3cm, 0) -- cycle;
            \filldraw[fill=green!30!white, draw=none, rotate=-60] (-1.5cm,0) --++ (3cm, 0)  -- ++ (0, 3cm) -- ++ (-3cm, 0) -- cycle;
            \filldraw[fill=green!30!white, draw=none, rotate=180] (-1.5cm,0) --++ (3cm, 0)  -- ++ (0, 3cm) -- ++ (-3cm, 0) -- cycle;

            \draw[rotate=60, very thick, dashed] (1.5cm, 3cm) -- ++ (-3cm, 0);
            \draw[rotate=-60, very thick, dashed] (1.5cm, 3cm) -- ++ (-3cm, 0);
            \draw[rotate=180, very thick, dashed] (1.5cm, 3cm) -- ++ (-3cm, 0);

            \draw[rotate=60, very thick] (-1.5cm, 0.825cm) -- (-1.5cm, 5cm);
            \draw[rotate=60, very thick] (1.5cm, 0.825cm) -- (1.5cm, 5cm);
            \draw[rotate=-60, very thick] (-1.5cm, 0.825cm) -- (-1.5cm, 5cm);
            \draw[rotate=-60, very thick] (1.5cm, 0.825cm) -- (1.5cm, 5cm);
            \draw[rotate=180, very thick] (-1.5cm, 0.825cm) -- (-1.5cm, 5cm);
            \draw[rotate=180, very thick] (1.5cm, 0.825cm) -- (1.5cm, 5cm);

            \node[] () at (0,0) {$M'$};

            \node[] () at (0, -6) {(b)};
        \end{scope}
        \begin{scope}[xshift=32cm, yshift=4.5cm]
            \filldraw[fill=blue!30!white, draw=none, rotate=60] (-1.5cm,0) --++ (3cm, 0)  -- ++ (0, 5cm) -- ++ (-3cm, 0) -- cycle;
            \filldraw[fill=blue!30!white, draw=none, rotate=-60] (-1.5cm,0) --++ (3cm, 0)  -- ++ (0, 5cm) -- ++ (-3cm, 0) -- cycle;
            \filldraw[fill=blue!30!white, draw=none, rotate=180] (-1.5cm,0) --++ (3cm, 0)  -- ++ (0, 5cm) -- ++ (-3cm, 0) -- cycle;

            \filldraw[fill=green!30!white, draw=none, rotate=60] (-1.5cm,0) --++ (3cm, 0)  -- ++ (0, 3cm) -- ++ (-3cm, 0) -- cycle;
            \filldraw[fill=green!30!white, draw=none, rotate=-60] (-1.5cm,0) --++ (3cm, 0)  -- ++ (0, 3cm) -- ++ (-3cm, 0) -- cycle;
            \filldraw[fill=green!30!white, draw=none, rotate=180] (-1.5cm,0) --++ (3cm, 0)  -- ++ (0, 3cm) -- ++ (-3cm, 0) -- cycle;

            \filldraw[fill=orange, draw=none, rotate=60] (-1.5cm,2.5cm) --++ (3cm, 0)  -- ++ (0, 0.5cm) -- ++ (-3cm, 0) -- cycle;
            \filldraw[fill=orange, draw=none, rotate=-60] (-1.5cm,2.5cm) --++ (3cm, 0)  -- ++ (0, 0.5cm) -- ++ (-3cm, 0) -- cycle;
            \filldraw[fill=orange, draw=none, rotate=180] (-1.5cm,2.5cm) --++ (3cm, 0)  -- ++ (0, 0.5cm) -- ++ (-3cm, 0) -- cycle;

            \draw[rotate=60, very thick, dashed] (1.5cm, 3cm) -- ++ (-3cm, 0);
            \draw[rotate=-60, very thick, dashed] (1.5cm, 3cm) -- ++ (-3cm, 0);
            \draw[rotate=180, very thick, dashed] (1.5cm, 3cm) -- ++ (-3cm, 0);

            \draw[rotate=60, very thick] (-1.5cm, 0.825cm) -- (-1.5cm, 5cm);
            \draw[rotate=60, very thick] (1.5cm, 0.825cm) -- (1.5cm, 5cm);
            \draw[rotate=-60, very thick] (-1.5cm, 0.825cm) -- (-1.5cm, 5cm);
            \draw[rotate=-60, very thick] (1.5cm, 0.825cm) -- (1.5cm, 5cm);
            \draw[rotate=180, very thick] (-1.5cm, 0.825cm) -- (-1.5cm, 5cm);
            \draw[rotate=180, very thick] (1.5cm, 0.825cm) -- (1.5cm, 5cm);

            \node[] () at (0,0) {$M'$};

            \node[inner sep=0pt] (n1p) at (90:4.5cm) {$N_1'$};
            \node[inner sep=0pt] (n2p) at (210:4.5cm) {$N_2'$};
            \node[inner sep=0pt] (n3p) at (330:4.5cm) {$N_3'$};

            \draw[->] (n1p) -- (115:3.25cm);
            \draw[->] (n2p) -- (235:3.25cm);
            \draw[->] (n3p) -- (355:3.25cm);

            \node[] () at (0, -6) {(c)};
        \end{scope}
    \end{tikzpicture}
    \caption{Frequently used subsystems in this paper. (a) The information convex set of $N'$ is isomorphic to $\Sigma(N)$ in Figure~\ref{fig:IC_types_original}(a). (b) The information convex set of $M'$ is isomorphic to $\Sigma(M)$ in Figure~\ref{fig:IC_types_original}(b). (c) The subsystem $M'$ includes $N_1', N_2',$ and $N_3'$.}
    \label{fig:IC_types_this_paper}
\end{figure}

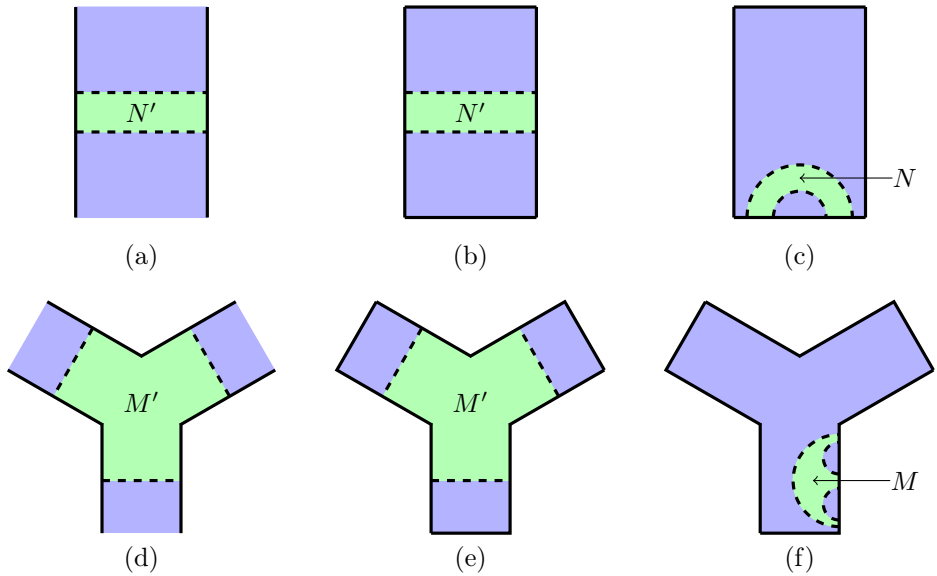
\begin{figure}[htbp]
    \centering
    \begin{tikzpicture}[scale=0.35]
        \filldraw[blue!30!white] (-2.5, 0) -- (2.5, 0) -- ++ (0, 8)  -- ++ (-5, 0)-- cycle;
        \filldraw[green!30!white] (-2.5, 3.25) -- ++ (5,0) -- ++ (0, 1.5) -- ++ (-5, 0)-- cycle;
        \draw[very thick] (-2.5, 0) -- ++ (0,8);
        \draw[very thick] (2.5, 0) -- ++ (0,8);

        \draw[dashed, very thick] (-2.5, 4-0.75) -- ++ (5,0);
        \draw[dashed, very thick] (-2.5, 5.5-0.75) -- ++ (5,0);

        \node[] () at (0, 4) {$N'$};
        
        \node[] () at (0, -1.5) {(a)};

        \begin{scope}[xshift=12.5cm]
        \filldraw[blue!30!white] (-2.5, 0) -- (2.5, 0) -- ++ (0, 8)  -- ++ (-5, 0)-- cycle;
        \filldraw[green!30!white] (-2.5, 3.25) -- ++ (5,0) -- ++ (0, 1.5) -- ++ (-5, 0)-- cycle;
        \draw[very thick] (-2.5, 0) -- ++ (0,8);
        \draw[very thick] (2.5, 0) -- ++ (0,8);
        \draw[very thick] (-2.5, 0)  -- ++ (5,0);
        \draw[very thick] (-2.5, 8)  -- ++ (5,0);

        \draw[dashed, very thick] (-2.5, 4-0.75) -- ++ (5,0);
        \draw[dashed, very thick] (-2.5, 5.5-0.75) -- ++ (5,0);

        \node[] () at (0, 4) {$N'$};
        
        \node[] () at (0, -1.5) {(b)};
            
        \end{scope}

        \begin{scope}[xshift=25cm]
        \filldraw[blue!30!white] (-2.5, 0) -- (2.5, 0) -- ++ (0, 8)  -- ++ (-5, 0)-- cycle;

        \filldraw[green!30!white] (-2, 0) arc (180:0:2cm);
        \filldraw[blue!30!white] (-1, 0) arc (180:0:1cm);
        \draw[very thick, dashed] (-2, 0) arc (180:0:2cm);
        \draw[very thick, dashed] (-1, 0) arc (180:0:1cm);

        \draw[very thick] (-2.5, 0) -- ++ (0,8);
        \draw[very thick] (2.5, 0) -- ++ (0,8);
        \draw[very thick] (-2.5, 0)  -- ++ (5,0);
        \draw[very thick] (-2.5, 8)  -- ++ (5,0);

        \node[inner sep=0pt] (N) at (4cm, 1.5cm) {$N$};
        \draw[->] (N) -- (0, 1.5cm);
        \node[] () at (0, -1.5) {(c)};
        \end{scope}

        \begin{scope}[yshift=-7cm]
            \filldraw[fill=blue!30!white, draw=none, rotate=60] (-1.5cm,0) --++ (3cm, 0)  -- ++ (0, 5cm) -- ++ (-3cm, 0) -- cycle;
            \filldraw[fill=blue!30!white, draw=none, rotate=-60] (-1.5cm,0) --++ (3cm, 0)  -- ++ (0, 5cm) -- ++ (-3cm, 0) -- cycle;
            \filldraw[fill=blue!30!white, draw=none, rotate=180] (-1.5cm,0) --++ (3cm, 0)  -- ++ (0, 5cm) -- ++ (-3cm, 0) -- cycle;

            \filldraw[fill=green!30!white, draw=none, rotate=60] (-1.5cm,0) --++ (3cm, 0)  -- ++ (0, 3cm) -- ++ (-3cm, 0) -- cycle;
            \filldraw[fill=green!30!white, draw=none, rotate=-60] (-1.5cm,0) --++ (3cm, 0)  -- ++ (0, 3cm) -- ++ (-3cm, 0) -- cycle;
            \filldraw[fill=green!30!white, draw=none, rotate=180] (-1.5cm,0) --++ (3cm, 0)  -- ++ (0, 3cm) -- ++ (-3cm, 0) -- cycle;

            \draw[rotate=60, very thick, dashed] (1.5cm, 3cm) -- ++ (-3cm, 0);
            \draw[rotate=-60, very thick, dashed] (1.5cm, 3cm) -- ++ (-3cm, 0);
            \draw[rotate=180, very thick, dashed] (1.5cm, 3cm) -- ++ (-3cm, 0);

            \draw[rotate=60, very thick] (-1.5cm, 0.825cm) -- (-1.5cm, 5cm);
            \draw[rotate=60, very thick] (1.5cm, 0.825cm) -- (1.5cm, 5cm);
            \draw[rotate=-60, very thick] (-1.5cm, 0.825cm) -- (-1.5cm, 5cm);
            \draw[rotate=-60, very thick] (1.5cm, 0.825cm) -- (1.5cm, 5cm);
            \draw[rotate=180, very thick] (-1.5cm, 0.825cm) -- (-1.5cm, 5cm);
            \draw[rotate=180, very thick] (1.5cm, 0.825cm) -- (1.5cm, 5cm);

            \node[] () at (0,0) {$M'$};

            \node[] () at (0, -6) {(d)};
        \end{scope}
        \begin{scope}[yshift=-7cm, xshift=12.5cm]
            \filldraw[fill=blue!30!white, draw=none, rotate=60] (-1.5cm,0) --++ (3cm, 0)  -- ++ (0, 5cm) -- ++ (-3cm, 0) -- cycle;
            \filldraw[fill=blue!30!white, draw=none, rotate=-60] (-1.5cm,0) --++ (3cm, 0)  -- ++ (0, 5cm) -- ++ (-3cm, 0) -- cycle;
            \filldraw[fill=blue!30!white, draw=none, rotate=180] (-1.5cm,0) --++ (3cm, 0)  -- ++ (0, 5cm) -- ++ (-3cm, 0) -- cycle;

            \filldraw[fill=green!30!white, draw=none, rotate=60] (-1.5cm,0) --++ (3cm, 0)  -- ++ (0, 3cm) -- ++ (-3cm, 0) -- cycle;
            \filldraw[fill=green!30!white, draw=none, rotate=-60] (-1.5cm,0) --++ (3cm, 0)  -- ++ (0, 3cm) -- ++ (-3cm, 0) -- cycle;
            \filldraw[fill=green!30!white, draw=none, rotate=180] (-1.5cm,0) --++ (3cm, 0)  -- ++ (0, 3cm) -- ++ (-3cm, 0) -- cycle;

            \draw[rotate=60, very thick, dashed] (1.5cm, 3cm) -- ++ (-3cm, 0);
            \draw[rotate=-60, very thick, dashed] (1.5cm, 3cm) -- ++ (-3cm, 0);
            \draw[rotate=180, very thick, dashed] (1.5cm, 3cm) -- ++ (-3cm, 0);

            \draw[rotate=60, very thick] (-1.5cm, 0.825cm) -- (-1.5cm, 5cm) -- ++ (3cm, 0);
            \draw[rotate=60, very thick] (1.5cm, 0.825cm) -- (1.5cm, 5cm);
            \draw[rotate=-60, very thick] (-1.5cm, 0.825cm) -- (-1.5cm, 5cm)-- ++ (3cm, 0);
            \draw[rotate=-60, very thick] (1.5cm, 0.825cm) -- (1.5cm, 5cm);
            \draw[rotate=180, very thick] (-1.5cm, 0.825cm) -- (-1.5cm, 5cm)-- ++ (3cm, 0);
            \draw[rotate=180, very thick] (1.5cm, 0.825cm) -- (1.5cm, 5cm);

            \node[] () at (0,0) {$M'$};

            \node[] () at (0, -6) {(e)};
        \end{scope}
        \begin{scope}[yshift=-7cm, xshift=25cm]
            \filldraw[fill=blue!30!white, draw=none, rotate=60] (-1.5cm,0) --++ (3cm, 0)  -- ++ (0, 5cm) -- ++ (-3cm, 0) -- cycle;
            \filldraw[fill=blue!30!white, draw=none, rotate=-60] (-1.5cm,0) --++ (3cm, 0)  -- ++ (0, 5cm) -- ++ (-3cm, 0) -- cycle;
            \filldraw[fill=blue!30!white, draw=none, rotate=180] (-1.5cm,0) --++ (3cm, 0)  -- ++ (0, 5cm) -- ++ (-3cm, 0) -- cycle;

            \filldraw[rotate=180, green!30!white] (-1.5cm, 4.75cm) arc (90:-90:1.75cm);
            \filldraw[rotate=180, blue!30!white] (-1.5cm, 4.5cm) arc (90:-90:0.6cm);
            \filldraw[rotate=180, blue!30!white] (-1.5cm, 2.75cm) arc (90:-90:0.6cm);

            \draw[rotate=180, very thick, dashed] (-1.5cm, 4.75cm) arc (90:-90:1.75cm);
            \draw[rotate=180, very thick, dashed] (-1.5cm, 4.5cm) arc (90:-90:0.6cm);
            \draw[rotate=180, very thick, dashed] (-1.5cm, 2.75cm) arc (90:-90:0.6cm);

            \draw[rotate=60, very thick] (-1.5cm, 0.825cm) -- (-1.5cm, 5cm) -- ++ (3cm, 0);
            \draw[rotate=60, very thick] (1.5cm, 0.825cm) -- (1.5cm, 5cm);
            \draw[rotate=-60, very thick] (-1.5cm, 0.825cm) -- (-1.5cm, 5cm)-- ++ (3cm, 0);
            \draw[rotate=-60, very thick] (1.5cm, 0.825cm) -- (1.5cm, 5cm);
            \draw[rotate=180, very thick] (-1.5cm, 0.825cm) -- (-1.5cm, 5cm)-- ++ (3cm, 0);
            \draw[rotate=180, very thick] (1.5cm, 0.825cm) -- (1.5cm, 5cm);

            \node[inner sep=0pt] (M) at (4, -3) {$M$};
            \draw[->] (M) -- (0.5, -3);

            \node[] () at (0, -6) {(f)};
        \end{scope}
    \end{tikzpicture}
    \caption{To establish $\Sigma(N')\cong \Sigma(N)$ and $\Sigma(M')\cong \Sigma(M)$, one can change the boundary conditions as in (b) and (e), and then use the isomorphism theorem [Theorem~\ref{thm:isomorphism_theorem}]. Note that $N'$ and $N$ can be smoothly deformed into each other in (b-c). Similarly, $M'$ and $M$ can be smoothly deformed into each other in (e-f).}
    \label{fig:equivalence_ICs}
\end{figure}

By the isomorphism, we can label the extreme points of $\Sigma(N')$ in terms of the label set of the boundary anyons $\mathcal{C}$. Moreover, we can decompose $\Sigma(M')$ further into $\Sigma_{ab}^c(M')$, associated with the extreme points corresponding to $a, b,$ and $c$ in $\Sigma(N_1')$, $\Sigma(N_2')$, and $\Sigma(N_3')$, respectively. The Hilbert space underlying $\Sigma_{ab}^c(M')$ is then $\mathbb{V}_{ab}^c$, whose dimension is $N_{ab}^c$.

Finally, we will need to know the information convex for one more topology: a disk with boundary.  In this case, the information convex is trivial.

\subsection{Extreme points}
\label{subsec:extreme_points}

As briefly reviewed in Section~\ref{subsec:information_convex_set}, the physical meaning of the extreme points of the information convex set depend on the topology of the underlying subsystem. For the $N$-type subsystems, each extreme point corresponds to a boundary anyon sector, which is an element of the set $\mathcal{C}$. For the $M$-type subsystems, upon fixing the label $a,b,c\in \mathcal{C}$, the extreme point of $\Sigma_{ab}^c(M)$ corresponds to a state in the Hilbert space $\mathbb{V}_{ab}^c$. Later we will associate each extreme point to a \emph{multiplicity label} [Section~\ref{sec:sn-review}].

All these extreme points enjoy a \emph{factorization} property~\cite[Section III.D]{shi2021entanglement}, which will play an important role in our analysis. As an example, consider a half-annulus $N$ and let $a\in \mathcal{C}$ be an anyon label corresponding to one of the extreme points of $\Sigma(N)$. The factorization property in this setup says that the ``inner'' part of the half-annulus is decoupled from any other non-adjacent (separated) subsystem. More precisely,
\begin{equation}
\left( S(D) + S(CD) - S(C) \right)_{\rho_a} =0, \label{eq:factorization_extreme_point}
\end{equation}
where $D$ is a half-annulus and the $C$ is the union of two half-annuli surrounding $D$  [Figure~\ref{fig:factorization}]. This implies that $D$ is decoupled from the purification of $CD$. A simple consequence of Eq.~\eqref{eq:factorization_extreme_point} is that $I(B:D|C)_{\rho_a}=0$, where $D$ is the complement of $BC$. This is due to the Eq.~\eqref{eq:cmi_upper_bound}. A similar conclusion applies to the extreme points of $\Sigma_{ab}^c(M)$, as we discuss in more detail in Section~\ref{subsec:basic_operations}. For a more general discussion on the factorization property, see~\cite[Section III.D]{shi2021entanglement}.

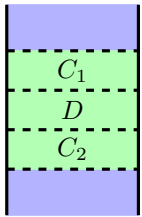
\begin{figure}[t]
    \centering
    \begin{tikzpicture}[scale=0.35]
        \begin{scope}[xshift=-20cm]
        \filldraw[blue!30!white] (-2.5, 0) -- (2.5, 0) -- ++ (0, 8)  -- ++ (-5, 0)-- cycle;
        \filldraw[green!30!white] (-2.5, 1.75) -- ++ (5,0) -- ++ (0, 4.5) -- ++ (-5, 0)-- cycle;
        \draw[very thick] (-2.5, 0) -- ++ (0,8);
        \draw[very thick] (2.5, 0) -- ++ (0,8);

        \draw[dashed, very thick] (-2.5, 2.5-0.75) -- ++ (5,0);
        \draw[dashed, very thick] (-2.5, 4-0.75) -- ++ (5,0);
        \draw[dashed, very thick] (-2.5, 5.5-0.75) -- ++ (5,0);
        \draw[dashed, very thick] (-2.5, 7-0.75) -- ++ (5,0);

        \node[] () at (0, 5.5) {$C_1$};
        \node[] () at (0, 4) {$D$};
        \node[] () at (0, 2.5) {$C_2$};
        
        % \node[] () at (0, -1.5) {(a)};
        \end{scope}
    \end{tikzpicture}
    \caption{ Half-annulus (green) anchored at the boundary is partitioned into $C_1, C_2,$ and $D$. Every extreme point $\rho_a$ satisfies $\left( S(D) + S(CD) - S(C) \right)_{\rho_a}=0$, where $C= C_1C_2$.}
    \label{fig:factorization}
\end{figure}

\subsection{Merging}
\label{sec:merging}

A useful technique frequently used in entanglement bootstrap is \emph{merging}. This entails combining two density matrices with overlapping supports to a density matrix supported on the union of their supports. The resulting global density matrix is the maximum-entropy state consistent with the two density matrices, meaning that its marginals (on different subsystems) are equal to the respective density matrices we started with. While  such a global density matrix may not exist in general,\footnote{As a simple example, consider bipartite density matrices over $AB$ and $BC$, both of which are EPR pairs. By the monogamy of entanglement, there does not exist a quantum state on $ABC$ consistent with these density matrices.} the axioms of the entanglement bootstrap [Section~\ref{sec:eb_review}] can guarantee this. 

To that end, let us first introduce the \emph{merging lemma}~\cite{Kato2016}.
\begin{lemma}~\cite{Kato2016}
    Let $\rho_{ABC}$ and $\sigma_{BCD}$ be density matrices such that $\rho_{BC} = \sigma_{BC}$ and $I(A:C|B)_{\rho} = I(B:D|C)_{\sigma}=0$. Then there exists a density matrix $\lambda_{ABCD}$ such that
    \begin{enumerate}
        \item $\lambda_{ABC} = \rho_{ABC}, \lambda_{BCD} = \sigma_{BCD}$ and 
        \item $I(A:CD|B)_{\lambda} = I(AB:D|C)_{\lambda}=0$.
    \end{enumerate}
    \label{lemma:merging_lemma}
\end{lemma}
\noindent 
We will later use the merging lemma to combine density matrices over some regions to build up a density matrix over a larger region. Crucially, the existence of the density matrix on the larger region need not be assumed; it follows simply from the properties of the density matrices over the smaller regions.

Here are some canonical examples of the applications of the merging lemma~\cite{shi2020fusion}. In Figure~\ref{fig:merging_lemma_examples}(a), from the axiom \textbf{A1}, one can show that $I(A:C|B)_{\sigma} = I(B:D|C)_{\sigma}=0$. Thus by merging the two states, we obtain state on $ABCD$ which is again a quantum Markov chain [Lemma~\ref{lemma:merging_lemma}]. This state is consistent with $\sigma_{ABCD}$ on both $ABC$ and $BCD$. Since $\sigma_{ABCD}$ is a Markov chain satisfying $I(A:CD|B)_{\sigma}=0$, by the uniqueness of the Markov chain, we can conlcude that the merged state is in fact $\sigma_{ABCD}$. 

It is important to note that the merged state is not always a reduced density matrix of the reference state. For instance, consider the subsystems shown in Figure~\ref{fig:merging_lemma_examples}(b). Using the axiom $\textbf{A0}$ and $\textbf{A1}$, it is again possible to show that $I(A:C|B)_{\sigma} = I(B:D|C)_{\sigma}=0$. Thus we can again merge these two states. However, the resulting state (denoted as $\tau$ below) is not necessarily $\sigma_{ABCD}$. Rather, it is the maximum-entropy state of the following form:
\begin{equation}
    \tau_{ABCD} = \sum_a \frac{d_a^2}{\mathcal{D}^2} \rho_a,
\end{equation}
where $d_a$ is the \emph{quantum dimension} of the sector $a$ and $\mathcal{D} = \sqrt{\sum_a d_a^2}$ is the total quantum dimension~\cite{shi2020fusion}. This is different from the reference state, which would be simply $\sigma = \rho_1$.

\begin{figure}[htbp]
    \centering
    \begin{tikzpicture}
    \filldraw[blue!30!white] (-1cm, -1cm) -- ++ (6cm, 0) -- ++ (0, 3cm) -- ++ (-6cm, 0) -- cycle;
        \draw[very thick, dashed] (0, 0) -- ++ (1,0) -- ++ (0, 1) -- ++ (-1, 0) -- cycle;
        \draw[very thick, dashed] (1, 0) -- ++ (1,0) -- ++ (0, 1) -- ++ (-1, 0) -- cycle;
        \draw[very thick, dashed] (2, 0) -- ++ (1,0) -- ++ (0, 1) -- ++ (-1, 0) -- cycle;
        \draw[very thick, dashed] (3, 0) -- ++ (1,0) -- ++ (0, 1) -- ++ (-1, 0) -- cycle;
        \node[] () at (0.5, 0.5) {$A$};
        \node[] () at (1.5, 0.5) {$B$};
        \node[] () at (2.5, 0.5) {$C$};
        \node[] () at (3.5, 0.5) {$D$};

        \node[] () at (2cm, -1.5cm) {(a)};

        \begin{scope}[xshift=9cm, yshift=0.5cm, scale= 0.6]
            \filldraw[blue!30!white] (-4cm, -2.5cm) -- ++ (8cm, 0) -- ++ (0, 5cm) -- ++ (-8cm, 0) -- cycle;
            \draw[very thick, dashed] (0,0) circle (2cm); 
            \draw[very thick, dashed] (0,0) circle (1cm);
            \draw[very thick, dashed] (90:1cm) -- (90:2cm);
            \draw[very thick, dashed] (270:1cm) -- (270:2cm);
            \draw[very thick, dashed] (120:1cm) -- (120:2cm);
            \draw[very thick, dashed] (240:1cm) -- (240:2cm);
            \draw[very thick, dashed] (60:1cm) -- (60:2cm);
            \draw[very thick, dashed] (300:1cm) -- (300:2cm);
            \node[] (A) at (180:1.5cm) {$A$};
            \node[] (B1) at (105:1.5cm) {$B$};
            \node[] (B2) at (255:1.5cm) {$B$};
            \node[] (C1) at (75:1.5cm) {$C$};
            \node[] (C2) at (285:1.5cm) {$C$};
            \node[] (D) at (0:1.5cm) {$D$};
            
            \node[] () at (0, -3.5cm) {(b)};
        \end{scope}
    \end{tikzpicture}
    \caption{Examples being merging. (a) Merging $\sigma_{ABC}$ and $\sigma_{BCD}$, we obtain $\sigma_{ABCD}$. (b) Merging $\sigma_{ABC}$ and $\sigma_{BCD}$, we obtain a maximum-entropy state consistent with the marginals, which is generally different from $\sigma_{ABCD}$.}
    \label{fig:merging_lemma_examples}
\end{figure}
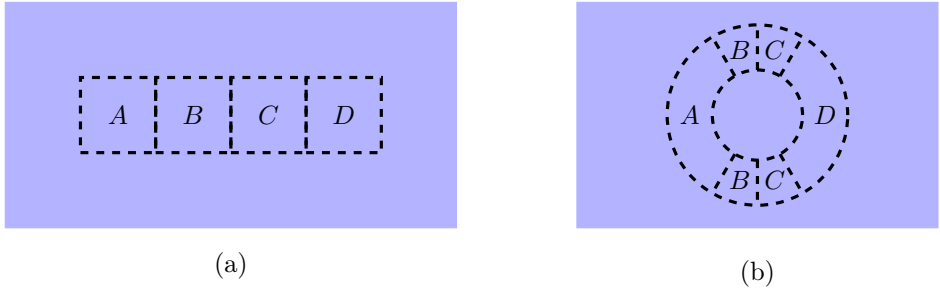

We remark that one can in fact make a stronger statement. Namely, if the density matrices $\rho$ and $\sigma$ are elements of the information convex set (say $\Sigma(A)$ and $\Sigma(B)$), the merged state $\lambda$ is an element of the information convex set $\Sigma(A\cup B)$~\cite{shi2020fusion,shi2021entanglement}. This is known as the merging theorem~\cite{shi2020fusion,shi2021entanglement}.

\begin{theorem}
\label{thm:merging_theorem}
    Consider two density matrices $\rho_{ABC}\in \Sigma(ABC)$ and $\lambda_{BCD} \in \Sigma(BCD)$. Consider the following three conditions.
    \begin{enumerate}
        \item $\rho_{BC} = \lambda_{BC}$ and $I(A:C|B)_{\rho} = I(B:D|C)_{\lambda}=0$.
        \item There exists a partition $B'C' = BC$ such that no disk of radius $r$ overlaps with both $AB'$ and $CD.$ (Here $r$ is the radius of the disk on which the axioms are imposed.)
        \item $I(A:C'|B')_{\rho} = I(B':D|C')_{\lambda}=0$.
    \end{enumerate}
\end{theorem}

\section{Unitary fusion category for boundary anyons}
\label{sec:boundary_anyons}

As reviewed in Section~\ref{subsec:information_convex_set}, the boundary anyon labels and their fusion spaces can be defined in terms of the information convex sets of half-annuli and $M$-shaped regions [Figure~\ref{fig:IC_types_original}]. We remark that the fusion multiplicities $N_{ab}^c$ satisfies the following set of identities~\cite[Section VI-A]{shi2021entanglement}.
\begin{equation}
    \begin{aligned}
        N_{1a}^b &= N_{a1}^b = \delta_{a,b},\\
        \forall a\in \mathcal{C}, \exists \overline{a}\in\mathcal{C} \text{ s.t. } N_{ab}^1&= \delta_{a,\bar{b}} =\delta_{\overline{a}, b},\\
        N_{ab}^{c} &= N_{\bar{b} \overline{a}}^{\bar{c}}, \\
        \sum_{e\in \mathcal{C}} N_{ab}^e N_{ec}^d &= \sum_{f} N_{af}^dN_{bc}^f.
    \end{aligned}
    \label{eq:ufc_almost}
\end{equation}
The first line means there exists a vacuum sector, and the second line means that there exists an antiparticle for every anyon. The last identity tells us that the composition of fusion multiplicities is associative.

 Eq.~\eqref{eq:ufc_almost} provides part of the data needed to define a unitary fusion category (UFC), which is the mathematical formalism that describes boundary anyons~\cite{kong2014some}. (See Appenix~\ref{sec:tensor_category} for a terse summary of tensor categories.) However, there is still missing data, given by the $F$-symbol. We shall define this operation in Section~\ref{sec:f-symbols}, thus deriving a UFC from the the entanglement properties of the ground state at the physical boundary. 

Our construction is based on a formalism developed in Ref.~\cite{kawagoe2020microscopic}, which assumes the existence of certain operators for manipulating anyons. Therefore, our first goal is to construct these basic unitary operators from the entanglement bootstrap. We discuss these operators in Section~\ref{subsec:basic_operations}.

\subsection{Basic operations}
\label{subsec:basic_operations}

Consider a disk-like region enclosed by a physical boundary . We shall restrict our attention to operations performed on a fixed interval-like subregion of the boundary. We choose a finite set of locations $\{x_i\}$ along the boundary, where we will define some anyons and operators. These locations must be some sufficiently far apart and we want to consider a boundary region of some sufficiently large total length. We must then choose a sufficiently large number of these locations along the boundary.  For each boundary point $x_i$ we define a disk-like sub-region $A_{[x_i]}$ along the boundary, containing location $x_i$, and disjoint for $i \neq j$.  Moreover for $i>j$ we define disk-like regions $A_{[x_i,x_{i+1},\ldots,x_j]}$ containing locations  $[x_i,\ldots,x_j]$, with disjoint disks for disjoint sub-intervals, and inclusions of disks for inclusions of intervals [Figure~\ref{fig:anyon-positions}].

\begin{figure}[htbp]
    \centering
    \begin{tikzpicture}[scale=0.325]
        \filldraw[blue!30!white] (-10, 0) -- (10, 0) -- ++ (0, 8) -- ++ (-20, 0) -- cycle;
        
    \draw[very thick, fill=green!30!white, dashed] (-8, 0) -- ++ (6, 0) -- ++ (0, 2) -- ++ (-6, 0) -- cycle;

    \draw[very thick, fill=green!30!white, dashed] (0, 0) -- ++ (2, 0) -- ++ (0, 2) -- ++ (-2, 0) -- cycle;
    
        \draw[very thick] (-10, 0) -- (10, 0);

\foreach \x in {-9,-7,...,9}
{
\draw[thick] (\x, -0.25) -- ++ (0, 0.5);
\node[below] () at (\x, -0.25) {$x_{\the\numexpr (\x + 9)/2\relax}$};

}
\node[above] () at (-5, 2) {$A_{[x_1, x_2, x_3]}$};
\node[above] () at (1, 2) {$A_{[x_5]}$};

    \end{tikzpicture}
    \caption{Boundary anyon positions and regions.}
    \label{fig:anyon-positions}
\end{figure}
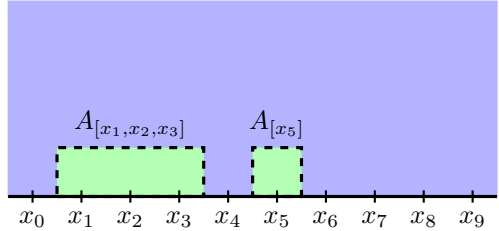

We will be defining three types of operations, called movement, splitting, and fusion. Roughly speaking, the movement operators move a boundary anyon at one location to another. The splitting and fusion operation splits and fuses anyons, respectively. We will prove the existence of these operations from the entanglement bootstrap.

Before we construct these operators, let us clarify what we mean by ``a boundary anyon at location $x_i$.'' Recall that we are considering a disk-like region surrounded by a physical boundary and that we have placed a finite set of points along the boundary. Without loss of generality, consider a state in $\Sigma_{a\overline{a}}^1(M)$ for some $M$-shaped region [Figure~\ref{fig:IC_types_original}(b)], where $a, \overline{a},$ and $1$ are the extreme points associated with $N_1, N_2, N_3\subset M$ [Figure~\ref{fig:IC_types_original}(c)]. Because $N_{a\overline{a}}^1=1$, such a state exists and is unique. In particular, this state is the extreme point of  $\Sigma_{a\overline{a}}^1(M)$. Using \cite[Proposition D.4]{shi2021entanglement}, we can conclude that this state (denoted as $\rho_M$) has the following structure. Let $N_i = N_{i, \text{in}}\cup N_{i, \text{out}}$, where `in' and `out' means the inner and the outer part of the annulus of $N_{i}$. There exists a decomposition of each Hilbert space
\begin{equation}
    \mathcal{H}_{N_{i, \text{in}}} = (\mathcal{H}_{N_{i, \text{in}}^{(1)}} \otimes \mathcal{H}_{N_{i, \text{in}}^{(2)}}) \oplus \mathcal{H}'_{N_i,\text{in}}
\end{equation}
\begin{equation}
    \rho_M = \left(\bigotimes_{i=1}^3\rho_{N_{i,\text{out}} \cup N_{i, \text{in}}^{(1)}}\right) \bigotimes \rho_{M' \cup_{i} N_{i,\text{in}}^{(2)}},
    \label{eq:rdm_structure}
\end{equation}
where $M' = M\setminus (\cup_i N_i)$. From Eq.~\eqref{eq:rdm_structure}, we can see there is a purification of of $\rho_M$ local to $N_1, N_2,$ and $N_3$. First note the reduced density matrix of $\rho_M$ over $N_3$ is equal to that of the reference state. Therefore, for each $N_1, N_2,$ and $N_3$, there is a canonical choice of purification we can take. This will be our canonical purification of $\rho_M$, the state in which the anyon $a$ and $\overline{a}$ are placed in regions enclosed by $N_1$ and $N_2$. 

Based on this purification, we can define the movement operators. Without loss of generality, consider two neighboring points $x_i$ and $x_{i+1}$. These will be the locations at which the boundary anyon $a$ will be placed. Its antiparticle $\overline{a}$ will be placed somewhere else. (Its precise location will not matter for our construction as long as it is neither $x_i$ nor $x_{i+1}$.) Then we can compare the two states, each corresponding to canonical purifications of the state in which $a$ is placed at $x_i$ and $x_{i+1}$, respectively. The two states are indistinguishable in a complement of $A_{[x_1, x_2]}$. By Uhlmann's theorem~\cite{Uhlmann1976}, there exists an isometry localized within $A_{[x_1, x_2]}$ that maps the former to the latter. This is the movement operator, denoted as $M_{x_i\to x_{i+1}}^a$ [Figure~\ref{fig:movement}]. Longer-range movement operators can be decomposed into the elementary movement operators. For instance, for $j>i$,
\begin{equation}
    M_{x_i\to x_j}^a = M_{x_{j-1} \to x_j}^a \cdots M^a_{x_{i+1} \to x_{i+2}} \cdot M_{x_i \to x_{i+1}}^a
\end{equation}
and for $j<i,$
\begin{equation}
    M^a_{x_j\to x_i} = (M_{x_i\to x_j}^a)^{\dagger}.
\end{equation}

\begin{figure}[t]
    \centering
    \begin{tikzpicture}[scale=0.65]
    \filldraw[blue!30!white] (-10, 0) -- (0, 0) -- ++ (0, 4) -- ++ (-10, 0) -- cycle;    
    \draw[fill=green!30!white, very thick, dashed] (-7.75, 0) arc (180:0:0.75);
    \draw[fill=blue!30!white, very thick, dashed] (-7.45, 0) arc (180:0:0.45);
    \draw[very thick] (-10, 0) -- (0, 0);
    \foreach \x in {-9,-7,...,-1}
    {
        \draw[thick] (\x, -0.25) -- ++ (0, 0.5);
        \node[below] () at (\x, -0.25) {$x_{\the\numexpr (\x + 9)/2\relax}$};
    }
    \node[above] () at (-7, 0.75) {$a$};
    \draw[dotted, thick] (-8, 0) -- ++ (0, 2) -- ++ (4, 0) -- ++ (0, -2);
    \draw[->] (0.5, 2) -- (4.5, 2) node [midway, above] {$M_{x_1 \to x_2}^a$};
    \begin{scope}[xshift= 15cm]
    \filldraw[blue!30!white] (-10, 0) -- (0, 0) -- ++ (0, 4) -- ++ (-10, 0) -- cycle;    
    \draw[fill=green!30!white, very thick, dashed] (-5.75, 0) arc (180:0:0.75);
    \draw[fill=blue!30!white, very thick, dashed] (-5.45, 0) arc (180:0:0.45);
    \draw[very thick] (-10, 0) -- (0, 0);
    \foreach \x in {-9,-7,...,-1}
    {
        \draw[thick] (\x, -0.25) -- ++ (0, 0.5);
        \node[below] () at (\x, -0.25) {$x_{\the\numexpr (\x + 9)/2\relax}$};
    }   
    \node[above] () at (-5, 0.75) {$a$};
    \draw[dotted, thick] (-8, 0) -- ++ (0, 2) -- ++ (4, 0) -- ++ (0, -2);
    \end{scope}
    \end{tikzpicture}
    \caption{Movement operator. The reduced density matrices outside the dotted lines are the same. The movement operator is localized in the dotted region.}
    \label{fig:movement}
\end{figure}
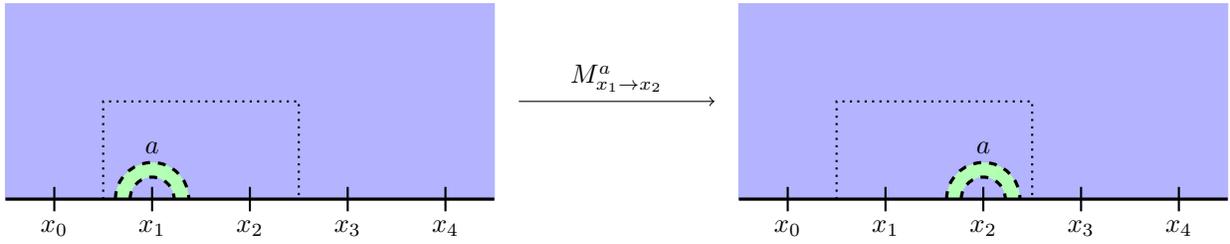

Now we can define the splitting operators. These are operators that splits a single anyon into two anyons. We remark that there can be a nontrivial multiplicity. For instance, if $c$ splits into $a$ and $b$, we end up with the information convex set $\Sigma_{ab}^c(M)$ of some $M$-shaped subsystem $M$ [Figure~\ref{fig:splitting}]. This information convex set contains more than one element if $N_{ab}^c >1$. In particular, one must specify which element of the underlying Hilbert space $\mathbb{V}_{ab}^c$ we get after splitting the anyons. We will specify this state with multiplicity label $\alpha$. 

The element of $\Sigma_{ab}^c(M)$ corresponding to $\alpha$ is an extreme point, and as such, it obeys the structure in Eq.~\eqref{eq:rdm_structure}. Therefore, again we can define a canonical purification that locally purifies $N_1, N_2,$ and $N_3$. Now, similar to how we defined the movement operator, we can consider two different states, the one which is a canonical purification of a state that has an anyon $c$ at location $x_i$ and the canonical purification which has anyon $a$ and $b$ located at $x_i$ and $x_{i+1}$ whose joint charge and fusion multiplicity is $c$ and $\alpha$, respectively [Figure~\ref{fig:splitting}]. (The antiparticle of $c$ should not be placed at $x_{i}$ or $x_{i+1}$ but otherwise its precise location does not matter.) These two states are indistinguishable in the complement of $A_{[x_i, x_{i+1}]}$. Therefore, by Uhlmann's theorem~\cite{Uhlmann1976}, there exists an isometry that maps the former to the latter. This is the splitting operator, denoted as $S_{c\to a,b;\alpha}^{x_i}$. The adjoint of the splitting operator  is denoted as 
\begin{align}
    S^{x_i}_{a,b; \alpha \to c}=  (S^{x_i}_{c \to a,b; \alpha})^\dagger
\end{align}
and it may be thought of as fusing anyons $a$ and $b$ and then projecting onto fusion outcome $c$ with multiplicity label $\alpha$, labeling a state in multiplicity space $\mathbb{V}_{ab}^c$.  We generally take $\alpha$ to be a discrete label running over an orthonormal basis of $\mathbb{V}_{ab}^c$.
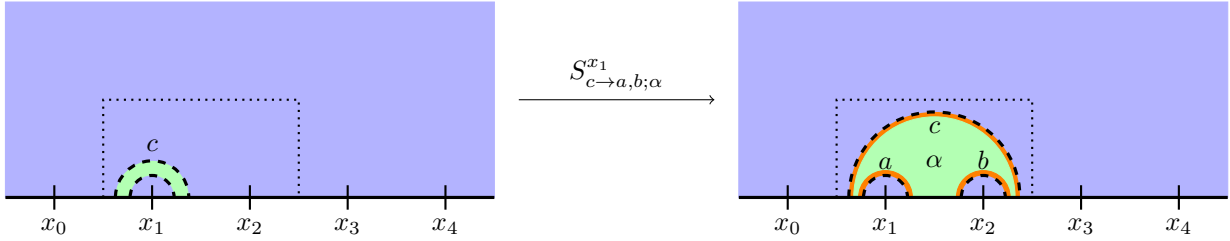
\begin{figure}[t]
    \centering
    
    \begin{tikzpicture}[scale=0.65]
    \filldraw[blue!30!white] (-10, 0) -- (0, 0) -- ++ (0, 4) -- ++ (-10, 0) -- cycle;    
    \draw[fill=green!30!white, very thick, dashed] (-7.75, 0) arc (180:0:0.75);
    \draw[fill=blue!30!white, very thick, dashed] (-7.45, 0) arc (180:0:0.45);
    \draw[very thick] (-10, 0) -- (0, 0);
    \foreach \x in {-9,-7,...,-1}
    {
        \draw[thick] (\x, -0.25) -- ++ (0, 0.5);
        \node[below] () at (\x, -0.25) {$x_{\the\numexpr (\x + 9)/2\relax}$};
    }
    \node[above] () at (-7, 0.75) {$c$};
    \draw[dotted, thick] (-8, 0) -- ++ (0, 2) -- ++ (4, 0) -- ++ (0, -2);
    \draw[->] (0.5, 2) -- (4.5, 2) node [midway, above] {$S_{c \to a, b;\alpha}^{x_1}$};
    \begin{scope}[xshift= 15cm]
    \filldraw[blue!30!white] (-10, 0) -- (0, 0) -- ++ (0, 4) -- ++ (-10, 0) -- cycle;    
    
    \draw[fill=orange, very thick, dashed] (-7.75, 0) arc (180:0:1.75);
    \filldraw[green!30!white] (-7.65, 0) arc (180:0:1.65);
    \filldraw[orange] (-5.55, 0) arc (180:0:0.55);
    \filldraw[orange] (-7.55, 0) arc (180:0:0.55);
    
    \draw[fill=blue!30!white, very thick, dashed] (-5.45, 0) arc (180:0:0.45);
    \draw[fill=blue!30!white, very thick, dashed] (-7.45, 0) arc (180:0:0.45);
    \draw[very thick] (-10, 0) -- (0, 0);
    \foreach \x in {-9,-7,...,-1}
    {
        \draw[thick] (\x, -0.25) -- ++ (0, 0.5);
        \node[below] () at (\x, -0.25) {$x_{\the\numexpr (\x + 9)/2\relax}$};
    }  
    \node[above] () at (-7, 0.4) {$a$};
    \node[above] () at (-5, 0.4) {$b$};
    \node[below] () at (-6, 1.75) {$c$};
    \node[] () at (-6, 0.75) {$\alpha$};
    \draw[dotted, thick] (-8, 0) -- ++ (0, 2) -- ++ (4, 0) -- ++ (0, -2);
    \end{scope}
    \end{tikzpicture}
    \caption{Splitting operator. The reduced density matrices outside the dotted lines are the same. The splitting operator is localized in the dotted region.}
    \label{fig:splitting}
\end{figure}

The splitting operators $ S^{x_i}_{c \to a,b; \alpha}$ are uniquely defined up to a complex phase (separately at each location $x_i$). We now discuss how to fix these phases. For the first location $x_1$, we choose the phases used to define $S^{x_1}_{c \to a,b; \alpha}$ arbitrarily.  For the remaining locations $x_i$ for $i>1$, we choose the phases such that, for region $D_{[x_1,x_{i+1}]}$ with a single anyon $c$ at $x_i$, we have
\begin{align}\label{eq:splitting-op-def}
   S^{x_i}_{c \to a,b; \alpha}=M^a_{x_1 \to x_i} M^b_{x_2 \to x_{i+1}} S^{x_1}_{c \to a,b; \alpha}M^c_{x_i \to x_1}.
\end{align}
This is illustrated diagramatically by Figure \ref{fig:splitting-op-diagram-new}, wherein the movement operators are represented by horizontal lines and the splitting operators (up to a normalization we discuss below) are represented by trivalent vertices and their incident edges.

\begin{figure}[htbp] 
    \centering 
    \begin{tikzpicture}
    \draw[thick] (0, 0) -- (5,0);
    \draw[thick] (1, -0.125cm) -- ++ (0, 0.25cm);
    \draw[thick] (4, -0.125cm) -- ++ (0, 0.25cm);
    \node[below] () at (1, -0.125cm) {$x_1$};
    \node[below] () at (4, -0.125cm) {$x_i$};
    \splitting{4cm}{0.5cm}{$c$}{$a$}{$b$}{$\alpha$}{}{}{};

    \node[] () at (6cm, 0.875cm) {$:=$};
    \begin{scope}[xshift=7cm]
        \draw[thick] (0, 0) -- (5,0);
        \draw[thick] (1, -0.125cm) -- ++ (0, 0.25cm);
        \draw[thick] (4, -0.125cm) -- ++ (0, 0.25cm);
        \node[below] () at (1, -0.125cm) {$x_1$};
        \node[below] () at (4, -0.125cm) {$x_i$};
        \splitting{1cm}{0.5cm}{}{}{}{$\alpha$}{}{}{};
        \movementtwo{4cm}{0.25cm}{-3cm}{0.25cm};
        \movementtwo{1.5cm}{1.5cm}{3cm}{0.25cm};
        \movementtwo{1cm}{1.5cm}{3cm}{0.5cm};
        \node[right] () at (4cm, 0.25cm) {$c$};
        \node[above] () at (4cm, 2cm) {$a$};
        \node[above] () at (4.5cm, 1.75cm) {$b$};
    \end{scope}
    
    \end{tikzpicture}
    \caption{Illustration of Eq.~\eqref{eq:splitting-op-def} in diagrammatic calculus. The trivalent vertex and its incident edges represent a splitting operator. The splitting operator at location $x_i$ is defined in terms of the splitting operator at location $x_1$ and the movement operators (horizontal lines). }
    \label{fig:splitting-op-diagram-new} 
\end{figure}

\subsection{Diagrammatics}
\label{subsec:diagrammatics}

Now we are in a position to define the diagrammatic calculus involving the movement, splitting, and fusion operators. We will use the following convention:
\begin{equation} \label{eq:diagram_conventions}
\begin{tikzpicture}[baseline={([yshift=-0.5ex]current bounding box.center)}]
\movement{0}{0}{1}{1};
\node[above] () at (0.5, 0.25) {$a$};
\end{tikzpicture}
:= M^a_{x_i\to x_j},
\quad 
    \begin{tikzpicture}[baseline={([yshift=-0.5ex]current bounding box.center)}]
        \splitting{0}{0}{$c$}{$a$}{$b$}{$\alpha$}{}{}{};
    \end{tikzpicture}:= \left(\frac{d_a d_b}{d_c} \right)^{\frac{1}{4}}S^{x_i}_{c\to a, b; \alpha},\quad  
    \begin{tikzpicture}[baseline={([yshift=-0.5ex]current bounding box.center)}]
        \merging{0}{0}{$c$}{$a$}{$b$}{$\alpha$}{}{}{};
    \end{tikzpicture}:=
    \left(\frac{d_a d_b}{d_c} \right)^{\frac{1}{4}} \left(S^{x_i}_{c\to a, b; \alpha} \right)^{\dagger},
\end{equation}
where the data about the locations are suppressed in the diagrams. (Unless specified otherwise, this data will be obvious from the given contexts.) In these diagrams ``time'' flows upward, i.e.\ they represent operators composed from bottom to top.
Here $d_a$ is a non-negative number called \emph{quantum dimension}; one way to define it is to ensure the following identity:
\begin{equation}
\label{eq:frobenius_schur}
    \begin{tikzpicture}[baseline={([yshift=-0.5ex]current bounding box.center)}]
        \movement{-1.75cm}{0}{0}{4cm};
        \node[right] () at (-1.75cm, 0) {$a$};
        \node[right] () at (-1.75cm, 4cm) {$a$};
        \draw[very thick] (-1.75cm, 2cm) circle (0.125cm);

        \node[] () at (-0.75cm, 2cm) {$:=$};
        
        \splitting{0}{0}{$a$}{}{}{}{}{dotted}{};
        \movement{0.5cm}{1cm}{0.5cm}{1cm};
        \splitting{0}{1cm}{}{}{}{}{dotted}{}{};
        \merging{0.5cm}{3cm}{}{$\overline{a}$}{}{}{dotted}{}{};
        \movement{0}{2cm}{0}{1cm};
        \merging{0}{4cm}{$a$}{}{}{}{}{}{dotted}{dotted}{}; 

        \node[] () at (1.75cm, 2cm) {$=\quad \varkappa_a$};

        \movement{2.5cm}{0}{0}{4cm};
        \node[right] () at (2.5cm, 0) {$a$};
        \node[right] () at (2.5cm, 4cm) {$a$};
    \end{tikzpicture}
\end{equation}
where the dotted line represents the vacuum sector and $\varkappa_a$ is the \emph{Frobenius-Schur indicator}, a unit complex number.  (Note: The middle of Eq.~\eqref{eq:frobenius_schur} can be decomposed into two splitting and merging operators and a movement operator.) Requiring condition~\eqref{eq:frobenius_schur} uniquely defines the quantities $d_a$ in definition~\eqref{eq:diagram_conventions}. As it stands, the Frobenius-Schur indicator may be a complex number (though by definition here its norm is $1$). We will later choose a gauge  in which this becomes $\pm 1$ [Section~\ref{sec:f-symbols}].  

Let us make a few side remarks. First, the information about the locations will be often immaterial to our analysis; note (i) that two topologically equivalent diagrams with the same endpoints represent the same operation when acting the space of states we discussed so far (potentially up to the Frobenius-Schur indicator) and (ii) that diagrams with different end points can be related to each other by applying the movement operators judiciously. Second, our convention follows that of Ref.~\cite{bonderson2008interferometry} and especially \cite[Section 2.1-2.2]{bonderson2012non}. This leads to diagrammatic rules that produce factors of quantum dimension $d_a$ when popping bubbles or using fusion resolutions of the identity. The benefit of this choice is that the diagrammatic calculus obtains (partial) isotopy invariance: bending lines does not produce factors of $d_a$ (though factors of $\varkappa_a = \pm 1$, the Frobenius-Schur indicator, do still appear).  See Eq.~2.22 of \cite{bonderson2012non} for an illustration. Lastly, due to the (partial) isotopy invariance, it will be often convenient to move the trivalent vertices. All that needs to be done during this process is to keep track of the bendings [Eq.~\eqref{eq:frobenius_schur}]. Therefore, in the rest of the paper, we will not strictly adhere to the precise diagrammatic form in Eq.~\eqref{eq:diagram_conventions}. For instance, the following re-drawing of the splitting and fusion operators should be understood as the elementary splitting and fusion operator in Eq.~\eqref{eq:diagram_conventions}, composed with the movement operators. 
\begin{equation}
     \begin{tikzpicture}[baseline={([yshift=-0.5ex]current bounding box.center)}, line width=1pt, scale=0.8]
     \draw[] (0,0) -- (270:1cm);
     \draw[] (0,0) --  (150:1cm);
     \draw[] (0,0) --  (30:1cm);
     \node[above] () at (0,0) {$\alpha$};
     \node[above] () at (150:1cm) {$a$};
     \node[above] () at (30:1cm) {$b$};
     \node[below] () at (270:1cm) {$c$};
    \end{tikzpicture},
    \begin{tikzpicture}[baseline={([yshift=-0.5ex]current bounding box.center)}, line width=1pt, scale=0.8]
    \draw[] (0,0) -- (90:1cm);
    \draw[] (0,0) -- (210:1cm);
    \draw[] (0,0) -- (-30:1cm);
    \node[below] () at (0,0) {$\alpha$};
    \node[below] () at (210:1cm) {$a$};
    \node[below] () at (-30:1cm) {$b$};
    \node[above] () at (90:1cm) {$c$};
    \end{tikzpicture}.
\end{equation}

Using our diagrammatic rule, it is straightforward to derive the following set of identities. The first is the \emph{completeness relation}.
\begin{equation}
\label{eq:completeness}
    \begin{tikzpicture}[baseline={([yshift=-0.5ex]current bounding box.center)},line width=1pt, scale=0.6]
        \draw[] (-1, -2) -- ++ (0,4);
        \draw[] (1, -2) -- ++ (0,4);
        \node[below] () at (-1, -2) {$a$};
        \node[below] () at (1, -2) {$b$};
        \node[above] () at (-1, 2) {$a$};
        \node[above] () at (1, 2) {$b$};
    \end{tikzpicture}
    =
    \sum_{c, \alpha} \sqrt{\frac{d_c}{d_a d_b}}
    \begin{tikzpicture}[baseline={([yshift=-0.5ex]current bounding box.center)},line width=1pt, scale=0.6]
        \draw[] (-1, -2) -- ++ (0, 1) -- ++ (1, 0.5) -- ++ (1, -0.5) -- ++ (0, -1);
        \draw[] (-1, 2) -- ++ (0, -1) -- ++ (1, -0.5) -- ++ (1, 0.5) -- ++ (0, 1);
        \draw[] (0, -0.5) -- (0, 0.5);
        \node[above] () at (0, 0.5) {$\alpha$};
        \node[below] () at (0, -0.5) {$\alpha$};
        \node[] () at (0.25, 0) {$c$};
        \node[below] () at (-1, -2) {$a$};
        \node[below] () at (1, -2) {$b$};
        \node[above] () at (-1, 2) {$a$};
        \node[above] () at (1, 2) {$b$};
    \end{tikzpicture}
\end{equation}
Secondly, we will choose the multiplicity label for each $(a,b,c)$ to run over an orthonormal bases of $\mathbb{V}_{ab}^c$, so that we have the following orthogonality relation. 
\begin{equation}
\label{eq:orthogonality}
    \begin{tikzpicture}[baseline={([yshift=-0.5ex]current bounding box.center)},line width=1pt, scale=0.6]
    \draw[] (-1, -1) -- ++ (0, 2) -- ++ (1, 0.5) -- ++ (1, -0.5) -- ++(0, -2) -- ++ (-1, -0.5) -- ++ (-1, 0.5);
    \draw[] (0, -1.5) -- ++ (0, -0.5);
    \draw[] (0, 1.5) -- ++ (0, 0.5);

    \node[left] () at (-1, 0) {$a$};
    \node[left] () at (1, 0) {$b$};
    \node[left] () at (0, 2) {$c$};
    \node[left] () at (0, -2) {$c'$};
    \node[below] () at (0, 1.5) {$\alpha$};
    \node[above] () at (0, -1.5) {$\beta$};
    \end{tikzpicture}
    = \begin{tikzpicture}[baseline={([yshift=-0.5ex]current bounding box.center)},line width=1pt, scale=0.6]
    \draw[] (0, -2) -- ++ (0, 4);
    \node[left] () at (0, 0) {$c$};
    \end{tikzpicture} \quad
    \sqrt{\frac{d_a d_b}{d_c}}  \delta_{\alpha, \beta} \delta_{c, c'}
\end{equation}
Lastly, we note the following identity, which we refer to as the \emph{vacuum identity}:
\begin{equation}
    \begin{tikzpicture}[baseline={([yshift=-0.5ex]current bounding box.center)}, line width=1pt, scale=0.5]
        \draw[dashed] (-2, -1) -- ++ (4,0) -- ++ (0, 2) -- ++ (-4, 0) -- cycle;
        \draw[] (-1, 1) -- ++ (0, 2);
        \draw[] (1, 1) -- ++ (0, 2);
        \node[left] () at (-1, 2) {$a$};
        \node[left] () at (1, 2) {$b$};
    \end{tikzpicture}
    =  \begin{tikzpicture}[baseline={([yshift=-0.5ex]current bounding box.center)}, line width=1pt, scale=0.5]
        \draw[dashed] (-2, -1) -- ++ (4,0) -- ++ (0, 2) -- ++ (-4, 0) -- cycle;
        \draw[] (-1, 1) -- ++ (0, 2);
        \draw[] (1, 1) -- ++ (0, 2);
        \node[left] () at (-1, 2) {$a$};
        \node[left] () at (1, 2) {$b$};
    \end{tikzpicture} \quad \delta_{a, \bar{b}},
    \label{eq:vacuum_identity}
\end{equation}
where the dashed region is an arbitrary diagram with no open edges (aside from the ones connected to $a$ and $b$). A similar identity holds for diagrams with the dashed region placed on the top (as opposed to the bottom).

\subsection{$F$-symbols} \label{sec:f-symbols}
We can now define the $F$-symbols in terms of the movement, splitting, and fusion operators. Following our diagrammatic convention [Section~\ref{subsec:diagrammatics}], we define the $F$-symbol as follows:
\begin{equation}
    \begin{tikzpicture}[baseline={([yshift=-0.5ex]current bounding box.center)}]
        \splitting{0}{0}{$d$}{$e$}{$c$}{$\mu$}{}{}{};
        \splitting{0}{1cm}{}{$a$}{$b$}{$\nu$}{}{}{};
        \movement{0.5cm}{1cm}{0.5cm}{1cm};
    \end{tikzpicture}
=
\sum_{f,\alpha,\beta}
F_{def;\mu \nu}^{abc;\alpha\beta}
    \begin{tikzpicture}[baseline={([yshift=-0.5ex]current bounding box.center)}]
        \splitting{0}{0}{$d$}{$a$}{$f$}{$\alpha$}{}{}{};
        \movement{0}{1cm}{0}{1cm};
        \splitting{0.5cm}{1cm}{}{$b$}{$c$}{$\beta$}{}{}{};
    \end{tikzpicture}.
\end{equation}
Alternatively, one may write this (perhaps more sensibly) as $\left(F_d^{abc}\right)_{(e,\mu, \nu), (f, \alpha, \beta)}= F_{def;\mu \nu}^{abc;\alpha\beta}$, which are unitary matrices for each $(a,b,c,d)$ when considered with two joint indices $(e,\mu,\nu), (f,\alpha,\beta)$. These two expressions are just two ways of denoting the same tensor; the first is more compact, while the second emphasizes the unitary matrix. If we have trivial fusion multiplicities, we will simply omit them.

The choice of orthonormal basis for multiplicty Hilbert space $\mathbb{V}_{ab}^c$ is arbitrary.  Under a change of basis for for the multiplicity spaces, the $F$-symbols transform accordingly, with unitary matrices acting on their multiplicity labels.  These are referred to as a gauge transformation of the $F$-symbols.

 Different choices of phases on the movement operators can also be absorbed and expressed as as gauge transformations of $F$.  Ultimately, no matter what choices we make when defining movement/splitting/fusion operators.  An argument appears in essentially the same setting in \cite{kawagoe2020microscopic}. 

A few basic properties of the $F$-symbol will follow from our definitions.  We obtain unitarity because $F$-symbols are inner products of two different orthonorml bases.  We obtain the property $F^{abc;\alpha \beta}_{def;\mu \nu}=1$ when $a$, $b$, or $c$ is 1 from the definition and the topological invariance.  (For abstract $F$-symbols, one can always gauge-fix to obtain this property, but we obtain it automatically here.) We will additionally gauge-fix the Frobenius-Schur indicator $\varkappa_a$ so that $F^{a,\overline{a},a}_a =\varkappa_a = \pm 1$, with $\varkappa_a=1$ if $a \neq \overline{a}$ [Eq.~\eqref{eq:frobenius_schur}]. After showing the existence of suitable movement and splitting operators, our construction of the $F$-symbol is consistent with the one introduced in Ref.~\cite{kawagoe2020microscopic}. Therefore, the conclusions they made about the $F$-symbols apply to ours as well. In particular, one can also show that the $F$-symbols satisfy the pentagon equation.  
~\cite{maclane1963natural, kitaev2006anyons}(The multiplicity labels are omitted for brevity.): 
\begin{equation}
\begin{tikzpicture}[baseline={([yshift=-0.5ex]current bounding box.center)}]
\begin{scope}
    %Tree 1
    \splitting{0}{0}{}{}{}{}{}{}{};
    \splitting{0}{1cm}{}{}{}{}{}{}{};
    \movement{0.5cm}{1cm}{1cm}{2cm};
    \splitting{0}{2cm}{}{}{}{}{}{}{};
    \movement{0.5cm}{2cm}{0.5cm}{1cm};
    \node[above] (a) at (0, 3cm) {$a$};
    \node[above] (b) at (0.5, 3cm) {$b$};
    \node[above] (c) at (1, 3cm) {$c$};
    \node[above] (d) at (1.5, 3cm) {$d$};
    \node[below] (e) at (0, 0cm) {$u$};
    \node[left] (e) at (0, 2cm) {$e$};
    \node[left] (e) at (0, 1cm) {$f$};
\end{scope}

\begin{scope}[xshift=6cm, yshift=2cm]
    %Tree 1
    \splitting{0}{0}{}{}{}{}{}{}{};
    \splitting{0}{1cm}{}{}{}{}{}{}{};
    \movement{0.5cm}{1cm}{0.5cm}{1cm};
    \movement{0}{2cm}{0}{1cm};
    \movement{0.5cm}{2cm}{0}{1cm};
    \splitting{1cm}{2cm}{}{}{}{}{}{}{};
    \node[above] (a) at (0, 3cm) {$a$};
    \node[above] (b) at (0.5, 3cm) {$b$};
    \node[above] (c) at (1, 3cm) {$c$};
    \node[above] (d) at (1.5, 3cm) {$d$};
    \node[below] (e) at (0, 0cm) {$u$};
    \node[left] (e) at (0, 1cm) {$e$};
    \node[right] (g) at (1cm, 2cm) {$g$};
\end{scope}

\begin{scope}[xshift=12cm, yshift=0cm]
    %Tree 1
    \splitting{0}{0}{}{}{}{}{}{}{};
    \movement{0}{1cm}{0}{2cm};
    \splitting{0.5cm}{1cm}{}{}{}{}{}{}{};
    \movement{0.5cm}{2cm}{0}{1cm};
    \splitting{1cm}{2cm}{}{}{}{}{}{}{};
    \node[above] (a) at (0, 3cm) {$a$};
    \node[above] (b) at (0.5, 3cm) {$b$};
    \node[above] (c) at (1, 3cm) {$c$};
    \node[above] (d) at (1.5, 3cm) {$d$};
    \node[below] (e) at (0, 0cm) {$u$};
    \node[right] (g) at (1cm, 2cm) {$g$};
    \node[right] (h) at (0.5cm, 1cm) {$h$};
\end{scope}

\begin{scope}[xshift=4cm, yshift=-3cm]
    %Tree 1
    \splitting{0}{0}{}{}{}{}{}{}{};
    \splitting{0}{1cm}{}{}{}{}{}{}{};
    \splitting{0.5cm}{2cm}{}{}{}{}{}{}{};
    \movement{0}{2cm}{0}{1cm};
    \movement{0.5cm}{1cm}{1cm}{2cm};
    \node[above] (a) at (0, 3cm) {$a$};
    \node[above] (b) at (0.5, 3cm) {$b$};
    \node[above] (c) at (1, 3cm) {$c$};
    \node[above] (d) at (1.5, 3cm) {$d$};
    \node[below] (e) at (0, 0cm) {$u$};
    \node[left] (e) at (0, 1cm) {$f$};
    \node[right] (i) at (0.5cm, 2cm) {$i$};
\end{scope}

\begin{scope}[xshift=8cm, yshift=-3cm]
    %Tree 1
    \splitting{0}{0}{}{}{}{}{}{}{};
    \movement{0}{1cm}{0}{2cm};
    \splitting{0.5cm}{1cm}{}{}{}{}{}{}{};
    \splitting{0.5cm}{2cm}{}{}{}{}{}{}{};
    \movement{1cm}{2cm}{0.5cm}{1cm};
    \node[above] (a) at (0, 3cm) {$a$};
    \node[above] (b) at (0.5, 3cm) {$b$};
    \node[above] (c) at (1, 3cm) {$c$};
    \node[above] (d) at (1.5, 3cm) {$d$};
    \node[below] (e) at (0, 0cm) {$u$};
    \node[right] (i) at (0.5cm, 2cm) {$i$};
    \node[right] (h) at (0.5cm, 1cm) {$h$};
\end{scope}

\draw[thick, ->] (2.5cm, 2cm) -- ++ (2.5cm, 1.5cm) node [pos=0.45, above=0.25cm] {$F_{ufg}^{ecd}$};
\draw[thick, ->] (8.5cm, 3.5cm) -- ++ (2.5cm, -1.5cm) node [pos=0.55, above=0.25cm] {$F_{ueh}^{abg}$};
\draw[thick, ->] (2cm, 1cm) -- ++ (1.5cm, -2cm) node [pos=0.35, right] {$F_{fei}^{abc}$};
\draw[thick, ->] (10cm, -1cm) -- ++ (1.5cm, 2cm) node [pos=0.65, left] {$F_{hig}^{bcd}$};
\draw[thick, ->] (6cm, -1cm) -- ++ (1.5cm, 0) node [midway, above] {$F_{ufh}^{aid}$};
\end{tikzpicture}.
\label{eq:pentagon_equation}
\end{equation}
We have omitted summations above; see Ref.~\cite{kawagoe2020microscopic}.

One difference in our development compared to  Ref.~\cite{kawagoe2020microscopic} is that we have developed a slightly more flexible diagrammatics, where we allow splitting and fusion operations to occur at any point in space along the boundary, defined in a consistent way. Therefore, as discussed above, topologically equivalent diagrams with fixed endpoints represent identical operators. With this in mind, it is slightly simpler to prove the pentagon equations than in the careful explanation of Ref.~\cite{kawagoe2020microscopic}.  Regardless, we have deliberately arranged the proof of Theorem~\ref{thm:mapping-to-sn} so that the $F$-symbols do not directly appear; instead, we make use of the topological invariance of the diagrammatics.

To summarize, we have arrived at a finite set of unitary matrices $\{F_{def; \rho \sigma}^{abc;\mu \nu} \}$ that satisfy the pentagon equation. This is the defining data of the unitary fusion category (UFC). As such, starting from the axioms of the entanglement bootstrap [Figure~\ref{fig:eb_axioms_bulk},~\ref{fig:eb_axioms_boundary}], we deduced the existence of a UFC.

\section{The Levin-Wen string-net model} \label{sec:sn-review}

Above, we described how to extract a unitary fusion category (UFC) from a ground state with gapped boundary.  Meanwhile, a UFC is also precisely the data used to define the Levin-Wen Hamiltonian and the associated ground states, also known as ``string-nets.''  Levin and Wen \cite{levin2005string} originally defined string-nets built on a subset of unitary fusion categories satisfying an extra constraint (``tetrahedral symmetry''), and later the construction was slightly generalized to arbitrary unitary fusion categories \cite{kitaev2012models,lin2021generalized,Hahn2020}.  When we refer to string-nets and the Levin-Wen model, we refer to this generalization.  In a sense this generalization is important for us: when we transform ground states to string-nets, we  naturally find the``generalized string-nets.''

In this Section, we provide a brief introduction to the Levin-Wen model, focusing on facts that are pertinent to our work. The Hilbert space of the Levin-Wen model can be defined over any trivalent graph on any 2D manifold, though  we restrict our attention to the honeycomb lattice on the 2D plane.\footnote{Note this choice of lattice for the definition of string-nets does not constrain the underlying lattice of our reference state.  Ultimately we will transform any suitable reference state, regardless of the underlying ``microscopic'' lattice, into a string-net on a coarse-grained honeycomb lattice, whose coarse-grained sites are still $O(1)$-size.} In the literature, the Hilbert space is often defined with degrees of freedom associated to both edges and vertices of the graph.  (For categories without fusion multiplicity, the vertex degrees of freedom are not necessary.)  In that convention, edges are assigned Hilbert space $\textrm{span}\{|a\rangle\}_a$ for anyon types (simple objects) $a \in \mathcal{C}$, and vertices are assigned Hilbert space $\Hom{1}{a \otimes b \otimes c}$, conditional on the states $|a\rangle, |b\rangle,|c\rangle$ of the incoming edges.

Here we will adopt a slightly different convention in which degrees of freedom are associated only to vertices.  Each vertex has associated Hilbert space
\begin{align}
\label{eq:hilb_vertex}
\mathcal{H}_v = \bigoplus_{a,b,c} \Hom{1}{a \otimes b \otimes c} 
\end{align}
with direct sum over simple objects $a,b,c$. The anyons $a,b,c$ are associated with the three edges attached to the vertex, arranged e.g.\ in counterclockwise fashion. At each vertex, a \emph{branching rule} must be satisfied. Specifically, the following configurations are allowed
\begin{equation}
    \begin{tikzpicture}[baseline={([yshift=-0.5ex]current bounding box.center)}]
        \draw[thick] (0,0) -- (270:1cm) node [midway, right] {$c$};
        \draw[thick] (0,0) -- (30:1cm) node [midway, above] {$b$};
        \draw[thick] (0,0) -- (150:1cm) node [midway, above] {$a$};
        \node[above] () at (0,0) {$\alpha$};

        \node[] () at (2cm, -0.5cm) {and};
        
        \begin{scope}[xshift=4cm, yshift=-0.5cm]
        \draw[thick] (0,0) -- (90:1cm) node [midway, right] {$c$};
        \draw[thick] (0,0) -- (-30:1cm) node [midway, above] {$b$};
        \draw[thick] (0,0) -- (210:1cm) node [midway, above] {$a$};
        \node[below] () at (0,0) {$\alpha$};
        \end{scope}
    \end{tikzpicture},
\end{equation}
only for certain choices of $a, b,$ and $c$. (Here $\alpha\in \Hom{1}{a\times b\times c}$.) We will denote this rule in terms of $\delta_{ab}^c$, which is $1$ if the configuration is allowed and $0$ otherwise.

Our convention can be related to the more standard convention by a simple reorganization of degrees of freedom, using a depth-1 unitary circuit.  Essentially, each edge may be split into two half-edges, and then each vertex degree of freedom can absorb its adjacent half-edges; for each edge, one can apply an isometry that implements $|a\rangle \to |a\rangle |a\rangle$ for any $a\in \mathcal{C}$, and then the two copies of $|a\rangle$ can be absorbed into each vertex. At this point, the entire string-net Hilbert space is the tensor product product over vertices,
\begin{align}
\mathcal{H} = \bigotimes_v \mathcal{H}_v = \bigotimes_v \bigoplus_{a,b,c} \Hom{1}{a \otimes b \otimes c} 
\end{align}
At times we may restrict to the subspace of the Hilbert space with ``stable edge labelings,'' i.e.\ when distributing the above tensor product of direct sums, we restrict to terms where every pair of adjacent vertices associates same anyon type to their shared leg (when viewed incoming to one vertex and outgoing from the other).  Sometimes this subspace is referred to as the ``string-net Hilbert space,'' and it can be organized as 
 \begin{align} \label{eq:sn_space}
\mathcal{H}_0 = \bigoplus_{\{a_i\}} \bigotimes_v \Hom{1}{a_{v_1} \otimes a_{v_2} \otimes a_{v_1}} 
\end{align}
where the sum is over all assignments of anyon types $a_i$ to edges $e_i$, and $a_{v_1},a_{v_2},a_{v_3}$ are the anyon types assigned to the three incoming edges of $v$.

\subsection{Ground state}
\label{sec:ground_state}

The string-net ground space is a subspace of the string-net Hilbert space $\mathcal{H}_0$, whose input data are the $F$-symbols of a UFC [Section~\ref{sec:f-symbols}]. Let us again choose the manifold to be the plane, in which case the ground space is one-dimensional. Then the ground state wavefunction of the string-net model is defined in terms of the amplitudes assigned to the basis states in Eq.~\eqref{eq:sn_space}. The graphical rule that assigns these amplitudes are well-known; see Ref.~\cite{lin2021generalized} for example. 

Fortunately, the diagrammatic rules defining string-nets are identical to the diagrammatic rules we developed in Section \ref{subsec:diagrammatics}.  (In some references, such as Ref.~\cite{lin2021generalized}, there is a quantity denoted $Y_c^{a,b}$ that allows additional gauge freedom. The freedom can be fixed by choosing $Y_c^{a,b}=\sqrt{(d_ad_b)/d_c}$.  We make this choice for consistency with our Eq.~\eqref{eq:diagram_conventions}.)

The ground state is defined by the following constraints:
\begin{equation}
    Q_I |\psi\rangle = |\psi\rangle, \quad B_p^s|\psi\rangle = d_s|\psi\rangle,
\end{equation}
where $Q_I$ and $B_p^s$ are the vertex and the plaquette operators whose supports are shown in Figure~\ref{fig:string_net}.  The factor of $d_s$ is associated with the equation
\begin{equation}
    d_s  = 
    \begin{tikzpicture}[baseline={([yshift=-0.5ex]current bounding box.center)}, scale=0.5]
        \draw[thick] (0,0) -- ++ (-1, 1) -- ++ (1, 1) node [midway, above] {$s$};
        \draw[thick] (0,0) -- ++ (1, 1) -- ++ (-1, 1) node [midway, above] {$\bar{s}$};
    \end{tikzpicture}.
    \label{eq:bubble-pop}
\end{equation}
Note that bivalent vertices are trivalent vertices in disguise, wherein the edge corresponding to the vacuum sector is suppressed; hence Eqe.~\eqref{eq:bubble-pop} is consistent with Eq.~\eqref{eq:diagram_conventions}.

The operators $Q_I$ and $B_p^s$ can be described graphically as follows:
\begin{equation}
    Q_I\,
\begin{tikzpicture}[baseline={([yshift=-0.5ex]current bounding box.center)},scale=0.7]
        \draw[thick] (-1, -1.125) -- ++ (0,2.25);
        \draw[thick] (0,0) -- (270:1cm) node [midway, right] {$c$};
        \draw[thick] (0,0) -- (30:1cm) node [midway, above] {$b$};
        \draw[thick] (0,0) -- (150:1cm) node [midway, above] {$a$};
        \node[above] () at (0,0) {$\alpha$};
        \draw[thick] (1,-1) -- ++ (0.25, 1) -- ++ (-0.25, 1);
        \end{tikzpicture} = \delta_{ab}^c \,
        \begin{tikzpicture}[baseline={([yshift=-0.5ex]current bounding box.center)},scale=0.7]
        \draw[thick] (-1, -1.125) -- ++ (0,2.25);
        \draw[thick] (0,0) -- (270:1cm) node [midway, right] {$c$};
        \draw[thick] (0,0) -- (30:1cm) node [midway, above] {$b$};
        \draw[thick] (0,0) -- (150:1cm) node [midway, above] {$a$};
        \node[above] () at (0,0) {$\alpha$};
        \draw[thick] (1,-1) -- ++ (0.25, 1) -- ++ (-0.25, 1);
        \end{tikzpicture}, \quad \quad 
        Q_I\,
\begin{tikzpicture}[baseline={([yshift=-0.5ex]current bounding box.center)},scale=0.7]
        \draw[thick] (-1, -1.125) -- ++ (0,2.25);
        \draw[thick] (0,0) -- (90:1cm) node [midway, right] {$c$};
        \draw[thick] (0,0) -- (-30:1cm) node [midway, above] {$b$};
        \draw[thick] (0,0) -- (210:1cm) node [midway, above] {$a$};
        \node[below] () at (0,0) {$\alpha$};
        \draw[thick] (1,-1) -- ++ (0.25, 1) -- ++ (-0.25, 1);
        \end{tikzpicture} = \delta_{ab}^c\,
        \begin{tikzpicture}[baseline={([yshift=-0.5ex]current bounding box.center)},scale=0.7]
        \draw[thick] (-1, -1.125) -- ++ (0,2.25);
        \draw[thick] (0,0) -- (90:1cm) node [midway, right] {$c$};
        \draw[thick] (0,0) -- (-30:1cm) node [midway, above] {$b$};
        \draw[thick] (0,0) -- (210:1cm) node [midway, above] {$a$};
        \node[below] () at (0,0) {$\alpha$};
        \draw[thick] (1,-1) -- ++ (0.25, 1) -- ++ (-0.25, 1);
        \end{tikzpicture},
        \label{eq:vertex}
\end{equation}
\begin{equation}
      B_p^s \, \begin{tikzpicture}[baseline={([yshift=-0.5ex]current bounding box.center)},scale=0.7]
    \draw[thick] (2, -1.75)  -- ++ (0.5, 1.75) -- ++ (-0.5, 1.75);
    \draw[thick] (-2.125, -1.75) -- ++ (0, 3.5);
    \draw[thick] (30:1cm) -- (90:1cm) node[midway, above] {$i_1$};
    \draw[thick] (90:1cm) -- (150:1cm) node[midway, above] {$i_6$};
    \draw[thick] (150:1cm) -- (210:1cm) node[midway, left] {$i_5$};
    \draw[thick] (210:1cm) -- (270:1cm) node[midway, below] {$i_4$};
    \draw[thick] (270:1cm) -- (330:1cm) node[midway, below] {$i_3$};
    \draw[thick] (330:1cm) -- (30:1cm) node[midway, right] {$i_2$};

    \draw[thick] (30:1cm) -- (30:1.5cm) node[right] {$e_2$};
    \draw[thick] (90:1cm) -- (90:1.5cm) node[above] {$e_1$};
    \draw[thick] (150:1cm) -- (150:1.5cm) node[left] {$e_6$};
    \draw[thick] (210:1cm) -- (210:1.5cm) node[left] {$e_5$};
    \draw[thick] (270:1cm) -- (270:1.5cm) node[below] {$e_4$};
    \draw[thick] (330:1cm) -- (330:1.5cm) node[right] {$e_3$};

    \end{tikzpicture} \, =\, 
    \begin{tikzpicture}[baseline={([yshift=-0.5ex]current bounding box.center)},scale=0.7]
    \draw[thick] (2, -1.75)  -- ++ (0.5, 1.75) -- ++ (-0.5, 1.75);
    \draw[thick] (-2.125, -1.75) -- ++ (0, 3.5);
    \draw[thick] (30:1cm) -- (90:1cm) node[midway, above] {$i_1$};
    \draw[thick] (90:1cm) -- (150:1cm) node[midway, above] {$i_6$};
    \draw[thick] (150:1cm) -- (210:1cm) node[midway, left] {$i_5$};
    \draw[thick] (210:1cm) -- (270:1cm) node[midway, below] {$i_4$};
    \draw[thick] (270:1cm) -- (330:1cm) node[midway, below] {$i_3$};
    \draw[thick] (330:1cm) -- (30:1cm) node[midway, right] {$i_2$};

    \draw[thick] (30:1cm) -- (30:1.5cm) node[right] {$e_2$};
    \draw[thick] (90:1cm) -- (90:1.5cm) node[above] {$e_1$};
    \draw[thick] (150:1cm) -- (150:1.5cm) node[left] {$e_6$};
    \draw[thick] (210:1cm) -- (210:1.5cm) node[left] {$e_5$};
    \draw[thick] (270:1cm) -- (270:1.5cm) node[below] {$e_4$};
    \draw[thick] (330:1cm) -- (330:1.5cm) node[right] {$e_3$};

    \draw[red!70!white, thick] (30:0.8cm) -- (90:0.8cm) -- (150:0.8cm) -- (210:0.8cm) -- (270:0.8cm) -- (330:0.8cm) -- cycle;
    \node[red] () at (-0.5cm, 0)  {$s$};
    \node[red] () at (0.5cm, 0)  {$\bar{s}$};
    \end{tikzpicture},
    \label{eq:plaquette}
\end{equation}
where the effect of the insertion of the red line can be computed using the $F$-symbols; see ~\cite[Eq.(18)]{Hahn2020} and~\cite[Eq.(33)]{lin2021generalized}. (For the plaquette operator, we suppress the multiplicity labels for brevity.) We will discuss the explicit matrix element of the plaquette operator in Section~\ref{subsec:sn_matrix_element}. 

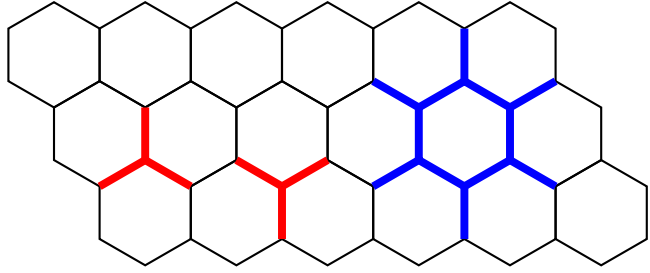
\begin{figure}[t]
    \centering
    \begin{tikzpicture}[scale=0.7]
    \foreach \x in {1, ..., 6} {
    \begin{scope}[xshift= \x* 1.732cm]
    \draw[thick] (30:1cm) -- (90:1cm) -- (150:1cm) -- (210:1cm) -- (270:1cm) -- (330:1cm) -- cycle;
    \end{scope}
    }
    \foreach \x in {1, ..., 6} {
    \begin{scope}[xshift= \x* 1.732cm - 0.866cm, yshift = 1.5cm]
    \draw[thick] (30:1cm) -- (90:1cm) -- (150:1cm) -- (210:1cm) -- (270:1cm) -- (330:1cm) -- cycle;
    \end{scope}
    }
    \foreach \x in {1, ..., 6} {
    \begin{scope}[xshift= (\x-1)* 1.732cm, yshift=3cm]
    \draw[thick] (30:1cm) -- (90:1cm) -- (150:1cm) -- (210:1cm) -- (270:1cm) -- (330:1cm) -- cycle;
    \end{scope}
    }
    \begin{scope}[xshift=1.732cm]
        \draw[red, line width=3pt] (30:1cm) -- (90:1cm);
        \draw[red, line width=3pt] (90:1cm) -- (150:1cm);
        \draw[red, line width=3pt] (90:1cm) --++ (0, 1cm);
    \end{scope}
    \begin{scope}[xshift=4.33cm, yshift=1.5cm]
        \draw[red, line width=3pt] (210:1cm) -- (270:1cm);
        \draw[red, line width=3pt] (270:1cm) -- (330:1cm);
        \draw[red, line width=3pt] (270:1cm) --++ (0, -1cm);
    \end{scope}
    \begin{scope}[xshift=7.794cm, yshift=1.5cm]
        \draw[blue, line width=3pt] (30:1cm) -- (90:1cm) -- (150:1cm) -- (210:1cm) -- (270:1cm) -- (330:1cm) -- cycle;
        \draw[blue, line width=3pt] (30:1cm) --++ (30:1cm);
        \draw[blue, line width=3pt] (90:1cm) --++ (90:1cm);
        \draw[blue, line width=3pt] (150:1cm) --++ (150:1cm);
        \draw[blue, line width=3pt] (210:1cm) --++ (210:1cm);
        \draw[blue, line width=3pt] (270:1cm) --++ (270:1cm);
        \draw[blue, line width=3pt] (330:1cm) --++ (330:1cm);
    \end{scope}
    \end{tikzpicture}
    \caption{Support of $Q_I$ (red) and $B_p^s$ (blue).}
    \label{fig:string_net}
\end{figure}

Note that $B_p^s$ is generally not Hermitian. However, one can define a Hermitian projection operator $B_p := \mathcal{D}^{-2} \sum_s d_s B_p^s $. Then the the following Hamiltonian is a parent Hamiltonian of the string-net wavefunction:
\begin{equation}
    H = -\sum_I Q_I - \sum_p B_p.
    \label{eq:string_net_ham}
\end{equation}
It should be clear that $\mathcal{H}_0$ already satisfies the vertex constraints $Q_I$ [Eq.~\eqref{eq:vertex}].

\subsection{Plaquette matrix element}
\label{subsec:sn_matrix_element}
Here we discuss the plaquette matrix element of the string-net model. The main goal of this Section is to recast Eq.~\eqref{eq:plaquette} in a form that can be identified with the matrix elements we derive later in Section~\ref{subsec:plaquette}. The string-net plaquette operator matrix element is  
\begin{equation}
\begin{aligned}
    B_p^s 
    \begin{tikzpicture}[baseline={([yshift=-0.5ex]current bounding box.center)},scale=0.8]
    \draw[thick] (-2.125, -1.75) -- ++ (0,3.5);

    \draw[thick] (30:1cm) -- (90:1cm) node[midway, above] {\footnotesize $i_1$};
    \draw[thick] (90:1cm) -- (150:1cm) node[midway, above] {\footnotesize $i_6$};
    \draw[thick] (150:1cm) -- (210:1cm) node[midway, left] {\footnotesize $i_5$};
    \draw[thick] (210:1cm) -- (270:1cm) node[midway, below] {\footnotesize $i_4$};
    \draw[thick] (270:1cm) -- (330:1cm) node[midway, below] {\footnotesize $i_3$};
    \draw[thick] (330:1cm) -- (30:1cm) node[midway, right] {\footnotesize $i_2$};

    \node[] () at (30:0.7cm) {\footnotesize $\alpha_2$};
    \node[] () at (-30:0.7cm) {\footnotesize $\alpha_3$};
    \node[] () at (-90:0.7cm) {\footnotesize $\alpha_4$};
    \node[] () at (-150:0.7cm) {\footnotesize $\alpha_5$};
    \node[] () at (-210:0.7cm) {\footnotesize $\alpha_6$};
    \node[] () at (-270:0.7cm) {\footnotesize $\alpha_1$};

    \draw[thick] (30:1cm) -- (30:1.5cm) node[right] {\footnotesize $e_2$};
    \draw[thick] (90:1cm) -- (90:1.5cm) node[above] {\footnotesize $e_1$};
    \draw[thick] (150:1cm) -- (150:1.5cm) node[left] {\footnotesize $e_6$};
    \draw[thick] (210:1cm) -- (210:1.5cm) node[left] {\footnotesize $e_5$};
    \draw[thick] (270:1cm) -- (270:1.5cm) node[below] {\footnotesize $e_4$};
    \draw[thick] (330:1cm) -- (330:1.5cm) node[right] {\footnotesize $e_3$};

    \draw[thick] (2, -1.75) -- ++ (0.5, 1.75) -- ++ (-0.5, 1.75);
    \end{tikzpicture} &=  
    \begin{tikzpicture}[baseline={([yshift=-0.5ex]current bounding box.center)},scale=0.8]
    \draw[thick] (-2.125, -1.75) -- ++ (0,3.5);
 \draw[thick] (30:1cm) -- (90:1cm) node[midway, above] {\footnotesize $i_1$};
    \draw[thick] (90:1cm) -- (150:1cm) node[midway, above] {\footnotesize $i_6$};
    \draw[thick] (150:1cm) -- (210:1cm) node[midway, left] {\footnotesize $i_5$};
    \draw[thick] (210:1cm) -- (270:1cm) node[midway, below] {\footnotesize $i_4$};
    \draw[thick] (270:1cm) -- (330:1cm) node[midway, below] {\footnotesize $i_3$};
    \draw[thick] (330:1cm) -- (30:1cm) node[midway, right] {\footnotesize $i_2$};

    \node[] () at (30:0.7cm) {\footnotesize $\alpha_2$};
    \node[] () at (-30:0.7cm) {\footnotesize $\alpha_3$};
    \node[] () at (-90:0.7cm) {\footnotesize $\alpha_4$};
    \node[] () at (-150:0.7cm) {\footnotesize $\alpha_5$};
    \node[] () at (-210:0.7cm) {\footnotesize $\alpha_6$};
    \node[] () at (-270:0.7cm) {\footnotesize $\alpha_1$};
    
    \draw[thick] (30:1cm) -- (30:1.5cm) node[right] {\footnotesize $e_2$};
    \draw[thick] (90:1cm) -- (90:1.5cm) node[above] {\footnotesize $e_1$};
    \draw[thick] (150:1cm) -- (150:1.5cm) node[left] {\footnotesize $e_6$};
    \draw[thick] (210:1cm) -- (210:1.5cm) node[left] {\footnotesize $e_5$};
    \draw[thick] (270:1cm) -- (270:1.5cm) node[below] {\footnotesize $e_4$};
    \draw[thick] (330:1cm) -- (330:1.5cm) node[right] {\footnotesize $e_3$};

    \draw[red!70!white, thick] (30:0.45cm) -- (90:0.45cm) -- (150:0.45cm) -- (210:0.45cm) -- (270:0.45cm) -- (330:0.45cm) -- cycle;
    \node[red] () at (-0.25cm, 0)  {\footnotesize $s$};
    \node[red] () at (0.25cm, 0)  {\footnotesize $\bar{s}$};

    \draw[thick] (2, -1.75) -- ++ (0.5, 1.75) -- ++ (-0.5, 1.75);
    \end{tikzpicture} \\
    &= 
    \sum_{\substack{\gamma_1,\ldots, \gamma_6 \\ i_1', \ldots, i_6'}} \left(\prod_{k=1}^6 \sqrt{\frac{d_{i_k}}{d_s d_{i_k'}}} \right)
        \begin{tikzpicture}[baseline={([yshift=-0.5ex]current bounding box.center)},scale=0.8]
   
    \draw[thick] (-2.125, -1.75) -- ++ (0,3.5);

    \draw[thick] (30:1cm) -- (90:1cm) node[midway, above] {\footnotesize $i_1'$};
    \draw[thick] (90:1cm) -- (150:1cm) node[midway, above] {\footnotesize $i_6'$};
    \draw[thick] (150:1cm) -- (210:1cm) node[midway, left] {\footnotesize $i_5'$};
    \draw[thick] (210:1cm) -- (270:1cm) node[midway, below] {\footnotesize $i_4'$};
    \draw[thick] (270:1cm) -- (330:1cm) node[midway, below] {\footnotesize $i_3'$};
    \draw[thick] (330:1cm) -- (30:1cm) node[midway, right] {\footnotesize $i_2'$};

    \draw[thick] (30:1cm) -- (30:1.5cm) node[right] {\footnotesize $e_2$};
    \draw[thick] (90:1cm) -- (90:1.5cm) node[above] {\footnotesize $e_1$};
    \draw[thick] (150:1cm) -- (150:1.5cm) node[left] {\footnotesize $e_6$};
    \draw[thick] (210:1cm) -- (210:1.5cm) node[left] {\footnotesize $e_5$};
    \draw[thick] (270:1cm) -- (270:1.5cm) node[below] {\footnotesize $e_4$};
    \draw[thick] (330:1cm) -- (330:1.5cm) node[right] {\footnotesize $e_3$};

    \draw[red, thick] (50:0.875cm) -- (10:0.875cm);
    \draw[red, thick] (110:0.875cm) -- (90:0.9cm) -- (70:0.875cm);
    \draw[red, thick] (170:0.875cm) -- (130:0.875cm);
    \draw[red, thick] (230:0.875cm) -- (190:0.875cm);
    \draw[red, thick] (290:0.875cm) -- (270:0.9cm) --(250:0.875cm);
    \draw[red, thick] (350:0.875cm) -- (310:0.875cm);

    \node[] () at (70:0.7cm) {\footnotesize \color{red} $\bar{s}$};
    \node[] () at (30:0.7cm) {\footnotesize \color{red} $\bar{s}$};    
    \node[] () at (-30:0.7cm) {\footnotesize \color{red} $\bar{s}$};
    \node[] () at (-70:0.7cm) {\footnotesize \color{red} $\bar{s}$};

    \node[] () at (110:0.7cm) {\footnotesize \color{red} $s$};
    \node[] () at (150:0.7cm) {\footnotesize \color{red} $s$};    
    \node[] () at (210:0.7cm) {\footnotesize \color{red} $s$};
    \node[] () at (-110:0.7cm) {\footnotesize \color{red} $s$};

    \draw[thick] (2, -1.75) -- ++ (0.5, 1.75) -- ++ (-0.5, 1.75);
    \end{tikzpicture} 
\end{aligned}
\label{eq:sn_plaquette_action}
\end{equation}
wherein the last line we suppressed the indices $i_1,\ldots, i_6$ and $\alpha_1, \ldots, \alpha_6$. Here $\gamma_1, \ldots ,\gamma_6$ are the multiplicity labels of the vertices touching the red edges, i.e., the vertices incident on the edges $i_1',\ldots i_6'$ and $s$ (or $\bar{s}$) string~\cite{Hahn2020,lin2021generalized}. 

It remains to explain what the second line of Eq.~\eqref{eq:sn_plaquette_action} means. Using the standard string-net identities, the diagrams with the additional red lines can be converted to diagrams without the red lines, by applying a sequence of moves related to the $F$-symbols [Section~\ref{sec:f-symbols}]. Such diagrammatic moves are well-studied; see Ref.~\cite{Hahn2020,lin2021generalized} for the detailed exposition. We shall simply refer the result of these diagrammatic moves as $b_1, \ldots, b_6$. They act separately on each vertex, in the following way:
\begin{equation}
\begin{aligned}
    &\begin{tikzpicture}[baseline={([yshift=-0.5ex]current bounding box.center)},line width=1pt, scale=1.2, yscale=-1]
        \draw[] (0,0) -- (30:1cm);
        \draw[] (0,0) -- (270:1cm);
        \draw[] (0,0) -- (150:1cm);
        \draw[red] (150:0.9cm) -- (0, 0.25cm) -- (30:0.9cm);
        \node[above] () at (150:0.9cm) {\footnotesize $\gamma_6$};
        \node[above] () at (30:0.9cm) {\footnotesize $\gamma_1$};
        \node[] () at (0.25,-0.125) {\footnotesize $\alpha_1$};
        \node[] () at (-0.25cm, 0.5cm) {\footnotesize \color{red} $s$};
        \node[] () at (0.25cm, 0.5cm) {\footnotesize \color{red} $\bar{s}$};
        \node[] () at (-0.5cm, 0.1cm) {\footnotesize $i_6$};
        \node[] () at (0.5cm, 0.1cm) {\footnotesize $i_1$};
        \node[above] () at (0, -1cm) {\footnotesize $e_1$};
        \node[below] () at (30:1cm) {\footnotesize $i_1'$};
        \node[below] () at (150:1cm) {\footnotesize $i_6'$};
    \end{tikzpicture}
    = \sum_{\alpha_1'} (b_{1})_{\alpha_1 \alpha_1'}
    \begin{tikzpicture}[baseline={([yshift=-0.5ex]current bounding box.center)},line width=1pt, scale=1.2, yscale=-1]
        \draw[] (0,0) -- (30:1cm);
        \draw[] (0,0) -- (270:1cm);
        \draw[] (0,0) -- (150:1cm);
        \node[] () at (0.25,-0.125) {\footnotesize $\alpha_1'$};
        \node[above] () at (0, -1cm) {\footnotesize $e_1$};
        \node[below] () at (30:1cm) {\footnotesize $i_1'$};
        \node[below] () at (150:1cm) {\footnotesize $i_6'$};
    \end{tikzpicture}, \quad 
    &\begin{tikzpicture}[baseline={([yshift=-0.5ex]current bounding box.center)},line width=1pt, scale=1.2, yscale=-1]
        \draw[] (0,0) -- (90:1cm);
        \draw[] (0,0) -- (-30:1cm);
        \draw[] (0,0) -- (210:1cm);
        \draw[red] (90:0.9cm) -- (210:0.9cm);
        \node[] () at (150:0.6cm) {\footnotesize {\color{red} $\bar{s}$}};
        \node[left] () at (0, 0.9cm) {\footnotesize $\gamma_2$};
        \node[] () at (190:0.8cm) {\footnotesize $\gamma_1$};
        \node[right] () at (90:0.5cm) {\footnotesize $i_2$};
        \node[above] () at (210:0.5cm) {\footnotesize $i_5$};
        \node[above] () at  (-30:1cm) {\footnotesize $e_2$};
        \node[above] () at (210:1cm) {\footnotesize $i_5'$};
        \node[below] () at (90:1cm) {\footnotesize $i_2'$};
        \node[above] () at (0,0) {\footnotesize $\alpha_2$};
    \end{tikzpicture}
    =\sum_{\alpha_2'} (b_2)_{\alpha_2 \alpha_2'}
    \begin{tikzpicture}[baseline={([yshift=-0.5ex]current bounding box.center)},line width=1pt, scale=1.2, yscale=-1]
        \draw[] (0,0) -- (90:1cm);
        \draw[] (0,0) -- (-30:1cm);
        \draw[] (0,0) -- (210:1cm);
        \node[above] () at  (-30:1cm) {\footnotesize $e_2$};
        \node[above] () at (210:1cm) {\footnotesize $i_5'$};
        \node[below] () at (90:1cm) {\footnotesize $i_2'$};
        \node[above] () at (0,0) {\footnotesize $\alpha_2'$};
    \end{tikzpicture}, \\
     &\begin{tikzpicture}[baseline={([yshift=-0.5ex]current bounding box.center)},line width=1pt, scale=1.2, yscale=-1]
        \draw[] (0,0) -- (30:1cm);
        \draw[] (0,0) -- (270:1cm);
        \draw[] (0,0) -- (150:1cm);
        \draw[red] (270:0.9cm) -- (150:0.9cm);
        \node[] () at (210:0.7cm) {\footnotesize {\color{red} $\bar{s}$}};
        \node[] () at (0.25,-0.125) {\footnotesize $\alpha_3$};
        \node[] () at (-0.3cm, 0cm) {\footnotesize $i_3$};
        \node[right] () at (0, -0.5cm) {\footnotesize  $i_2$};
        \node[] () at (170:0.8cm) {\footnotesize $\gamma_3$};
        \node[left] () at (270:0.9cm) {\footnotesize $\gamma_2$};
        \node[above] () at (270:1cm) {\footnotesize $i_2'$};
        \node[below] () at (30:1cm) {\footnotesize $e_3$};
        \node[below] () at (150:1cm) {\footnotesize $i_3'$};
    \end{tikzpicture}
    = \sum_{\alpha_3'} (b_3)_{\alpha_3 \alpha_3'}
    \begin{tikzpicture}[baseline={([yshift=-0.5ex]current bounding box.center)},line width=1pt, scale=1.2, yscale=-1]
        \draw[] (0,0) -- (30:1cm);
        \draw[] (0,0) -- (270:1cm);
        \draw[] (0,0) -- (150:1cm);
        \node[] () at (0.25,-0.125) {\footnotesize $\alpha_3'$};
        \node[above] () at (270:1cm) {\footnotesize $i_2'$};
        \node[below] () at (30:1cm) {\footnotesize $e_3$};
        \node[below] () at (150:1cm) {\footnotesize $i_3'$};
    \end{tikzpicture}, \quad
    &
    \begin{tikzpicture}[baseline={([yshift=-0.5ex]current bounding box.center)},line width=1pt, scale=1.2, yscale=-1]
        \draw[] (0,0) -- (90:1cm);
        \draw[] (0,0) -- (-30:1cm);
        \draw[] (0,0) -- (210:1cm);
        
        \draw[red] (210:0.9cm) -- (270:0.25cm) -- (330:0.9cm);
        \node[below] () at (210:0.9cm) {\footnotesize $\gamma_4$};
        \node[below] () at (330:0.9cm) {\footnotesize $\gamma_3$};
        \node[] () at (-0.25cm, -0.5cm) {\footnotesize \color{red} $s$};
        \node[] () at (0.25cm, -0.5cm) {\footnotesize \color{red} $\bar{s}$};
        \node[above] () at (-30:1cm) {\footnotesize $i_3'$};
        \node[above] () at (210:1cm) {\footnotesize $i_4'$};
        \node[below] () at (90:1cm) {\footnotesize $e_4$};
        \node[] () at  (-0.2cm, 0.125cm) {\footnotesize $\alpha_4$};
        \node[below] () at (-30:0.5cm) {\footnotesize $i_3$};
        \node[below] () at (210:0.5cm) {\footnotesize $i_4$};
    \end{tikzpicture}
    =
    \sum_{\alpha_4'} (b_4)_{\alpha_4 \alpha_4'}
    \begin{tikzpicture}[baseline={([yshift=-0.5ex]current bounding box.center)},line width=1pt, scale=1.2, yscale=-1]
        \draw[] (0,0) -- (90:1cm);
        \draw[] (0,0) -- (-30:1cm);
        \draw[] (0,0) -- (210:1cm);
        \node[above] () at (-30:1cm) {\footnotesize $i_3'$};
        \node[above] () at (210:1cm) {\footnotesize $i_4'$};
        \node[below] () at (90:1cm) {\footnotesize $e_4$};
        \node[above] () at  (0, 0) {\footnotesize $\alpha_4'$};
    \end{tikzpicture},\\
    &\begin{tikzpicture}[baseline={([yshift=-0.5ex]current bounding box.center)},line width=1pt, scale=1.2, yscale=-1]
        \draw[] (0,0) -- (30:1cm);
        \draw[] (0,0) -- (270:1cm);
        \draw[] (0,0) -- (150:1cm);
        \node[] () at (0.25,-0.125) {\footnotesize $\alpha_5$};
        \node[below] () at (150:1cm) {\footnotesize $e_5$};
        \node[] () at (0.3cm, 0.35cm) {\footnotesize $i_4$};
        \node[left] () at (0, -0.5cm) {\footnotesize $i_5$};
        \node[above] () at (0, -1cm) {\footnotesize $i_5'$};
        \draw[red] (270:0.9cm) -- (30:0.9cm);
        \node[] () at (330:0.6cm) {\footnotesize {\color{red} $s$}};

        \node[right] () at (270:0.9cm) {\footnotesize $\gamma_5$};
        \node[] () at (10:0.85cm) {\footnotesize $\gamma_4$};
        \node[below] () at (30:1cm) {\footnotesize $i_4'$};
    \end{tikzpicture}
    =\sum_{\alpha_5'} (b_5)_{\alpha_5 \alpha_5'}
    \begin{tikzpicture}[baseline={([yshift=-0.5ex]current bounding box.center)},line width=1pt, scale=1.2, yscale=-1]
        \draw[] (0,0) -- (30:1cm);
        \draw[] (0,0) -- (270:1cm);
        \draw[] (0,0) -- (150:1cm);
        \node[] () at (0.25,-0.125) {\footnotesize $\alpha_5'$};
        \node[below] () at (150:1cm) {\footnotesize $e_5$};
        \node[above] () at (0, -1cm) {\footnotesize $i_5'$};
        \node[below] () at (30:1cm) {\footnotesize $i_4'$};
    \end{tikzpicture},
    &\quad \begin{tikzpicture}[baseline={([yshift=-0.5ex]current bounding box.center)},line width=1pt, scale=1.2, yscale=-1]
        \node[] () at (0.4cm, -0.4cm) {\footnotesize $i_6$};
        \draw[] (0,0) -- (90:1cm);
        \draw[] (0,0) -- (-30:1cm);
        \draw[] (0,0) -- (210:1cm);
        \node[left] () at (90:0.5cm) {\footnotesize $i_5$};
        \draw[red] (-30:0.9cm) -- (90:0.9cm);
        \node[] () at (30:0.6cm) {\footnotesize {\color{red} $s$}};
        \node[left] () at (90:0.9cm) {\footnotesize $\gamma_5$};
        \node[] () at (-10:0.8cm) {\footnotesize $\gamma_6$};
        \node[below] () at (90:1cm) {\footnotesize $i_5'$};
        \node[above] () at (-30:1cm) {\footnotesize $i_6'$};
        \node[above] () at (210:1cm) {\footnotesize $e_6$};
        \node[above] () at (0,0) {\footnotesize $\alpha_6$};
    \end{tikzpicture}
    =\sum_{\alpha_6'} (b_6)_{\alpha_6 \alpha_6'}
    \begin{tikzpicture}[baseline={([yshift=-0.5ex]current bounding box.center)},line width=1pt, scale=1.2, yscale=-1]
        \draw[] (0,0) -- (90:1cm);
        \draw[] (0,0) -- (-30:1cm);
        \draw[] (0,0) -- (210:1cm);
        \node[below] () at (90:1cm) {\footnotesize $i_5'$};
        \node[above] () at (-30:1cm) {\footnotesize $i_6'$};
        \node[above] () at (210:1cm) {\footnotesize $e_6$};
        \node[above] () at (0,0) {\footnotesize $\alpha_6'$};
    \end{tikzpicture}.
\end{aligned}
\label{eq:plaquette_bs}
\end{equation}

To summarize, the action of $B_p^s$ on the string-net basis state can be written as follows:
\begin{equation}
     B_p^s 
    \begin{tikzpicture}[baseline={([yshift=-0.5ex]current bounding box.center)},scale=0.8]
    \draw[thick] (-2.125, -1.75) -- ++ (0,3.5);

    \draw[thick] (30:1cm) -- (90:1cm) node[midway, above] {\footnotesize $i_1$};
    \draw[thick] (90:1cm) -- (150:1cm) node[midway, above] {\footnotesize $i_6$};
    \draw[thick] (150:1cm) -- (210:1cm) node[midway, left] {\footnotesize $i_5$};
    \draw[thick] (210:1cm) -- (270:1cm) node[midway, below] {\footnotesize $i_4$};
    \draw[thick] (270:1cm) -- (330:1cm) node[midway, below] {\footnotesize $i_3$};
    \draw[thick] (330:1cm) -- (30:1cm) node[midway, right] {\footnotesize $i_2$};

    \node[] () at (30:0.7cm) {\footnotesize $\alpha_2$};
    \node[] () at (-30:0.7cm) {\footnotesize $\alpha_3$};
    \node[] () at (-90:0.7cm) {\footnotesize $\alpha_4$};
    \node[] () at (-150:0.7cm) {\footnotesize $\alpha_5$};
    \node[] () at (-210:0.7cm) {\footnotesize $\alpha_6$};
    \node[] () at (-270:0.7cm) {\footnotesize $\alpha_1$};

    \draw[thick] (30:1cm) -- (30:1.5cm) node[right] {\footnotesize $e_2$};
    \draw[thick] (90:1cm) -- (90:1.5cm) node[above] {\footnotesize $e_1$};
    \draw[thick] (150:1cm) -- (150:1.5cm) node[left] {\footnotesize $e_6$};
    \draw[thick] (210:1cm) -- (210:1.5cm) node[left] {\footnotesize $e_5$};
    \draw[thick] (270:1cm) -- (270:1.5cm) node[below] {\footnotesize $e_4$};
    \draw[thick] (330:1cm) -- (330:1.5cm) node[right] {\footnotesize $e_3$};

    \draw[thick] (2, -1.75) -- ++ (0.5, 1.75) -- ++ (-0.5, 1.75);
    \end{tikzpicture} = 
    \sum_{\substack{i_1', \ldots , i_6' \\ \alpha_1', \ldots \alpha_6'}} B_{p, \{ i_k\}, \{\alpha_k \}}^{\{i_k'\}, \{\alpha_k'\}}\left(\{e_k\}\right)
    \begin{tikzpicture}[baseline={([yshift=-0.5ex]current bounding box.center)},scale=0.8]
    \draw[thick] (-2.125, -1.75) -- ++ (0,3.5);

    \draw[thick] (30:1cm) -- (90:1cm) node[midway, above] {\footnotesize $i_1'$};
    \draw[thick] (90:1cm) -- (150:1cm) node[midway, above] {\footnotesize $i_6'$};
    \draw[thick] (150:1cm) -- (210:1cm) node[midway, left] {\footnotesize $i_5'$};
    \draw[thick] (210:1cm) -- (270:1cm) node[midway, below] {\footnotesize $i_4'$};
    \draw[thick] (270:1cm) -- (330:1cm) node[midway, below] {\footnotesize $i_3'$};
    \draw[thick] (330:1cm) -- (30:1cm) node[midway, right] {\footnotesize $i_2'$};

    \node[] () at (30:0.7cm) {\footnotesize $\alpha_2'$};
    \node[] () at (-30:0.7cm) {\footnotesize $\alpha_3'$};
    \node[] () at (-90:0.7cm) {\footnotesize $\alpha_4'$};
    \node[] () at (-150:0.7cm) {\footnotesize $\alpha_5'$};
    \node[] () at (-210:0.7cm) {\footnotesize $\alpha_6'$};
    \node[] () at (-270:0.7cm) {\footnotesize $\alpha_1'$};

    \draw[thick] (30:1cm) -- (30:1.5cm) node[right] {\footnotesize $e_2$};
    \draw[thick] (90:1cm) -- (90:1.5cm) node[above] {\footnotesize $e_1$};
    \draw[thick] (150:1cm) -- (150:1.5cm) node[left] {\footnotesize $e_6$};
    \draw[thick] (210:1cm) -- (210:1.5cm) node[left] {\footnotesize $e_5$};
    \draw[thick] (270:1cm) -- (270:1.5cm) node[below] {\footnotesize $e_4$};
    \draw[thick] (330:1cm) -- (330:1.5cm) node[right] {\footnotesize $e_3$};

    \draw[thick] (2, -1.75) -- ++ (0.5, 1.75) -- ++ (-0.5, 1.75);
    \end{tikzpicture},
\end{equation}
where
\begin{equation}
    B_{p, \{ i_k\}, \{\alpha_k \}}^{\{i_k'\}, \{\alpha_k'\}}\left(\{e_k\}\right) = \left(\prod_{k=1}^6 \sqrt{\frac{d_{i_k}}{d_s d_{i_k'}}} \right) \sum_{\gamma_1, \ldots, \gamma_6} \prod_{k=1}^6 (b_k)_{\alpha_k \alpha_k'}.
    \label{eq:b_definition}
\end{equation}
Note that $\{b_1, \ldots, b_6\}$ depends on $\{\gamma_1, \ldots, \gamma_6\}$ [Eq.~\eqref{eq:plaquette_bs}], though we suppressed the dependence for the sake of brevity.

\subsection{Local topological quantum order}
\label{subsec:ltqo}
A useful property of the string-net Hamiltonian is that it satisfies a \emph{local topological quantum order} condition. This is a condition introduced in Ref.~\cite{Michalakis2013}, which ensures gap stability.
\begin{definition}[Local topological quantum order]
\label{definition:ltqo}
    Let $A$ be a ball of radius $r$ and $A(\ell)$ be a set of sites that are distance at most $\ell$ away from $A$. Let 
    \begin{equation}
        c_{\ell}(O_A) := \frac{\text{Tr}(P_{A(\ell)} O_A)}{\text{Tr}(P_{A(\ell)})}
    \end{equation}
    for any bounded operator $O_A$ acting on $\mathcal{H}_A$, for any $r\leq L^*$. The local topological quantum order condition is satisfied if
    \begin{equation}
      P_{A(\ell)} O_A P_{A(\ell)} =  c_{\ell}(O_A) P_{A(\ell)}  \label{eq:ltqo_equation}
    \end{equation}
    for any $O_A$ for every $\ell >0$. 
\end{definition}
\noindent
Here $L^{*}$ is the cutoff parameter\footnote{For the string-net model on a plane, this can be taken to be infinity.} that depends on the details of the model and $P_{X}$ is the ground state projector of a Hamiltonian on region $X\subset \Lambda$. Specifically, we can consider a restricted Hamiltonian $H_{X}$ as a sum of local terms in $H$ that are supported strictly on $X$. Note that we modified the original definition in Ref.~\cite{Michalakis2013} slightly. In that version, Eq.~\eqref{eq:ltqo_equation} need not be satisfied exactly. 

What matters to us is that the string-net Hamiltonian satisfies the local topological quantum order condition [Definition~\ref{definition:ltqo}], verified in \cite[Theorem A]{jones2023local} and \cite{qiu2020ground}. We thus obtain the following corollary.

\begin{corollary}
    Let $A\subset \Lambda$ be a disk. For the string-net Hamiltonian Eq.~\eqref{eq:string_net_ham}, for any ground state $|\psi\rangle$ of $H_{A(1)}$, 
    \begin{equation}
        \text{Tr}_{A(1)\setminus A}\left( |\psi{\rangle}{\langle}\psi| \right) = \sigma_A,
    \end{equation}
    where $\sigma_A$ is the ground state reduced density matrix of the string-net Hamiltonian over $A$.
\end{corollary}

\section{Proof}
\label{sec:proof_main}

We prove Theorem \ref{thm:mapping-to-sn}, stating that states satisfying Definition \ref{def:gappable-boundary-EB} can be mapped to string-net ground states  [Section~\ref{sec:ground_state}]. More precisely, we establish a procedure that, for any disk, one can convert the interior of the disk to a string-net. 

For a high-level summary, please begin with Section \ref{sec:proof_summary} and Figure~\ref{fig:circuit-summary} especially. Again, our proof consists of three important steps. First, we apply a depth-$1$ circuit $U_1$ to the reference state $\sigma$, creating a new state $\sigma^{(1)}$ with holes that have gapped boundary, forming a fattened lattice. The local Hilbert space on the thick edges of the fat lattice can be decomposed into sectors corresponding to boundary anyon types. A depth-$1$ circuit $U_2$ applies disentanglers to these edge regions, conditional on their sector.  Then the state is a superposition of products over disentangled vertex states, mirroring the string-net Hilbert space.  Finally, a depth-$1$ circuit $U_3$ identifies a subspace of the physical Hilbert space (where the current state is supported) with the string-net Hilbert space $\mathcal{H}_0$ [Section~\ref{sec:sn-review}].  The vertex constraints of the string-net ground state are satisfied by construction.  The plaquette constraints are shown by developing a diagrammatic calculus for anyon operations in the physical Hilbert space that mirrors the string-net diagrammatics.

\subsection{Punching holes}
\label{subsec:punching_holes}

In this Section, we aim to convert the reference state $\sigma$ into a state with many holes; see Figure~\ref{fig:punching-holes} for an illustration. Each hole can be created by applying a unitary localized in the vicinity of the hole. If the holes are sufficiently far apart, these unitaries are supported on disjoint regions. As such, they can be applied in parallel.

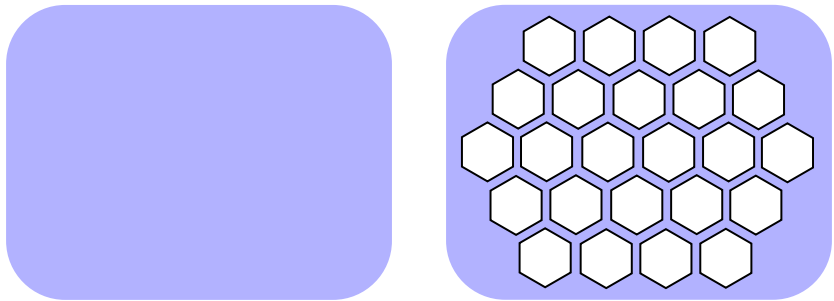
\begin{figure}[t]
        \centering
        \tikzset{every picture/.style={line width=0.75pt}} %set default line width to 0.75pt        

\begin{tikzpicture}[x=0.75pt,y=0.75pt,yscale=-1,xscale=1, scale=0.7]
%uncomment if require: \path (0,972); %set diagram left start at 0, and has height of 972

%Shape: Path Data [id:dp15313631076157752] 
\draw  [fill={rgb, 255:red, 178; green, 178; blue, 255 }  ,fill opacity=1 ] (353.32,767.09) .. controls (353.32,744.26) and (371.83,725.75) .. (394.66,725.75) -- (584.04,725.75) .. controls (606.87,725.75) and (625.39,744.26) .. (625.39,767.09) -- (625.39,891.13) .. controls (625.39,913.97) and (606.87,932.48) .. (584.04,932.48) -- (394.66,932.48) .. controls (371.83,932.48) and (353.32,913.97) .. (353.32,891.13) (556.33,773.7) -- (574.74,763.07) -- (574.74,741.81) -- (556.33,731.18) -- (537.92,741.81) -- (537.92,763.07) -- (556.33,773.7) -- cycle (512.76,773.19) -- (531.17,762.56) -- (531.17,741.3) -- (512.76,730.67) -- (494.35,741.3) -- (494.35,762.56) -- (512.76,773.19) -- cycle (469.57,773.44) -- (487.98,762.81) -- (487.98,741.55) -- (469.57,730.92) -- (451.16,741.55) -- (451.16,762.81) -- (469.57,773.44) -- cycle (426.26,773.51) -- (444.67,762.88) -- (444.67,741.62) -- (426.26,730.99) -- (407.85,741.62) -- (407.85,762.88) -- (426.26,773.51) -- cycle (403.92,811.72) -- (422.33,801.09) -- (422.33,779.83) -- (403.92,769.2) -- (385.51,779.83) -- (385.51,801.09) -- (403.92,811.72) -- cycle (381.77,849.75) -- (400.18,839.12) -- (400.18,817.86) -- (381.77,807.23) -- (363.36,817.86) -- (363.36,839.12) -- (381.77,849.75) -- cycle (402.37,888.34) -- (420.78,877.71) -- (420.78,856.45) -- (402.37,845.82) -- (383.96,856.45) -- (383.96,877.71) -- (402.37,888.34) -- cycle (423.32,926.1) -- (441.73,915.47) -- (441.73,894.21) -- (423.32,883.58) -- (404.91,894.21) -- (404.91,915.47) -- (423.32,926.1) -- cycle (510.32,926.72) -- (528.73,916.09) -- (528.73,894.83) -- (510.32,884.2) -- (491.9,894.83) -- (491.9,916.09) -- (510.32,926.72) -- cycle (553.5,926.46) -- (571.91,915.83) -- (571.91,894.57) -- (553.5,883.94) -- (535.09,894.57) -- (535.09,915.83) -- (553.5,926.46) -- cycle (447.17,811.85) -- (465.58,801.22) -- (465.58,779.96) -- (447.17,769.33) -- (428.76,779.96) -- (428.76,801.22) -- (447.17,811.85) -- cycle (490.98,812.05) -- (509.39,801.42) -- (509.39,780.16) -- (490.98,769.53) -- (472.57,780.16) -- (472.57,801.42) -- (490.98,812.05) -- cycle (534.17,811.8) -- (552.58,801.17) -- (552.58,779.91) -- (534.17,769.28) -- (515.76,779.91) -- (515.76,801.17) -- (534.17,811.8) -- cycle (468.51,849.85) -- (486.92,839.22) -- (486.92,817.96) -- (468.51,807.33) -- (450.09,817.96) -- (450.09,839.22) -- (468.51,849.85) -- cycle (424.34,849.67) -- (442.75,839.04) -- (442.75,817.78) -- (424.34,807.15) -- (405.93,817.78) -- (405.93,839.04) -- (424.34,849.67) -- cycle (467.32,926.72) -- (485.73,916.09) -- (485.73,894.83) -- (467.32,884.2) -- (448.9,894.83) -- (448.9,916.09) -- (467.32,926.72) -- cycle (575.11,888.14) -- (593.52,877.51) -- (593.52,856.25) -- (575.11,845.62) -- (556.7,856.25) -- (556.7,877.51) -- (575.11,888.14) -- cycle (577.03,812.23) -- (595.44,801.6) -- (595.44,780.34) -- (577.03,769.71) -- (558.62,780.34) -- (558.62,801.6) -- (577.03,812.23) -- cycle (597.9,850.41) -- (616.31,839.78) -- (616.31,818.52) -- (597.9,807.89) -- (579.49,818.52) -- (579.49,839.78) -- (597.9,850.41) -- cycle (445.51,887.85) -- (463.92,877.22) -- (463.92,855.96) -- (445.51,845.33) -- (427.09,855.96) -- (427.09,877.22) -- (445.51,887.85) -- cycle (489.32,888.05) -- (507.73,877.42) -- (507.73,856.16) -- (489.32,845.53) -- (470.9,856.16) -- (470.9,877.42) -- (489.32,888.05) -- cycle (532.5,887.8) -- (550.91,877.17) -- (550.91,855.91) -- (532.5,845.28) -- (514.09,855.91) -- (514.09,877.17) -- (532.5,887.8) -- cycle (512.32,850.05) -- (530.73,839.42) -- (530.73,818.16) -- (512.32,807.53) -- (493.9,818.16) -- (493.9,839.42) -- (512.32,850.05) -- cycle (555.5,849.8) -- (573.91,839.17) -- (573.91,817.91) -- (555.5,807.28) -- (537.09,817.91) -- (537.09,839.17) -- (555.5,849.8) -- cycle ;
%Shape: Path Data [id:dp7920384871032651] 
\draw  [draw opacity=0][fill={rgb, 255:red, 178; green, 178; blue, 255 }  ,fill opacity=1 ] (394.47,722.45) -- (587.4,722.45) .. controls (610.89,722.45) and (629.93,741.5) .. (629.93,764.99) -- (629.93,892.59) .. controls (629.93,916.08) and (610.89,935.12) .. (587.4,935.12) -- (394.47,935.12) .. controls (370.98,935.12) and (351.93,916.08) .. (351.93,892.59) -- (351.93,764.99) .. controls (351.93,741.5) and (370.98,722.45) .. (394.47,722.45) -- cycle (361.93,767.79) -- (361.93,887.79) .. controls (361.93,909.88) and (379.84,927.79) .. (401.93,927.79) -- (583.27,927.79) .. controls (605.36,927.79) and (623.27,909.88) .. (623.27,887.79) -- (623.27,767.79) .. controls (623.27,745.7) and (605.36,727.79) .. (583.27,727.79) -- (401.93,727.79) .. controls (379.84,727.79) and (361.93,745.7) .. (361.93,767.79) -- cycle ;
%Shape: Path Data [id:dp6175711837517941] 
\draw  [draw opacity=0][fill={rgb, 255:red, 178; green, 178; blue, 255 }  ,fill opacity=1 ] (77.47,722.54) -- (270.4,722.54) .. controls (293.89,722.54) and (312.93,741.58) .. (312.93,765.08) -- (312.93,892.68) .. controls (312.93,916.17) and (293.89,935.21) .. (270.4,935.21) -- (77.47,935.21) .. controls (53.98,935.21) and (34.93,916.17) .. (34.93,892.68) -- (34.93,765.08) .. controls (34.93,741.58) and (53.98,722.54) .. (77.47,722.54) -- cycle (44.93,767.88) -- (44.93,887.88) .. controls (44.93,909.97) and (62.84,927.88) .. (84.93,927.88) -- (266.27,927.88) .. controls (288.36,927.88) and (306.27,909.97) .. (306.27,887.88) -- (306.27,767.88) .. controls (306.27,745.78) and (288.36,727.88) .. (266.27,727.88) -- (84.93,727.88) .. controls (62.84,727.88) and (44.93,745.78) .. (44.93,767.88) -- cycle ;
%Rounded Rect [id:dp08678570918115303] 
\draw  [draw opacity=0][fill={rgb, 255:red, 178; green, 178; blue, 255 }  ,fill opacity=1 ] (40,766.58) .. controls (40,743.85) and (58.42,725.42) .. (81.15,725.42) -- (268.05,725.42) .. controls (290.78,725.42) and (309.2,743.85) .. (309.2,766.58) -- (309.2,890.03) .. controls (309.2,912.76) and (290.78,931.19) .. (268.05,931.19) -- (81.15,931.19) .. controls (58.42,931.19) and (40,912.76) .. (40,890.03) -- cycle ;

\end{tikzpicture}
    \caption{By applying a depth-$1$ circuit, we convert the reference state $\sigma$ (left) to a state $\sigma^{(1)}$ consisting of many holes (right).}
    \label{fig:punching-holes}
\end{figure}

We now discuss why such unitaries exist. To that end, let us first note the following fact.
\begin{lemma}
\label{lemma:unitary_map_decoupled}
    Suppose $\rho_{ABC}$ and $\sigma_{ABC}$ satisfy $\rho_{AB} = \sigma_{AB}$ and $\left(S(BC) + S(C) - S(B)\right)_{\rho} = \left(S(BC) + S(C) - S(B)\right)_{\sigma}=0$. There is a unitary $U_{BC}$ acting on $BC$ such that
    \begin{equation}
        \rho_{ABC} = U_{BC}\sigma U_{BC}^{\dagger}.
    \end{equation}
\end{lemma}
\begin{proof}
    Let $D$ be the purification and denote the purified states of $\rho$ and $\sigma$ as $|\psi_{\rho}\rangle$ and $|\psi_{\sigma}\rangle$, respectively. It follows that $I(A:CD)_{\rho} = I(A:CD)_{\sigma}=0$. By Uhlmann's theorem~\cite{Uhlmann1976}, there is a decomposition of  $B=B_LB_R$ such that
    \begin{equation}
        \begin{aligned}
            |\psi_{\rho}\rangle &= |\psi_{\rho,1}\rangle_{AB_LD} \otimes |\psi_{\rho,2}\rangle_{B_RC}, \\
            |\psi_{\sigma}\rangle &= |\psi_{\sigma,1}\rangle_{AB_LD} \otimes |\psi_{\sigma,2}\rangle_{B_RC}.
        \end{aligned}
    \end{equation}
    Because $\rho$ and $\sigma$ are consistent on $BC,$
    \begin{equation}
        \begin{aligned}
            \rho_{ABC} &=  \rho_{AB_L} \otimes |\psi_{\rho,1}\rangle_{B_R C}\langle \psi_{\rho, 1}|,\\
            \sigma_{ABC} &= \rho_{AB_L} \otimes |\psi_{\sigma,1}\rangle_{B_R C}\langle \psi_{\sigma, 1}|,
        \end{aligned}
    \end{equation}
    from which the claim follows immediately.
\end{proof}

At this step of the proof, the main idea is to identify two density matrices that satisfy the requisite conditions in Lemma~\ref{lemma:unitary_map_decoupled}. Specifically, we choose $\sigma$ as the reference state and $\rho$ as a state with a gapped boundary. We will focus on a partition $ABC$, where $B$ is an annulus in the bulk that surrounds a hole, $C$ is a region enclosed by $B$, and $A$ is the region lying outside of $B$.

Because of \textbf{A0} [Figure~\ref{fig:eb_axioms_bulk}], it should be obvious that $\sigma$ satisfies $\left(S(BC) + S(B) - S(C)\right)_{\sigma}=0$. However, the existence of $\rho$ that satisfies $\rho_{AB} = \sigma_{AB}$ is less obvious, because we did not necessarily assume the existence of such a state in Section~\ref{sec:setup-results}. 

Nevertheless, such a state does exist for the following reason. The main idea is to start with $\sigma$, take a partial trace on a disk in the interior, and build up the gapped boundary by a sequence of Petz maps. We now explain this procedure in more detail. Starting with $\sigma$, let us first trace out all of $C$ except an annulus adjacent to $B$. The resulting state is defined over $ABC'$ where $C'\subset C$, which we represent diagrammatically as follows 
\begin{equation}
    \sigma_{ABC'}=
    \begin{tikzpicture}[baseline={([yshift=-0.5ex]current bounding box.center)}, scale=1]
    \filldraw[blue!30!white] (30:1cm)-- (90:1cm) -- (150:1cm) -- (210:1cm) -- (270:1cm) -- (330:1cm) -- cycle;
    \draw[dashed, thick] (30:0.85cm)-- (90:0.85cm) -- (150:0.85cm) -- (210:0.85cm) -- (270:0.85cm) -- (330:0.85cm) -- cycle;
    \draw[dashed, thick] (30:0.7cm)-- (90:0.7cm) -- (150:0.7cm) -- (210:0.7cm) -- (270:0.7cm) -- (330:0.7cm) -- cycle;
    \draw[draw=none, fill=white] (30:0.6cm)-- (90:0.6cm) -- (150:0.6cm) -- (210:0.6cm) -- (270:0.6cm) -- (330:0.6cm) -- cycle;
    \end{tikzpicture}.
\end{equation}

Now we will introduce a procedure local to the inner boundary. Here we will briefly use notation $\sigma_{\partial}$ for the reference fragments with gapped boundary. Consider a local region in the vicinity of such boundary [Figure~\ref{fig:local_extension_gapped_boundary}]. Note that $I(C_1':ABC'\setminus (C_1'C_2')|C_2')_{\sigma}=0$ and $I(\delta_C:C_2'|C_1')_{\sigma_{\partial}}=0$ because of axiom \textbf{A1} and that $\sigma$ and $\sigma_{\partial}$ are consistent with each other over $C_1'C_2'$ [Section~\ref{sec:setup-results}]. Therefore, by the merging theorem [Theorem~\ref{thm:merging_theorem}], we can obtain state on $ABC'\delta C$ by applying a Petz map. Moreover, the resulting state is in the information convex set $\Sigma(ABC' \delta C)$.

\begin{figure}[ht]
    \centering
    \begin{tikzpicture}
        \filldraw[blue!30!white] (0,0) --++ (5, 0) --++ (0, 2) -- ++ (-5, 0) -- ++ (0, -2);
        \filldraw[blue!30!white] (2, 2) -- ++ (1, 0) -- ++ (0, 1) -- ++ (-1, 0) -- cycle;
        \draw[very thick] (2, 3) -- ++ (1,0);
        \draw[very thick, dashed] (2,2) -- ++ (1,0);
        \begin{scope}[xshift=2.5cm, yshift=2cm]
            \draw[very thick, dashed] (180:1cm) arc (180:360:1cm);
            \draw[very thick, dashed] (180:1.75cm) arc (180:360:1.75cm);
        \end{scope}
        \node[] () at (2.5, 2.5) {$\delta C$};
        \node[] () at (2.5, 1.5) {$C_1'$};
        \node[] () at (2.5, 0.625) {$C_2'$};
    \end{tikzpicture}
    \caption{Local extension of bulk reduced density matrix to the gapped boundary.}
    \label{fig:local_extension_gapped_boundary}
\end{figure}
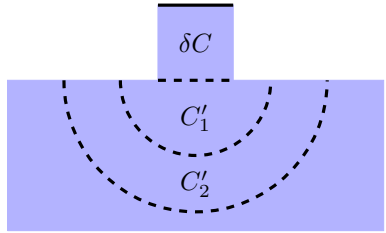

We can repeat this procedure in parallel by choosing disjoint sets of $C_1'C_2' \delta C$. Finally, using boundary \textbf{A0} [Figure~\ref{fig:eb_axioms_boundary}], one can deduce the existence of a unitary that completes the boundary. We will call the resulting state as $\rho_{ABC}$. By construction, $\rho_{BC} = \sigma_{BC}$. Moreover, $\rho_{ABC}$ is in the information convex set of the annulus region surrounding the gapped boundary. Moreover, because $\rho_{ABC}$ was obtained from $\sigma_{ABC'}$ --- an extreme point of $\Sigma(ABC')$ --- $\rho_{ABC}$ is also an extreme point [Theorem~\ref{thm:isomorphism_theorem}].\footnote{The isomorphism theorem guarantees that this map is bijective, and bijective linear maps preserve extreme points.} The factorization of extreme point [Eq.~\eqref{eq:factorization_extreme_point}] then implies $\left(S(BC) + S(B) - S(C)\right)_{\rho}=0$. Therefore, the requisite conditions in Lemma~\ref{lemma:unitary_map_decoupled} are satisfied.

We thus conclude that there is a unitary on $BC$ that converts $\sigma_{ABC}$ to $\rho_{ABC}$. Clearly, we can apply such unitaries to a set of disks  that are sufficiently far apart. Upon applying such unitaries, we are able to punch holes, yielding a state $\sigma^{(1)}$ with many gapped boundaries [Figure~\ref{fig:punching-holes}].

\subsection{Identifying the string-net Hilbert space}
\label{subsec:identifying_string_net}
While $\sigma$ may not be a priori pure, we can disentangle this state into two regions, the disk on which the disentangling unitaries in Section~\ref{subsec:punching_holes} were applied and the rest. Because of axiom \textbf{A0}, the disk can be disentangled from the rest by applying a unitary along the boundary of the disk. Using an argument akin to Lemma~\ref{lemma:unitary_map_decoupled}, the interior state is pure. Thus, without loss of generality, we can write $\sigma$ as a tensor product of a pure state that describes the disk and some other state which will not play a role in our analysis. We will refer to this state as $|\sigma\rangle$. Because $\sigma$ and $\sigma^{(1)}$ are related to each other by a unitary acting in the interior of the disk, $\sigma^{(1)}$ can be also disentangled in the same way. We will refer to the corresponding pure state on the disk as $|\sigma^{(1)}\rangle$. 

The state vector $|\sigma^{(1)}\rangle$ is locally indistinguishable from the reference state $\sigma$ in the bulk, and  $|\sigma^{(1)}\rangle$ is also locally indistinguishable, along its boundaries, from the reference fragments with gapped boundary. The set of states that satisfy the local indistinguishability condition can  be labeled in terms of the information convex set of the $N$-shaped subsystems [Figure~\ref{fig:IC_types_this_paper} (a)] , which will be chosen as the ``bridges'' between neighboring hexagons;  see Figure~\ref{fig:bridge_single}(a) for an example. Recall that for such subsystem $N$, there is a finite set of extreme points, each labeled by the simple objects in a UFC $\mathcal{C}$. Because these extreme points are orthogonal to each other, there is a set of orthogonal projections $\{P_a: a\in \mathcal{C} \}$ such that $P_a \rho_b = \delta_{ab} \rho_b$ for any extreme point of $\Sigma(N)$. Because the label set $\mathcal{C}$ includes \emph{all} extreme points of the information convex set, we deduce $\left(\sum_{a\in \mathcal{C}} P_a\right) |\sigma^{(1)}{\rangle}{\langle} \sigma^{(1)}| =  |\sigma^{(1)}{\rangle}{\langle} \sigma^{(1)}|$. Therefore, without loss of generality, we can write $|\sigma^{(1)}\rangle$ as a linear combination of states, each of which are labeled by the labels $a\in \mathcal{C}$.
\begin{equation}
    |\sigma^{(1)}\rangle = \sum_{a\in\mathcal{C}} |\sigma_a'\rangle,
\end{equation}
where $|\sigma_a'\rangle = P_a |\sigma^{(1)}\rangle$ and $P_a$ is associated to a single bridge.  Applying this decomposition on every bridge between neighboring hexagons, we obtain 
\begin{align}
    |\sigma^{(1)}\rangle = \sum_{\{a_e\}_e} \prod_e P_e |\sigma^{(1)}\rangle
\end{align}
where the sum runs over all assignments of sector labels $a_e \in \mathcal{C}$ to each edge $e$. So each term in the superposition is labeled by an assignment of (simple) objects in $\mathcal{C}$ to each edge.

\begin{figure}[ht]
    \centering
    \begin{tikzpicture}[scale=1.25]
    \begin{scope}[xshift=-2cm]
        \begin{scope}[yshift=1.5cm]
        \filldraw[blue!30!white] (30:1cm) -- (90:1cm) -- (150:1cm) -- (210:1cm) -- (270:1cm) -- (330:1cm) -- cycle;
        \begin{scope}[xshift=1.732cm]
        \filldraw[blue!30!white] (30:1cm) -- (90:1cm) -- (150:1cm) -- (210:1cm) -- (270:1cm) -- (330:1cm) -- cycle;
        \end{scope}
        \begin{scope}[xshift=3.464cm]
        \filldraw[blue!30!white] (30:1cm) -- (90:1cm) -- (150:1cm) -- (210:1cm) -- (270:1cm) -- (330:1cm) -- cycle;
        \draw[fill=white, thick] (30:0.8cm) -- (90:0.8cm) -- (150:0.8cm) -- (210:0.8cm) -- (270:0.8cm) -- (330:0.8cm) -- cycle;
        \end{scope}
        \begin{scope}[xshift=0.866cm, yshift=1.5cm]
        \filldraw[blue!30!white] (30:1cm) -- (90:1cm) -- (150:1cm) -- (210:1cm) -- (270:1cm) -- (330:1cm) -- cycle;
        \draw[fill=white, thick] (30:0.8cm) -- (90:0.8cm) -- (150:0.8cm) -- (210:0.8cm) -- (270:0.8cm) -- (330:0.8cm) -- cycle;
        \end{scope}
        \begin{scope}[xshift=2.598cm, yshift=1.5cm]
        \filldraw[blue!30!white] (30:1cm) -- (90:1cm) -- (150:1cm) -- (210:1cm) -- (270:1cm) -- (330:1cm) -- cycle;
        \draw[fill=white, thick] (30:0.8cm) -- (90:0.8cm) -- (150:0.8cm) -- (210:0.8cm) -- (270:0.8cm) -- (330:0.8cm) -- cycle;
        \end{scope}
        \begin{scope}[xshift=0.866cm, yshift=-1.5cm]
        \filldraw[blue!30!white] (30:1cm) -- (90:1cm) -- (150:1cm) -- (210:1cm) -- (270:1cm) -- (330:1cm) -- cycle;
        \draw[fill=white, thick] (30:0.8cm) -- (90:0.8cm) -- (150:0.8cm) -- (210:0.8cm) -- (270:0.8cm) -- (330:0.8cm) -- cycle;
        \end{scope}
        \begin{scope}[xshift=2.598cm, yshift=-1.5cm]
        \filldraw[blue!30!white] (30:1cm) -- (90:1cm) -- (150:1cm) -- (210:1cm) -- (270:1cm) -- (330:1cm) -- cycle;
        \draw[fill=white, thick] (30:0.8cm) -- (90:0.8cm) -- (150:0.8cm) -- (210:0.8cm) -- (270:0.8cm) -- (330:0.8cm) -- cycle;
        \end{scope}

        \filldraw[green!50!white] (0.6928, -0.25) -- ++ (0.3464, 0) -- ++ (0, 0.5) -- ++ (-0.3464, 0) -- cycle;
        \draw[dashed, thick] (0.6928, -0.25) -- ++ (0.3464, 0) -- ++ (0, 0.5) -- ++ (-0.3464, 0) -- cycle;
        \begin{scope}[xshift=1.732cm]
        \filldraw[green!50!white] (0.6928, -0.25) -- ++ (0.3464, 0) -- ++ (0, 0.5) -- ++ (-0.3464, 0) -- cycle;
        \draw[dashed, thick] (0.6928, -0.25) -- ++ (0.3464, 0) -- ++ (0, 0.5) -- ++ (-0.3464, 0) -- cycle;
        \end{scope}

        \begin{scope}[xshift=1.732cm]
        \draw[fill=yellow!50!white,rotate=60,shift={(0cm, 0)}, draw=none] (0.6928, -0.25) -- ++ (0.3464, 0) -- ++ (0, 0.5) -- ++ (-0.3464, 0) -- cycle;           
        \draw[dashed, thick, rotate=60,shift={(0cm, 0)}] (0.6928, -0.25) -- ++ (0.3464, 0) -- ++ (0, 0.5) -- ++ (-0.3464, 0) -- cycle;
        \draw[fill=yellow!50!white,rotate=120,shift={(0cm, 0)}, draw=none] (0.6928, -0.25) -- ++ (0.3464, 0) -- ++ (0, 0.5) -- ++ (-0.3464, 0) -- cycle;           
        \draw[dashed, thick, rotate=120,shift={(0cm, 0)}] (0.6928, -0.25) -- ++ (0.3464, 0) -- ++ (0, 0.5) -- ++ (-0.3464, 0) -- cycle;
        \draw[fill=yellow!50!white,rotate=240,shift={(0cm, 0)}, draw=none] (0.6928, -0.25) -- ++ (0.3464, 0) -- ++ (0, 0.5) -- ++ (-0.3464, 0) -- cycle;           
        \draw[dashed, thick, rotate=240,shift={(0cm, 0)}] (0.6928, -0.25) -- ++ (0.3464, 0) -- ++ (0, 0.5) -- ++ (-0.3464, 0) -- cycle;
        \draw[fill=yellow!50!white,rotate=300,shift={(0cm, 0)}, draw=none] (0.6928, -0.25) -- ++ (0.3464, 0) -- ++ (0, 0.5) -- ++ (-0.3464, 0) -- cycle;           
        \draw[dashed, thick, rotate=300,shift={(0cm, 0)}] (0.6928, -0.25) -- ++ (0.3464, 0) -- ++ (0, 0.5) -- ++ (-0.3464, 0) -- cycle;
        \end{scope}

        \draw[fill=white, thick] (30:0.8cm) -- (90:0.8cm) -- (150:0.8cm) -- (210:0.8cm) -- (270:0.8cm) -- (330:0.8cm) -- cycle;
        \begin{scope}[xshift=1.732cm]
        \draw[fill=white, thick] (30:0.8cm) -- (90:0.8cm) -- (150:0.8cm) -- (210:0.8cm) -- (270:0.8cm) -- (330:0.8cm) -- cycle;
        \end{scope}
        \begin{scope}[xshift=3.464cm]
        \draw[fill=white, thick] (30:0.8cm) -- (90:0.8cm) -- (150:0.8cm) -- (210:0.8cm) -- (270:0.8cm) -- (330:0.8cm) -- cycle;
        \end{scope}
        \begin{scope}[xshift=0.866cm, yshift=1.5cm]
        \draw[fill=white, thick] (30:0.8cm) -- (90:0.8cm) -- (150:0.8cm) -- (210:0.8cm) -- (270:0.8cm) -- (330:0.8cm) -- cycle;
        \end{scope}
        \begin{scope}[xshift=2.598cm, yshift=1.5cm]
        \draw[fill=white, thick] (30:0.8cm) -- (90:0.8cm) -- (150:0.8cm) -- (210:0.8cm) -- (270:0.8cm) -- (330:0.8cm) -- cycle;
        \end{scope}
        \begin{scope}[xshift=0.866cm, yshift=-1.5cm]
        \draw[fill=white, thick] (30:0.8cm) -- (90:0.8cm) -- (150:0.8cm) -- (210:0.8cm) -- (270:0.8cm) -- (330:0.8cm) -- cycle;
        \end{scope}
        \begin{scope}[xshift=2.598cm, yshift=-1.5cm]        \draw[fill=white, thick] (30:0.8cm) -- (90:0.8cm) -- (150:0.8cm) -- (210:0.8cm) -- (270:0.8cm) -- (330:0.8cm) -- cycle;
        \end{scope}

        \end{scope}
        \node[] () at (1.732, -1.5) {(a)};
    \end{scope}

        \begin{scope}[xshift=5cm]
        \filldraw[blue!30!white] (30:1cm) -- (90:1cm) -- (150:1cm) -- (210:1cm) -- (270:1cm) -- (330:1cm) -- cycle;
        \begin{scope}[xshift=1.732cm]
        \filldraw[blue!30!white] (30:1cm) -- (90:1cm) -- (150:1cm) -- (210:1cm) -- (270:1cm) -- (330:1cm) -- cycle;
        \end{scope}
        \begin{scope}[xshift=0.866cm, yshift=1.5cm]
        \filldraw[blue!30!white] (30:1cm) -- (90:1cm) -- (150:1cm) -- (210:1cm) -- (270:1cm) -- (330:1cm) -- cycle;
        \end{scope}

        \draw[fill=green!50!white, thick, dashed] (0,0) -- ++ (1.732, 0) -- ++ (-0.866, 1.5) -- cycle;

        \draw[fill=white, thick] (30:0.8cm) -- (90:0.8cm) -- (150:0.8cm) -- (210:0.8cm) -- (270:0.8cm) -- (330:0.8cm) -- cycle;
        \begin{scope}[xshift=1.732cm]
        \draw[fill=white, thick] (30:0.8cm) -- (90:0.8cm) -- (150:0.8cm) -- (210:0.8cm) -- (270:0.8cm) -- (330:0.8cm) -- cycle;
        \end{scope}
        \begin{scope}[xshift=0.866cm, yshift=1.5cm]
        \draw[fill=white, thick] (30:0.8cm) -- (90:0.8cm) -- (150:0.8cm) -- (210:0.8cm) -- (270:0.8cm) -- (330:0.8cm) -- cycle;
        \end{scope}
        \node[] () at (0.866, -1.5) {(b)};
        \end{scope}
    \end{tikzpicture}
    \caption{(a) There is a $N$-type subsystem (green, yellow) between two neighboring hexagons. Later we will distinguish vertical (green) and non-vertical (yellow) edge regions.
    (b) There is a $M$-type subsystem (green) surrounded by three hexagons.}
    \label{fig:bridge_single}
\end{figure}
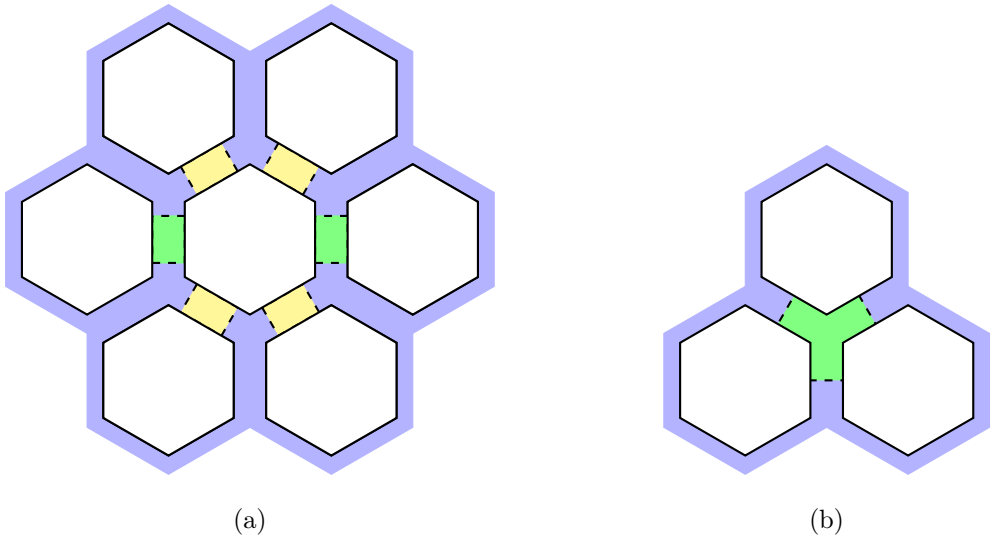

Every vector appearing in such linear combination can be decomposed further, by noting that there can be nontrivial multiplicites for the $M$-shaped subsystems [Figure~\ref{fig:bridge_single}(b)]. Without loss of generality, let the $N$-type sectors to be $a, b,$ and $c$. The remaining degrees of freedom are labeled by the multiplicities [Section~\ref{subsec:information_convex_set}].

Therefore, we can expand $|\sigma^{(1)}\rangle$ in terms of the basis vectors labeled by (i) sector defined over every $N$-type subsystems and (ii) the multiplicities on the $M$-type subsystems. Moreover, the multiplicities must lie in the fusion space, fixed by the adjacent $N$-type sectors. Note this basis set is tantalizingly similar to the string-net Hilbert space $\mathcal{H}_0$ [Section~\ref{sec:sn-review}]!

One important difference is that the string-net Hilbert space has a tensor product decomposition over the vertex degrees of freedom.  (Recall our convention for string-nets uses only vertex degrees of freedom, without additional edge degrees of freedom). On the other hand, the state $|\sigma^{(1)}\rangle$ is not yet a superposition of product states with respect to vertex degrees of freedom, at least not in a natural way.  We can achieve this by disentangling the edge regions, conditional on their sector type. Without loss of generality, consider an edge region labeled by $e$ and suppose the $N$-type sector is $a$. The region $e$ is shown in Figure~\ref{fig:edge-disentanglers-simple}(a).  We want to transform it to look like Figure~\ref{fig:edge-disentanglers-simple}(c).  By the factorization of the extreme points and Lemma~\ref{lemma:unitary_map_decoupled},  there is a unitary $V_a^{(e)}$ that achieves this while acting only in the region $C$.  The action of $V_a^{(e)}$ on (a purification of) the state in Figure~\ref{fig:edge-disentanglers-simple} is then determined determined up to a complex phase.  The precise choice of phase is important for ultimately producing a string-net and it will be specified in Section~\ref{subsec:plaquette} near Figure~\ref{fig:edge-disentanglers_action}.

 More generally,  we can apply the following unitary acting on an edge region $e$,
\begin{equation}
    V^{(e)} = \left(\sum_{a\in \mathcal{C}} V_a^{(e)} P_a^{(e)}\right) + (I - \sum_{a\in \mathcal{C}} P_a^{(e)}),
\end{equation}
which disentangles the edge region using $V_a^{(e)}$, conditional on the sector $a$. (The second term only serves to make it unitary.) 

Implementing this disentangling action at every edge,
\begin{align}
U_2 = \prod_e  V^{(e)}
\end{align}
we call the resulting unitary $U_2$ the \emph{edge disentangler}.  (Again, the precise phases desired to define $V_a^{(e)}$ and hence $U_2$ will be specified later in Section~\ref{subsec:plaquette}.)

\begin{figure}[htbp]
\label{fig:edge-disentanglers-simple}      
    \centering 
\begin{tikzpicture}
[x=0.75pt,y=0.75pt,yscale=-1,xscale=1]
%uncomment if require: \path (0,145); %set diagram left start at 0, and has height of 145

%Shape: Path Data [id:dp009086579330055633] 
\draw  [draw opacity=0][fill={rgb, 255:red, 178; green, 178; blue, 255 }  ,fill opacity=1 ] (194.89,85.59) -- (209.48,62.51) -- (238.77,79.57) -- (225.56,101.21) -- (223.43,103.95) -- (194.89,85.59) -- cycle ;
%Shape: Path Data [id:dp3367508291473156] 
\draw  [draw opacity=0][fill={rgb, 255:red, 178; green, 178; blue, 255 }  ,fill opacity=1 ] (88.98,75.01) -- (118.57,91.73) -- (104.01,116.71) -- (75.6,100.92) -- (39.06,79.22) -- (52.81,55.64) -- (88.98,75.01) -- cycle ;
%Straight Lines [id:da978519109663935] 
\draw    (39.06,79.22) -- (104.01,116.71) ;
%Straight Lines [id:da8648867565708662] 
\draw    (52.81,55.64) -- (118.57,91.73) ;
%Shape: Path Data [id:dp23460759031536615] 
\draw  [draw opacity=0][fill={rgb, 255:red, 178; green, 178; blue, 255 }  ,fill opacity=1 ] (159.45,65.04) -- (173.94,41.9) -- (202.47,58.36) -- (186.43,82.23) -- (159.45,65.04) -- cycle ;
%Straight Lines [id:da9027932912487089] 
\draw    (159.45,65.04) -- (187.07,81.92) ;
%Straight Lines [id:da5043644132313376] 
\draw    (173.94,41.9) -- (202.81,58.91) ;
%Straight Lines [id:da795509896077041] 
\draw    (187.07,81.92) -- (202.81,58.91) ;
%Straight Lines [id:da3624221191259427] 
\draw    (194.48,86.01) -- (223.43,103.95) ;
%Straight Lines [id:da9281902885933286] 
\draw    (209.48,62.51) -- (238.7,80.22) ;
%Straight Lines [id:da19319039591737863] 
\draw    (194.48,86.01) -- (209.48,62.51) ;
%Shape: Ellipse [id:dp836733465485463] 
\draw  [draw opacity=0][fill={rgb, 255:red, 208; green, 2; blue, 27 }  ,fill opacity=1 ] (189.8,54.15) .. controls (189.8,56.23) and (191.57,57.92) .. (193.76,57.92) .. controls (195.95,57.92) and (197.73,56.23) .. (197.73,54.15) .. controls (197.73,52.06) and (195.95,50.37) .. (193.76,50.37) .. controls (191.57,50.37) and (189.8,52.06) .. (189.8,54.15) -- cycle ;
%Shape: Ellipse [id:dp42146295897745456] 
\draw  [draw opacity=0][fill={rgb, 255:red, 208; green, 2; blue, 27 }  ,fill opacity=1 ] (215.61,64.89) .. controls (213.76,65.85) and (213.08,68.2) .. (214.09,70.14) .. controls (215.1,72.08) and (217.42,72.88) .. (219.27,71.92) .. controls (221.11,70.95) and (221.79,68.6) .. (220.78,66.66) .. controls (219.77,64.72) and (217.46,63.92) .. (215.61,64.89) -- cycle ;
%Straight Lines [id:da47849902868326266] 
\draw  [dash pattern={on 1.5pt off 1.5pt}]  (69.9,63.9) -- (55.2,88.43) ;
%Straight Lines [id:da9336202138140786] 
\draw  [dash pattern={on 1.5pt off 1.5pt}]  (101.7,82) -- (87.06,106.86) ;
%Straight Lines [id:da44378791977042353] 
\draw  [dash pattern={on 1.5pt off 1.5pt}]  (187.4,49.4) -- (172.7,73.93) ;
%Straight Lines [id:da9932719920087978] 
\draw  [dash pattern={on 1.5pt off 1.5pt}]  (226.7,72) -- (212.06,96.86) ;
%Shape: Brace [id:dp7917925442628353] 
\draw   (174.2,77.93) .. controls (171.59,81.79) and (172.21,85.03) .. (176.07,87.65) -- (178.25,89.12) .. controls (183.77,92.85) and (185.22,96.65) .. (182.6,100.52) .. controls (185.22,96.65) and (189.29,96.59) .. (194.81,100.33)(192.32,98.64) -- (196.98,101.8) .. controls (200.85,104.41) and (204.09,103.79) .. (206.7,99.92) ;

% Text Node
\draw (191.67,34.03) node [anchor=north west][inner sep=0.75pt]    {$a$};
% Text Node
\draw (220.67,49.03) node [anchor=north west][inner sep=0.75pt]    {$\overline{a}$};
% Text Node
\draw (19.67,14.26) node [anchor=north west][inner sep=0.75pt]    {$a)$};
% Text Node
\draw (132.67,11.2) node [anchor=north west][inner sep=0.75pt]    {$b)$};
% Text Node
\draw (33.33,87.54) node [anchor=north west][inner sep=0.75pt]  [font=\small]  {$B_{1}$};
% Text Node
\draw (83.93,113.84) node [anchor=north west][inner sep=0.75pt]  [font=\small]  {$B_{2}$};
% Text Node
\draw (59.93,100.34) node [anchor=north west][inner sep=0.75pt]  [font=\small]  {$C$};
% Text Node
\draw (152.33,75.04) node [anchor=north west][inner sep=0.75pt]  [font=\small]  {$B_{1}$};
% Text Node
\draw (205.93,105.34) node [anchor=north west][inner sep=0.75pt]  [font=\small]  {$B_{2}$};
% Text Node
\draw (169.93,102.84) node [anchor=north west][inner sep=0.75pt]  [font=\small]  {$C$};
\end{tikzpicture}
\caption{We illustrate the action of the edge disentangler $U_2$ for a state with the $N$-type sector $a$. The state in (a) is in sector $a$ of its information convex set.  The state in (b) is a product between the upper-left and lower-right pieces, hence we say the edge region is ``factorized'' or ``disentangled.'' The disentangling unitary maps state (a) to state (c), essentially by replacing region $C$ of (a) with region $C$ of (b).}
\end{figure}
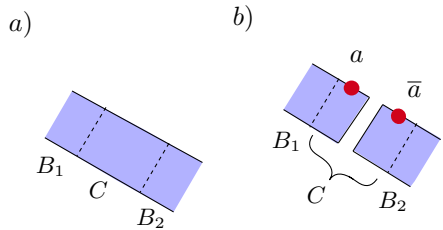

After applying $U_2$, we obtain a new state $\sigma^{(2)}$. Again using the same purification, we get
\begin{equation}
    |\sigma^{(2)}\rangle = U_2 |\sigma^{(1)}\rangle.
\end{equation}
The state $|\sigma^{(2)}\rangle$ can be written as a superposition of a particular basis set. This basis set consists of a tensor product of pure states that represent a disk-like fragment with a gapped boundary, each hosting at most three boundary anyons (corresponding to the $N$-type sector) and their multiplicities (corresponding to the $M$-type sector, once the boundary anyon types are fixed) [Figure~\ref{fig:vertex-states_explanation} and~\ref{fig:bridges}]. Once we fix the anyon sectors as $a, b,$ and $c$, we obtain a Hilbert space $S_v^{a,b,c}\cong \mathbb{V}_{ab}^c$ associated with the vertex $v$. The resulting global Hilbert space is 
\begin{align}
    \mathcal{H}_0' = \bigoplus_{\{a_v\}, \{b_v \}, \{c_v \}} \bigotimes_v S^{a_v,b_v,c_v}_v, \label{eq:string_net_basis_from_eb}
\end{align} 
with an additional constraint that two neighboring sectors connected by an edge must be identical. This is isomorphic to the ``string-net Hilbert space'' $\mathcal{H}_0$ of Eq.~\eqref{eq:sn_space}, described in Section \ref{sec:sn-review}. Therefore, there is an isometry from the string-net Hilbert space  to $\mathcal{H}_0'$, which is clearly a depth-$1$ circuit. After introducing extra ancillary degrees of freedom to turn the isometry into a unitary, we shall refer to the inverse of this unitary as $U_3$. This unitary implements 
\begin{equation}
    U_3|\sigma^{(2)}\rangle = |\sigma^{(3)}\rangle |\phi\rangle, 
    \label{eq:sigma2_to_sigma3}
\end{equation}
where $|\sigma^{(3)}\rangle$ is a state living in the string-net Hilbert space $\mathcal{H}_0$ and $|\phi\rangle$ is a product state.

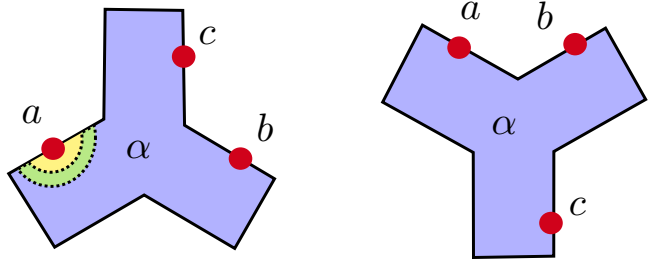
\begin{figure}[htbp]
    \centering
    \scalebox{1.5}{
     \tikzset{every picture/.style={line width=0.75pt}} %set default line width to 0.75pt        

\begin{tikzpicture}[x=0.75pt,y=0.75pt,yscale=-1,xscale=1]
%uncomment if require: \path (0,209); %set diagram left start at 0, and has height of 209

%Shape: Polygon [id:ds08324314589781534] 
\draw  [fill={rgb, 255:red, 178; green, 178; blue, 255 }  ,fill opacity=1 ] (58.25,14.61) -- (84.36,14.98) -- (84.95,53.63) -- (115.21,71.64) -- (101.82,95.07) -- (71.38,77.1) -- (41.28,94.65) -- (25.83,69.91) -- (57.81,51.58) -- cycle ;
%Shape: Path Data [id:dp3768412021772529] 
\draw  [draw opacity=0][fill={rgb, 255:red, 255; green, 243; blue, 131 }  ,fill opacity=1 ] (37.63,71.5) .. controls (34.41,71.5) and (31.41,70.38) .. (28.87,68.45) -- (54,54.05) .. controls (52.96,63.89) and (46.03,71.5) .. (37.63,71.5) -- cycle ;
%Shape: Polygon [id:ds8393240064587804] 
\draw  [fill={rgb, 255:red, 178; green, 178; blue, 255 }  ,fill opacity=1 ] (226.5,20.95) -- (240,44.95) -- (209.12,62.46) -- (209.09,97.68) -- (182.1,98.13) -- (182.02,62.78) -- (151.56,45.85) -- (164.94,19.94) -- (197.04,38.08) -- cycle ;
%Shape: Ellipse [id:dp4364589548720583] 
\draw  [draw opacity=0][fill={rgb, 255:red, 208; green, 2; blue, 27 }  ,fill opacity=1 ] (36.8,61.51) .. controls (36.8,63.6) and (38.57,65.29) .. (40.76,65.29) .. controls (42.95,65.29) and (44.73,63.6) .. (44.73,61.51) .. controls (44.73,59.43) and (42.95,57.74) .. (40.76,57.74) .. controls (38.57,57.74) and (36.8,59.43) .. (36.8,61.51) -- cycle ;
%Shape: Ellipse [id:dp03897625138304672] 
\draw  [draw opacity=0][fill={rgb, 255:red, 208; green, 2; blue, 27 }  ,fill opacity=1 ] (80.69,30.53) .. controls (80.69,32.61) and (82.47,34.3) .. (84.66,34.3) .. controls (86.84,34.3) and (88.62,32.61) .. (88.62,30.53) .. controls (88.62,28.44) and (86.84,26.75) .. (84.66,26.75) .. controls (82.47,26.75) and (80.69,28.44) .. (80.69,30.53) -- cycle ;
%Shape: Ellipse [id:dp9251646247039829] 
\draw  [draw opacity=0][fill={rgb, 255:red, 208; green, 2; blue, 27 }  ,fill opacity=1 ] (99.8,64.85) .. controls (99.8,66.93) and (101.57,68.62) .. (103.76,68.62) .. controls (105.95,68.62) and (107.73,66.93) .. (107.73,64.85) .. controls (107.73,62.76) and (105.95,61.07) .. (103.76,61.07) .. controls (101.57,61.07) and (99.8,62.76) .. (99.8,64.85) -- cycle ;
%Shape: Ellipse [id:dp761253015817025] 
\draw  [draw opacity=0][fill={rgb, 255:red, 208; green, 2; blue, 27 }  ,fill opacity=1 ] (173.3,27.51) .. controls (173.3,29.6) and (175.07,31.29) .. (177.26,31.29) .. controls (179.45,31.29) and (181.23,29.6) .. (181.23,27.51) .. controls (181.23,25.43) and (179.45,23.74) .. (177.26,23.74) .. controls (175.07,23.74) and (173.3,25.43) .. (173.3,27.51) -- cycle ;
%Shape: Ellipse [id:dp26950908394419737] 
\draw  [draw opacity=0][fill={rgb, 255:red, 208; green, 2; blue, 27 }  ,fill opacity=1 ] (212.3,26.35) .. controls (212.3,28.43) and (214.07,30.12) .. (216.26,30.12) .. controls (218.45,30.12) and (220.23,28.43) .. (220.23,26.35) .. controls (220.23,24.26) and (218.45,22.57) .. (216.26,22.57) .. controls (214.07,22.57) and (212.3,24.26) .. (212.3,26.35) -- cycle ;
%Shape: Ellipse [id:dp20991317235479046] 
\draw  [draw opacity=0][fill={rgb, 255:red, 208; green, 2; blue, 27 }  ,fill opacity=1 ] (204.19,86.53) .. controls (204.19,88.61) and (205.97,90.3) .. (208.16,90.3) .. controls (210.34,90.3) and (212.12,88.61) .. (212.12,86.53) .. controls (212.12,84.44) and (210.34,82.75) .. (208.16,82.75) .. controls (205.97,82.75) and (204.19,84.44) .. (204.19,86.53) -- cycle ;
%Shape: Path Data [id:dp41551407873939894] 
\draw  [fill={rgb, 255:red, 184; green, 233; blue, 134 }  ,fill opacity=1 ][dash pattern={on 0.75pt off 0.75pt}] (55,57.58) .. controls (55,66.76) and (48.34,74.2) .. (40.13,74.2) .. controls (35.63,74.2) and (31.6,71.98) .. (28.87,68.45) -- (32.64,66.3) .. controls (34.53,68.25) and (37.13,69.45) .. (40,69.45) .. controls (45.8,69.45) and (50.5,64.53) .. (50.5,58.45) .. controls (50.5,57.67) and (50.42,56.92) .. (50.28,56.18) -- (54.59,53.71) .. controls (54.86,54.95) and (55,56.25) .. (55,57.58) -- cycle ;

% Text Node
\draw (24.67+4,39.67+6) node [anchor=north west][inner sep=0.75pt]    {$a$};
% Text Node
\draw (107.67,44.33+4) node [anchor=north west][inner sep=0.75pt]    {$b$};
% Text Node
\draw (88.17,15+4) node [anchor=north west][inner sep=0.75pt]    {$c$};
% Text Node
\draw (63.67,49.33+8) node [anchor=north west][inner sep=0.75pt]    {$\alpha $};
% Text Node
\draw (176.17,2.5+8) node [anchor=north west][inner sep=0.75pt]    {$a$};
% Text Node
\draw (201.67,4.5+6) node [anchor=north west][inner sep=0.75pt]    {$b$};
% Text Node
\draw (212.67,72.5+4) node [anchor=north west][inner sep=0.75pt]    {$c$};
% Text Node
\draw (186.67,41.33+8) node [anchor=north west][inner sep=0.75pt]    {$\alpha $};

\end{tikzpicture}}
    \caption{For each vertex $v$, for each choice of anyons $a,b,c$ with $N_{ab}^c>0$, and each multiplicity label $\alpha$ labeling a basis of $\mathbb{V}^c_{ab}$, we define a corresponding pure state $|\psi_v^{a,b,c;\alpha}\rangle_v$ on the associated vertex region. The vertex regions come in two shapes (pictured left and right). The state $|\psi_v^{a,b,c;\alpha}\rangle_v$ has anyons $a,b,c$ located at the red dots, meaning e.g.\ the green half-annulus surrounding $a$ is in the sector of the information convex corresponding to $a$. The yellow region with the anyon $a$ at location $x$ must match the canonical purification associated with $a$, discussed in Section~\ref{subsec:basic_operations}.}
    
    \label{fig:vertex-states_explanation}
\end{figure}

\begin{figure}[htbp]
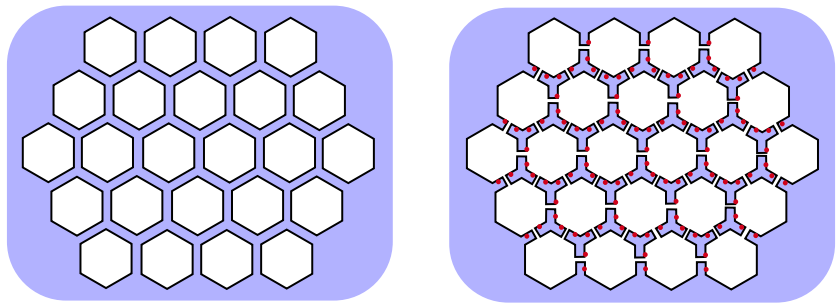

        \centering
        \include{figs/disentangle_step2}
    \caption{By applying a depth-$1$ circuit, we can disentangle $\sigma^{(1)}$ (a) to a state $\sigma^{(2)}$ which is a superposition of vertex fragments (b) with boundary anyons (red).}
    \label{fig:bridges}
\end{figure}

To summarize, we constructed three unitaries, $U_1,$ $U_2$, and $U_3$ such that if we apply $U_3U_2U_1$ on the reference state, we obtain a state $|\sigma^{(3)}\rangle$ in the string-net Hilbert space $\mathcal{H}_0$ and some ancillary degrees of freedom in a product state. By construction, the resulting state $|\sigma^{(3)}\rangle$ satisfies the vertex constraint $Q_I$; see Section~\ref{sec:ground_state}. Now the remaining question is whether the plaquette constraint is satisfied on the state we obtained. 

\subsection{Choosing string operators and basis states}
\label{subsec:choosing-phases}

To check the state constructed in the previous Section is indeed a string-net, we will verify that it is the ground state of the Levin-Wen Hamiltonian.  
Before we delve into these details, let us discuss an important subtlety we will need to resolve first. The subtlety comes from the fact that we currently do not have any single ``reference frame'' for the string operators. For instance, consider the two string operators that act on two disjoint gapped boundaries. These two operators are a priori unrelated to each other because the basic operations [Section~\ref{subsec:basic_operations}] that make up these two operators are independent. In particular, the $F$-symbols defined on these two disjoint boundaries will be generally different from one another. In contrast, in the string-net Hamiltonian, the choice of $F$-symbol at different locations are defined consistently. To derive a string-net Hamiltonian, we must devise a convention that defines all the string operators in a consistent manner. 

Our approach to this problem is to choose a single disk from which all the string operators can be constructed in a consistent way. To that end, consider a sprawling disk-like region with boundary in Fig.~\ref{fig:merge-spanning-tree}(b).  A state on this region can be built from elementary fragments in Figure~\ref{fig:merge-spanning-tree}(a) and the edge fragments using the process described in Section~\ref{subsec:punching_holes}. One can view this region as a \emph{spanning tree} of a finite disk-like subset of the hexagonal lattice [Figure~\ref{fig:string_net}], which is then fattened.\footnote{A spanning tree of a graph is a subgraph containing every vertex and no loops.} Because the tree has no loops, the fattened tree is a disk-like 2D region. We therefore refer to this region as the (fattened) tree. If we consider the boundary of the fattened tree and excise a point (black `x' in Fig.~\ref{fig:merge-spanning-tree}(b)), we are left with a boundary interval. Now we can define the boundary anyon operations (movement, splitting, and fusion operations) just as described in Section \ref{sec:boundary_anyons}.

\begin{figure}[htbp]
    \centering
    \scalebox{1.3}{
     \input{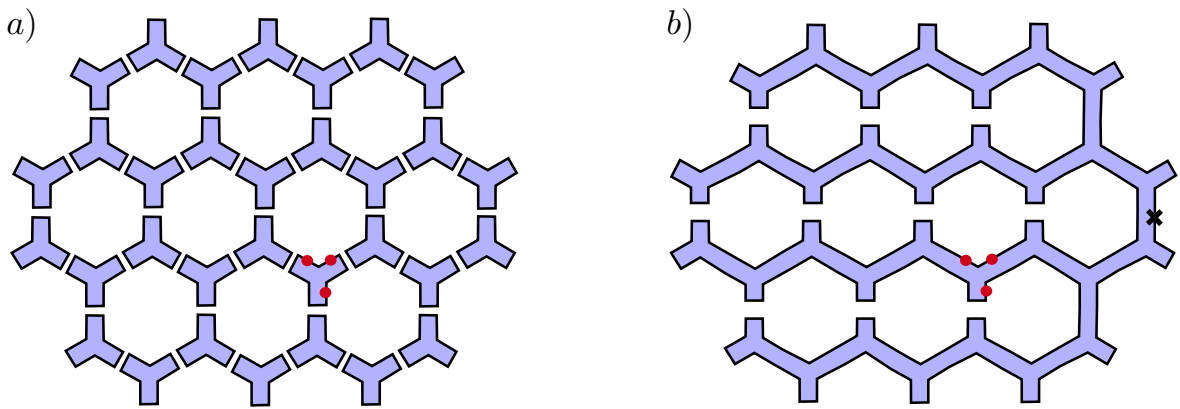}}
    \caption{ The states on the vertex fragments in (a) are merged to create the sprawling disk-like region in (b), using the method from Section~\ref{subsec:punching_holes}. The disk-like region in (b) is called the ``(fat) tree.'' It is used to define boundary anyon operations. The black `x' marks an excised point on the boundary.}
    \label{fig:merge-spanning-tree}
\end{figure}

\begin{figure}[htbp]
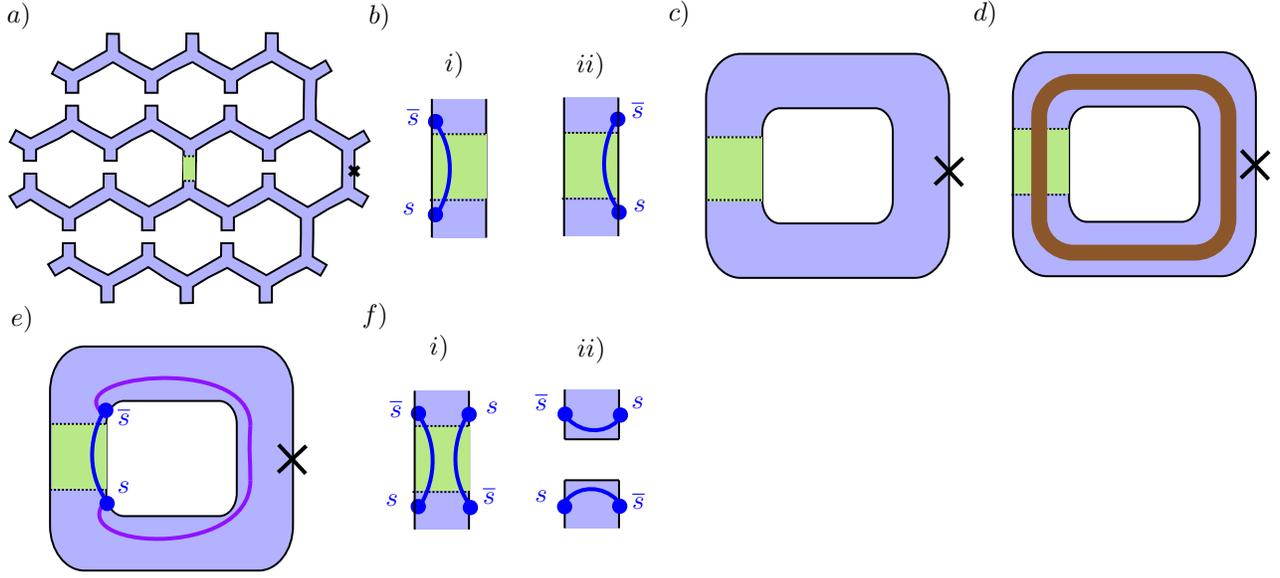

    \centering
    \include{figs/spanning_tree_with_a_junction}
    \caption{Subsystems used in the definition of $S_L$ and $S_R$, the blue string operators in (b)(i) and (b)(ii), respectively.
    \textbf{(a)} A union of the spanning tree and the junction (green). \textbf{(c)} The resulting system is topologically an annulus. \textbf{(d)} We consider the state consistent with the vacuum sector on the brown annulus. \textbf{(e)} The operator $S_R$ is chosen such that its action on the annulus is equal to the action of the splitting + movement operator (purple). \textbf{(f)} The action of $S_L$ is defined by relating the states in (f)(ii) and (f)(ii); see main text.}
\label{fig:spanning_tree_with_a_junction}
\end{figure}

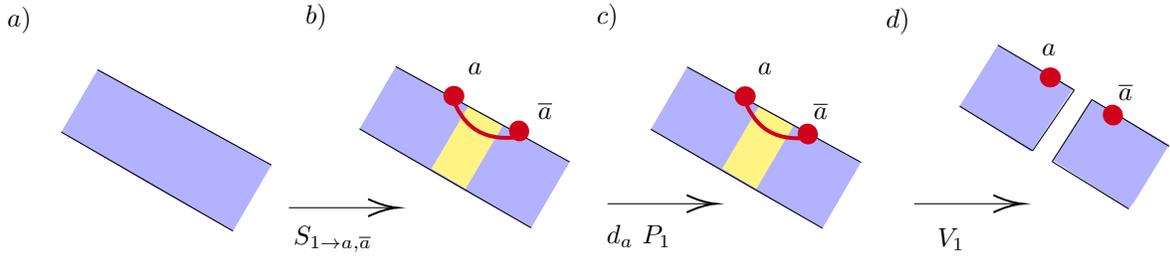
\begin{figure}[htbp] 
    \centering 
\begin{tikzpicture}[x=1pt,y=1pt,yscale=-1,xscale=1]
%uncomment if require: \path (0,420); %set diagram left start at 0, and has height of 420

%Shape: Path Data [id:dp15312747088555256] 
\draw  [draw opacity=0][fill={rgb, 255:red, 178; green, 178; blue, 255 }  ,fill opacity=1 ] (408.62,229.19) -- (423.21,206.11) -- (452.5,223.17) -- (439.3,244.81) -- (437.17,247.55) -- (408.62,229.19) -- cycle ;
%Shape: Path Data [id:dp7982948052605896] 
\draw  [draw opacity=0][fill={rgb, 255:red, 178; green, 178; blue, 255 }  ,fill opacity=1 ] (83.31,214.21) -- (112.9,230.93) -- (98.34,255.91) -- (69.93,240.12) -- (33.39,218.42) -- (47.14,194.84) -- (83.31,214.21) -- cycle ;
%Straight Lines [id:da45932395708294727] 
\draw    (33.39,218.42) -- (98.34,255.91) ;
%Straight Lines [id:da29388510573900994] 
\draw    (47.14,194.84) -- (112.9,230.93) ;
%Shape: Path Data [id:dp8096474343608298] 
\draw  [draw opacity=0][fill={rgb, 255:red, 178; green, 178; blue, 255 }  ,fill opacity=1 ] (373.19,208.64) -- (387.68,185.5) -- (416.2,201.96) -- (400.16,225.83) -- (373.19,208.64) -- cycle ;
%Straight Lines [id:da4062902836519826] 
\draw    (373.19,208.64) -- (400.81,225.52) ;
%Straight Lines [id:da8419374752826014] 
\draw    (387.68,185.5) -- (416.55,202.51) ;
%Straight Lines [id:da7576592557565944] 
\draw    (400.81,225.52) -- (416.55,202.51) ;
%Straight Lines [id:da07960546115247569] 
\draw    (408.21,229.61) -- (437.17,247.55) ;
%Straight Lines [id:da8007891292583906] 
\draw    (423.21,206.11) -- (452.43,223.82) ;
%Straight Lines [id:da6143570263655975] 
\draw    (408.21,229.61) -- (423.21,206.11) ;
%Shape: Ellipse [id:dp7435589945522265] 
\draw  [draw opacity=0][fill={rgb, 255:red, 208; green, 2; blue, 27 }  ,fill opacity=1 ] (403.53,197.75) .. controls (403.53,199.83) and (405.31,201.52) .. (407.5,201.52) .. controls (409.68,201.52) and (411.46,199.83) .. (411.46,197.75) .. controls (411.46,195.66) and (409.68,193.97) .. (407.5,193.97) .. controls (405.31,193.97) and (403.53,195.66) .. (403.53,197.75) -- cycle ;
%Shape: Ellipse [id:dp3784457761769491] 
\draw  [draw opacity=0][fill={rgb, 255:red, 208; green, 2; blue, 27 }  ,fill opacity=1 ] (429.34,208.49) .. controls (427.49,209.45) and (426.81,211.8) .. (427.82,213.74) .. controls (428.83,215.68) and (431.15,216.48) .. (433,215.52) .. controls (434.85,214.55) and (435.53,212.2) .. (434.52,210.26) .. controls (433.51,208.32) and (431.19,207.52) .. (429.34,208.49) -- cycle ;
%Shape: Path Data [id:dp1044923408815337] 
\draw  [draw opacity=0][fill={rgb, 255:red, 178; green, 178; blue, 255 }  ,fill opacity=1 ] (196.11,213.01) -- (225.7,229.73) -- (211.14,254.71) -- (182.73,238.92) -- (146.19,217.22) -- (159.94,193.64) -- (196.11,213.01) -- cycle ;
%Straight Lines [id:da2538468302145893] 
\draw    (146.19,217.22) -- (211.14,254.71) ;
%Shape: Path Data [id:dp10422095861205216] 
\draw  [draw opacity=0][fill={rgb, 255:red, 178; green, 178; blue, 255 }  ,fill opacity=1 ] (306.31,213.16) -- (335.9,229.88) -- (321.34,254.86) -- (292.93,239.07) -- (256.39,217.38) -- (270.14,193.79) -- (306.31,213.16) -- cycle ;
%Straight Lines [id:da09382506240590205] 
\draw    (256.39,217.38) -- (321.34,254.86) ;
%Shape: Rectangle [id:dp9409886957654647] 
\draw  [draw opacity=0][fill={rgb, 255:red, 255; green, 243; blue, 131 }  ,fill opacity=1 ] (297.41,208.53) -- (310.33,216.14) -- (296.32,239.93) -- (283.39,232.31) -- cycle ;
%Straight Lines [id:da9590810496238644] 
\draw    (270.14,193.79) -- (335.9,229.88) ;
%Shape: Ellipse [id:dp447236612433368] 
\draw  [draw opacity=0][fill={rgb, 255:red, 208; green, 2; blue, 27 }  ,fill opacity=1 ] (288.13,204.9) .. controls (288.13,206.98) and (289.91,208.67) .. (292.1,208.67) .. controls (294.28,208.67) and (296.06,206.98) .. (296.06,204.9) .. controls (296.06,202.81) and (294.28,201.12) .. (292.1,201.12) .. controls (289.91,201.12) and (288.13,202.81) .. (288.13,204.9) -- cycle ;
%Shape: Boxed Bezier Curve [id:dp9774628504663223] 
\draw [color={rgb, 255:red, 208; green, 2; blue, 27 }  ,draw opacity=1 ][line width=1.5]    (316.44,220.47) .. controls (309.16,221.38) and (303.67,220.17) .. (299.48,217.07) .. controls (295.93,214.44) and (293.31,210.45) .. (291.33,205.24) ;
%Shape: Ellipse [id:dp15482087449771909] 
\draw  [draw opacity=0][fill={rgb, 255:red, 208; green, 2; blue, 27 }  ,fill opacity=1 ] (313.94,215.64) .. controls (312.09,216.6) and (311.41,218.95) .. (312.42,220.9) .. controls (313.43,222.84) and (315.75,223.63) .. (317.6,222.67) .. controls (319.45,221.71) and (320.13,219.35) .. (319.12,217.41) .. controls (318.11,215.47) and (315.79,214.68) .. (313.94,215.64) -- cycle ;
%Straight Lines [id:da9558100749721876] 
\draw    (119.5,247.07) -- (158.4,247.07) ;
\draw [shift={(160.4,247.07)}, rotate = 180] [color={rgb, 255:red, 0; green, 0; blue, 0 }  ][line width=0.75]    (10.93,-3.29) .. controls (6.95,-1.4) and (3.31,-0.3) .. (0,0) .. controls (3.31,0.3) and (6.95,1.4) .. (10.93,3.29)   ;
%Straight Lines [id:da319319068261098] 
\draw    (240,245.57) -- (278.9,245.57) ;
\draw [shift={(280.9,245.57)}, rotate = 180] [color={rgb, 255:red, 0; green, 0; blue, 0 }  ][line width=0.75]    (10.93,-3.29) .. controls (6.95,-1.4) and (3.31,-0.3) .. (0,0) .. controls (3.31,0.3) and (6.95,1.4) .. (10.93,3.29)   ;
%Straight Lines [id:da7436558858781501] 
\draw    (356,245.57) -- (394.9,245.57) ;
\draw [shift={(396.9,245.57)}, rotate = 180] [color={rgb, 255:red, 0; green, 0; blue, 0 }  ][line width=0.75]    (10.93,-3.29) .. controls (6.95,-1.4) and (3.31,-0.3) .. (0,0) .. controls (3.31,0.3) and (6.95,1.4) .. (10.93,3.29)   ;
%Shape: Rectangle [id:dp06132336558766194] 
\draw  [draw opacity=0][fill={rgb, 255:red, 255; green, 243; blue, 131 }  ,fill opacity=1 ] (187.41,208.53) -- (200.33,216.14) -- (186.32,239.93) -- (173.39,232.31) -- cycle ;
%Shape: Boxed Bezier Curve [id:dp03522941389854228] 
\draw [color={rgb, 255:red, 208; green, 2; blue, 27 }  ,draw opacity=1 ][line width=1.5]    (206.24,220.31) .. controls (198.96,221.23) and (193.47,220.02) .. (189.28,216.91) .. controls (185.73,214.28) and (183.11,210.3) .. (181.13,205.09) ;
%Straight Lines [id:da7247195112195417] 
\draw    (159.94,193.64) -- (183.18,206.39) -- (225.7,229.73) ;
%Shape: Ellipse [id:dp09036890570477696] 
\draw  [draw opacity=0][fill={rgb, 255:red, 208; green, 2; blue, 27 }  ,fill opacity=1 ] (177.93,204.75) .. controls (177.93,206.83) and (179.71,208.52) .. (181.9,208.52) .. controls (184.08,208.52) and (185.86,206.83) .. (185.86,204.75) .. controls (185.86,202.66) and (184.08,200.97) .. (181.9,200.97) .. controls (179.71,200.97) and (177.93,202.66) .. (177.93,204.75) -- cycle ;
%Shape: Ellipse [id:dp35196288755667693] 
\draw  [draw opacity=0][fill={rgb, 255:red, 208; green, 2; blue, 27 }  ,fill opacity=1 ] (204.74,214.49) .. controls (202.89,215.45) and (202.21,217.8) .. (203.22,219.74) .. controls (204.23,221.68) and (206.55,222.48) .. (208.4,221.52) .. controls (210.25,220.55) and (210.93,218.2) .. (209.92,216.26) .. controls (208.91,214.32) and (206.59,213.52) .. (204.74,214.49) -- cycle ;

% Text Node

\draw (403.4,174.23+10) node [anchor=north west][inner sep=0.75pt]    {$a$};

% Text Node

\draw (432.4,189.23+10) node [anchor=north west][inner sep=0.75pt]    {$\overline{a}$};

% Text Node

\draw (12,150.06+20) node [anchor=north west][inner sep=0.75pt]    {$a)$};

% Text Node

\draw (344.4,150.4+20) node [anchor=north west][inner sep=0.75pt]    {$d)$};

% Text Node

\draw (124.8,148.86+20) node [anchor=north west][inner sep=0.75pt]    {$b)$};

% Text Node

\draw (178.3+8,181.23+10) node [anchor=north west][inner sep=0.75pt]    {$a$};

% Text Node

\draw (206.8+6,196.23+10) node [anchor=north west][inner sep=0.75pt]    {$\overline{a}$};

% Text Node

\draw (235,149.01+20) node [anchor=north west][inner sep=0.75pt]    {$c)$};

% Text Node

\draw (288+8,181.39+10) node [anchor=north west][inner sep=0.75pt]    {$a$};
% Text Node
\draw (317,196.39+10) node [anchor=north west][inner sep=0.75pt]    {$\overline{a}$};
% Text Node
\draw (112.3+8,251.23+2) node [anchor=north west][inner sep=0.75pt]    {$S_{1\rightarrow a,\overline{a}}$};
% Text Node
\draw (238.8,252.73) node [anchor=north west][inner sep=0.75pt]    {$d_{a} \ P_{1}$};
% Text Node
\draw (363.8,253.73) node [anchor=north west][inner sep=0.75pt]    {$V_{1}$};

\end{tikzpicture}
    \caption{ We illustrate the edge disentangler $V_a^{(e)}$ for an edge region with the $N$-type sector $a$. \textbf{(a)} The state begin in sector $a$ of its information convex set. \textbf{(b)} We apply splitting operator $S_{1 \to a, \bar{a}}$ . The state in the yellow region is then in the maximum entropy state of its information convex.  \textbf{(c)} We apply the operator $P_1$ on the yellow region to project it to its vacuum sector, with factor $d_a$ to preserve the norm. \textbf{(d)} We apply the vacuum disentangler $V_1^{(e)}$ to the yellow region. By Uhlmann's theorem, there is an isometry that maps (a) to (d). This is our definition of $V_a^{(e)}$.}
    The state in now a product between the upper-left and lower-right pieces, hence we say the edge region is ``factorized'' or ``disentangled.''
    \label{fig:edge-disentanglers_action} 
\end{figure}

Note the edge regions have three possible alignments, as in Figure~\ref{fig:bridge_single}(a).  We call these vertical edges (shown there in green) and non-vertical edges (yellow).  Note the non-vertical edge regions also appear on the spanning tree.  We can then define anyon operations on the non-vertical edge regions by using those defined for the spanning tree.  
Meanwhile, the vertical edge regions do not appear on the spanning tree. Later we will need string operators defined on these vertical edge regions too, and we will want them to be somehow consistent with the anyon operations on the remaining regions, allowing consistent diagrammatic calculations in Section \ref{subsec:plaquette}. 

To define the string operators on vertical edge regions, we follow Figure~\ref{fig:spanning_tree_with_a_junction}. We are concerned with a region such as the green region in Figure~\ref{fig:spanning_tree_with_a_junction}(a); we call the green region a \textit{junction}, because it does not appear on the spanning tree, but it does appear on the underlying geometry of the state $|\sigma^{(1)}\rangle$.  Our goal is to define splitting operators  $S_{1 \to \bar{s},s}$ that straddle the junction. We define two versions, $S_L$ and $S_R$.  These are the blue splitting operators shown in Figure~\ref{fig:spanning_tree_with_a_junction}(b)(i) and (ii). The action of $S_L, S_R$ is already determined up to complex phases, and our goal is to specify these phases.

We begin by defining $S_R$.  If we consider the junction region merged into the spanning tree, as in Figure~\ref{fig:spanning_tree_with_a_junction}(a), the resulting state has the geometry of an annulus, re-drawn in Figure~\ref{fig:spanning_tree_with_a_junction}(c) for convenience.  It will be helpful to consider a particular state with this geometry, namely where the brown annulus in Figure~\ref{fig:spanning_tree_with_a_junction}(a) is in its vacuum sector. This state is pure because the vacuum sector is an extreme point [Section~\ref{subsec:extreme_points}]; we will denote this state as $|\phi_e\rangle$. Then define the operator $S_R$ such that $S_R|\phi_e\rangle = M^s|\phi_e\rangle$, where $M^s$ is the purple splitting + movement operator shown in Figure~\ref{fig:spanning_tree_with_a_junction}(e) .
 Note $M^s$ is already defined because it lives on the spanning tree. The existence of such  a $S_R$ follows from Uhlmann's theorem \cite{Uhlmann1976}. Note that the complement of the blue string is a union of two disjoint disks. The two states $M^s|\phi_e\rangle$ and $|\phi_e\rangle$ are indistinguishable over each disk. Moreover, because of the boundary \textbf{A0}, the reduced state over the union of the two are in a product form. Therefore, $M^s|\phi_e\rangle$ and $|\phi_e\rangle$ are indistinguishable over the union of the two disks and the existence of the $S_R$ with the requisite property follows. 

Finally, we define the operator $S_L$ [Figure~\ref{fig:spanning_tree_with_a_junction}(b)(i)]  by referring to Figure~\ref{fig:spanning_tree_with_a_junction}(f).  We require that if we act with $S_L S_R$ in Figure~\ref{fig:spanning_tree_with_a_junction}(f),
then project the green junction to its vacuum sector, and then apply the vacuum disentangler to the green region, we obtain the same state (up to a real factor $d_s$) as when we act with the two blue string operators in Figure~\ref{fig:spanning_tree_with_a_junction}(f)(ii); the latter are already defined on the spanning tree.

The above choices uniquely define the complex phases associated to the splitting operators  $S_L, S_R$  at the vertical junctions. In particular, their action on $|\sigma^{(2)}\rangle$ can be reduced to the action of the movement and splitting operators acting on the spanning tree. The fact that this is possible may not be obvious because the operators $S_L$ and $S_R$ are defined with respect to some auxiliary states (over the union of the spanning tree and a junction). To that end, let us note the following fact.
\begin{lemma}
\label{lemma:action_unitary_different_states}
    Let $|\psi\rangle_{AB}$ and $|\phi\rangle_{AB'}$ be two states that are indistinguishable over $A$. For any operator $O$ acting on $A$, $(O\otimes I_B)|\psi\rangle_{AB} = |\psi\rangle_{AB}$ if and only if $(O\otimes I_{B'})|\psi\rangle_{AB'} = |\psi\rangle_{AB'}$.
\end{lemma}
\begin{proof}
    The proof follows straightforwardly from Uhlmann's theorem~\cite{Uhlmann1976}.
\end{proof}

With these string operations defined, we can return to an omission from Section \ref{subsec:identifying_string_net}. There we noted the existence of a unitary $V_a^{(e)}$, the ``edge disentangler,'' shown in Figure \ref{fig:edge-disentanglers-simple}, disentangling an edge region $e$ in sector $a$. Its action is determined up to an ambiguous complex phase.  We specify this phase by defining the action of  $V_a^{(e)}$ via Figure~\ref{fig:edge-disentanglers_action}.  By Lemma~\ref{lemma:action_unitary_different_states}, the action of  $V_a^{(e)}$ is then specified on any state whose reduced density matrix over edge region $e$ is in the sector $a$ of the information convex.

\begin{figure}[htbp]
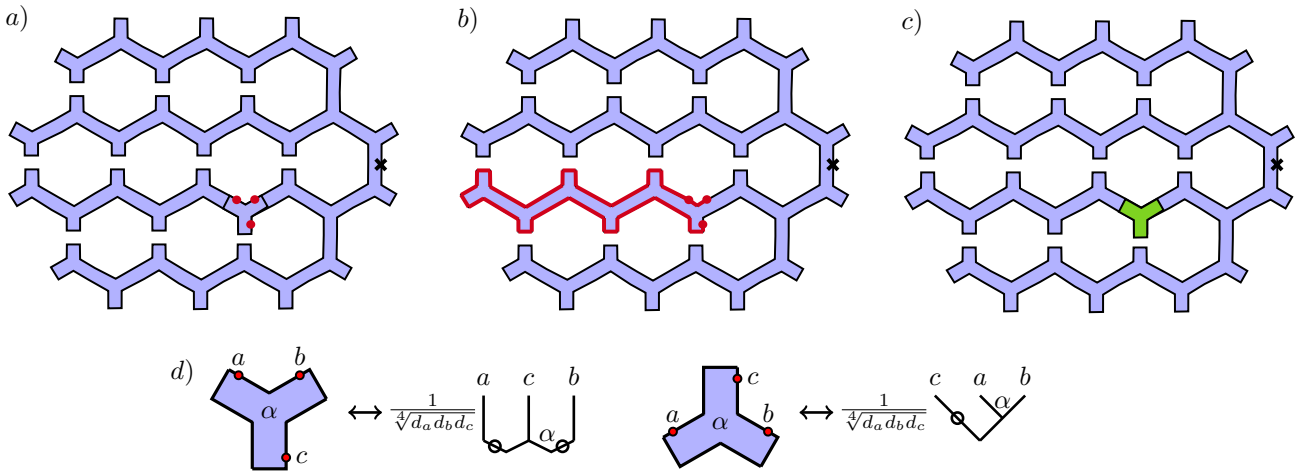

    \centering
    \include{figs/defining_basis}
    \caption{\textbf{(a)} A vertex fragment $v$, viewed as a subset of the spanning tree. \textbf{(b)} Application of the movement/splitting operators (red) on the spanning tree , yielding $|\psi_{\text{tree}}^{v;a,b,c;\alpha}\rangle$. \textbf{(c)} By Uhlmann's theorem, there is a unitary $U_{v}^{a,b,c;\alpha}$ acting on the vertex fragment (green) such that Eq.~\eqref{eq:vertex_uhlmann} is true. \textbf{(d)}  The anyon operations schematically shown in (b) are specified by these diagrams, for the two vertex types.}
    \label{fig:defining_basis}
\end{figure}

We also want to define canonical basis states, including properly chosen phases, for the embedded string-net Hilbert space identified in Section \ref{subsec:identifying_string_net}.
 Let us first recall the structure of the Hilbert space:
\begin{equation}
    \mathcal{H}_0' = \bigoplus_{\{a_v\}, \{b_v \}, \{c_v\}} \bigotimes_v S_v^{a_v, b_v, c_v},
\end{equation}
where $a_v, b_v,$ and $c_v$ are the boundary anyon labels for a vertex $v$. We define an explicit set of basis states for the space $S_v^{a_v, b_v, c_v}$. First, since the vacuum sector (corresponding to $a_v=b_v=c_v=1$) has no multiplicity, there is a single state $|\psi_v^{1,1,1}\rangle_v \in S_v^{1,1,1}$ up to a phase. We choose these phases arbitrarily and hence fix states $|\psi_v^{1,1,1}\rangle_v$ for all $v$.

While the phase chosen for $|\psi_v^{1,1,1}\rangle$ was arbitrary for every $v$, the phase chosen for other basis states will not be. They will be defined in terms of a set of operators that are nontrivially related to each other. More precisely, these states are defined in terms of the movement and the splitting operators [Section~\ref{subsec:basic_operations}], on the spanning tree [Figure~\ref{fig:merge-spanning-tree}(b)]. Consider a vacuum state of the spanning tree, denoted as $|\psi_{\text{tree}}\rangle$. Let us focus on a subsystem of the spanning tree, which corresponds to the vertex fragment $v$ [Figure~\ref{fig:defining_basis}(a)]. There are anyon operations on the spanning tree that create the boundary anyon $a_v, b_v, c_v$, with a multiplicity $\alpha$ [Figure~\ref{fig:defining_basis}(b)]; the resulting state is denoted as $|\psi_{\text{tree}}^{v;a,b,c;\alpha}\rangle$. These anyon operations together are schematically denoted with the red string.  More specifically, we choose the anyon operations on the spanning tree specified by the diagram in Figure~\ref{fig:defining_basis}(d).

Note that $|\psi_{\text{tree}}^{v;a,b,c;\alpha}\rangle$ and $|\psi_{\text{tree}}\rangle$ are indistinguishable over the complement of the vertex $v$ over the spanning tree. Thus by Uhlmann's theorem, there is a unitary $U_v^{a,b,c;\alpha}$ acting on $v$ such that the following equation holds true~[Figure~\ref{fig:defining_basis}(c)]:
\begin{equation}
    |\psi_{\text{tree}}^{v;a,b,c;\alpha}\rangle = U_v^{a,b,c;\alpha}|\psi_{\text{tree}}\rangle. \label{eq:vertex_uhlmann}
\end{equation}
We define the basis vector for the fragment $|\psi_v^{a,b,c;\alpha}\rangle_v\in S_v^{a,b,c}$ as 
\begin{equation}
\label{eq:vertex_basis}
    |\psi_v^{a,b,c;\alpha}\rangle_v = U_v^{a,b,c;\alpha}|\psi_v^{1,1,1}\rangle_v.
\end{equation}

\subsection{Plaquette operators}
\label{subsec:plaquette}

In this Section, we will define a plaquette operator and study its action on $|\sigma^{(3)}\rangle$. We denote this plaquette operator $\tilde{B}_p^{s}$, and it is designed to mirror the behavior of the Levin-Wen plaquette operator $B_p$ [Eq.~\eqref{eq:plaquette}].  However, $\tilde{B}_p^{s}$ is defined independently, acting on the physical Hilbert space associated to $\sigma$. 
We also compute the matrix elements of $\tilde{B}_p^{s}$ with respect to the string-net basis identified in [Eq.~\eqref{eq:string_net_basis_from_eb}], which can be viewed as an embedding of the canonical string-net basis states into the physical Hilbert space. 

The main result of this are that (1) $\tilde{B}_p^{s}$ stabilizes $|\sigma^{(3)}\rangle$, and (2) the matrix elements of $\tilde{B}_p^{s}$ with respect to the string-net basis states identified within the physical Hilbert space coincide exactly with those of the Levin-Wen plaquette operator $B_p$.   This will allow us to conclude $|\sigma^{(3)}\rangle$ matches the string-net ground state defined with respect to the string-net basis states embedded on the physical lattice.

It will be useful to consider how $\tilde{B}_p$ acts before and after the circuits $U_2, U_3$.  We therefore adopt the notation
\begin{align}
    (\tilde{B}_p^s)^{(2)} &= U_3^\dagger \tilde{B}_p^s U_3 \\
     (\tilde{B}_p^s)^{(1)} &= U_2^\dagger U_3^\dagger \tilde{B}_p^s U_3 U_2 
\end{align}
so that $(\tilde{B}_p^s)^{(2)}$ naturally acts on $|\sigma^{(2)}\rangle$ and  $(\tilde{B}_p^s)^{(1)}$ naturally acts on $|\sigma^{(1)}\rangle$.

First define how $(\tilde{B}_p^s)^{(1)}$ acts on $|\sigma^{(1)}\rangle$, where the edge regions are still entangled.  (Then the action of $\tilde{B}_p^s$ is defined via inverse conjugation.) We define
\begin{equation}
\label{eq:plaquette_op_eb}
    (\tilde{B}_p^s)^{(1)} \,
    \begin{tikzpicture}[baseline={([yshift=-0.5ex]current bounding box.center)}, scale=1]
    \draw[thick] (-1, -1) -- ++ (0, 2);
    \draw[thick] (1, -0.9) -- (1.25, 0) -- (1, 0.9);
        \filldraw[blue!30!white] (30:1cm) -- (90:1cm) -- (150:1cm) -- (210:1cm) -- (270:1cm) -- (330:1cm) -- cycle;
        
        \draw[fill=white, thick] (30:0.7cm) -- (90:0.7cm) -- (150:0.7cm) -- (210:0.7cm) -- (270:0.7cm) -- (330:0.7cm) -- cycle;
    \end{tikzpicture}
  = 
    \varkappa_s \,
    \begin{tikzpicture}[baseline={([yshift=-0.5ex]current bounding box.center)}, scale=1]
    
    \draw[thick] (-1.2, -1) -- ++ (0, 2);
    \draw[thick] (1, -0.9) -- (1.25, 0) -- (1, 0.9);
        \filldraw[blue!30!white] (30:1cm) -- (90:1cm) -- (150:1cm) -- (210:1cm) -- (270:1cm) -- (330:1cm) -- cycle;
        
        \draw[fill=white, thick] (30:0.7cm) -- (90:0.7cm) -- (150:0.7cm) -- (210:0.7cm) -- (270:0.7cm) -- (330:0.7cm) -- cycle;

        \draw [red, thick] plot [smooth, tension=0.5] coordinates { (-0.606, 0.2)  (-0.706, 0.4) (0, 0.8) (0.706, 0.4) (0.606, 0.2)};

        \node[] () at (120:1cm) {\color{red} $s$};
        \node[] () at (60:1cm) {\color{red} $\bar{s}$};

        \node[] () at (240:1cm) {\color{red} $\bar{s}$};
        \node[] () at (300:1cm) {\color{red} $s$};

        \node[] () at (0.966, 0.2) {\color{blue} $s$};
        \node[] () at (0.966, -0.2) {\color{blue} $\bar{s}$};

        \node[] () at (-0.966, 0.2) {\color{blue} $\bar{s}$};
        \node[] () at (-0.966, -0.2) {\color{blue} $s$};

        \draw [red, thick] plot [smooth, tension=0.5] coordinates { (-0.606, -0.2)  (-0.706, -0.4) (0, -0.8) (0.706, -0.4) (0.606, -0.2)};

        \draw[blue, thick] plot [smooth, tension=0.5] coordinates {(-0.606, -0.2) (-0.75, 0) (-0.606, 0.2)};

        \draw[blue, thick] plot [smooth, tension=0.5] coordinates {(0.606, -0.2) (0.75, 0) (0.606, 0.2)};
    \end{tikzpicture},
\end{equation}
where $\varkappa_s$ is the Frobenius-Schur indicator of $s$. The red strings are the splitting operators $S_{1\to s, \bar{s}}$  defined with respect to the spanning tree, while the blue strings are given by $S_{s, \bar{s} \to 1}$ using the splitting operators defined at the vertical junctions [Section \ref{subsec:choosing-phases}].

Below we will show the plaquette operator $(\tilde{B}_p^s)^{(1)}$ defined by Eq.~\eqref{eq:plaquette_op_eb} satisfies
\begin{equation}
\label{eq:plaquette_stabilized}
\begin{aligned}
   (\tilde{B}_p^s)^{(1)} \,
    \begin{tikzpicture}[baseline={([yshift=-0.5ex]current bounding box.center)}, scale=1]
    \draw[thick] (-1, -1) -- ++ (0, 2);
    \draw[thick] (1, -0.9) -- (1.25, 0) -- (1, 0.9);
        \filldraw[blue!30!white] (30:1cm) -- (90:1cm) -- (150:1cm) -- (210:1cm) -- (270:1cm) -- (330:1cm) -- cycle;
        
        \draw[fill=white, thick] (30:0.7cm) -- (90:0.7cm) -- (150:0.7cm) -- (210:0.7cm) -- (270:0.7cm) -- (330:0.7cm) -- cycle;
    \end{tikzpicture}
  &= \begin{tikzpicture}[baseline={([yshift=-0.5ex]current bounding box.center)}, scale=0.5]
        \draw[thick] (0,0) -- ++ (1, -1) node [midway, left] {$s$} -- ++ (1, 1) node [midway, left] {$\bar{s}$} -- ++ (1, -1) node [midway, right] {$s$} -- ++ (1,1) node [midway, right] {$\bar{s}$} -- ++ (-2, 2) -- ++ (-2, -2);
    \end{tikzpicture} \cdot \varkappa_s 
    \begin{tikzpicture}[baseline={([yshift=-0.5ex]current bounding box.center)}, scale=1]
    \draw[thick] (-1, -1) -- ++ (0, 2);
    \draw[thick] (1, -0.9) -- (1.25, 0) -- (1, 0.9);
        \filldraw[blue!30!white] (30:1cm) -- (90:1cm) -- (150:1cm) -- (210:1cm) -- (270:1cm) -- (330:1cm) -- cycle;
        
        \draw[fill=white, thick] (30:0.7cm) -- (90:0.7cm) -- (150:0.7cm) -- (210:0.7cm) -- (270:0.7cm) -- (330:0.7cm) -- cycle;
    \end{tikzpicture}
    \\
&= 
  d_s 
  \begin{tikzpicture}[baseline={([yshift=-0.5ex]current bounding box.center)}, scale=1]
    \draw[thick] (-1, -1) -- ++ (0, 2);
    \draw[thick] (1, -0.9) -- (1.25, 0) -- (1, 0.9);
        \filldraw[blue!30!white] (30:1cm) -- (90:1cm) -- (150:1cm) -- (210:1cm) -- (270:1cm) -- (330:1cm) -- cycle;
        
        \draw[fill=white, thick] (30:0.7cm) -- (90:0.7cm) -- (150:0.7cm) -- (210:0.7cm) -- (270:0.7cm) -- (330:0.7cm) -- cycle;
    \end{tikzpicture}.
\end{aligned}
\end{equation}

Thus the state $|\sigma^{(1)}\rangle$ is stabilized by the plaquette operator $(\tilde{B}_p^s)^{(1)}$, up to a factor $d_s$. Likewise, $(\tilde{B}_p^s)^{(2)}$ stabilizes $|\sigma^{(2)}\rangle$, and $\tilde{B}_p^s$ stabilizes $|\sigma^{(3)}\rangle$.

To show Eq.~\eqref{eq:plaquette_stabilized}, first we focus on the claim that $(\tilde{B}_p^s)^{(1)}$ preserves $|\sigma^{(1)}\rangle$ up to some scalar factor, and then we will determine that factor as $d_s$. To see $(\tilde{B}_p^s)^{(1)}|\sigma^{(2)}\rangle \propto |\sigma^{(2)}\rangle $, the key facts about the plaquette operator ~\eqref{eq:plaquette_op_eb} are that (1) it only acts on a neighborhood of the boundary of the hole, (2) while it creates anyon excitations at intermediate steps, afterward there are no anyons, and (3) the annular region surrounding the hole on the left-hand side of Eq.~\eqref{eq:plaquette_op_eb} has zero correlation between its inner and outer boundary regions.  Fact (3) follows because the original state $\sigma$ satisfies $\textbf{A0}$ and the hole was created with a local unitary.  

As a re-statement of fact (3) above, in the left-hand side of the below diagram~\eqref{eq:compactify}, the pink shaded region has zero correlation with any separated regions, including the black inner boundary. We identify the pink region as a single point, so that the resulting topology is a disk with boundary, topologically re-drawn on the right-hand side of diagram~\eqref{eq:compactify}.  

\begin{equation}
\label{eq:compactify}
    \begin{tikzpicture}[baseline={([yshift=-0.5ex]current bounding box.center)}, scale=1]
        \filldraw[blue!30!white] (30:1cm) -- (90:1cm) -- (150:1cm) -- (210:1cm) -- (270:1cm) -- (330:1cm) -- cycle;
        
        \draw[fill=white, thick] (30:0.7cm) -- (90:0.7cm) -- (150:0.7cm) -- (210:0.7cm) -- (270:0.7cm) -- (330:0.7cm) -- cycle;

         \draw[very thick, red!30!white] (30:1cm) -- (90:1cm) -- (150:1cm) -- (210:1cm) -- (270:1cm) -- (330:1cm) -- cycle;
    \end{tikzpicture}
 \leftrightarrow
    \begin{tikzpicture}[baseline={([yshift=-0.5ex]current bounding box.center)}, scale=1]
        \filldraw[blue!30!white] (30:1cm) -- (90:1cm) -- (150:1cm) -- (210:1cm) -- (270:1cm) -- (330:1cm) -- cycle;
        \draw[fill=blue!30!white, thick] (30:1cm) -- (90:1cm) -- (150:1cm) -- (210:1cm) -- (270:1cm) -- (330:1cm) -- cycle;
        \filldraw[red!30!white] (90:0cm) circle (0.1cm);
    \end{tikzpicture}
\end{equation}
Then the disk on the right-hand side above satisfies axioms $\textbf{A0}$ and $\textbf{A1}$ for a disk with boundary, including for regions containing the pink point, using the above fact.  Therefore its information convex is trivial [Section \ref{subsec:information_convex_set}]. Returning to the plaquette operator, it is localized away from the pink region, and it leaves no anyons.  Thus it produces the unique state in the information convex, up to a scalar.  We conclude that $(\tilde{B}_p^s)^{(1)}|\sigma^{(2)}\rangle \propto |\sigma^{(2)}\rangle$.

Now we check the scalar value in the above proportionality.  The key point is that the reduced density matrices over the union of the blue and the purple string in Figure~\ref{fig:spanning_tree_with_a_junction}(e) are indistinguishable from the same region in $|\sigma^{(2)}\rangle$. This is because the annulus formed by the two strings is in the vacuum sector by construction, which uniquely determines the element within the information convex set of this annulus. Therefore, the action of the blue strings in \eqref{eq:plaquette_op_eb} on $|\sigma^{(2)}\rangle$ are equal to that strings in Figure~\ref{fig:spanning_tree_with_a_junction}(e) and (f)(ii) on the same state, yielding the diagram involving a jagged $s$-loop in Eq.~\eqref{eq:plaquette_stabilized}.

We now study the matrix element of the plaquette operator (specifically $(\tilde{B}_p^s)^{(2)}$) in terms of a basis of the Hilbert space in Eq.~\eqref{eq:string_net_basis_from_eb}.  We use the canonical basis states $|\psi_v^{a,b,c;\alpha}\rangle_v\in S_v^{a,b,c}$ defined in Eq.~\eqref{eq:vertex_basis} via Figure~\ref{fig:merge-spanning-tree}.

The basic setup for this computation is described in Figure~\ref{fig:plaquette_setup}. For each vertex fragment, we choose the operators $U_v^{a,b,c;\alpha}$ corresponding to the procedure in Figure~\ref{fig:plaquette_setup}(a). This choice of $U_v^{a,b,c;\alpha}$ will be useful for later diagrammatic manipulations.  The $d_i$ factors ensure the state is normalized. A basis set for a plaquette and the corresponding diagram (specifying a physical process applied on the spanning tree) are described in Figure~\ref{fig:plaquette_setup}(b) and (c), respectively. In order to compute the matrix element, we apply entanglers over every junction, apply the plaquette operator, and disentangle. 

Part of this process can be described in terms of the set of movement, splitting, and fusion operators defined on the spanning tree. Contributions from these processes can be calculated using the diagrammatic rules introduced in Section~\ref{subsec:diagrammatics}. 

However, there are also other operations not immediately defined on the spanning tree. We describe these processes and explain how we deal with them. First, there are the entanglers and disentanglers. When we apply the entangler, we can effectively convert this into a diagram in which a sector $a$ and $\bar{a}$ fuses into the vacuum, multiplied by a factor of $d_a^{\frac{1}{2}}$. For the disentangler, we instead use  a process in which the vacuum splits into $a$ and $\bar{a}$, with a factor of $d_a^{-\frac{3}{2}}$ [Figure~\ref{fig:edge-disentanglers_action}].\footnote{The factor of $d_a$ and $d_a^{-1}$ come from the edge entangler/disentangler. The additional factor of $d_a^{\frac{1}{2}}$ comes from a conversion of the normalized string operator to a diagrammatic representation of that operator (which differs by a multiplicative factor of $d_a^{\frac{1}{2}}$).} If the sector changes between the two processes, the entangler and disentangler together yield a ratio of the two quantum dimensions.

Second, there are contributions from the string operators acting on the vertical junctions, e.g., the junction between $i_5$ and $\bar{i}_5$ in Figure~\ref{fig:plaquette_setup}(b). While these operators are not supported strictly on the spanning tree, their action on the underlying state can be converted to the action of the operators acting on the spanning tree [Figure~\ref{fig:merge-spanning-tree}].

Therefore, we will organize the overall computation as follows. The contributions from the entangler/disentangler are $\prod_{k=1}^6 (d_{i_k}/d_{i_{k}'}^3)^{\frac{1}{2}}$, where $i_k$ and $i_k'$ are the sector labels we choose for the $k$'th junction.  We will then convert the contribution from the vertical junctions to the contribution on the spanning tree. At this point, all that remains is a process that is amenable to our diagrammatic calculus. The bulk of our analysis will be the simplification of the resulting diagrams.

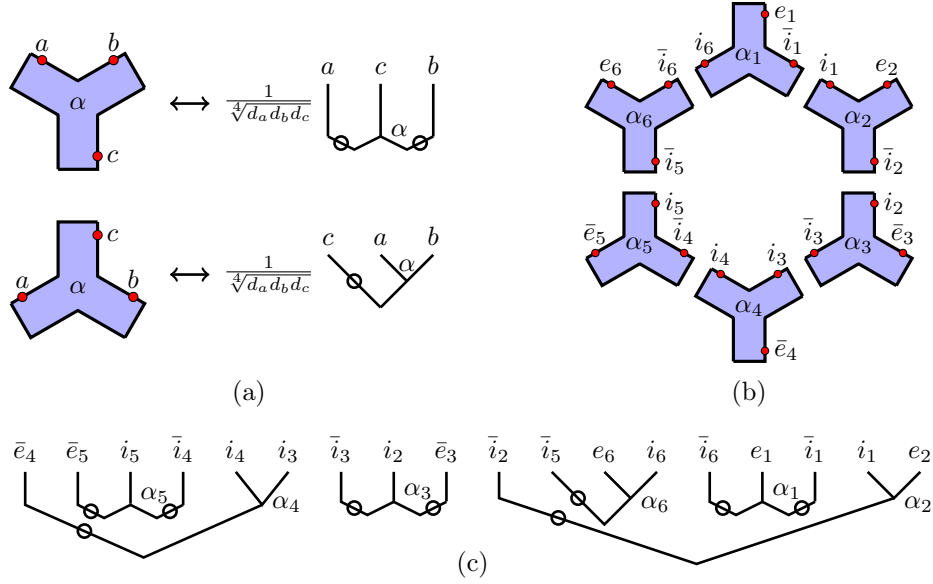
\begin{figure}[htbp]
    \centering
        \begin{tikzpicture}[scale=0.175]
        \vertexone{-6cm}{0}{$\alpha$}{$a$}{$b$}{$c$};
        \begin{scope}[xshift=17cm, yshift=-2.5cm, line width = 1pt]
            \draw[] (0,0) -- ++ (0, 4cm);
            \draw[] (4cm,0) -- ++ (0, 4cm);
            \draw[] (-4cm,0) -- ++ (0, 4cm);
            \draw[] (-4cm,0) -- ++ (2cm, -1cm) -- ++ (2cm, 1cm) -- ++ (2cm, -1cm) -- ++ (2cm, 1cm);
            \node[above] () at (-4cm, 4cm) {$a$};
            \node[above] () at (0cm, 4cm) {$c$};
            \node[above] () at (4cm, 4cm) {$b$};
            \node[right] () at (0, 0.5cm) {$\alpha$};
            \draw[] (3cm, -0.5cm) circle (0.5cm);
            \draw[] (-3cm, -0.5cm) circle (0.5cm);
        \end{scope}
        \node[] () at (8.5cm, 0) {$\frac{1}{\sqrt[4]{d_a d_b d_c}}$};
        \node[] () at (8.5cm, -13cm) {$\frac{1}{\sqrt[4]{d_a d_b d_c}}$};
        \draw[<->, line width=1pt] (1cm, 0) -- ++ (3cm, 0);
        \draw[<->, line width=1pt] (1cm, -13cm) -- ++ (3cm, 0);
        \vertextwo{-6cm}{-14cm}{$\alpha$}{$a$}{$b$}{$c$};
        \begin{scope}[xshift=17cm, yshift=-15.5cm, line width = 1pt]
            \draw[] (-4cm, 4cm) -- (0,0);
            \draw[] (4cm, 4cm) -- (0,0);
            \draw[] (2cm, 2cm) -- (0, 4cm);
            \node[above] () at (-4cm, 4cm) {$c$};
            \node[above] () at (0cm, 4cm) {$a$};
            \node[above] () at (4cm, 4cm) {$b$};
            \node[above] () at (2cm, 2cm) {$\alpha$};
            \draw[] (-2cm, 2cm) circle (0.5cm);
        \end{scope}
        
        \node[] () at (7cm, -22cm) {(a)};

        \begin{scope}[xshift=45cm, yshift=-6cm, scale=0.8]
            \vertexone{-10.392cm}{6cm}{$\alpha_6$}{$e_6$}{$\bar{i}_6$}{$\bar{i}_5$};
            \vertextwo{-10.392cm}{-6cm}{$\alpha_5$}{$\bar{e}_5$}{$\bar{i}_4$}{$i_5$};
            \vertexone{10.392cm}{6cm}{$\alpha_2$}{$i_1$}{$e_2$}{$\bar{i}_2$};
            \vertextwo{10.392cm}{-6cm}{$\alpha_3$}{$\bar{i}_3$}{$\bar{e}_3$}{$i_2$};
            \vertextwo{0cm}{12cm}{$\alpha_1$}{$i_6$}{$\bar{i}_1$}{$e_1$};
            \vertexone{0cm}{-12cm}{$\alpha_4$}{$i_4$}{$i_3$}{$\bar{e}_4$};

            \node[] () at (0, -20cm) {(b)};
        \end{scope}

        \begin{scope}[xshift= -10cm, yshift=-32cm, line width=1pt]
        \node[above] (e4bar) at (0, 4cm) {$\bar{e}_4$};
        \node[above] (e5bar) at (4cm, 4cm) {$\bar{e}_5$};
        \node[above] (i5) at (8cm, 4cm) {$i_5$};
        \node[above] (i4bar) at (12cm, 4cm) {$\bar{i}_4$};
        \node[above] (i4) at (16cm, 4cm) {$i_4$};
        \node[above] (i3) at (20cm, 4cm) {$i_3$};
        \node[above] (i3bar) at (24cm, 4cm) {$\bar{i}_3$};
        \node[above] (i2) at (28cm, 4cm) {$i_2$};
        \node[above] (e3bar) at (32cm, 4cm) {$\bar{e}_3$};
        \node[above] (i2bar) at (36cm, 4cm) {$\bar{i}_2$};
        \node[above] (i5bar) at (40cm, 4cm) {$\bar{i}_5$};
        \node[above] (e6) at (44cm, 4cm) {$e_6$};
        \node[above] (i6) at (48cm, 4cm) {$i_6$};
        \node[above] (i6bar) at (52cm, 4cm) {$\bar{i}_6$};
        \node[above] (e1) at (56cm, 4cm) {$e_1$};
        \node[above] (i1bar) at (60cm, 4cm) {$\bar{i}_1$};
        \node[above] (i1) at (64cm, 4cm) {$i_1$};
        \node[above] (e2) at (68cm, 4cm) {$e_2$};

        \draw[] (e4bar) -- ++ (0, -4cm) -- ++ (9cm, -4cm) -- ++ (9cm, 4cm); 
        \draw[] (e5bar) -- ++ (0, -4cm) -- ++ (2cm, -1cm) -- ++ (2cm, 1cm) -- ++ (2cm, -1cm) -- ++ (2cm, 1cm);
        \draw[] (i5) -- ++ (0, -4cm);
        \draw[] (i4bar) -- ++ (0, -4.25cm);
        \draw[] (16cm, 4cm) -- ++ (2cm, -2.5cm) -- ++ (2cm, 2.5cm);
        \node[right] () at (8cm, 2.25cm) {$\alpha_5$};
        \node[right] () at (18cm, 1.5cm) {$\alpha_4$};
        \draw[] (4.5cm, -0.5cm) circle (0.5cm);
        \draw[] (5cm, 1cm) circle (0.5cm);
        \draw[] (11cm, 1cm) circle (0.5cm);

        \draw[] (i3bar) -- ++ (0, -4cm) -- ++ (2cm, -1cm) -- ++ (2cm, 1cm) -- ++ (2cm, -1cm) -- ++ (2cm, 1cm);
        \draw[] (28cm, 4cm) -- ++ (0, -2.25cm);
        \draw[] (32cm, 4cm) -- ++ (0, -2.25cm);
        \node[right]  () at (28cm, 2.5cm) {$\alpha_3$};
        \draw[] (25cm, 1.25cm) circle (0.5cm);
        \draw[] (31cm, 1.25cm) circle (0.5cm);

        \draw[] (40cm, 4cm) -- ++ (4cm, -4cm) -- ++ (4cm, 4cm);
        \draw[] (44cm, 4cm) -- ++ (2cm, -2cm);
        \node[right] () at (46cm, 1.5cm) {$\alpha_6$};
        \draw[] (42cm, 2cm) circle (0.5cm);

        \draw[] (i6bar) -- ++ (0, -4cm) -- ++ (2cm, -1cm) -- ++ (2cm, 1cm) -- ++ (2cm, -1cm) -- ++ (2cm, 1cm);
        \draw[] (56cm, 4cm) -- ++ (0, -2.25cm);
        \draw[] (60cm, 4cm) -- ++ (0, -2.25cm);
        \node[right]  () at (56cm, 2.5cm) {$\alpha_1$};
        \draw[] (53cm, 1.25cm) circle (0.5cm);
        \draw[] (59cm, 1.25cm) circle (0.5cm);

        \draw[] (64cm, 4cm) -- ++ (2cm, -2cm) -- ++ (2cm, 2cm);
        \draw[] (66cm, 2cm) -- ++ (-15cm, -5cm) -- ++ (-15cm, 5cm)-- ++ (0, 2cm);
        \node[right] () at (66cm, 1.5cm) {$\alpha_2$};
        \draw[] (40.5cm, 0.5cm) circle (0.5cm);

        \node[] () at (34cm, -3cm) {(c)};
            
        \end{scope}
    \end{tikzpicture}
    \caption{(a) Diagrammatic expression for each vertex fragment. (b) A basis set for a plaquette. (c) Applying the unitaries associated with the vertex fragments in (b) on the spanning tree, we obtain a state described by this diagrammatic expression (with $d_i$ factors omitted).}
    \label{fig:plaquette_setup}
\end{figure}

We now discuss how to convert the contributions from the vertical junctions to the contributions on the spanning tree. Recall operations on vertical junctions were specified by Figures~\ref{fig:spanning_tree_with_a_junction} and \ref{fig:edge-disentanglers_action}.  Figure~\ref{fig:edge-disentanglers_action} defines the edge disentangler, and implicitly defines the entangling operation by taking the inverse of the associated unitary operation.  The important point is that the operations that take place straddling the junction can be deformed to an operation   \emph{away} from the junction, on the spanning tree.

To organize this analysis, we insert an operator under which the underlying state remains unchanged. This operator is supported on a region that is slightly larger than the junction. More precisely, consider a schematic description of the union of the junction and the spanning tree below. We will consider the region enclosed by the dashed line.
\begin{equation*}
    \begin{tikzpicture}[baseline={([yshift=-0.5ex]current bounding box.center)}, line width=1pt,scale=0.7]
        \draw[fill=blue!30!white] (0,0) circle (2cm);
        \draw[fill=white] (0,0) circle (1cm);
        \draw[fill=green!30!white, draw= none] (165:1cm) arc (165:195:1cm) -- (195:2cm) arc (195:165:2cm) -- cycle;
        \draw[] (0,0) circle (2cm);
        \draw[] (0,0) circle (1cm);

        \draw[dashed] (210:0.8cm) arc (210:150:0.8cm) -- (150:1.5cm) arc (150:210:1.5cm) -- cycle;
    \end{tikzpicture}
\end{equation*}
We will insert an operator localized in this region, shown below. (The region drawn below is rotated 90 degrees clockwise relative to the diagram above.)
\begin{equation}
    \label{eq:operator_insertion_left}
    \sum_{\alpha, \beta, i_5'} C_{s, i_5}
    \begin{tikzpicture}[baseline={([yshift=-0.5ex]current bounding box.center)}, line width=1pt,scale=0.7]

    \filldraw[blue!30!white] (-2, 0) -- ++ (4, 0) -- ++ (0, 5) -- ++ (-4, 0) -- cycle;
    \filldraw[green!30!white] (-0.25, 0) -- ++ (0.5, 0) -- ++ (0, 5) -- ++ (-0.5, 0) -- cycle;
    \draw[] (-2, 0)  -- ++ (4, 0);

    \begin{scope}[yshift=1.25cm]
    \draw[] (-1, 0) -- ++ (1, 1) -- ++ (1, -1) -- ++ (-1, -1) -- cycle;
    \draw[] (-0.5, 0.5) -- ++ (0.5, -0.5) -- ++ (0.5, 0.5);
    \node[left] () at (-0.5, 0.5) {$\alpha$};
    \node[right] () at (0.5, 0.5) {$\beta$};
    \node[below] () at (-0.5, -0.5) {$s$};
    \node[below] () at (-0.25, 0.25) {$i_5$};
    \end{scope}

    \begin{scope}[yshift=3.75cm]
    \draw[] (-1, 0) -- ++ (1, 1) -- ++ (1, -1) -- ++ (-1, -1) -- cycle;
    \draw[] (-0.5, -0.5) -- ++ (0.5, 0.5) -- ++ (0.5, -0.5);
    \node[left] () at (-0.5, -0.5) {$\alpha$};
    \node[right] () at (0.5, -0.5) {$\beta$};
    \node[below] () at (-0.25, -0.75) {$i_5'$};
    \node[above] () at (-0.5, 0.5) {$s$};
    \node[above] () at (-0.25, -0.25) {$i_5$};
    \end{scope}
    \end{tikzpicture}
\end{equation}
Here $C_{s, i_5} = \frac{1}{d_s^2 d_{i_5}^2}$ is a normalization constant such that the action of this operator on the state of interest is trivial.

Now we can apply Eq.~\eqref{eq:operator_insertion_left} to the state in which $i_5$ and $\bar{i}_5$ are fused into the vacuum followed by the application of the plaquette operator, localized to this region. Diagrammatically, this can be expressed as follows:
\begin{equation}
    % \label{eq:operator_insertion_left}
    \sum_{\alpha, \beta, i_5'} C_{s, i_5}
    \begin{tikzpicture}[baseline={([yshift=-0.5ex]current bounding box.center)}, line width=1pt,scale=0.7]

    \filldraw[blue!30!white] (-2, 0) -- ++ (4, 0) -- ++ (0, 6) -- ++ (-4, 0) -- cycle;
    \filldraw[green!30!white] (-0.25, 0) -- ++ (0.5, 0) -- ++ (0, 6) -- ++ (-0.5, 0) -- cycle;
    \draw[] (-2, 0)  -- ++ (4, 0);

    \begin{scope}[yshift=2.25cm]
    \draw[] (-1, 0) -- ++ (1, 1) -- ++ (1, -1) -- ++ (-1, -1) -- cycle;
    \draw[] (-0.5, 0.5) -- ++ (0.5, -0.5) -- ++ (0.5, 0.5);
    \node[left] () at (-0.5, 0.5) {$\alpha$};
    \node[right] () at (0.5, 0.5) {$\beta$};
    \node[below] () at (-0.5, -0.5) {$s$};
    \node[below] () at (-0.25, 0.25) {$i_5$};
    \end{scope}

    \begin{scope}[yshift=4.75cm]
    \draw[] (-1, 0) -- ++ (1, 1) -- ++ (1, -1) -- ++ (-1, -1) -- cycle;
    \draw[] (-0.5, -0.5) -- ++ (0.5, 0.5) -- ++ (0.5, -0.5);
    \node[left] () at (-0.5, -0.5) {$\alpha$};
    \node[right] () at (0.5, -0.5) {$\beta$};
    \node[below] () at (-0.25, -0.75) {$i_5'$};
    \node[above] () at (-0.5, 0.5) {$s$};
    \node[above] () at (-0.25, -0.25) {$i_5$};
    \end{scope}

    \draw[] (-2, 0.25) -- ++ (2, 0.25) -- ++ (2, -0.25);
    \draw[] (-2, 0.75) -- ++ (2, 0.25) -- ++ (2, -0.25);
    \node[left] () at (-2, 0.25) {$i_5$};
    \node[left] () at (-2, 0.75) {$s$};
    \end{tikzpicture}
    = d_s d_{i_5} C_{s, i_5}
    \sum_{\alpha, \beta, i_5'} 
    \begin{tikzpicture}[baseline={([yshift=-0.5ex]current bounding box.center)}, line width=1pt,scale=0.7]

    \filldraw[blue!30!white] (-2, 0) -- ++ (4, 0) -- ++ (0, 6) -- ++ (-4, 0) -- cycle;
    \filldraw[green!30!white] (-0.25, 0) -- ++ (0.5, 0) -- ++ (0, 6) -- ++ (-0.5, 0) -- cycle;
    \draw[] (-2, 0)  -- ++ (4, 0);

    \begin{scope}[yshift=2.25cm]
    \draw[] (-1, 0) -- ++ (1, 1) -- ++ (1, -1);
    \draw[] (-1, 0) -- (-2, -1.5);
    \draw[] (1, 0) -- (2, -1.5);
    \draw[] (-0.5, 0.5) -- (-1, -1.875);
    \draw[] (0.5, 0.5) -- (1, -1.875);
    \node[left] () at (-0.5, 0.5) {$\alpha$};
    \node[right] () at (0.5, 0.5) {$\beta$};
    \node[below] () at (-0.5, -0.5) {$s$};
    \node[below] () at (-0.25, 0.25) {$i_5$};
    \node[above] () at (-0.25, 0.75) {$i_5'$};
    \end{scope}

    \begin{scope}[xshift= 4cm, yshift=2.25cm]
    \draw[] (-1, 0) -- ++ (1, 1) -- ++ (1, -1) -- ++ (-1, -1) -- cycle;
    \draw[] (-0.5, -0.5) -- ++ (0.5, 0.5) -- ++ (0.5, -0.5);
    \node[left] () at (-0.5, -0.5) {$\alpha$};
    \node[right] () at (0.5, -0.5) {$\beta$};
    \node[below] () at (-0.25, -0.75) {$i_5'$};
    \node[above] () at (-0.5, 0.5) {$i_5$};
    \node[above] () at (-0.25, -0.25) {$s$};
    \end{scope}

    \draw[] (-2, 0.25) -- (-1, 0.375);
    \draw[] (2, 0.25) -- (1, 0.375);
    \node[left] () at (-2, 0.25) {$i_5$};
    \node[left] () at (-2, 0.75) {$s$};
    \end{tikzpicture},
\end{equation}
where the diagram in the white region represents a process occurring on the spanning tree. This identity can be obtained by using the diagrammatic identities in Section~\ref{subsec:diagrammatics} and noting the fact that the action of the string operator acting across the junction (green) is equivalent to the action of the string operator acting on the spanning tree.

A similar argument can be applied to the other vertical junction. After applying the splitting operator on the vertical junctions, we obtain the following state:
\begin{equation}
\begin{aligned}
&\varkappa_s C_N C_J\sum_{\substack{\alpha, \beta, \alpha', \beta' \\ i_2', i_5'}} d_{i_2'} d_{i_5'} d_{i_2} d_{i_5} d_s^2  C_{s, i_5} C_{s, i_2} 
\begin{tikzpicture}[baseline={([yshift=-0.5ex]current bounding box.center)}, line width=1pt, scale=0.7]
\draw[] (-1, 0) -- ++ (1, 1) -- ++ (1, -1) -- ++ (-1, -1) -- cycle;
    \draw[] (-0.5, -0.5) -- ++ (0.5, 0.5) -- ++ (0.5, -0.5);
    \node[left] () at (-0.5, -0.5) {$\alpha$};
    \node[right] () at (0.5, -0.5) {$\beta$};
    \node[below] () at (-0.25, -0.75) {$i_5'$};
    \node[above] () at (-0.5, 0.5) {$i_5$};
    \node[above] () at (-0.25, -0.25) {$s$};

    \begin{scope}[xshift=2.25cm]
    \draw[] (-1, 0) -- ++ (1, 1) -- ++ (1, -1) -- ++ (-1, -1) -- cycle;
    \draw[] (-0.5, -0.5) -- ++ (0.5, 0.5) -- ++ (0.5, -0.5);
    \node[left] () at (-0.5, -0.5) {$\alpha'$};
    \node[right] () at (0.5, -0.5) {$\beta'$};
    \node[below] () at (-0.25, -0.75) {$i_2'$};
    \node[above] () at (-0.5, 0.5) {$\bar{s}$};
    \node[above] () at (-0.25, -0.25) {$i_2$};
    \end{scope}
\end{tikzpicture}
\\
    &\begin{tikzpicture}[baseline={([yshift=-0.5ex]current bounding box.center)}, line width=1pt, scale=0.175]
        \node[above] (e4bar) at (0, 4cm) {$\bar{e}_4$};
        \node[above] (e5bar) at (4cm, 4cm) {$\bar{e}_5$};
        \node[left] (i5) at (8cm, 5.5cm) {$i_5$};
        \node[left] (i4bar) at (12cm, 5.5cm) {$\bar{i}_4$};
        \node[] (i4) at (16.5cm, 5.5cm) {$i_4$};
        \node[] (i3) at (19.5cm, 5.5cm) {$i_3$};
        \node[right] (i3bar) at (24cm, 4cm) {$\bar{i}_3$};
        \node[right] (i2) at (28cm, 5.5cm) {$i_2$};
        \node[above] (e3bar) at (32cm, 4cm) {$\bar{e}_3$};
        \node[left] (i2bar) at (36cm, 5.5cm) {$\bar{i}_2$};
        \node[] (i5bar) at (41.5cm, 5.5cm) {$\bar{i}_5$};
        \node[above] (e6) at (44cm, 4cm) {$e_6$};
        \node[] (i6) at (47.5cm, 5.5cm) {$i_6$};
        \node[right] (i6bar) at (52cm, 5.5cm) {$\bar{i}_6$};
        \node[above] (e1) at (56cm, 4cm) {$e_1$};
        \node[left] (i1bar) at (60cm, 5.5cm) {$\bar{i}_1$};
        \node[above] (i1) at (64.5cm, 4cm) {$i_1$};
        \node[above] (e2) at (68cm, 4cm) {$e_2$};

        \draw[] (e4bar) -- ++ (0, -4cm) -- ++ (9cm, -4cm) -- ++ (9cm, 4cm); 
        \draw[] (e5bar) -- ++ (0, -4cm) -- ++ (2cm, -1cm) -- ++ (2cm, 1cm) -- ++ (2cm, -1cm) -- ++ (2cm, 1cm);
        \draw[] (8cm, 1.25cm) -- ++ (0, 12cm);
        \node[above] () at (8cm, 13.25cm) {$i_5'$};
        \node[left] () at (8cm, 11.25cm) {$\alpha$};
        \node[right] () at (28cm, 11.25cm) {$\alpha'$};
        \draw[red] (8cm, 11.25cm) -- ++ (10cm, -2cm) -- ++ (10cm, 2cm);
        \node[above] () at (13cm, 10.25cm) {$s$};
        \node[above] () at (23cm, 10.25cm) {$\bar{s}$};
        \draw[] (12cm, 1.375cm) -- ++ (0, 7.5cm);
        \draw[] (16cm, 4cm) -- ++ (-4cm, 5cm);
        \draw[] (16cm, 4cm) -- ++ (2cm, -2.5cm) -- ++ (2cm, 2.5cm);
        \node[right] () at (8cm, 2.25cm) {$\alpha_5$};
        \node[right] () at (18cm, 1.5cm) {$\alpha_4$};
        \draw[] (4.5cm, -0.5cm) circle (0.5cm);
        \draw[] (5cm, 1cm) circle (0.5cm);
        \draw[] (11cm, 1cm) circle (0.5cm);

        \draw[] (20cm, 4cm) -- ++ (4cm, 5cm)-- ++ (0, -5cm);

        \draw[] (24cm, 5.75cm) -- ++ (0, -4cm) -- ++ (2cm, -1cm) -- ++ (2cm, 1cm) -- ++ (2cm, -1cm) -- ++ (2cm, 1cm);
        \draw[] (28cm, 4cm) -- ++ (0, -2.25cm);
        \draw[] (28cm, 3.25cm) -- ++ (0, 10cm);
        \node[above] () at (28cm, 13.25cm) {$i_2'$};
        \draw[] (32cm, 4cm) -- ++ (0, -2.25cm);
        \node[right]  () at (28cm, 2.5cm) {$\alpha_3$};
        \draw[] (25cm, 1.25cm) circle (0.5cm);
        \draw[] (31cm, 1.25cm) circle (0.5cm);

        \draw[] (36cm, 3.875cm) -- ++ (0, 6cm);
        \draw[] (40cm, 3.875cm) -- ++ (0, 6cm);
        \node[above] () at (36cm, 9.875cm) {$\bar{i}_2'$};
        \node[above] () at (40cm, 9.875cm) {$\bar{i}_5'$};

        \node[] () at (36.75cm, 5.5cm) {$s$};
        \node[] () at (39.25cm, 5.5cm) {$\bar{s}$};

        \node[left] () at (36cm, 7.875cm) {$\beta'$};
        \node[right] () at (40cm, 7.875cm) {$\beta$};
        \draw[red] (36cm, 7.875cm) -- ++ (2cm, -2cm) -- ++ (2cm, 2cm);

        \draw[] (40cm, 4cm) -- ++ (4cm, -4cm) -- ++ (4cm, 4cm);
        \draw[] (44cm, 4cm) -- ++ (2cm, -2cm);
        \node[right] () at (46cm, 1.5cm) {$\alpha_6$};
        \draw[] (42cm, 2cm) circle (0.5cm);

        \draw[] (48cm, 4cm) -- ++ (4cm, 4cm) -- ++ (0, -4cm);

        \draw[] (52cm, 5.75cm) -- ++ (0, -4cm) -- ++ (2cm, -1cm) -- ++ (2cm, 1cm) -- ++ (2cm, -1cm) -- ++ (2cm, 1cm);
        \draw[] (56cm, 4cm) -- ++ (0, -2.25cm);
        \draw[] (60cm, 4cm) -- ++ (0, -2.25cm);
        \node[right]  () at (56cm, 2.5cm) {$\alpha_1$};
        \draw[] (53cm, 1.25cm) circle (0.5cm);
        \draw[] (59cm, 1.25cm) circle (0.5cm);

        \draw[] (64cm, 4cm) -- ++ (-4cm, 4cm) -- ++ (0, -4cm);

        \draw[] (64cm, 4cm) -- ++ (2cm, -2cm) -- ++ (2cm, 2cm);
        \draw[] (66cm, 2cm) -- ++ (-15cm, -5cm) -- ++ (-15cm, 5cm)-- ++ (0, 2cm);
        \node[right] () at (66cm, 1.5cm) {$\alpha_2$};
        \draw[] (40.5cm, 0.5cm) circle (0.5cm);            
        \end{tikzpicture},
\end{aligned}
        \label{eq:entire_state}
\end{equation}
where $C_{s, i_2} = \frac{1}{d_s^2 d_{i_2}^2}$ is the normalization factor coming from the other junction, $C_N = \prod_{k=1}^6 d_{i_k}^{-\frac{1}{2}} d_{e_k}^{-\frac{1}{4}}$ is the normalization constant of the initial state, and $C_J = \prod_{k=1}^6 (d_{i_k}/d_{i_k'}^3)^{\frac{1}{2}}$ are the aforementioned contributions from the edge entanglers and disentanglers. Now the matrix element can be obtained by applying the disentangling operations on the remaining edges and then re-expressing the resulting state as a linear combination of our basis states [Figure~\ref{fig:plaquette_setup}(c)]. We will proceed with the following steps. First, we will simplify the diagrams in the second line in Eq.~\eqref{eq:entire_state}, treating the left and the right diagram separately. This process will involve the string operators for $s$ and $\bar{s}$, which we colored in red. After this simplification, we obtain a set of diagrams similar to Eq.~\eqref{eq:plaquette_bs}, modulo some additional diagrams. Second, we combine these additional diagrams with the diagrams in the first line of Eq.~\eqref{eq:entire_state}. These diagrams, together with the summation over $\beta$ and $\beta'$, yield delta functions. Combining these delta functions with the rest, we obtain a simplified expression for the resulting state.

To that end, we carry out the following diagrammatic calculation.
\begin{equation}
\begin{aligned}
\begin{tikzpicture}[baseline={([yshift=-0.5ex]current bounding box.center)}, line width=1pt, scale=0.175]
        \node[left] (i2bar) at (36cm, 5.5cm) {$\bar{i}_2$};
        \node[] (i5bar) at (41.5cm, 5.5cm) {$\bar{i}_5$};
        \node[above] (e6) at (44cm, 4cm) {$e_6$};
        \node[] (i6) at (47.5cm, 5.5cm) {$i_6$};
        \node[right] (i6bar) at (52cm, 5.5cm) {$\bar{i}_6$};
        \node[above] (e1) at (56cm, 4cm) {$e_1$};
        \node[left] (i1bar) at (60cm, 5.5cm) {$\bar{i}_1$};
        \node[above] (i1) at (64.5cm, 4cm) {$i_1$};
        \node[above] (e2) at (68cm, 4cm) {$e_2$};

        \draw[] (36cm, 3.875cm) -- ++ (0, 6cm);
        \draw[] (40cm, 3.875cm) -- ++ (0, 6cm);
        \node[above] () at (36cm, 9.875cm) {$\bar{i}_2'$};
        \node[above] () at (40cm, 9.875cm) {$\bar{i}_5'$};

        \node[] () at (36.75cm, 5.5cm) {$s$};
        \node[] () at (39.25cm, 5.5cm) {$\bar{s}$};

        \node[left] () at (36cm, 7.875cm) {$\beta'$};
        \node[right] () at (40cm, 7.875cm) {$\beta$};
        \draw[red] (36cm, 7.875cm) -- ++ (2cm, -2cm) -- ++ (2cm, 2cm);

        \draw[] (40cm, 4cm) -- ++ (4cm, -4cm) -- ++ (4cm, 4cm);
        \draw[] (44cm, 4cm) -- ++ (2cm, -2cm);
        \node[right] () at (46cm, 1.5cm) {$\alpha_6$};
        \draw[] (42cm, 2cm) circle (0.5cm);

        \draw[] (48cm, 4cm) -- ++ (4cm, 4cm) -- ++ (0, -4cm);

        \draw[] (52cm, 5.75cm) -- ++ (0, -4cm) -- ++ (2cm, -1cm) -- ++ (2cm, 1cm) -- ++ (2cm, -1cm) -- ++ (2cm, 1cm);
        \draw[] (56cm, 4cm) -- ++ (0, -2.25cm);
        \draw[] (60cm, 4cm) -- ++ (0, -2.25cm);
        \node[right]  () at (56cm, 2.5cm) {$\alpha_1$};
        \draw[] (53cm, 1.25cm) circle (0.5cm);
        \draw[] (59cm, 1.25cm) circle (0.5cm);

        \draw[] (64cm, 4cm) -- ++ (-4cm, 4cm) -- ++ (0, -4cm);

        \draw[] (64cm, 4cm) -- ++ (2cm, -2cm) -- ++ (2cm, 2cm);
        \draw[] (66cm, 2cm) -- ++ (-15cm, -5cm) -- ++ (-15cm, 5cm)-- ++ (0, 2cm);
        \node[right] () at (66cm, 1.5cm) {$\alpha_2$};
        \draw[] (40.5cm, 0.5cm) circle (0.5cm);            
        \end{tikzpicture}
   &=
    \begin{tikzpicture}[baseline={([yshift=-0.5ex]current bounding box.center)}, line width=1pt, scale=0.175]
        \node[left] (i2bar) at (36cm, 5.5cm) {$\bar{i}_2$};
        \node[] (i5bar) at (41.5cm, 5.5cm) {$\bar{i}_5$};
        \node[above] (e6) at (44cm, 4cm) {$e_6$};
        \node[] (i6) at (47.5cm, 5.5cm) {$i_6$};
        \node[right] (i6bar) at (52cm, 5.5cm) {$\bar{i}_6$};
        \node[above] (e1) at (56cm, 4cm) {$e_1$};
        \node[left] (i1bar) at (60cm, 5.5cm) {$\bar{i}_1$};
        \node[above] (i1) at (64.5cm, 4cm) {$i_1$};
        \node[above] (e2) at (68cm, 4cm) {$e_2$};

        \draw[] (36cm, 3.875cm) -- ++ (0, 6cm);
        \draw[] (40cm, 3.875cm) -- ++ (0, 6cm);
        \node[above] () at (36cm, 9.875cm) {$\bar{i}_2'$};
        \node[above] () at (40cm, 9.875cm) {$\bar{i}_5'$};

        \node[] () at (36.75cm, 5.5cm) {$s$};
        \node[] () at (39.25cm, 5.5cm) {$\bar{s}$};

        \node[left] () at (36cm, 7.875cm) {$\beta'$};
        \node[right] () at (40cm, 7.875cm) {$\beta$};
        \draw[red] (36cm, 7.875cm) -- ++ (2cm, -2cm) -- ++ (2cm, 2cm);

        \draw[] (40cm, 4cm) -- ++ (4cm, -4cm) -- ++ (4cm, 4cm);
        \draw[] (44cm, 4cm) -- ++ (2cm, -2cm);
        \node[right] () at (46cm, 1.5cm) {$\alpha_6$};
        \draw[] (42cm, 2cm) circle (0.5cm);

        \draw[] (48cm, 4cm) -- ++ (4cm, 4cm) -- ++ (0, -4cm);

        \draw[] (52cm, 5.75cm) -- ++ (0, -4cm) -- ++ (2cm, -1cm) -- ++ (2cm, 1cm) -- ++ (2cm, -1cm) -- ++ (2cm, 1cm);
        \draw[] (56cm, 4cm) -- ++ (0, -2.25cm);
        \draw[] (60cm, 4cm) -- ++ (0, -2.25cm);
        \node[right]  () at (56cm, 2.5cm) {$\alpha_1$};
        \draw[] (53cm, 1.25cm) circle (0.5cm);
        \draw[] (59cm, 1.25cm) circle (0.5cm);

        \draw[] (64cm, 4cm) -- ++ (-4cm, 4cm) -- ++ (0, -4cm);

        \draw[] (64cm, 4cm) -- ++ (2cm, -2cm) -- ++ (2cm, 2cm);
        \draw[] (66cm, 2cm)  -- ++ (0, -4cm)   -- ++ (-15cm, -5cm) -- ++ (-15cm, 5cm)-- ++ (0, 6cm);
        \node[right] () at (66cm, 1.5cm) {$\alpha_2$};
        \draw[] (40.5cm, -3.5cm) circle (0.5cm);            
        \end{tikzpicture}
        \\
        &= 
            \begin{tikzpicture}[baseline={([yshift=-0.5ex]current bounding box.center)}, line width=1pt, scale=0.175]
        \node[left] (i2bar) at (36cm, 5.5cm) {$\bar{i}_2$};
        \node[] (i5bar) at (41.5cm, 5.5cm) {$\bar{i}_5$};
        \node[above] (e6) at (44cm, 4cm) {$e_6$};
        \node[] (i6) at (47.5cm, 5.5cm) {$i_6$};
        \node[above] (e1) at (56cm, 15cm) {$e_1$};
        \node[above] (i1) at (64.5cm, 4cm) {$i_1$};
        \node[above] (e2) at (68cm, 4cm) {$e_2$};

        \draw[] (36cm, 3.875cm) -- ++ (0, 6cm);
        \draw[] (40cm, 3.875cm) -- ++ (0, 6cm);
        \node[above] () at (36cm, 9.875cm) {$\bar{i}_2'$};
        \node[above] () at (40cm, 9.875cm) {$\bar{i}_5'$};

        \node[] () at (36.75cm, 5.5cm) {$s$};
        \node[] () at (39.25cm, 5.5cm) {$\bar{s}$};

        \node[left] () at (36cm, 7.875cm) {$\beta'$};
        \node[right] () at (40cm, 7.875cm) {$\beta$};
        \draw[red] (36cm, 7.875cm) -- ++ (2cm, -2cm) -- ++ (2cm, 2cm);

        \draw[] (40cm, 4cm) -- ++ (4cm, -4cm) -- ++ (4cm, 4cm);
        \draw[] (44cm, 4cm) -- ++ (2cm, -2cm);
        \node[right] () at (46cm, 1.5cm) {$\alpha_6$};
        \draw[] (42cm, 2cm) circle (0.5cm);

        \draw[] (48cm, 4cm) -- ++ (8cm, 8cm) -- ++ (4cm, -4cm);
        \draw[] (56cm, 12cm) -- ++ (0, 3cm);

        \node[right]  () at (56cm, 12cm) {$\alpha_1$};

        \draw[] (64cm, 4cm) -- ++ (-4cm, 4cm);

        \draw[] (64cm, 4cm) -- ++ (2cm, -2cm) -- ++ (2cm, 2cm);
        \draw[] (66cm, 2cm) -- ++ (0, -4cm) -- ++ (-15cm, -5cm) -- ++ (-15cm, 5cm)-- ++ (0, 6cm);
        \node[right] () at (66cm, 1.5cm) {$\alpha_2$};
        \draw[] (40.5cm, -3.5cm) circle (0.5cm);
        \end{tikzpicture} \\
        &= 
        \begin{tikzpicture}[baseline={([yshift=-0.5ex]current bounding box.center)}, line width=1pt, scale=0.175]
        \node[left] (i2bar) at (36cm, 5.5cm) {$\bar{i}_2$};
        \node[] (i5bar) at (41.5cm, 4.5cm) {$\bar{i}_5$};
        \node[above] (e6) at (44cm, 4cm) {$e_6$};
        \node[] (i6) at (47.5cm, 5.5cm) {$i_6$};
        \node[above] (e1) at (56cm, 15cm) {$e_1$};
        \node[above] (i1) at (64.5cm, 4cm) {$i_1$};
        \node[above] (e2) at (68cm, 4cm) {$e_2$};

        \draw[] (36cm, 3.875cm) -- ++ (0, 6cm);
        \draw[] (40cm, 3.875cm) -- ++ (0, 6cm);
        \node[above] () at (36cm, 9.875cm) {$\bar{i}_2'$};
        \node[above] () at (40cm, 9.875cm) {$\bar{i}_5'$};

        \node[] () at (36.75cm, 1cm) {$s$};
        \node[] () at (39.25cm, 7cm) {$\bar{s}$};

        \node[left] () at (36cm, 1cm) {$\beta'$};
        \node[right] () at (40cm, 7.475cm) {$\beta$};
        \draw[red] (36cm, 0cm) -- ++ (15cm, -5cm) -- ++ (14cm, 4.5cm) -- ++ (0, 2cm) -- ++ (-6cm, 6cm)-- ++ (-3cm, 1cm) -- ++ (-3cm, -1cm) --++ (-9cm, -9cm) -- ++ (-6cm, 6cm)-- ++ (2cm, 2cm);

        \draw[red] (38cm, 4.5cm) circle (0.5cm);

        \draw[] (40cm, 4cm) -- ++ (4cm, -4cm) -- ++ (4cm, 4cm);
        \draw[] (44cm, 4cm) -- ++ (2cm, -2cm);
        \node[above] () at (46cm, 2.5cm) {$\alpha_6$};
        \draw[] (42cm, 2cm) circle (0.5cm);

        \draw[] (48cm, 4cm) -- ++ (8cm, 8cm) -- ++ (4cm, -4cm);
        \draw[] (56cm, 12cm) -- ++ (0, 3cm);

        \node[right]  () at (56cm, 12cm) {$\alpha_1$};

        \draw[] (64cm, 4cm) -- ++ (-4cm, 4cm);

        \draw[] (64cm, 4cm) -- ++ (2cm, -2cm) -- ++ (2cm, 2cm);
        \draw[] (66cm, 2cm) -- ++ (0, -4cm) -- ++ (-15cm, -5cm) -- ++ (-15cm, 5cm)-- ++ (0, 6cm);
        \node[right] () at (66cm, 1.5cm) {$\alpha_2$};
        \draw[] (40.5cm, -3.5cm) circle (0.5cm);
        \end{tikzpicture},
    \end{aligned}
\end{equation}
where in the second line we bent the edges labeled by $\bar{i}_6$ and $\bar{i}_1$, removing the Frobenius-Schur indicators. Now we can use the completeness relation [Eq.~\eqref{eq:completeness}] on four places: the edges labeled by $i_1, i_2, i_5, i_6$. This entails summing over a set of fusion outcomes and multiplicity labels, which shall be denoted as $\tilde{i}_1, \tilde{i}_2, \tilde{i}_5, \tilde{i}_6$ and $\gamma_1, \gamma_2, \gamma_5, \gamma_6$. After bending the edges back down and applying the vacuum identity [Eq.~\eqref{eq:vacuum_identity}] and the completeness relation [Eq.~\eqref{eq:completeness}], we obtain the following diagram:
\begin{equation}
\begin{aligned}
\begin{tikzpicture}[baseline={([yshift=-0.5ex]current bounding box.center)}, line width=1pt, scale=0.175]
        \node[left] (i2bar) at (36cm, 5.5cm) {$\bar{i}_2$};
        \node[] (i5bar) at (41.5cm, 5.5cm) {$\bar{i}_5$};
        \node[above] (e6) at (44cm, 4cm) {$e_6$};
        \node[] (i6) at (47.5cm, 5.5cm) {$i_6$};
        \node[right] (i6bar) at (52cm, 5.5cm) {$\bar{i}_6$};
        \node[above] (e1) at (56cm, 4cm) {$e_1$};
        \node[left] (i1bar) at (60cm, 5.5cm) {$\bar{i}_1$};
        \node[above] (i1) at (64.5cm, 4cm) {$i_1$};
        \node[above] (e2) at (68cm, 4cm) {$e_2$};

        \draw[] (36cm, 3.875cm) -- ++ (0, 6cm);
        \draw[] (40cm, 3.875cm) -- ++ (0, 6cm);
        \node[above] () at (36cm, 9.875cm) {$\bar{i}_2'$};
        \node[above] () at (40cm, 9.875cm) {$\bar{i}_5'$};

        \node[] () at (36.75cm, 5.5cm) {$s$};
        \node[] () at (39.25cm, 5.5cm) {$\bar{s}$};

        \node[left] () at (36cm, 7.875cm) {$\beta'$};
        \node[right] () at (40cm, 7.875cm) {$\beta$};
        \draw[red] (36cm, 7.875cm) -- ++ (2cm, -2cm) -- ++ (2cm, 2cm);

        \draw[] (40cm, 4cm) -- ++ (4cm, -4cm) -- ++ (4cm, 4cm);
        \draw[] (44cm, 4cm) -- ++ (2cm, -2cm);
        \node[right] () at (46cm, 1.5cm) {$\alpha_6$};
        \draw[] (42cm, 2cm) circle (0.5cm);

        \draw[] (48cm, 4cm) -- ++ (4cm, 4cm) -- ++ (0, -4cm);

        \draw[] (52cm, 5.75cm) -- ++ (0, -4cm) -- ++ (2cm, -1cm) -- ++ (2cm, 1cm) -- ++ (2cm, -1cm) -- ++ (2cm, 1cm);
        \draw[] (56cm, 4cm) -- ++ (0, -2.25cm);
        \draw[] (60cm, 4cm) -- ++ (0, -2.25cm);
        \node[right]  () at (56cm, 2.5cm) {$\alpha_1$};
        \draw[] (53cm, 1.25cm) circle (0.5cm);
        \draw[] (59cm, 1.25cm) circle (0.5cm);

        \draw[] (64cm, 4cm) -- ++ (-4cm, 4cm) -- ++ (0, -4cm);

        \draw[] (64cm, 4cm) -- ++ (2cm, -2cm) -- ++ (2cm, 2cm);
        \draw[] (66cm, 2cm) -- ++ (-15cm, -5cm) -- ++ (-15cm, 5cm)-- ++ (0, 2cm);
        \node[right] () at (66cm, 1.5cm) {$\alpha_2$};
        \draw[] (40.5cm, 0.5cm) circle (0.5cm);            
        \end{tikzpicture}
         =
        \sum_{\substack{\gamma_1, \gamma_2, \gamma_5, \gamma_6 \\\tilde{i}_1, \tilde{i}_2, \tilde{i}_5, \tilde{i}_6}} \sqrt{ \frac{d_{\tilde{i}_1}}{d_sd_{i_1}} \frac{d_{\tilde{i}_2}}{d_s d_{i_2}} \frac{d_{\tilde{i}_5}}{d_s d_{i_5}} \frac{d_{\tilde{i}_6}}{d_s d_{i_6}}} \frac{1}{d_{i_2'} d_{i_5'}}\delta_{\tilde{i}_5, i_5'} \delta_{\tilde{i}_2, i_2'}\\
        \times 
        \begin{tikzpicture}[baseline={([yshift=-0.5ex]current bounding box.center)}, line width=1pt, scale=0.8]
    \draw[] (-1, 0) -- ++ (1, 1) -- ++ (1, -1) -- ++ (-1, -1) -- cycle;
    \draw[] (-0.5, 0.5) -- ++ (0.5, -0.5) -- ++ (0.5, 0.5);
    \node[left] () at (-0.5, 0.5) {$\beta'$};
    \node[right] () at (0.5, 0.5) {$\gamma_2$};
    \node[above] () at (0.25, 0.75) {$\tilde{i}_2$};
    \node[below] () at (0.5, -0.5) {$i_2$};
    \node[below]  () at (0.25, 0.25) {$\bar{s}$};
    \draw[] (-0.25, 0.75) circle (0.1cm);
    \draw[] (-0.5, -0.5) circle (0.1cm);
    \draw[] (-0.25, 0.25) circle (0.1cm);

    \begin{scope}[xshift=2.25cm]
    \draw[] (-1, 0) -- ++ (1, 1) -- ++ (1, -1) -- ++ (-1, -1) -- cycle;
    \draw[] (-0.5, 0.5) -- ++ (0.5, -0.5) -- ++ (0.5, 0.5);
    \node[left] () at (-0.5, 0.5) {$\beta$};
    \node[right] () at (0.5, 0.5) {$\gamma_5$};
    \node[above] () at (0.25, 0.75) {$\tilde{i}_5$};
    \node[below] () at (0.5, -0.5) {$s$};
    \node[below]  () at (0.25, 0.25) {$i_5$};

    \draw[] (-0.25, 0.75) circle (0.1cm);
    \draw[] (-0.25, 0.25) circle (0.1cm);
    \end{scope}
\end{tikzpicture}
        \begin{tikzpicture}[baseline={([yshift=-0.5ex]current bounding box.center)}, line width=1pt, scale=0.175]
        \node[above] (e6) at (44cm, 4cm) {$e_6$};
        \node[] (i6) at (47.5cm, 5.5cm) {$i_6$};
        \node[right] (i6bar) at (52cm, 5.5cm) {$\bar{\tilde{i}}_6$};
        \node[above] (e1) at (56cm, 4cm) {$e_1$};
        \node[left] (i1bar) at (60cm, 5.5cm) {$\bar{\tilde{i}}_1$};
        \node[above] (i1) at (64.5cm, 4cm) {$i_1$};
        \node[above] (e2) at (68cm, 4cm) {$e_2$};

        \draw[] (36cm, 3.875cm) -- ++ (0, 6cm);
        \draw[] (40cm, 3.875cm) -- ++ (0, 6cm);
        \node[above] () at (36cm, 9.875cm) {$\bar{i}_2'$};
        \node[above] () at (40cm, 9.875cm) {$\bar{i}_5'$};

        \draw[] (40cm, 4cm) -- ++ (4cm, -4cm) -- ++ (4cm, 4cm);
        \draw[] (44cm, 4cm) -- ++ (2cm, -2cm);
        \node[] () at (46cm, 3.5cm) {$\alpha_6$};
        \draw[] (42cm, 2cm) circle (0.5cm);

        \draw[] (48cm, 4cm) -- ++ (4cm, 4cm) -- ++ (0, -4cm);

        \draw[] (52cm, 5.75cm) -- ++ (0, -4cm) -- ++ (2cm, -1cm) -- ++ (2cm, 1cm) -- ++ (2cm, -1cm) -- ++ (2cm, 1cm);
        \draw[] (56cm, 4cm) -- ++ (0, -2.25cm);
        \draw[] (60cm, 4cm) -- ++ (0, -2.25cm);
        \node[right]  () at (56cm, 2.5cm) {$\alpha_1$};
        \draw[] (53cm, 1.25cm) circle (0.5cm);
        \draw[] (59cm, 1.25cm) circle (0.5cm);

        \draw[] (64cm, 4cm) -- ++ (-4cm, 4cm) -- ++ (0, -4cm);

        \draw[] (64cm, 4cm) -- ++ (2cm, -2cm) -- ++ (2cm, 2cm);
        \draw[] (66cm, 2cm) -- ++ (-15cm, -5cm) -- ++ (-15cm, 5cm)-- ++ (0, 2cm);
        \node[right] () at (66cm, 1.5cm) {$\alpha_2$};
        \draw[] (40.5cm, 0.5cm) circle (0.5cm);  

        % s-strings
        \draw[red] (62cm, 6cm) -- ++ (0, -5.25cm);
        \draw[red] (50cm, 6cm) -- ++ (-1.75cm, -3.5cm) -- ++(-3.5cm, -1.75cm);
        
        \draw[red] (56cm, 1.25cm) -- ++ (1.5cm, -0.25cm);
        \draw[red] (56cm, 1.25cm) -- ++ (-1.5cm, -0.25cm);
        \end{tikzpicture},
        \end{aligned}
\end{equation}
where the indices $\gamma_1, \gamma_2, \gamma_5, \gamma_6$ are suppressed in the right diagram. (Each $\gamma_k$ is a vertex incident on edges $i_k$ and $\tilde{i}_k$.) Carrying out this sum over $\tilde{i}_2$ and $\tilde{i}_5$ and replacing the summation variable $\tilde{i}_1$ and $\tilde{i}_6$ to $i_1'$ and $i_6'$ respectively, we obtain the following. 
\begin{equation}
\begin{aligned}
\begin{tikzpicture}[baseline={([yshift=-0.5ex]current bounding box.center)}, line width=1pt, scale=0.175]
        \node[left] (i2bar) at (36cm, 5.5cm) {$\bar{i}_2$};
        \node[] (i5bar) at (41.5cm, 5.5cm) {$\bar{i}_5$};
        \node[above] (e6) at (44cm, 4cm) {$e_6$};
        \node[] (i6) at (47.5cm, 5.5cm) {$i_6$};
        \node[right] (i6bar) at (52cm, 5.5cm) {$\bar{i}_6$};
        \node[above] (e1) at (56cm, 4cm) {$e_1$};
        \node[left] (i1bar) at (60cm, 5.5cm) {$\bar{i}_1$};
        \node[above] (i1) at (64.5cm, 4cm) {$i_1$};
        \node[above] (e2) at (68cm, 4cm) {$e_2$};

        \draw[] (36cm, 3.875cm) -- ++ (0, 6cm);
        \draw[] (40cm, 3.875cm) -- ++ (0, 6cm);
        \node[above] () at (36cm, 9.875cm) {$\bar{i}_2'$};
        \node[above] () at (40cm, 9.875cm) {$\bar{i}_5'$};

        \node[] () at (36.75cm, 5.5cm) {$s$};
        \node[] () at (39.25cm, 5.5cm) {$\bar{s}$};

        \node[left] () at (36cm, 7.875cm) {$\beta'$};
        \node[right] () at (40cm, 7.875cm) {$\beta$};
        \draw[red] (36cm, 7.875cm) -- ++ (2cm, -2cm) -- ++ (2cm, 2cm);

        \draw[] (40cm, 4cm) -- ++ (4cm, -4cm) -- ++ (4cm, 4cm);
        \draw[] (44cm, 4cm) -- ++ (2cm, -2cm);
        \node[right] () at (46cm, 1.5cm) {$\alpha_6$};
        \draw[] (42cm, 2cm) circle (0.5cm);

        \draw[] (48cm, 4cm) -- ++ (4cm, 4cm) -- ++ (0, -4cm);

        \draw[] (52cm, 5.75cm) -- ++ (0, -4cm) -- ++ (2cm, -1cm) -- ++ (2cm, 1cm) -- ++ (2cm, -1cm) -- ++ (2cm, 1cm);
        \draw[] (56cm, 4cm) -- ++ (0, -2.25cm);
        \draw[] (60cm, 4cm) -- ++ (0, -2.25cm);
        \node[right]  () at (56cm, 2.5cm) {$\alpha_1$};
        \draw[] (53cm, 1.25cm) circle (0.5cm);
        \draw[] (59cm, 1.25cm) circle (0.5cm);

        \draw[] (64cm, 4cm) -- ++ (-4cm, 4cm) -- ++ (0, -4cm);

        \draw[] (64cm, 4cm) -- ++ (2cm, -2cm) -- ++ (2cm, 2cm);
        \draw[] (66cm, 2cm) -- ++ (-15cm, -5cm) -- ++ (-15cm, 5cm)-- ++ (0, 2cm);
        \node[right] () at (66cm, 1.5cm) {$\alpha_2$};
        \draw[] (40.5cm, 0.5cm) circle (0.5cm);            
        \end{tikzpicture}
         = \sum_{\substack{\gamma_1, \gamma_2, \gamma_5, \gamma_6 \\i_1', i_6'}} \sqrt{ \frac{d_{i_1'}}{d_sd_{i_1}} \frac{d_{i_2'}}{d_s d_{i_2}} \frac{d_{i_5'}}{d_s d_{i_5}} \frac{d_{i_6'}}{d_s d_{i_6}}} \frac{1}{d_{i_2'} d_{i_5'}}\\
        \times 
        \begin{tikzpicture}[baseline={([yshift=-0.5ex]current bounding box.center)}, line width=1pt, scale=0.8]
    \draw[] (-1, 0) -- ++ (1, 1) -- ++ (1, -1) -- ++ (-1, -1) -- cycle;
    \draw[] (-0.5, 0.5) -- ++ (0.5, -0.5) -- ++ (0.5, 0.5);
    \node[left] () at (-0.5, 0.5) {$\beta'$};
    \node[right] () at (0.5, 0.5) {$\gamma_2$};
    \node[above] () at (0.25, 0.75) {$i_2'$};
    \node[below] () at (0.5, -0.5) {$i_2$};
    \node[below]  () at (0.25, 0.25) {$\bar{s}$};

    \draw[] (-0.25, 0.75) circle (0.1cm);
    \draw[] (-0.5, -0.5) circle (0.1cm);
    \draw[] (-0.25, 0.25) circle (0.1cm);

    \begin{scope}[xshift=2.25cm]
    \draw[] (-1, 0) -- ++ (1, 1) -- ++ (1, -1) -- ++ (-1, -1) -- cycle;
    \draw[] (-0.5, 0.5) -- ++ (0.5, -0.5) -- ++ (0.5, 0.5);
    \node[left] () at (-0.5, 0.5) {$\beta$};
    \node[right] () at (0.5, 0.5) {$\gamma_5$};
    \node[above] () at (0.25, 0.75) {$i_5'$};
    \node[below] () at (0.5, -0.5) {$s$};
    \node[below]  () at (0.25, 0.25) {$i_5$};

    \draw[] (-0.25, 0.75) circle (0.1cm);
    \draw[] (-0.25, 0.25) circle (0.1cm);
    \end{scope}
\end{tikzpicture}
        \begin{tikzpicture}[baseline={([yshift=-0.5ex]current bounding box.center)}, line width=1pt, scale=0.175]
        %\node[left] (i2bar) at (36cm, 5.5cm) {$\bar{i}_2$};
        %\node[] (i5bar) at (41.5cm, 5.5cm) {$\bar{i}_5$};
        \node[above] (e6) at (44cm, 4cm) {$e_6$};
        \node[] (i6) at (47.5cm, 5.5cm) {$i_6$};
        \node[right] (i6bar) at (52cm, 5.5cm) {$\bar{i}_6'$};
        \node[above] (e1) at (56cm, 4cm) {$e_1$};
        \node[left] (i1bar) at (60cm, 5.5cm) {$\bar{i}_1'$};
        \node[above] (i1) at (64.5cm, 4cm) {$i_1$};
        \node[above] (e2) at (68cm, 4cm) {$e_2$};

        \draw[] (36cm, 3.875cm) -- ++ (0, 6cm);
        \draw[] (40cm, 3.875cm) -- ++ (0, 6cm);
        \node[above] () at (36cm, 9.875cm) {$\bar{i}_2'$};
        \node[above] () at (40cm, 9.875cm) {$\bar{i}_5'$};

        \draw[] (40cm, 4cm) -- ++ (4cm, -4cm) -- ++ (4cm, 4cm);
        \draw[] (44cm, 4cm) -- ++ (2cm, -2cm);
        \node[] () at (46cm, 3.5cm) {$\alpha_6$};
        \draw[] (42cm, 2cm) circle (0.5cm);

        \draw[] (48cm, 4cm) -- ++ (4cm, 4cm) -- ++ (0, -4cm);

        \draw[] (52cm, 5.75cm) -- ++ (0, -4cm) -- ++ (2cm, -1cm) -- ++ (2cm, 1cm) -- ++ (2cm, -1cm) -- ++ (2cm, 1cm);
        \draw[] (56cm, 4cm) -- ++ (0, -2.25cm);
        \draw[] (60cm, 4cm) -- ++ (0, -2.25cm);
        \node[right]  () at (56cm, 2.5cm) {$\alpha_1$};
        \draw[] (53cm, 1.25cm) circle (0.5cm);
        \draw[] (59cm, 1.25cm) circle (0.5cm);

        \draw[] (64cm, 4cm) -- ++ (-4cm, 4cm) -- ++ (0, -4cm);

        \draw[] (64cm, 4cm) -- ++ (2cm, -2cm) -- ++ (2cm, 2cm);
        \draw[] (66cm, 2cm) -- ++ (-15cm, -5cm) -- ++ (-15cm, 5cm)-- ++ (0, 2cm);
        \node[right] () at (66cm, 1.5cm) {$\alpha_2$};
        \draw[] (40.5cm, 0.5cm) circle (0.5cm);  

        % s-strings
        \draw[red] (62cm, 6cm) -- ++ (0, -5.25cm);
        \draw[red] (50cm, 6cm) -- ++ (-1.75cm, -3.5cm) -- ++(-3.5cm, -1.75cm);
        
        \draw[red] (56cm, 1.25cm) -- ++ (1.5cm, -0.25cm);
        \draw[red] (56cm, 1.25cm) -- ++ (-1.5cm, -0.25cm);

        \end{tikzpicture}.
        \end{aligned}
        \label{eq:right_diagram_final}
\end{equation}

We can carry out a similar diagrammatic calculation for the left diagram on the second line of Eq.~\eqref{eq:entire_state}. This calculation consists of the following steps. We first bend down the edges labeled by $\bar{i}_4$ and $\bar{i}_3$, removing the Frobenius-Schur indicators. Then we use the completeness relation [Eq.~\eqref{eq:completeness}] over edges $i_4$ and $i_3$. Bending the bent-down edges back up, we obtain the following expression.
\begin{equation}
\begin{aligned}
    \begin{tikzpicture}[baseline={([yshift=-0.5ex]current bounding box.center)}, line width=1pt, scale=0.175]
        \node[above] (e4bar) at (0, 4cm) {$\bar{e}_4$};
        \node[above] (e5bar) at (4cm, 4cm) {$\bar{e}_5$};
        \node[left] (i5) at (8cm, 5.5cm) {$i_5$};
        \node[left] (i4bar) at (12cm, 5.5cm) {$\bar{i}_4$};
        \node[] (i4) at (16.5cm, 5.5cm) {$i_4$};
        \node[] (i3) at (19.5cm, 5.5cm) {$i_3$};
        \node[right] (i3bar) at (24cm, 4cm) {$\bar{i}_3$};
        \node[right] (i2) at (28cm, 5.5cm) {$i_2$};
        \node[above] (e3bar) at (32cm, 4cm) {$\bar{e}_3$};

        \draw[] (e4bar) -- ++ (0, -4cm) -- ++ (9cm, -4cm) -- ++ (9cm, 4cm); 
        \draw[] (e5bar) -- ++ (0, -4cm) -- ++ (2cm, -1cm) -- ++ (2cm, 1cm) -- ++ (2cm, -1cm) -- ++ (2cm, 1cm);
        \draw[] (8cm, 1.25cm) -- ++ (0, 12cm);
        \node[above] () at (8cm, 13.25cm) {$i_5'$};
        \node[left] () at (8cm, 11.25cm) {$\alpha$};
        \node[right] () at (28cm, 11.25cm) {$\alpha'$};
        \draw[red] (8cm, 11.25cm) -- ++ (10cm, -2cm) -- ++ (10cm, 2cm);
        \node[above] () at (13cm, 10.25cm) {$s$};
        \node[above] () at (23cm, 10.25cm) {$\bar{s}$};
        \draw[] (12cm, 1.375cm) -- ++ (0, 7.5cm);
        \draw[] (16cm, 4cm) -- ++ (-4cm, 5cm);
        \draw[] (16cm, 4cm) -- ++ (2cm, -2.5cm) -- ++ (2cm, 2.5cm);
        \node[right] () at (8cm, 2.25cm) {$\alpha_5$};
        \node[right] () at (18cm, 1.5cm) {$\alpha_4$};
        \draw[] (4.5cm, -0.5cm) circle (0.5cm);
        \draw[] (5cm, 1cm) circle (0.5cm);
        \draw[] (11cm, 1cm) circle (0.5cm);

        \draw[] (20cm, 4cm) -- ++ (4cm, 5cm)-- ++ (0, -5cm);

        \draw[] (24cm, 5.75cm) -- ++ (0, -4cm) -- ++ (2cm, -1cm) -- ++ (2cm, 1cm) -- ++ (2cm, -1cm) -- ++ (2cm, 1cm);
        \draw[] (28cm, 4cm) -- ++ (0, -2.25cm);
        \draw[] (28cm, 3.25cm) -- ++ (0, 10cm);
        \node[above] () at (28cm, 13.25cm) {$i_2'$};
        \draw[] (32cm, 4cm) -- ++ (0, -2.25cm);
        \node[right]  () at (28cm, 2.5cm) {$\alpha_3$};
        \draw[] (25cm, 1.25cm) circle (0.5cm);
        \draw[] (31cm, 1.25cm) circle (0.5cm);
        \end{tikzpicture} = \sum_{\substack{\gamma_3, \gamma_4, \\ i_3', i_4'}} \sqrt{\frac{d_{i_3'}}{d_sd_{i_3}}\frac{d_{i_4'}}{d_sd_{i_4}}} 
    \begin{tikzpicture}[baseline={([yshift=-0.5ex]current bounding box.center)}, line width=1pt, scale=0.175]
        \node[above] (e4bar) at (0, 4cm) {$\bar{e}_4$};
        \node[above] (e5bar) at (4cm, 4cm) {$\bar{e}_5$};
        \node[left] (i4bar) at (12cm, 5.5cm) {$\bar{i}_4'$};
        \node[] (i4) at (16.5cm, 5.5cm) {$i_4'$};
        \node[] (i3) at (19.5cm, 5.5cm) {$i_3'$};
        \node[right] (i3bar) at (24cm, 4cm) {$\bar{i}_3'$};
        \node[above] (e3bar) at (32cm, 4cm) {$\bar{e}_3$};

        \draw[] (e4bar) -- ++ (0, -4cm) -- ++ (9cm, -4cm) -- ++ (9cm, 4cm); 
        \draw[] (e5bar) -- ++ (0, -4cm) -- ++ (2cm, -1cm) -- ++ (2cm, 1cm) -- ++ (2cm, -1cm) -- ++ (2cm, 1cm);
        \draw[] (8cm, 1.25cm) -- ++ (0, 6cm);
        \node[above] () at (8cm, 7.25cm) {$i_5'$};
        \node[left] () at (8cm, 5cm) {$\alpha$};
        \node[right] () at (28cm, 5cm) {$\alpha'$};
        \draw[red] (8cm, 5cm) -- (9cm, 1cm);
        \draw[red] (28cm, 5cm) -- (27cm, 1.25cm);
        \draw[red] (18cm, 3cm) -- (20cm, 4cm);
        \draw[red] (18cm, 3cm) -- (16cm, 4cm);
        \draw[] (12cm, 1.375cm) -- ++ (0, 7.5cm);
        \draw[] (16cm, 4cm) -- ++ (-4cm, 5cm);
        \draw[] (16cm, 4cm) -- ++ (2cm, -2.5cm) -- ++ (2cm, 2.5cm);
        \node[right] () at (8cm, 2.25cm) {$\alpha_5$};
        \node[right] () at (18cm, 1.5cm) {$\alpha_4$};
        \draw[] (4.5cm, -0.5cm) circle (0.5cm);
        \draw[] (5cm, 1cm) circle (0.5cm);
        \draw[] (11cm, 1cm) circle (0.5cm);

        \draw[] (20cm, 4cm) -- ++ (4cm, 5cm)-- ++ (0, -5cm);

        \draw[] (24cm, 5.75cm) -- ++ (0, -4cm) -- ++ (2cm, -1cm) -- ++ (2cm, 1cm) -- ++ (2cm, -1cm) -- ++ (2cm, 1cm);
        \draw[] (28cm, 4cm) -- ++ (0, -2.25cm);
        \draw[] (28cm, 3.25cm) -- ++ (0, 4cm);
        \node[above] () at (28cm, 7.25cm) {$i_2'$};
        \draw[] (32cm, 4cm) -- ++ (0, -2.25cm);
        \node[right]  () at (28cm, 2.5cm) {$\alpha_3$};
        \draw[] (25cm, 1.25cm) circle (0.5cm);
        \draw[] (31cm, 1.25cm) circle (0.5cm);
        \end{tikzpicture}.
        \end{aligned}
        \label{eq:left_diagram_final}
\end{equation}
Here again we are suppressing the indices $\gamma_3$ and $\gamma_4$; each $\gamma_k$ is located on the vertex incident on edges $i_k$ and $i_k'$.

At this point, the remaining task is to plug in Eq.~\eqref{eq:right_diagram_final} and~\eqref{eq:left_diagram_final} to Eq.~\eqref{eq:entire_state} and simplify the resulting expression. The key observation here is that the sum over $\beta$ and $\beta'$ yields delta functions:

\begin{equation}
\begin{aligned}
    \sum_{\beta'}  
    \begin{tikzpicture}[baseline={([yshift=-0.5ex]current bounding box.center)}, line width=1pt, scale=0.8]
    \draw[] (-1, 0) -- ++ (1, 1) -- ++ (1, -1) -- ++ (-1, -1) -- cycle;
    \draw[] (-0.5, -0.5) -- ++ (0.5, 0.5) -- ++ (0.5, -0.5);
    \node[left] () at (-0.5, -0.5) {$\alpha'$};
    \node[right] () at (0.5, -0.5) {$\beta'$};
    \node[below] () at (-0.25, -0.75) {$i_2'$};
    \node[above] () at (-0.5, 0.5) {$\bar{s}$};
    \node[above] () at (-0.25, -0.25) {$i_2$};

    \begin{scope}[xshift=2.5cm]
    \draw[] (-1, 0) -- ++ (1, 1) -- ++ (1, -1) -- ++ (-1, -1) -- cycle;
    \draw[] (-0.5, 0.5) -- ++ (0.5, -0.5) -- ++ (0.5, 0.5);
    \node[left] () at (-0.5, 0.5) {$\beta'$};
    \node[right] () at (0.5, 0.5) {$\gamma_2$};
    \node[above] () at (0.25, 0.75) {$i_2'$};
    \node[below] () at (0.5, -0.5) {$i_2$};
    \node[below]  () at (0.25, 0.25) {$\bar{s}$};

    \draw[] (-0.25, 0.75) circle (0.1cm);
    \draw[] (-0.5, -0.5) circle (0.1cm);
    \draw[] (-0.25, 0.25) circle (0.1cm);
    \end{scope}
\end{tikzpicture} &=  d_s d_{i_2} d_{i_2'} \delta_{\alpha', \gamma_2}, 
    \\
    \sum_{\beta}
    \begin{tikzpicture}[baseline={([yshift=-0.5ex]current bounding box.center)}, line width=1pt, scale=0.8]
    \draw[] (-1, 0) -- ++ (1, 1) -- ++ (1, -1) -- ++ (-1, -1) -- cycle;
    \draw[] (-0.5, -0.5) -- ++ (0.5, 0.5) -- ++ (0.5, -0.5);
    \node[left] () at (-0.5, -0.5) {$\alpha$};
    \node[right] () at (0.5, -0.5) {$\beta$};
    \node[below] () at (-0.25, -0.75) {$i_5'$};
    \node[above] () at (-0.5, 0.5) {$i_5$};
    \node[above] () at (-0.25, -0.25) {$s$};
    \begin{scope}[xshift=2.5cm]
    \draw[] (-1, 0) -- ++ (1, 1) -- ++ (1, -1) -- ++ (-1, -1) -- cycle;
    \draw[] (-0.5, 0.5) -- ++ (0.5, -0.5) -- ++ (0.5, 0.5);
    \node[left] () at (-0.5, 0.5) {$\beta$};
    \node[right] () at (0.5, 0.5) {$\gamma_5$};
    \node[above] () at (0.25, 0.75) {$i_5'$};
    \node[below] () at (0.5, -0.5) {$s$};
    \node[below]  () at (0.25, 0.25) {$i_5$};

    \draw[] (-0.25, 0.75) circle (0.1cm);
    \draw[] (-0.25, 0.25) circle (0.1cm);
    \end{scope}
\end{tikzpicture} &= d_s d_{i_5} d_{i_5'} \varkappa_s \delta_{\alpha, \gamma_5}. \label{eq:delta_identity}
\end{aligned}
\end{equation}
The derivations of the two identities are similar, and as such, we only explain the first one.
\begin{equation}
\begin{aligned}
     \sum_{\beta'}  
    \begin{tikzpicture}[baseline={([yshift=-0.5ex]current bounding box.center)}, line width=1pt, scale=0.8]
    \draw[] (-1, 0) -- ++ (1, 1) -- ++ (1, -1) -- ++ (-1, -1) -- cycle;
    \draw[] (-0.5, -0.5) -- ++ (0.5, 0.5) -- ++ (0.5, -0.5);
    \node[left] () at (-0.5, -0.5) {$\alpha'$};
    \node[right] () at (0.5, -0.5) {$\beta'$};
    \node[below] () at (-0.25, -0.75) {$i_2'$};
    \node[above] () at (-0.5, 0.5) {$\bar{s}$};
    \node[above] () at (-0.25, -0.25) {$i_2$};

    \begin{scope}[xshift=2.5cm]
    \draw[] (-1, 0) -- ++ (1, 1) -- ++ (1, -1) -- ++ (-1, -1) -- cycle;
    \draw[] (-0.5, 0.5) -- ++ (0.5, -0.5) -- ++ (0.5, 0.5);
    \node[left] () at (-0.5, 0.5) {$\beta'$};
    \node[right] () at (0.5, 0.5) {$\gamma_2$};
    \node[above] () at (0.25, 0.75) {$i_2'$};
    \node[below] () at (0.5, -0.5) {$i_2$};
    \node[below]  () at (0.25, 0.25) {$\bar{s}$};

    \draw[] (-0.25, 0.75) circle (0.1cm);
    \draw[] (-0.5, -0.5) circle (0.1cm);
    \draw[] (-0.25, 0.25) circle (0.1cm);
    \end{scope}
\end{tikzpicture}
&= \sum_{\beta}
\begin{tikzpicture}[baseline={([yshift=-0.5ex]current bounding box.center)}, line width=1pt, scale=0.8]
    \begin{scope}[xshift=-1cm, yshift=2.5cm]
    \draw[] (-1, 0) -- ++ (1, 1) -- ++ (1, -1) -- ++ (-1, -1) -- cycle;
    \draw[] (-0.5, -0.5) -- ++ (0.5, 0.5) -- ++ (0.5, -0.5);
    \node[left] () at (-0.5, -0.5) {$\alpha'$};
    \node[right] () at (0.5, -0.5) {$\beta'$};
    \node[below] () at (-0.25, -0.75) {$i_2'$};
    \node[above] () at (-0.5, 0.5) {$\bar{s}$};
    \node[above] () at (-0.25, -0.25) {$i_2$};
    \end{scope}

    \draw[] (-0.5, 0.5) -- ++ (-0.5, -0.5) -- ++ (1, -1) -- ++ (1, 1) -- ++ (-0.5, 0.5);
    \draw[] (-0.5, 0.5) -- ++ (0.5, -0.5) -- ++ (0.5, 0.5);
    \draw[] (-0.5, 0.5) -- ++ (-0.5, 0.5) -- ++ (-1, -1) -- ++ (2, -2) -- ++ (2, 2) -- ++ (-1, 1) -- ++ (-0.5, -0.5);
    \node[left] () at (-0.5, 0.5) {$\beta'$};
    \node[right] () at (0.5, 0.5) {$\gamma_2$};
    \node[right] () at (-0.75, 1) {$\bar{i}_2'$};
    \node[below] () at (0.5, -0.5) {$i_2$};
    \node[below]  () at (0.25, 0.25) {$\bar{s}$};

    \draw[] (-0.5, -0.5) circle (0.1cm);
    \draw[] (-0.25, 0.25) circle (0.1cm);
\end{tikzpicture}\\
&=d_{i_2'}
\sum_{\beta}
\begin{tikzpicture}[baseline={([yshift=-0.5ex]current bounding box.center)}, line width=1pt, scale=0.8]
    \begin{scope}[xshift=-1cm, yshift=2.5cm]
    \draw[] (-0.5, -0.5) -- ++ (-0.5, 0.5) -- ++ (1, 1) -- ++ (1, -1) --++ (-0.5, -0.5);
    \draw[] (-0.5, -0.5) -- ++ (0, -1.5);
    \draw[] (0.5, -0.5) -- ++ (0, -1.5);
    \draw[] (-0.5, -0.5) -- ++ (0.5, 0.5) -- ++ (0.5, -0.5);
    \node[left] () at (-0.5, -0.5) {$\alpha'$};
    \node[right] () at (0.5, -0.5) {$\beta'$};
    \node[left] () at (-0.5, -1) {$i_2'$};
    \node[above] () at (-0.5, 0.5) {$\bar{s}$};
    \node[above] () at (-0.25, -0.25) {$i_2$};
    \end{scope}

    \draw[] (-0.5, 0.5) -- ++ (-0.5, -0.5) -- ++ (1, -1) -- ++ (1, 1) -- ++ (-0.5, 0.5);
    \draw[] (-0.5, 0.5) -- ++ (0.5, -0.5) -- ++ (0.5, 0.5);
    \draw[] (-1.5, 0.5) -- ++ (-0.5, -0.5) -- ++ (2, -2) -- ++ (2, 2) -- ++ (-1, 1) -- ++ (-0.5, -0.5);
    \node[left] () at (-0.5, 0.5) {$\beta'$};
    \node[right] () at (0.5, 0.5) {$\gamma_2$};
    \node[right] () at (-0.5, 1) {$\bar{i}_2'$};
    \node[below] () at (0.5, -0.5) {$i_2$};
    \node[below]  () at (0.25, 0.25) {$\bar{s}$};

    \draw[] (-0.5, -0.5) circle (0.1cm);
    \draw[] (-0.25, 0.25) circle (0.1cm);
\end{tikzpicture} \\
&= \sqrt{d_{s}d_{i_2} d_{i_2'}}\begin{tikzpicture}[baseline={([yshift=-0.5ex]current bounding box.center)}, line width=1pt, scale=0.8]
    \begin{scope}[xshift=-1cm, yshift=2.5cm]
    \draw[] (-0.5, -0.5) -- ++ (-0.5, 0.5) -- ++ (1, 1) -- ++ (1, -1) --++ (-0.25, -0.25);
    \draw[] (-0.5, -0.5) -- ++ (0, -1.5);
    \draw[] (0.25, -0.25) -- ++ (0, -2);
    \draw[] (0.75, -0.25) -- ++ (0, -2);
    \draw[] (-0.5, -0.5) -- ++ (0.5, 0.5) -- ++ (0.25, -0.25);
    \node[left] () at (-0.5, -0.5) {$\alpha'$};
    \node[left] () at (-0.5, -1) {$i_2'$};
    \node[above] () at (-0.5, 0.5) {$\bar{s}$};
    \node[above] () at (-0.25, -0.25) {$i_2$};
    \end{scope}

    \draw[] (-0.75, 0.25) -- ++ (-0.25, -0.25) -- ++ (1, -1) -- ++ (1, 1) -- ++ (-0.5, 0.5);
    \draw[] (-0.25, 0.25) -- ++ (0.25, -0.25) -- ++ (0.5, 0.5);
    \draw[] (-1.5, 0.5) -- ++ (-0.5, -0.5) -- ++ (2, -2) -- ++ (2, 2) -- ++ (-1, 1) -- ++ (-0.5, -0.5);
    \node[right] () at (0.5, 0.5) {$\gamma_2$};
    \node[below] () at (0.5, -0.5) {$i_2$};
    \node[below]  () at (0.25, 0.25) {$\bar{s}$};

    \draw[] (-0.5, -0.5) circle (0.1cm);
    \draw[] (-0.125, 0.125) circle (0.1cm);
\end{tikzpicture}\\
&=d_s d_{i_2} d_{i_2'} \delta_{\alpha', \gamma_2}.
\end{aligned}
\end{equation}
The first line uses the definition of the Frobenius-Schur indicator $\varkappa_{i_2'}$. The second uses the completeness relation [Eq.~\eqref{eq:completeness}] and the vacuum identity [Eq.~\eqref{eq:vacuum_identity}]. The third line uses the completeness relation [Eq.~\eqref{eq:completeness}]. The last line uses the orthogonality of the states corresponding to different multiplicity labels [Eq.~\eqref{eq:orthogonality}] and the definition of the Frobenius-Schur indicator.

Finally, we can disentangle each edge, obtaining the following state: 
\begin{equation}
    \frac{1}{d_s^3} C_N \sum_{\substack{\gamma_1, \ldots, \gamma_6 \\ i_1',\ldots, i_6'}}  
    \begin{tikzpicture}[baseline={([yshift=-0.5ex]current bounding box.center)}, scale=0.175, line width=1pt]
            \node[above] (e4bar) at (0, 4cm) {$\bar{e}_4$};
        \node[above] (e5bar) at (4cm, 4cm) {$\bar{e}_5$};
        \node[above] (i4bar) at (12cm, 4cm) {$\bar{i}_4'$};
        \node[] (i4) at (16.5cm, 5.5cm) {$i_4'$};
        \node[] (i3) at (19.5cm, 5.5cm) {$i_3'$};
        \node[above] (i3bar) at (24cm, 4cm) {$\bar{i}_3'$};
        \node[above] (e3bar) at (32cm, 4cm) {$\bar{e}_3$};

        \draw[] (e4bar) -- ++ (0, -4cm) -- ++ (9cm, -4cm) -- ++ (9cm, 4cm); 
        \draw[] (e5bar) -- ++ (0, -4cm) -- ++ (2cm, -1cm) -- ++ (2cm, 1cm) -- ++ (2cm, -1cm) -- ++ (2cm, 1cm);
        \draw[] (8cm, 1.25cm) -- ++ (0, 6cm);
        \node[above] () at (8cm, 7.25cm) {$i_5'$};
        \node[left] () at (8cm, 5cm) {$\alpha$};
        \node[right] () at (28cm, 5cm) {$\alpha'$};
        \draw[red] (8cm, 5cm) -- (9cm, 1cm);
        \draw[red] (28cm, 5cm) -- (27cm, 1.25cm);
        \draw[red] (18cm, 3cm) -- (20cm, 4cm);
        \draw[red] (18cm, 3cm) -- (16cm, 4cm);
%        \draw[red] (8cm, 11.25cm) -- ++ (10cm, -2cm) -- ++ (10cm, 2cm);
%        \node[above] () at (13cm, 10.25cm) {$s$};
%        \node[above] () at (23cm, 10.25cm) {$\bar{s}$};
        \draw[] (12cm, 1.375cm) -- ++ (i4bar);
        \draw[] (16cm, 4cm) -- ++ (-2cm, 2.5cm);
        \draw[] (16cm, 4cm) -- ++ (2cm, -2.5cm) -- ++ (2cm, 2.5cm);
        \node[right] () at (8cm, 2.25cm) {$\alpha_5$};
        \node[right] () at (18cm, 1.5cm) {$\alpha_4$};
        \draw[] (4.5cm, -0.5cm) circle (0.5cm);
        \draw[] (5cm, 1cm) circle (0.5cm);
        \draw[] (11cm, 1cm) circle (0.5cm);

        \draw[] (20cm, 4cm) -- ++ (2cm, 2.5cm);

        \draw[] (24cm, 4cm) -- ++ (0, -2.25cm) -- ++ (2cm, -1cm) -- ++ (2cm, 1cm) -- ++ (2cm, -1cm) -- ++ (2cm, 1cm);
        \draw[] (28cm, 4cm) -- ++ (0, -2.25cm);
        \draw[] (28cm, 3.25cm) -- ++ (0, 4cm);
        \node[above] () at (28cm, 7.25cm) {$i_2'$};
        \draw[] (32cm, 4cm) -- ++ (0, -2.25cm);
        \node[right]  () at (28cm, 2.5cm) {$\alpha_3$};
        \draw[] (25cm, 1.25cm) circle (0.5cm);
        \draw[] (31cm, 1.25cm) circle (0.5cm);
    
        \node[above] (i2bar) at (36cm, 4cm) {$\bar{i}_2'$};
        \node[above] (i5bar) at (40cm, 4cm) {$\bar{i}_5'$};
        \node[above] (e6) at (44cm, 4cm) {$e_6$};
        \node[] (i6) at (47.5cm, 5.5cm) {$i_6'$};
        \node[above] (i6bar) at (52cm, 4cm) {$\bar{i}_6'$};
        \node[above] (e1) at (56cm, 4cm) {$e_1$};
        \node[above] (i1bar) at (60cm, 4cm) {$\bar{i}_1'$};
        \node[above] (i1) at (64.5cm, 4cm) {$i_1'$};
        \node[above] (e2) at (68cm, 4cm) {$e_2$};

        \draw[] (40cm, 4cm) -- ++ (4cm, -4cm) -- ++ (4cm, 4cm);
        \draw[] (44cm, 4cm) -- ++ (2cm, -2cm);
        \node[] () at (46cm, 3.5cm) {$\alpha_6$};
        \draw[] (42cm, 2cm) circle (0.5cm);

%        \draw[] (48cm, 4cm) -- ++ (4cm, 4cm) -- ++ (0, -4cm);

        \draw[] (52cm, 4cm) -- ++ (0, -2.25cm) -- ++ (2cm, -1cm) -- ++ (2cm, 1cm) -- ++ (2cm, -1cm) -- ++ (2cm, 1cm);
        \draw[] (56cm, 4cm) -- ++ (0, -2.25cm);
        \draw[] (60cm, 4cm) -- ++ (0, -2.25cm);
        \node[right]  () at (56cm, 2.5cm) {$\alpha_1$};
        \draw[] (53cm, 1.25cm) circle (0.5cm);
        \draw[] (59cm, 1.25cm) circle (0.5cm);

        \draw[] (64cm, 4cm) -- ++ (2cm, -2cm) -- ++ (2cm, 2cm);
        \draw[] (66cm, 2cm) -- ++ (-15cm, -5cm) -- ++ (-15cm, 5cm)-- ++ (0, 2cm);
        \node[right] () at (66cm, 1.5cm) {$\alpha_2$};
        \draw[] (40.5cm, 0.5cm) circle (0.5cm);  

        % s-strings
        \draw[red] (65cm, 3cm) -- ++ (0, -1.5cm);
        \draw[red] (44.75cm, 0.75cm) -- ++ (1.5cm, 0.75cm) -- ++ (0.75cm, 1.5cm);
        
        \draw[red] (56cm, 1.25cm) -- ++ (1.5cm, -0.25cm);
        \draw[red] (56cm, 1.25cm) -- ++ (-1.5cm, -0.25cm);
    \end{tikzpicture}
\end{equation}
Note that each of the six diagrams can be converted to a diagram without the $s$-string using the transformation in Eq.~\eqref{eq:plaquette_bs}. This is because the diagrammatic rules we use are exactly the ones used in deriving these coefficients~\cite{Hahn2020}. Accounting for the normalization factor, we obtain
\begin{equation}
    (\tilde{B}_p^s)^{(2)} \,\,
    \begin{tikzpicture}[baseline={([yshift=-0.5ex]current bounding box.center)}, scale=0.15, line width=1pt]
    \draw[] (-17cm, -16cm) -- ++ (0, 32cm);
    \draw[] (17cm, -14cm) -- ++ (3cm, 14cm) -- ++ (-3cm, 14cm);
        \vertexone{-10.392cm}{6cm}{$\alpha_6$}{$e_6$}{$\bar{i}_6$}{$\bar{i}_5$};
        \vertextwo{-10.392cm}{-6cm}{$\alpha_5$}{$\bar{e}_5$}{$\bar{i}_4$}{$i_5$};
        \vertexone{10.392cm}{6cm}{$\alpha_2$}{$i_1$}{$e_2$}{$\bar{i}_2$};
        \vertextwo{10.392cm}{-6cm}{$\alpha_3$}{$\bar{i}_3$}{$\bar{e}_3$}{$i_2$};
        \vertextwo{0cm}{12cm}{$\alpha_1$}{$i_6$}{$\bar{i}_1$}{$e_1$};
        \vertexone{0cm}{-12cm}{$\alpha_4$}{$i_4$}{$i_3$}{$\bar{e}_4$};
    \end{tikzpicture}
    =
    \sum_{\substack{i_1', \ldots, i_6' \\ \alpha_1' \ldots, \alpha_6'}} B_{p, \{i_k\}, \{\alpha_k \}}^{ \{ i_k'\} \{\alpha_k' \}} \left(\{ e_k\}\right)
    \begin{tikzpicture}[baseline={([yshift=-0.5ex]current bounding box.center)}, scale=0.15, line width=1pt]
    \draw[] (-17cm, -16cm) -- ++ (0, 32cm);
    \draw[] (17cm, -14cm) -- ++ (3cm, 14cm) -- ++ (-3cm, 14cm);
        \vertexone{-10.392cm}{6cm}{$\alpha_6'$}{$e_6$}{$\bar{i}_6'$}{$\bar{i}_5'$};
        \vertextwo{-10.392cm}{-6cm}{$\alpha_5'$}{$\bar{e}_5$}{$\bar{i}_4'$}{$i_5'$};
        \vertexone{10.392cm}{6cm}{$\alpha_2'$}{$i_1'$}{$e_2$}{$\bar{i}_2'$};
        \vertextwo{10.392cm}{-6cm}{$\alpha_3'$}{$\bar{i}_3'$}{$\bar{e}_3$}{$i_2'$};
        \vertextwo{0cm}{12cm}{$\alpha_1'$}{$i_6'$}{$\bar{i}_1'$}{$e_1$};
        \vertexone{0cm}{-12cm}{$\alpha_4'$}{$i_4'$}{$i_3'$}{$\bar{e}_4$};
    \end{tikzpicture},
    \label{eq:plaquette_matrix_element_ours}
\end{equation}
where the  shown vector is a tensor product of the six normalized vectors, each defined as follows:
\begin{equation}
\begin{aligned}
\begin{tikzpicture}[baseline={([yshift=-0.5ex]current bounding box.center)}, line width=1pt, scale=0.175]
\draw[] (-6cm, -5cm) -- ++ (0, 10cm);
\draw[] (6cm, -4.5cm) -- ++ (1cm, 4.5cm) -- ++ (-1cm, 4.5cm);
 \vertexone{0}{0}{$\alpha$}{$a$}{$b$}{$c$};
\end{tikzpicture}
&:= \frac{1}{\sqrt[4]{d_a d_b d_c}} 
\begin{tikzpicture}[baseline={([yshift=-0.5ex]current bounding box.center)}, line width=1pt, scale=0.175]
\draw[] (0,0) -- ++ (0, 4cm);
            \draw[] (4cm,0) -- ++ (0, 4cm);
            \draw[] (-4cm,0) -- ++ (0, 4cm);
            \draw[] (-4cm,0) -- ++ (2cm, -1cm) -- ++ (2cm, 1cm) -- ++ (2cm, -1cm) -- ++ (2cm, 1cm);
            \node[above] () at (-4cm, 4cm) {$a$};
            \node[above] () at (0cm, 4cm) {$c$};
            \node[above] () at (4cm, 4cm) {$b$};
            \node[right] () at (0, 0.5cm) {$\alpha$};
            \draw[] (3cm, -0.5cm) circle (0.5cm);
            \draw[] (-3cm, -0.5cm) circle (0.5cm);    
\end{tikzpicture},
\\
\begin{tikzpicture}[baseline={([yshift=-0.5ex]current bounding box.center)}, line width=1pt, scale=0.175]
\draw[] (-6cm, -5cm) -- ++ (0, 10cm);
\draw[] (6cm, -4.5cm) -- ++ (1cm, 4.5cm) -- ++ (-1cm, 4.5cm);
 \vertextwo{0}{0}{$\alpha$}{$a$}{$b$}{$c$};
\end{tikzpicture}
&:= \frac{1}{\sqrt[4]{d_a d_b d_c}}
\begin{tikzpicture}[baseline={([yshift=-0.5ex]current bounding box.center)}, line width=1pt, scale=0.175]
     \draw[] (-4cm, 4cm) -- (0,0);
            \draw[] (4cm, 4cm) -- (0,0);
            \draw[] (2cm, 2cm) -- (0, 4cm);
            \node[above] () at (-4cm, 4cm) {$c$};
            \node[above] () at (0cm, 4cm) {$a$};
            \node[above] () at (4cm, 4cm) {$b$};
            \node[above] () at (2cm, 2cm) {$\alpha$};
            \draw[] (-2cm, 2cm) circle (0.5cm);
\end{tikzpicture}
\end{aligned}.
\end{equation}
The coefficients in Eq.~\eqref{eq:plaquette_matrix_element_ours} are the ones defined precisely for the string-net model [Eq.~\ref{eq:b_definition}]. Thus we conclude that our definition of the plaquette operator coincides exactly with the definition of the plaquette operator of the string-net model.

\section{Bulk anyons and circuit equivalences}
\label{sec:classification}

We elaborate on the implications of Theorem \ref{thm:mapping-to-sn} (mapping states to string-nets) on the classification of topological phases. The key claims were already summarized in Section \ref{sec:setup-results}.
Our main goal is to prove Corollary \ref{cor:classification-main}.  This result characterizes the equivalence classes of entanglement bootstrap states with gappable boundary, modulo constant-depth circuits.

To begin, we codify how to extract the full anyon data from the ground state, assuming the entanglement bootstrap axioms.  Here we
are only concerned with the ``bulk'' anyons, rather than the anyons associated to a gapped boundary.  Accordingly, we only assume the entanglement bootstrap axioms in the bulk (Definition \ref{def:bulk-EB}).  

The bulk anyons are described by a UMTC.  The paper \cite{shi2020fusion}, largely introducing entanglement bootstrap ideas, demonstrated how to extract part of this data.  Essentially, they demonstrated how to extract the underlying fusion ring, which tells you which anyons are produced when two anyons are fused.  However, it does not tell you about the complex phases produced during fusion, encoded by the $F$-symbols, nor does it reveal the braiding properties, encoded by the $R$-symbols.  Both of these are required to define the full UMTC describing the anyons.  

Fortunately, the methods of \cite{shi2020fusion} can be combined with the results of Kawagoe and Levin \cite{kawagoe2020microscopic} to extract the full UMTC from a state satisfying the entanglement bootstrap axioms.  Moreover, this UMTC is circuit-invariant.
\begin{theorem}[Bulk anyon contents]
\label{thm:bulk-anyons}
    For every 2D state satisfying the entanglement bootstrap axioms in the bulk (Definition \ref{def:bulk-EB}), we can assign a UMTC $\mathcal{C}$, following \cite{shi2020fusion,kawagoe2020microscopic}.  For two such states connected by a constant-depth circuit, the associated UMTCs are equivalent.  
\end{theorem}
\begin{proof}
First we discuss how to define the UMTC. In fact, we already detailed a similar procedure in Section \ref{sec:boundary_anyons}, where we extracted the UFC (describing the boundary anyons) from a state with gapped boundary satisfying the entanglement bootstrap axioms.  Because of the similarity, we will be brief. 

First, we identify the bulk anyon types, fusion rules, and string operators precisely as in \cite{shi2020fusion}. Using these operators, we define the closely related movement, splitting, and fusion operators. These operations for bulk anyons are defined analogously to those for boundary anyons already discussed in Section \ref{sec:boundary_anyons}.  The movement, splitting, and fusion operators are precisely the inputs needed for the procedure developed in \cite{kawagoe2020microscopic}, for calculating both the $F$-symbols and $R$-symbols.  Thus we obtain a full unitary braided tensor category.  We have not yet seen that it is modular, i.e., has nondegenerate braiding. However, the nondegeneracy of the $S$-matrix for entanglement bootstrap states was shown by \cite{shi2020verlinde}.  Thus we the category we obtain is a UMTC.

 Now we can show the circuit-invariance of the associated UMTC. Consider states $\rho_1$ and $\rho_2$ satisfying the bulk entanglement bootstrap axioms.  Let $\{S^{(1)}\}$ and $\{S^{(2)}\}$ schematically denote the corresponding movement operators (together with the fusion and splitting operators).  Using the above procedure, we associate UMTCs $\mathcal{C}_1$ and $\mathcal{C}_2$, describing their bulk anyons.  Assume $U \rho_1 U^\dagger = \rho_2$ for constant-depth circuit $U$.  Then $\{U^\dagger S^{(2)} U \}$ provides an alternative set of operations for $\rho_1$, in addition to the operations $\{S^{(1)}\}$ already prescribed. The new operations $\{U^\dagger S^{(2)} U \}$ are still a valid input to the procedure described in  \cite{kawagoe2020microscopic} for extracting the $F$-symbols and $R$-symbols, so we can compute the corresponding UMTC $\mathcal{C}_1'$.  

 Kawagoe and Levin \cite{kawagoe2020microscopic} demonstrate that all valid sets of movement, fusion, and splitting operators for the same system yield the same $F$-symbols and $R$-symbols, up to gauge transformation of $F$ and $R$. 
 That is, they yield an equivalent UMTC, so $\mathcal{C}_1' \cong \mathcal{C}_1$.
 
 On the other hand, $\mathcal{C}_1' = \mathcal{C}_2$, because using operations $\{U^\dagger S^{(2)} U \}$ to manipulate state $\rho_1 = U^\dagger \rho_2 U$ will manifestly yield the same matrix-elements as using operations $\{ S^{(2)} \}$ to manipulate state $\rho_2$. 
 
 Then $\mathcal{C}_1$ and $\mathcal{C}_2$ are both equivalent to $\mathcal{C}_1'$, and we are done.
\end{proof}

Recall that a string-net state is associated to each UFC $\mathcal{C}$. The bulk anyon contents of the string-net are expected to be given by the ``center'' of $\mathcal{C}$, which is a UMTC denoted $Z(\mathcal{C})$.  In some sense this is shown in \cite{kitaev2012models}.  However, we would like to know that when the bulk anyons are manipulated with string operators, they actually behave according to the $F$-symbols and $R$-symbols of $Z(\mathcal{C})$.  In other words, we would like to show:
\begin{conjecture}
    Consider the Levin-Wen model built on an arbitrary unitary fusion category $\mathcal{C}$. Then there exist string operators associated to the simple objects of $Z(\mathcal{C})$, such that when these are used to define $F$-symbols and $R$-symbols with the procedure in \cite{kawagoe2020microscopic}, one obtains $Z(\mathcal{C})$. 
    \label{conjecture:SN-bulk-anyons}
\end{conjecture}
This is a formalization of a folk theorem widely believed by experts, but to our knowledge there is no proof in the literature.  There are string operators described in \cite{kitaev2012models,lan2014topological,hu2018full,christian2023lattice} that correspond to simple objects in $Z(\mathcal{C})$, but their precise behavior under fusion and braiding does not appear to be shown.

Suppose Conjecture \ref{conjecture:SN-bulk-anyons} holds.  Then Corollary \ref{cor:classification-main} states that two gapped ground states with gappable boundary satisfying the axioms of Definition \ref{def:gappable-boundary-EB} are connected by a constant-depth circuit if and only if they have the same bulk anyon contents, as defined in Theorem \ref{thm:bulk-anyons}. 
We are ready to prove this.
\begin{proof}
Consider two states $\rho_1$, $\rho_2$  satisfying the axioms of Definition \ref{def:gappable-boundary-EB}. By Theorem \ref{thm:mapping-to-sn}, there are constant-depth circuits mapping them to string-net states labeled by UFCs $\mathcal{C}_1$ and $\mathcal{C}_2$, respectively.  

To prove one direction, assume $\rho_1$ are $\rho_2$ are related by a constant-depth circuit.  Then by Theorem \ref{thm:bulk-anyons}, they have the same bulk anyon contents.

To prove the other direction, assume $\rho_1$ and $\rho_2$ have the same bulk anyon contents.  That is, they are assigned equivalent UMTCs. Then the associated string-nets also have the same bulk anyon contents, again by Theorem \ref{thm:bulk-anyons}.  By Conjecture \ref{conjecture:SN-bulk-anyons}, the UMTCs are given by $Z(\mathcal{C}_1)$ and $Z(\mathcal{C}_2)$, so  $Z(\mathcal{C}_1) \cong Z(\mathcal{C}_2)$, i.e., $\mathcal{C}_1$ and $\mathcal{C}_2$ are Morita-equivalent.  By Ref.~\cite{lootens2022mapping}, Morita-equivalent string-nets are connected by constant-depth circuits, as desired.
\end{proof}
We conclude that the equivalence classes of states satisfying the criteria of Definition \ref{def:gappable-boundary-EB} under constant-depth circuits are precisely labeled by doubled UMTCs, by which we mean UMTCs of the form $Z(\mathcal{C})$ for some UFC $\mathcal{C}$.

Note that for a system with boundary, two states with the same bulk anyon content may have distinct boundaries. Meanwhile, the circuits we construct in Theorem \ref{thm:mapping-to-sn} only connect the two states in the sense of making them identical in the bulk, not necessarily on the boundary. Thus it may be the case that the boundary regions of the two states cannot be connected with a constant-depth circuit.  A natural conjecture is that there is a constant-depth circuit that maps the boundary regions as well if and only if the boundary UFCs are equivalent.  In fact we expect this may be proven using the techniques in this paper.

\section{Doubled states and gapped boundaries}
\label{sec:double}

In this section, we sketch how to prove a variant of Theorem \ref{thm:mapping-to-sn} for a class of ``doubled'' systems, without explicitly assuming gappable boundary.

Because there exist 2D topological phases with ungappable boundary that cannot be mapped to string-nets (though perhaps not shown rigorously), we expect our assumption of a gappable boundary in Theorem \ref{thm:mapping-to-sn} is necessary.  However, gapped boundaries do arise automatically for doubled systems.  This may be explained through the so-called folding trick.  Consider a 2D system living on a flat sheet of paper.  Now fold the sheet in half.  Then we have two half-sheets of paper that meet at the fold.  We can consider the two stacked half-sheets as a single 2D system, with twice as many local degrees of freedom.   Moreover, the boundary where the sheet was folded is naturally a \textit{gapped} boundary; the gapped parent Hamiltonian is just the parent Hamiltonian of the original sheet, re-arranged in the new geometry. 

If we study the folded system in the bulk (away from the fold), it looks like two stacked copies of the original system, but one copy is spatially reflected.  To describe the bulk of the doubled system, we start with a gapped ground state $\sigma$ on the 2D plane, and we stack it with its reflection to form the state
\begin{align} \label{eq:double}
    \sigma_{\text{double}} =  \sigma \boxtimes  \sigma^{\text{flip}}.
\end{align}
Here $\sigma_{\text{double}}$ is the two-layer ``doubled'' state, $\sigma^{\text{flip}}$ denotes the a copy of the 2D state $\sigma$ but with the qudits spatially reflected across a chosen axis, and $\boxtimes$ denotes the ordinary tensor product, but specially marked to indicate that we imagine stacking one sheet on top of the other to form a single 2D geometry.  The state $\sigma_{\text{double}}$ need not have a physical boundary (if we imagine it on the infinite 2D plane), but it does have a $\textit{gappable}$ boundary, demonstrated by the above folding trick. 

Because doubled states of the form \eqref{eq:double} appear to have gappable boundary, we should be able to map them to string-nets with constant-depth circuits. Here we sketch how to prove this using Theorem \ref{thm:mapping-to-sn}, but without explicitly assuming a gappable boundary in the sense of Definition \ref{def:gappable-boundary-EB}. 

That is, we can show:
\begin{theorem}[Doubled states; proof only sketched.]
\label{thm:doubled}
    Let $\sigma$ be a 2D state that is translation-invariant and satisfies the bulk entanglement bootstrap axioms (Definition \ref{def:bulk-EB}), and consider the doubled state $\sigma_{\text{double}}$ of Eq.~\eqref{eq:double}.
    Then there exists a constant-depth unitary circuit that maps $\sigma_{\text{double}}$ to a string-net state.
\end{theorem}
More precisely, for any disk subregion $A$ of the doubled system, there exists a constant-depth circuit $U$ such that $\Tr_{A \backslash A^-} U \sigma_{\text{double}} U^\dagger = \Tr_{A \backslash A^-} \sigma_{SN}$, where $\sigma_{SN}$ is a canonical string-net state, and $A^-$ is the interior of $A$. 

To prove Theorem \ref{thm:doubled}, we want to apply Theorem \ref{thm:mapping-to-sn}.  The only hard part is to show that $\sigma_{\text{double}}$ genuinely satisfies the notion of gappable boundary in Definition \ref{def:gappable-boundary-EB}.  This requires explicitly demonstrating that various regions can be given a gapped boundary. Note we require $\sigma$ is translation-invariant, unlike in Theorem \ref{thm:mapping-to-sn}; we expect this is an artefact of our current proof.

The key fact is that if a state satisfies the bulk entanglement bootstrap axioms (Definition \ref{def:bulk-EB}), then the folded system obeys the boundary entanglement bootstrap axioms (Definition \ref{def:boundary-EB}).  This is straightforward to check, because the boundary regions in Figure~\ref{fig:eb_axioms_boundary} for the folded system are precisely the bulk regions in Figure ~\ref{fig:eb_axioms_bulk} for the unfolded system.

Recall Definition \ref{def:gappable-boundary-EB} requires $\sigma$ have the property that disks $A$ can be given a gapped boundary, i.e., there exists some state $\sigma'_A$ satisfying the entanglement bootstrap axioms for gapped boundary and which matches $\sigma$ on the interior of $A$.  We call these states and associated regions ``fragments.'' As explained in Appendix~\ref{sec:elementary-fragments}, in fact the proof of Theorem~\ref{thm:mapping-to-sn} only requires a certain set of ``elementary fragments.''

\begin{figure}
    \centering
\tikzset{every picture/.style={line width=0.75pt}} %set default line width to 0.75pt        

\begin{tikzpicture}[x=0.75pt,y=0.75pt,yscale=-1,xscale=1]
%uncomment if require: \path (0,426); %set diagram left start at 0, and has height of 426

%Rounded Rect [id:dp9803868231485453] 
\draw  [color={rgb, 255:red, 245; green, 166; blue, 35 }  ,draw opacity=0.49 ][fill={rgb, 255:red, 0; green, 0; blue, 0 }  ,fill opacity=0.3 ][line width=2.25]  (10.67,248.48) .. controls (10.67,236.86) and (20.08,227.45) .. (31.69,227.45) -- (210.84,227.45) .. controls (222.45,227.45) and (231.87,236.86) .. (231.87,248.48) -- (231.87,311.56) .. controls (231.87,323.17) and (222.45,332.58) .. (210.84,332.58) -- (31.69,332.58) .. controls (20.08,332.58) and (10.67,323.17) .. (10.67,311.56) -- cycle ;
%Shape: Polygon [id:ds697008261165595] 
\draw  [fill={rgb, 255:red, 178; green, 178; blue, 255 }  ,fill opacity=0.57 ] (153.17,287.83) -- (163.78,303.92) -- (127.5,327.88) -- (128.18,369.3) -- (107.98,369.65) -- (108.56,327.62) -- (76.83,304.26) -- (87.44,288.51) -- (118.84,310.71) -- cycle ;
%Shape: Polygon [id:ds3892699566772728] 
\draw  [fill={rgb, 255:red, 178; green, 255; blue, 189 }  ,fill opacity=0.57 ] (27.76,317.51) -- (17.19,301.39) -- (53.55,277.53) -- (52.99,236.11) -- (73.19,235.83) -- (72.49,277.85) -- (105.52,300.97) -- (93.49,317.02) -- (62.16,294.73) -- cycle ;
%Shape: Polygon [id:ds842927613206865] 
\draw  [fill={rgb, 255:red, 208; green, 2; blue, 27 }  ,fill opacity=0.57 ] (142.42,318.18) -- (131.86,302.06) -- (168.22,278.2) -- (167.66,236.78) -- (187.85,236.5) -- (187.16,278.52) -- (220.18,301.63) -- (208.15,317.69) -- (176.82,295.4) -- cycle ;
\end{tikzpicture}
\caption{We illustrate the construction of a vertex fragment of the doubled state with gapped boundary. The above depicts the non-doubled state, highlighting two vertex regions (green, red). The gray shaded disk is considered a sphere by identifying its orange boundary as a single point.  Then the green and red regions of the sphere are pressed together, and the remainder of the gray is coarse-grained to form a gapped boundary of the resulting doubled vertex region.
}
\label{fig:vertex-fold-1}
\end{figure}
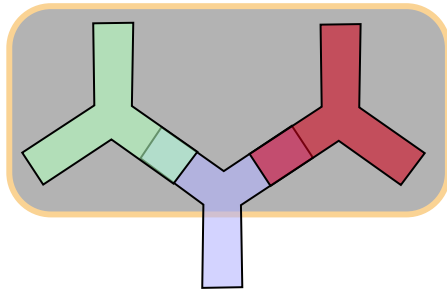

To prove Theorem~\ref{thm:doubled}, we must construct these fragments with gapped boundary, rather than assume their existence. We will construct elementary fragments as in Figure~\ref{fig:long-vertex-fragments}.  These are all vertex fragments, and they come in two types, corresponding to the two vertex types on the hexagonal lattice. Consider these vertex regions on the non-doubled state $\sigma$; they are simply subregions of the bulk state, and they do not have gapped boundary.  However, we can construct a fragment with gapped boundary for the vertex region of the \textit{doubled} state.  

Consider the shaded region in Figure~\ref{fig:vertex-fold-1} for the non-doubled state.  It contains two adjacent (non-overlapping) vertex regions of the same type, green and red.  Take the shaded region and  imagine identifying the orange boundary as a single point.  Then we obtain a sphere, where the green and red vertex regions are subregions of the sphere on opposite sides.  Then press the red and green regions together to form a doubled vertex region, and locally coarse-grain the remaining shaded region to form the boundary of this doubled region. 

Through the above manipulations, we obtain a doubled vertex region with gapped boundary, i.e., satisfying Definition~\ref{def:boundary-EB}. Call this state $\tau_A$, where $A$ is the vertex region. This will serve as an ordinary fragment with boundary for the doubled state $\sigma_{\text{double}}$. In our current convention, we should imagine $\sigma_{\text{double}}$ was doubled by reflecting across the vertical axis of Figure~\ref{fig:vertex-fold-1}.  Then $\tau_A$ matches $\sigma_{\text{double}}$ on the interior of $A$ due to the assumed translation-invariance of $\sigma$.  We repeat an analogous process to define a fragment for the second vertex type.

In this way, we obtain elementary fragments with gapped boundary for $\sigma_{\text{double}}$.  To complete the proof, all that remains is to check the consistency of these fragments with each other where they have overlapping boundaries.  This requires a specific procedure for coarse-graining the gray region of Figure~\ref{fig:vertex-fold-1} into the boundary of the doubled vertex region. That final procedure is hard to illustrate; we hope to provide further explanation in a future edition. For that reason, we refer to this proof of Theorem~\ref{thm:doubled} as a sketch.

\section{Discussion}
\label{sec:discussion}
We considered 2D topological phases with gappable boundary, built on lattices of qudits (rather than using fermionic degrees of freedom), and without imposing symmetries.  One could attempt generalizations to chiral phases, higher dimensions, symmetry-protected topological phases, fermions, and so on.  But first, restricted to our current setting, we can ask: how restrictive were our assumptions, and should our results be expected?

More precisely, we considered a class of quantum states satisfying certain conditions [Definition~\ref{def:gappable-boundary-EB}] that stringently codify the notion of being a gapped ground state with gappable boundary.   We showed these states can be mapped to string-nets using constant-depth circuits.  So if there is some undiscovered exotic phase beyond the Levin-Wen models, it must somehow fail to have any representative satisfying our assumptions.  

If there is a phase with no representative satisfying our assumptions, perhaps the assumptions can still be approximately satisfied, in which case the question becomes whether our proof is robust to small errors.  On the other hand, could there exist exotic phases, for which no representative state of the phase satisfies our assumptions, even approximately?  The assumptions essentially amount to ``zero mutual information at long distance’’ (\textbf{A0}) and ``zero conditional mutual information at long distance’’ (\textbf{A1}). But whereas the former is expected to hold true for general gapped systems (with a suitable approximation error), the latter generally may not.
It is a logical possibility that phases exist wherein no representative satisfies \textbf{A1}, even approximately. In fact, there are examples in higher dimensions: fracton models \cite{pretko2017higher} in three dimensions can fail axiom \textbf{A1} \cite{huang2023knots}, and one can show that all states connected by a constant-depth circuit also fail this property.  

On the other hand, fracton models exhibit pathologies when considered as a ``topological phases,’’ even though on can still consider their equivalence classes under circuits.  Rather than lament the failure of \textbf{A1}, one may instead interpret it as guaranteeing a genuinely topological or ``liquid’’ phase: the class of phases one might expect to be classified by topological quantum field theory.  It is plausible that in two dimensions, all translation-invariant gapped ground states do have a representative satisfying \textbf{A1}, so whether our results actually require \textbf{A1} in the translation-invariant case remains to be shown.  But apparently in three dimensions, axiom \textbf{A1} (or some alternative) is important for distinguishing liquid phases from fracton phases and beyond; the latter likely constitute a wilder variety of equivalence classes under circuits, which may be less natural to classify.

Finally, recall that states satisfying \textbf{A1} may severely violate it after a constant-depth circuit, in a phenomenon related to spurious topological entanglement entropy \cite{Zou2016,Williamson2019,Kato2020,kim2023universal}.  This is consistent with the above discussion, since we have only been discussing the criterion that \text{some} representative of the phase satisfies the axiom.  Nonetheless, it suggests \textbf{A1} is not the most natural axiom to consider.  We hope to remedy this in future work.

String-net states may be defined for more general fusion categories \cite{freedman2012galois, runkel2020string, geer2022non}, beyond the unitary fusion categories.  Perhaps some of these may be understood as defining distinct topological phases. Yet in our accounting of gapped phases with gappable boundary, we always produce a string-net built on a \textit{unitary} fusion category.  Have we somehow missed a larger class of phases, perhaps artificially ruled out by our assumptions?  For a string-net built on a non-unitary fusion category, the canonical parent Hamiltonian defined by Levin-Wen \cite{levin2005string} may be non-Hermitian.  Moreover, at least for some non-unitary string-net states, it has been shown \cite{freedman2012galois} there exists no Hermitian (and gapped, local) parent Hamiltonian. Meanwhile, the states we classify do have Hermitian parent Hamiltonians.  All these results appear consistent, but it would be good to more carefully consider these generalizations, including the sense in which they may correspond to physically realizable topological phases.

\section{Acknowledgments}

We especially thank Michael Levin and Xiao-Gang Wen for helpful intuition on string-nets and gapped boundaries, and Bowen Shi for related collaboration and early discussion.  We also thank Andi Bauer, Laurens Lootens, Alexei Kitaev, and Brandon Rayhaun. DR thanks Edward Mazenc for a figure.
We are grateful to the long term workshop YITP-T-23-01 held at YITP, Kyoto University, where a part of this work was done. IK acknowledges support from NSF under award number PHY-2337931.  DR acknowledges support from from NTT Grant AGMT DTD 9/24/20.

\bibliographystyle{unsrt}
\bibliography{references}

\appendix

\section{Elementary fragments}
\label{sec:elementary-fragments}

When we assume that the reference state $\sigma$ has gappable boundary in Definition~\ref{def:gappable-boundary-EB}, we require that all $O(1)$-size disks $A$ can be given a gapped boundary.  That is, we assume that there exists a state $\sigma'_A$ with gapped boundary in the sense of Definition~\ref{def:boundary-EB}, and which matches $\sigma$ on the interior of $A$.  Sometimes we call these disks and associated regions ``fragments.''

For the sake of parsimony, it would be nice to make this assumption about only certain $O(1)$-sized regions, explicitly specified. 
 We refer to these as ``elementary fragments.'' While specifying these is optional in the context of Theorem~\ref{thm:mapping-to-sn}, we will use them to argue the corollary in Section~\ref{sec:double}.

Figure~\ref{fig:fragment_assumption_elementary} illustrates one particular choice of elementary fragments.  They come in two types: vertex regions and edge regions.  A vertex region overlaps its associated edge regions, and we demand the states are consistent on their overlap. 

\begin{figure}
    \centering
    \input{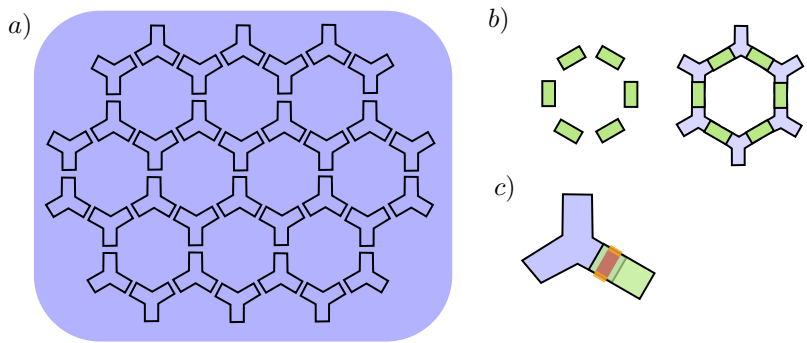}
    \caption{We specify our assumption about the existence of ``fragments'' more explicitly.
    \textbf{(a)} We assume the existence of states on the disk-like regions outlined in black.  These are assumed to have gapped boundary and match the reference state on their interiors. \textbf{(b)} We also assume the existence of states with gapped boundary on the green regions (which overlap the regions from (a)).  \textbf{(c)} An example of two overlapping regions above.  The states are assumed to match on the interior of their overlap (red) including along their shared boundary (orange).}
    \label{fig:fragment_assumption_elementary}
\end{figure}
\begin{figure}
    \centering
     \tikzset{every picture/.style={line width=0.75pt}} %set default line width to 0.75pt        

\begin{tikzpicture}[x=0.75pt,y=0.75pt,yscale=-1,xscale=1]
%uncomment if require: \path (0,426); %set diagram left start at 0, and has height of 426

%Shape: Polygon [id:ds3284062324011363] 
\draw  [fill={rgb, 255:red, 178; green, 178; blue, 255 }  ,fill opacity=1 ] (76.37,42.19) -- (102.48,42.56) -- (103.07,81.22) -- (133.33,99.23) -- (119.94,122.66) -- (89.5,104.69) -- (59.4,122.23) -- (43.95,97.5) -- (75.93,79.16) -- cycle ;
%Shape: Polygon [id:ds9830879286202776] 
\draw  [fill={rgb, 255:red, 178; green, 178; blue, 255 }  ,fill opacity=1 ] (207.17,102.27) -- (220.18,124.91) -- (187.24,145.15) -- (187.21,180.37) -- (160.22,180.82) -- (160.14,145.47) -- (129.68,128.54) -- (143.06,102.63) -- (175.16,120.77) -- cycle ;
%Shape: Rectangle [id:dp8206506465656149] 
\draw  [fill={rgb, 255:red, 184; green, 233; blue, 134 }  ,fill opacity=1 ] (374.22,101.44) -- (388.17,77.28) -- (340.41,49.71) -- (326.46,73.87) -- cycle ;
%Shape: Rectangle [id:dp046974249658844114] 
\draw  [fill={rgb, 255:red, 184; green, 233; blue, 134 }  ,fill opacity=0.55 ] (149.83,139.48) -- (163.78,115.31) -- (116.02,87.74) -- (102.07,111.9) -- cycle ;
%Shape: Polygon [id:ds1192114504939572] 
\draw  [fill={rgb, 255:red, 178; green, 178; blue, 255 }  ,fill opacity=1 ] (280.42,43.96) -- (306.53,44.33) -- (307.11,82.98) -- (338.85,102.24) -- (323.99,124.43) -- (293.55,106.46) -- (263.44,124) -- (248,99.27) -- (279.98,80.93) -- cycle ;
%Shape: Polygon [id:ds6186376645432883] 
\draw  [fill={rgb, 255:red, 178; green, 178; blue, 255 }  ,fill opacity=1 ] (411.21,104.04) -- (424.23,126.68) -- (391.29,146.92) -- (391.25,182.14) -- (364.27,182.59) -- (364.18,147.24) -- (332.83,129.43) -- (347.11,104.4) -- (379.2,122.54) -- cycle ;
%Shape: Path Data [id:dp888141516331928] 
\draw  [fill={rgb, 255:red, 178; green, 178; blue, 255 }  ,fill opacity=1 ] (523.5,55.64) -- (524.09,94.29) -- (564.08,115.71) -- (596.17,133.85) -- (628.18,115.35) -- (641.2,137.99) -- (608.26,158.23) -- (608.22,193.45) -- (581.24,193.89) -- (581.15,158.54) -- (550.7,141.62) -- (510.52,117.76) -- (480.42,135.31) -- (464.97,110.57) -- (496.95,92.24) -- (497.39,55.27) -- (523.5,55.64) -- cycle ;
%Straight Lines [id:da9504010060367545] 
\draw  [dash pattern={on 1.5pt off 1.5pt}]  (320.8,90.6) -- (306.16,114.65) ;
%Straight Lines [id:da745674254040432] 
\draw  [dash pattern={on 1.5pt off 1.5pt}]  (329.6,96.2) -- (316.16,119.45) ;
%Straight Lines [id:da3735787761726681] 
\draw  [dash pattern={on 1.5pt off 1.5pt}]  (358.4,110.6) -- (343.76,134.65) ;
%Straight Lines [id:da38018558663062496] 
\draw  [dash pattern={on 1.5pt off 1.5pt}]  (367.2,116.2) -- (352.56,141.05) ;
%Straight Lines [id:da09745558740686966] 
\draw  [dash pattern={on 1.5pt off 1.5pt}]  (349.8,55.6) -- (335.16,79.65) ;
%Straight Lines [id:da3851310113061224] 
\draw  [dash pattern={on 1.5pt off 1.5pt}]  (359.8,61.1) -- (345.16,85.15) ;
%Straight Lines [id:da4709058421746708] 
\draw  [dash pattern={on 1.5pt off 1.5pt}]  (368.8+1,65.1+1) -- (354.16+1,89.15+1) ;
%Straight Lines [id:da9043363180352324] 
\draw  [dash pattern={on 1.5pt off 1.5pt}]  (379.3,71.6) -- (364.66,95.65) ;

% Text Node
\draw (29,15+2) node [anchor=north west][inner sep=0.75pt]    {$a)$};
% Text Node
\draw (241,15+2) node [anchor=north west][inner sep=0.75pt]    {$b)$};
% Text Node
\draw (449,17+2) node [anchor=north west][inner sep=0.75pt]    {$c)$};
% Text Node
\draw (282+2,80+2) node [anchor=north west][inner sep=0.75pt]    {$A$};
% Text Node
\draw (368.8+3,126.4+3) node [anchor=north west][inner sep=0.75pt]    {$A$};
% Text Node
\draw (341.53+3,34.94+5) node [anchor=north west][inner sep=0.75pt]  [font=\small]  {$B$};
% Text Node
\draw (350.93+5,38.94+5) node [anchor=north west][inner sep=0.75pt]  [font=\small]  {$C$};
% Text Node
\draw (373.93+5,51.14+5) node [anchor=north west][inner sep=0.75pt]  [font=\small]  {$C$};
% Text Node
\draw (385.73+3,56.14+5) node [anchor=north west][inner sep=0.75pt]  [font=\small]  {$B$};
% Text Node
\draw (361.93+5,45.34+5) node [anchor=north west][inner sep=0.75pt]  [font=\small]  {$D$};
% Text Node
\draw (296.33+4,112.34+4) node [anchor=north west][inner sep=0.75pt]  [font=\small]  {$B$};
% Text Node
\draw (305.73+5,118.34+5) node [anchor=north west][inner sep=0.75pt]  [font=\small]  {$C$};
% Text Node
\draw (324.93+5,130.14+5) node [anchor=north west][inner sep=0.75pt]  [font=\small]  {$C$};
% Text Node
\draw (335.93+5,136.14+5) node [anchor=north west][inner sep=0.75pt]  [font=\small]  {$B$};

\end{tikzpicture}
    \caption{A simple example of merging fragments, i.e.\ states on disk-like regions with physical boundaries. \textbf{(a)} We consider two purple and one green fragment that overlap as shown, with existence given by assumption shown in Fig.~\ref{fig:fragment_assumption_elementary}.  \textbf{(b)} We draw the green fragment displaced upward for visual clarity.  The union of the two purple fragments is divided into regions $ABC$, and the green fragment is divided into regions $BCD$. \textbf{(c)} We use Lemma~\ref{lemma:merging_lemma} to create a state on $ABCD$. We obtain single disk-like region with physical boundary.}
    \label{fig:merging_two_disks}
\end{figure}
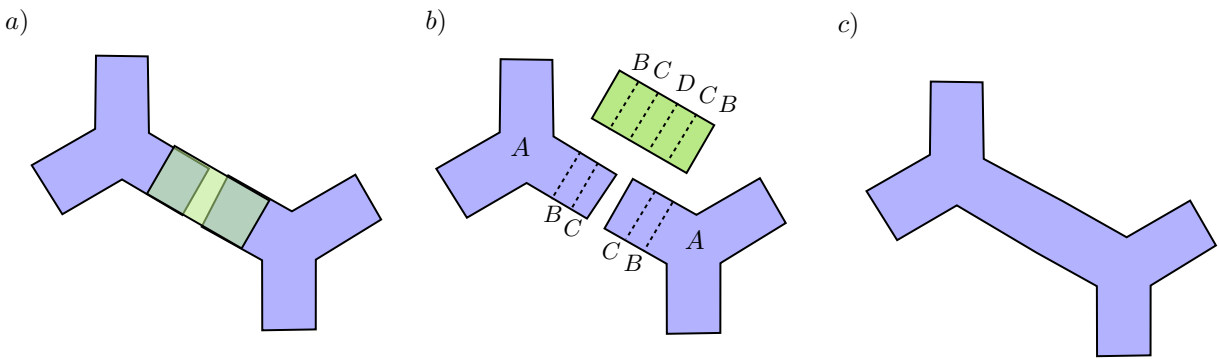

\begin{figure}[htb]
    \centering
     \input{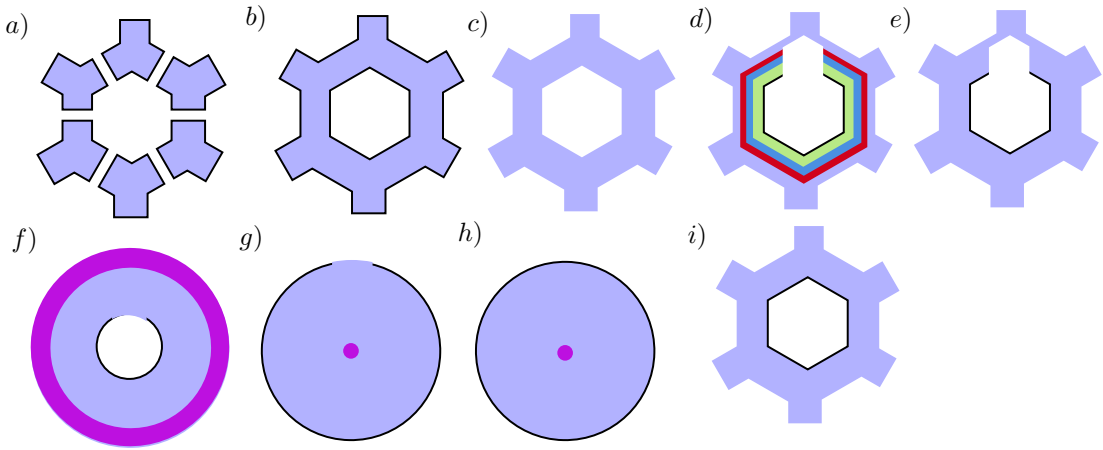}
    \caption{We illustrate the construction of a state on an annulus, which looks like the reference state on the interior, but which has  gapped inner boundary, as in (i).
    \textbf{(a)} We begin with six fragments forming a plaquette as described in Figure \ref{fig:fragment_assumption_elementary}(a).  \textbf{(b)} We merge these into an annulus with gapped boundary. While this state matches the reference state \textit{locally} on the interior, in general it does not match on the full interior. \textbf{(c)} Consider the reference state on the interior of the annulus region. \textbf{(d)} We merge the states from (b) and (c).  In particular, let $A,B,C,D$ denote the regions shown here from outermost to innermost (purple, red, blue, green). Note the states in (b) and (c) match on $BC$. We use the merging lemma with region $ABC$ from state (c) and region $BCD$ from state (b) to obtain a state on $ABCD$.  \textbf{(e)}  We obtain a state with a partial gapped boundary, which matches the reference state on the interior. \textbf{(f)} We re-draw (e) in a topologically equivalent fashion, for simplicity. Let $B$ denote the pink region, and $C$ the complement.  Note $S(B)=S(C)+S(BC)$, so there exists factorization  $B=B_1 \otimes B_2$ with $B_1$ purifying $C$.   \textbf{(g)} We use the inversion trick: we re-draw (f) in an inverted fashion, with the pink region viewed as a single coarse-grained lattice site. If we ignore the (unentangled) factor $B_2$, then we have a state on a disk that satisfies the entanglement bootstrap axioms (in the bulk, and on the partial boundary). \textbf{(h)} We restore the missing piece of boundary with the merging lemma. \textbf{(i)} Returning to the perspective of the drawing (e), we have filled in the missing piece of boundary, and the state matches the reference on the interior, as desired. }
    \label{fig:create-annulus}
\end{figure}

We can merge the elementary fragments together into larger regions, just as we merge fragments in the main text.  Below we explain how assuming only the existnece of these elementary fragments, we can replicate all the constructions used in our Proof of Theorem \ref{thm:mapping-to-sn}.  In particular, we need to construct the regions used for punching holes [Section~\ref{subsec:punching_holes}] and the spanning tree [Section \ref{subsec:choosing-phases}].

First we explain a simple example in Fig.~\ref{fig:merging_two_disks}.  There we merge two vertex regions and one edge region into a single larger disk-like region. This follows from the fact that $I(A:C|B)_{\sigma}=0$ and $I(B:D|C)_{\sigma}=0$, both of which follow from the boundary \textbf{A0} [Figure~\ref{fig:eb_axioms_boundary}]. (Alternatively, one can use the factorization of the extreme points [Section~\ref{subsec:extreme_points}].) This process illustrates the elementary step involved in forming larger regions. The state we obtain this way satisfies both bulk and boundary entanglement bootstrap axioms on every small enough disks, because of the consistency of the merged state with the fragments we started with. This implies that the axioms hold everywhere on the enlarged system, even at a larger scale. Through a sequence of similar merging operations, we can build the spanning tree region of Section~\ref{subsec:choosing-phases}.

Next we illustrate how to build the annulus with gapped boundary used for punching holes in Section~\ref{subsec:punching_holes}.  We cannot simply merge elementary fragments into an annulus, because we need the annulus to match the reference state on the interior.  Instead, we merge a combination of elementary fragments and a portion of the reference state. We begin with the fragments in Figure~\ref{fig:create-annulus}(a), along with additional fragments on edge regions (as in Figure~ \ref{fig:fragment_assumption_elementary}(b)) to obtain a state on an annulus with physical boundary. The remaining process is explained in Figure~\ref{fig:create-annulus}.

For the sake of the argument in Section~\ref{sec:double}, it will also be helpful to define elementary fragments that correspond only to vertex regions.  We can do this by elongating the vertex regions so that they overlap.  Such a choice is shown in Figure~\ref{fig:long-vertex-fragments}.  Using these vertex fragments alone, one can reproduce the edge fragments of Figure~\ref{fig:fragment_assumption_elementary}(b), and then proceed with the development of the remaining constructions required for the proof of Theorem~\ref{thm:mapping-to-sn}.

\begin{figure}[htb]
    \centering      
\begin{tikzpicture}[x=0.75pt,y=0.75pt,yscale=-1,xscale=1]
%uncomment if require: \path (0,426); %set diagram left start at 0, and has height of 426

%Shape: Polygon [id:ds37343890140206315] 
\draw  [fill={rgb, 255:red, 178; green, 178; blue, 255 }  ,fill opacity=0.57 ] (163.17,91.16) -- (173.78,107.25) -- (137.5,131.21) -- (138.18,172.64) -- (117.98,172.98) -- (118.56,130.95) -- (86.83,107.59) -- (97.44,91.85) -- (128.84,114.04) -- cycle ;
%Shape: Polygon [id:ds24363243326355666] 
\draw  [fill={rgb, 255:red, 178; green, 178; blue, 255 }  ,fill opacity=0.57 ] (37.76,120.84) -- (27.19,104.72) -- (63.55,80.87) -- (62.99,39.44) -- (83.19,39.16) -- (82.49,81.19) -- (115.52,104.3) -- (103.49,120.36) -- (72.16,98.06) -- cycle ;
%Shape: Polygon [id:ds5954715858761461] 
\draw  [fill={rgb, 255:red, 178; green, 178; blue, 255 }  ,fill opacity=0.57 ] (203.45,110.92) -- (192.88,94.8) -- (229.24,70.94) -- (228.68,29.52) -- (248.88,29.23) -- (248.18,71.26) -- (281.21,94.37) -- (269.18,110.43) -- (237.85,88.14) -- cycle ;
%Shape: Polygon [id:ds7373771591992406] 
\draw  [fill={rgb, 255:red, 178; green, 178; blue, 255 }  ,fill opacity=0.57 ] (315.92,99.38) -- (326.53,115.47) -- (290.25,139.43) -- (290.93,180.85) -- (270.73,181.19) -- (271.3,139.17) -- (239.58,115.81) -- (250.19,100.06) -- (281.59,122.26) -- cycle ;

% Text Node
\draw (20.33,21) node [anchor=north west][inner sep=0.75pt]    {$a)$};
% Text Node
\draw (188.33,22.52) node [anchor=north west][inner sep=0.75pt]    {$b)$};
\end{tikzpicture}
    \caption{We illustrate a choice of elementary fragments when only vertex fragments are used.  \textbf{(a)} Adjacent vertex fragments overlap along an edge region . \textbf{(b)} Two overlapping fragments are drawn slightly displaced, for visual clarity.}
    \label{fig:long-vertex-fragments}
\end{figure}
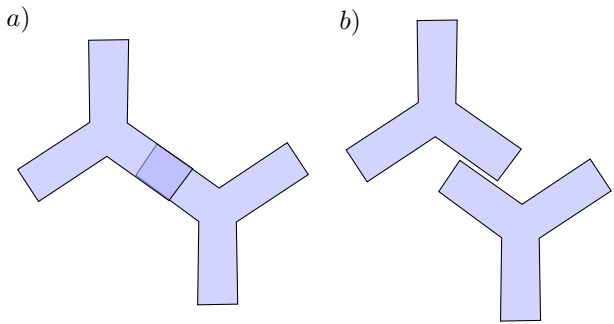

\section{Tensor categories, briefly}
\label{sec:tensor_category}

We make some terse comments on tensor categories.  We refer the reader to \cite{kitaev2006anyons,kitaev2012models} for a comprehensive physical perspective, \cite{delaney2019categorical} for accessible mathematical discussion, and \cite{etingof2017fusion} for a comprehensive textbook. 

First, some comments for the completely uninitiated.  Category theory is a sprawling subject, while its use in condensed matter theory is generally restricted to a few types of tensor  categories.  Moreover, in category theory, one often considers a category of mathematical objects (such as the category of groups, or algebras).  Here, we are more concerned with treating the tensor category as an algebraic object itself -- like a group or an algebra, but specified by different data.

We will only mention unitary fusion categories (UFCs) and unitary modular tensor categories (UMTCs). Note that these may also take other names and acronyms, depending on which adjectives one chooses to emphasize. The UFC characterizes the boundary anyon theory, and the UMTC characterizes the bulk anyon theory. A UMTC is in particular a UFC, but it is additionally equipped with a braiding that satisfies certain properties. We will take the "skeletal" perspective on tensor categories \cite{delaney2019categorical}, in which the category is specified by certain concrete numerical data, up to some gauge equivalence. 

The follow descriptions are \textit{not} comprehensive. A UFC $\mathcal{C}$ is specified by an abstract list of simple objects $\{a\}_{a \in \mathcal{C}}$ (corresponding to anyon types), fusion rules encoded by non-negative integers $N^{ab}_c$, quantum dimensions $d_a \geq 1$, and $F$-symbols $F^{abc; \mu \nu}_{def;\rho \sigma}$ with complex entries.  (In fact the fusion rules and quantum dimensions may be recovered from the $F$-symbols.)  These data are required to satisfy certain properties.  There exists a unit object (representing the trivial anyon, or vacuum sector), which we denote $1 \in \mathcal{C}$.  For every simple object $a$, there is also a dual object $\bar{a}$, possibly with $a = \bar{a}$.  This corresponds to the ``anti-particle,'' and the bar operation is an involution.  The $F$-symbols are required to satisfy the pentagon equation.  They are also required to satisfy a unitarity condition. Two specifications of a UFC by this concrete skeletal data are considered equivalent if they differ by a permutation of the object labels that preserves the above structures, along with a gauge-transformation of the $F$-symbols.

A UMTC is a UFC that additionally has a braiding structure, specified by some $R^{ab; \mu}_{c;\nu}$.  The $R$-symbol must satisfy the hexagon equation relating $F, R$, and must also satisfy a unitarity condition. From $R$ one can define the ``modular data'' denoted $S$ and $T$.  The ``modularity'' of the UMTC requires the $S$ matrix to be unitary. All this data is specified up to gauge-equivalence under a joint gauge transformation of $F$ and $R$.  The modular data $S, T$ is gauge-invariant.  In many cases it uniquely determines the category, but not in all cases, due to the phenomenon of ``modular isotopes'' \cite{mignard2021modular,delaney2021zesting}.

Given a UFC $\mathcal{C}$, one way to produce a UMTC is to take the categorical center, or ``Drinfeld center,'' denoted $Z(\mathcal{C})$. For a boundary anyon theory described by $\mathcal{C}$, the UMTC $Z(\mathcal{C})$ describes the bulk anyon theory. Two UFCs $\mathcal{C},\mathcal{D}$ are  ``Morita-equivalent'' if and only if $Z(\mathcal{C})=Z(\mathcal{D})$. Two distinct UFCs may be Morita-equivalent, corresponding to a bulk anyon theory with distinct gapped boundary conditions that are described by distinct UFCs.

\end{document}